\documentclass[english]{article}
\usepackage[utf8]{inputenc}
\usepackage{geometry}
\usepackage{babel}

\RequirePackage{amsthm,amsmath,amsfonts,amssymb}
\RequirePackage[numbers]{natbib}
\RequirePackage[colorlinks,citecolor=blue,urlcolor=blue]{hyperref}
\RequirePackage{graphicx}

\usepackage{bm, bbm}
\usepackage{amssymb}
\usepackage{tikz}
\usetikzlibrary{positioning}
\usepackage{upgreek}
\usepackage[shortlabels]{enumitem}
\usepackage{cleveref}
\usepackage{xargs}

\theoremstyle{plain}
\newtheorem{theorem}{Theorem}
\newtheorem{prop}{Proposition}
\newtheorem{lemma}{Lemma}
\theoremstyle{remark}

\newtheorem{remark}{Remark}
\newtheorem{example}{Example}
\newtheorem{assumption}{\textbf{A}\hspace{-3pt}}
\crefformat{assumption}{{\textbf{A}}#2#1#3}
\DeclareMathOperator{\prox}{prox}
\DeclareMathOperator{\sgn}{sgn}
\DeclareMathOperator{\var}{Var}
\newcommand{\real}{\ensuremath{\mathbb{R}}}
\newcommand{\pr}{\ensuremath{\mathbb{P}}}
\newcommand{\norm}[1]{\ensuremath{\Vert #1 \Vert}}
\newcommand{\expe}[1]{\ensuremath{\mathbb{E}\left[ #1 \right]}}
\newcommand{\expeLine}[1]{\ensuremath{\mathbb{E}[ #1 ]}}
\newcommand{\indicator}[2]{\ensuremath{\mathbbm{1}_{#1}( #2 )}}
\newcommand{\indicatorD}[2]{\ensuremath{\mathbbm{1}\{#1\}}}
\newcommand{\indicatorDD}[1]{\ensuremath{\mathbbm{1}\{#1\}}}
\newcommand{\floor}[1]{\ensuremath{\lfloor #1\rfloor}}
\newcommand{\ceil}[1]{\ensuremath{\lceil #1\rceil}}
\def\rmd{\mathrm{d}}
\def\dprime{\prime\prime}

\newcommand{\1}{\mathbbm{1}}
\newcommand{\rmL}{\mathrm{L}}
\newcommand{\bmW}{\bm{\mathrm{W}}}

\def\eqsp{\;}

\newcommand{\ccint}[1]{\left[#1\right]}

\def\rset{\mathbb{R}}
\def\msd{\mathsf{D}}
\def\mrl{\mathrm{L}}
\def\mrL{\mathrm{L}}
\def\bigO{\mathcal{O}}

\def\nset{\mathbb{N}}
\def\nsets{\mathbb{N}^*}

\def\rme{\mathrm{e}}
\def\Sigmabf{\boldsymbol{\Sigma}}

\def\varphibf{\boldsymbol{\varphi}}

\def\bmy{\bm{y}}
\def\bmu{\bm{u}}
\def\bmx{\bm{x}}

\def\plusinfty{+\infty}

\newcommand{\abs}[1]{\left\vert #1 \right\vert}

\newcommandx{\Vnorm}[2][1=V]{\| #2 \|_{#1}}
\newcommandx{\VnormEq}[2][1=V]{\left\| #2 \right\|_{#1}}
\newcommandx{\normLigne}[2][1=]{\ifthenelse{\equal{#1}{}}{\Vert #2 \Vert}{\Vert #2\Vert^{#1}}}

\newcommand{\parenthese}[1]{\left(#1 \right)}
\newcommand{\parentheseLigne}[1]{(#1 )}

\newcommand{\defEns}[1]{\left\lbrace #1 \right\rbrace }

\def\rmd{\mathrm{d}}

\def\rmC{\mathrm{C}}
\def\rmc{\mathrm{c}}

\def\piLaplace{\pi^{\mrL}}

\def\MALA{\textrm{MALA}}
\def\PMALA{\textrm{P-MALA}}
\renewcommand*{\iint}[2]{\lbrace #1 , \dots , #2\rbrace }
\def\acc{\mathsf{A}}
\def\tacc{\widetilde{\mathsf{A}}}

\usepackage{autonum}
\usepackage{authblk}

\usepackage[textwidth=1.8cm, textsize=scriptsize]{todonotes}

\graphicspath{{Images/}}

\title{Optimal Scaling Results for Moreau-Yosida Metropolis-adjusted Langevin Algorithms}
\author[1]{Francesca R. Crucinio\thanks{Corresponding author: francesca\_romana.crucinio@kcl.ac.uk}}
\author[2]{Alain Durmus}
\author[3]{Pablo Jim\'{e}nez}
\author[4]{Gareth O. Roberts}
\affil[1]{King's College London}
\affil[2]{Centre de Math\'{e}matiques Appliqu\'{e}es, Ecole
  Polytechnique, France, Institut Polytechnique de Paris}
\affil[3]{Sorbonne Université and Université Paris Cité, CNRS, Laboratoire de Probabilités, Statistique et Modélisation, F-75005 Paris, France}
\affil[4]{Department of Statistics, University of Warwick}
\date{ }
\begin{document}
\maketitle
\begin{abstract}
We consider a recently proposed class of MCMC methods which uses proximity maps instead of gradients to build proposal mechanisms which can be employed for both differentiable and non-differentiable targets.	
These methods have been shown to be stable for a wide class of targets, making them a valuable alternative to Metropolis-adjusted Langevin algorithms (MALA); and have found wide application in imaging contexts.
The wider stability properties are obtained by building the Moreau-Yosida envelope for the target of interest, which depends on a parameter $\lambda$.
In this work, we investigate the optimal scaling problem for this class of algorithms, which encompasses MALA, and provide practical guidelines for the implementation of these methods.
\end{abstract}

\tableofcontents

\section{Introduction}
Gradient-based Markov chain Monte Carlo (MCMC) methods have proved to be very successful at sampling from high-dimensional target distributions
\cite{brooks2011handbook}.
The key to their success is that in many cases their mixing time appears to 
scale better in dimension than
competitor algorithms which do not use gradient information (see for example \cite{roberts1998optimal}),
while their implementation has similar computational cost. Indeed, 
gradients of target densities can often be computed with computational complexity (in dimension $d$) which scales no worse than evaluation of the target density itself.

Gradient-based MCMC methods are mainly motivated from stochastic processes constructed to have the target density as limiting distribution \cite{ma2015complete,bouchard2018bouncy,bierkens2019zig,goldman2021gradient}. Our analysis will concentrate on the Metropolis Adjusted Langevin Algorithm (MALA) and its proximal variants which are based on the Langevin diffusion
\begin{equation}
\label{eq:LangDiff}
\rmd{\bf L}_t =  \rmd{\bf B}_t + {
\frac{
\nabla \log \pi ({\bf L}_t)
}{2}} 
\rmd t,
\end{equation}
where $\pi $ denotes the target density with respect to the Lebesgue measure and $({\bf B}_t)_{t\geq 0
}$ a standard Brownian motion. It is well-known that under appropriate conditions,~\eqref{eq:LangDiff} defines a continuous-time Markov process associated with a Markov semigroup which is reversible with respect to $\pi$. From this observation, it has been suggested to use a Euler-Maruyama (EM) approximation of~\eqref{eq:LangDiff}. This scheme has been popularized in statistics by \cite{grenander1994representations} and referred to as the {\em Unadjusted Langevin Algorithm (ULA)} in \cite{roberts1996exponential}.
Due to time-discretization, ULA typically does not have $\pi $ as stationary distribution. To address this problem,  \cite{rossky:doll:friedman:1978} and independently Besag in his contribution to \cite{grenander1994representations} proposed to add a Metropolis acceptance step at each iteration of the EM scheme, leading to the Metropolis Adjusted Langevin Algorithm (MALA) following \cite{roberts1996exponential} who also derive basic stability analysis. The accept/reject step in this algorithm confers two significant advantages: it ensures that the resulting algorithm has exactly the correct invariant distribution, while step sizes can be chosen larger than in the unadjusted case as there is not need to make the step size small to reduce discretization error. On the other hand, MALA algorithms are typically hard to analyze theoretically
(see e.g. \cite{bou2013nonasymptotic,durmus:moulines:2021:mala,malafast}). However, \cite{roberts1998optimal} (see also \cite{stuart:2009,stuart:2012}) have established that MALA has better convergence properties than the Random Walk Metropolis (RWM) algorithm with respect to the dimension $d$ from an optimal scaling perspective (see also \cite{gelman1997weak}).

Whereas gradient-based methods have been successively applied and offer interesting features, they are typically less robust than their vanilla alternatives (for example see \cite{roberts1996exponential}); while intuition suggests, and existing underpinning theory requires, that target densities need to be sufficiently smooth for the gradients to be aiding Markov chain convergence. 
Moreover, while gradient-based MCMC have been successful for smooth densities, there is no reason to believe that they should be effective for densities which are not differentiable at a subset $\msd\subseteq \rset^d$.
For non-smooth densities, \cite{pereyra2016proximal} proposes modified gradient-based algorithms. Their proposed P-MALA algorithm is inspired by the proximal algorithms popular in the optimization literature (e.g. \cite{parikh2014proximal}). 
The main idea is to approximate the (possibly non differentiable but) log-concave target density $\pi\propto \exp(-G)$ by substituting the potential $G$ with its Moreau-Yosida envelope $G^\lambda$ (see~\eqref{eq:glambda} below for its definition), to obtain a distribution $\pi^\lambda$ whose level of smoothness is controlled by the proximal parameter $\lambda>0$, so that $G^{0}=G$. 
Given this smooth approximation to $\pi$ one can then build proposals based on time discretizations of the Langevin diffusion targeting $\pi^\lambda$ \cite{pereyra2016proximal,durmus2018efficient}:
\begin{equation}
  \label{eq:langevin_intro}
\xi_{k+1}=\xi_k-\frac{\sigma^2}{2} \nabla G^\lambda(\xi_k)+\sigma Z_{k+1} ,
\end{equation}
where $\sigma^2 >0$ is a fixed stepsize and $(Z_k)_{k\in\nsets}$ is a sequence of i.i.d. zero-mean Gaussian random variables with identity covariance matrix. While \cite{pereyra2016proximal} mostly considers the case $\lambda = \sigma^2/2$, our aims in this paper are broadly to provide theoretical underpinning for the slightly larger family of {\em Moreau-Yosida Regularised MALA} (MY-MALA) algorithms obtained when $\lambda\neq \sigma^2/2$, analyze how these methods scale with dimension, and to give insights and practical guidance into how they should be implemented supported by the theory we establish. 

Proximal optimization and MCMC methods proved to be particularly well-suited for image estimation, where penalties involving the sparsity inducing norms are common \cite{pereyra2016proximal, durmus2018efficient, vono2019bayesian}. 
Similar targets are also common in sparse regression contexts \cite{atchade2015moreau, goldman2021gradient, zhou2022proximal}. In these situations, the set of non-differentiability points for the target density $\pi$ is a null set under Lebesgue measure, and, following \cite{durmus2017optimal}, we shall focus on this case. However, in contrast to the conclusions of \cite{durmus2017optimal} for RWM, we shall demonstrate that optimal scaling of MY-MALA is significantly affected by non-smoothness.

In this work, we first extend the results of \cite{pillai2022optimal}, considering a wider range of MY-MALA algorithms, as well as a more general class of finite dimensional target distributions.
We begin by comparing MALA and its proximal cousin in cases where MALA is well-defined, i.e., where target densities are sufficiently differentiable.
In some cases the proximal operator for a given distribution $\pi$ is less expensive to compute than $\nabla \log \pi$ \cite{parikh2014proximal, combettes2011proximal, pereyra2016proximal}, so we anticipate that MY-MALA with an appropriately tuned $\lambda$ might provide a computationally more efficient alternative to MALA, whilst retaining similar scaling properties. In our study, we let both the steps size $\sigma^2$ and the regularization parameter $\lambda$ depend on the dimension $d$ of the target and find that the scaling properties of MY-MALA depend on the relative speed at which $\lambda$ and $\sigma$ converge to $0$ as $d\to \infty$.
When $\lambda$ goes to $0$ at least as fast as $\sigma^2$, we find that the scaling properties of MY-MALA are equivalent to those of MALA (i.e., $\sigma^2$ should decay as $d^{-1/3}$; see Theorem~\ref{theo:differentiable}--\ref{case:b}, Theorem~\ref{theo:differentiable}--\ref{case:c}); when $\lambda$ converges to $0$ more slowly than $\sigma^2$, MY-MALA is less efficient than MALA with $\sigma^2$ decaying as $d^{-1/2}$ (Theorem~\ref{theo:differentiable}--\ref{case:a}).

We then turn to the optimal scaling of MY-MALA applied to the Laplace distribution $\pi(x)\propto \rme^{-\vert x\vert}$. We focus on this particular non-smooth target since it is the most widely used in applications of MY-MALA, including image deconvolution \cite{ pereyra2016proximal, durmus2018efficient, vono2019bayesian}, LASSO, and sparse regression \cite{atchade2015moreau, goldman2021gradient, zhou2022proximal}. We establish that non-differentiability of the target even at one point leads to a different optimal scaling than MALA. In particular, the step size has to scale as  $d^{-2/3}$ and not as $d^{-1/3}$ (Theorem~\ref{theo:acceptance_laplace}).
We thus uncover a new optimal scaling scenario for Metropolis MCMC algorithms which lies in between those of RWM and  MALA. 

The proof of the result for the differentiable case extends that of \cite{roberts1998optimal} for MALA, while 
the structure of the proof for the Laplace target is similar to that of \cite{durmus2017optimal} and constitutes the main element of novelty in this paper. As a special case of the result for the Laplace distribution, we also obtain the optimal scaling for MALA on Laplace targets. 
We point out that the strategy adopted in the proof of this result is not unique to the Laplace distribution, and could be applied to other distributions provided that the required integrals can be obtained.

To sum up, our main contributions are:
\begin{enumerate}[wide, labelwidth=!, labelindent=0pt, label = \arabic*), noitemsep,topsep=0pt,parsep=0pt,partopsep=0pt]
    \item We extend the result of \cite{pillai2022optimal} beyond the Gaussian case, covering all finite dimensional (sufficiently) differentiable targets, and show that, in some cases, MY-MALA affords the same scaling properties of MALA if the proximal parameter $\lambda$ is chosen appropriately.
    \item Motivated by applications in imaging and sparse regression applications, we study the scaling of MY-MALA methods for the Laplace target, and show that for values of $\lambda$ decaying sufficiently fast, the optimal scaling of MY-MALA, i.e., the choice for $\sigma^2$, is different from the one for MALA on differentiable targets and is of order $d^{-2/3}$.
    \item We use the insights obtained with the aforementioned results to provide practical guidelines for the selection of the proximal parameter $\lambda$.
\end{enumerate}

The paper is structured as follows.
In \Cref{sec:proximal_mala}, we rigorously introduce the class of MY-MALA algorithms that are studied and discuss related works on optimal scaling for MCMC algorithms. 
In Section~\ref{sec:differentiable} we state the main result for regular targets, showing that the scaling properties of MY-MALA depend on the relative speed at which $\lambda$ goes to 0 with respect to $\sigma$. In Section~\ref{sec:laplace} we obtain a scaling limit for MY-MALA when $\pi$ is a Laplace distribution, as a special case of our result we also obtain the scaling properties of a sub-gradient version of MALA for this target. We collect in Section~\ref{sec:implications} the main practical takeaways from these results and discuss possible extensions in Section~\ref{sec:discussion}. All proofs are available in the appendix.

\section{Background}
\label{sec:proximal_mala}
\subsection{MY-MALA algorithms}

We now introduce the general class of MY-MALA algorithms, first studied in \cite{pereyra2016proximal}. This class of algorithms aims at sampling from a probability density with respect to the Lebesgue measure on $\mathbb{R}^d$ of the form $\pi(\bm{x})= \exp ( -G(\bm{x}))/\int_{\rset^d} \exp ( -G(\tilde{\bm{x}})) \rmd \tilde{\bmx}$, with $G$ satisfying the following assumption
\begin{assumption}
\label{ass:a0}
The function $G:\real^d\to 
\mathbb{R}$ is convex, proper and lower semi-continuous. 
\end{assumption}

The main idea behind MY-MALA is to approximate the (possibly non differentiable) target density $\pi$ by approximating the potential $G$ with its Moreau-Yosida envelope $G^{\lambda} : \mathbb{R}^d \to \mathbb{R}$ defined for $\lambda >0$ by
\begin{equation}
\label{eq:glambda}
    G^\lambda(\bm{x}) =\min_{\bm{u}\in\real^d} [G(\bm{u})+\norm{\bm{u}-\bm{x}}^2/(2\lambda)] .
\end{equation}
Since $G$ is supposed to be convex, by \cite[Theorem 2.26]{rockafellar2009variational}, the Moreau-Yosida envelope is well-defined, convex and continuously differentiable with 
\begin{equation}
    \label{eq:proximity}
  \nabla G^{\lambda}(\bm{x}) = \lambda^{-1}(\bm{x}-\prox_G^\lambda(\bm{x})) , \quad
    \prox_G^\lambda(\bm{x}) = \arg\min_{\bm{u}\in\real^d} [G(\bm{u})+\norm{\bm{u}-\bm{x}}^2/(2\lambda)] .
\end{equation}
The proximity operator $ \bmx \mapsto \prox_G^\lambda(\bmx)$ behaves similarly to a gradient mapping and moves points in the direction of the minimizers of $G$. In the limit $\lambda\to 0$ the quadratic penalty dominates~\eqref{eq:proximity} and the proximity operator coincides with the identity operator, 
i.e., $\prox_G^\lambda(\bm{x})=\bm{x}$; in the limit $\lambda\to \infty$,
the quadratic penalty term vanishes and~\eqref{eq:proximity} maps all points to the set of minimizers of $G$.


It was shown in \cite[Proposition 1]{durmus2018efficient} that, under \Cref{ass:a0}, $\int_{\real^d} \exp(-G^{\lambda}(\bm{x})) \rmd \bm{x} < \infty$, and therefore the probability density $\pi^\lambda \propto \exp(-G^{\lambda})$ is well-defined. In addition, it has been shown that $\Vert \pi - \pi^\lambda \Vert_{\mathrm{TV}} \to 0$ as $\lambda \to 0$.
Based on this observation and since as we have emphasized $\pi^\lambda$ is now continuously differentiable, it has been suggested in \cite{pereyra2016proximal, durmus2018efficient} to use the discretization of the Langevin diffusion associated with $\pi^\lambda$ given by ~\eqref{eq:langevin_intro}, which can be rewritten using~\eqref{eq:proximity} as
\begin{align}
\label{eq:lambda_choice}
\xi_{k+1}&=\left(1-\frac{\sigma^2}{2\lambda}\right)\xi_k+\frac{\sigma^2}{2\lambda} \prox_{G}^\lambda(\xi_k)+\sigma Z_{k+1} . 
\end{align}

Similarly to other MCMC methods based on discretizations of the Langevin diffusion (e.g. \cite{roberts1996exponential}), one can build unadjusted schemes which target $\pi^\lambda$, expecting draws from these schemes to be close to draws from $\pi$ for small enough $\lambda$, or add a Metropolis-Hastings step to ensure that the resulting algorithm targets $\pi$.
Unadjusted Moreau-Yosida MCMC methods have been analyzed in \cite{durmus2018efficient}; in this paper we focus on Metropolis adjusted Moreau-Yosida MCMC methods and study their scaling properties.
More precisely, at each step $k$ and given the current state of the Markov chain $X_k$, a candidate $Y_{k+1}$ is generated from the transition density associated to~\eqref{eq:lambda_choice}, $(\bmx,\bmy) \mapsto q(\bmx,\bmy)  = \varphibf(\bmy ;[1-\sigma^2/(2\lambda)] \bmx+ \sigma^2 \prox_{G}^\lambda(\bmx)/2\lambda, \sigma^2 \operatorname{I}_{d})$, 
where
$\varphibf(\cdot\ ; \bmu, \Sigmabf)$ stands for the $d$-dimension Gaussian
density with mean $\bmu$ and covariance matrix $\Sigmabf$.
Given $X_k$ and $Y_{k+1}$,
Then, the next state is set as:
\begin{equation}
  \label{eq:p_mala_def}
  X_{k+1} = Y_{k+1} \mathrm{b}_{k+1} + X_k(1-\mathrm{b}_{k+1}) , \mathrm{b}_{k+1} = \1_{\rset_+}\left(\parenthese{\frac{\pi(Y_{k+1})q(Y_{k+1}, X_k)}{\pi(X_k)q(X_k, Y_{k+1})}\wedge 1\right)  - U_{k+1}} ,
\end{equation}
where $(U_{i})_{i\in\nsets}$ is a sequence of i.i.d. uniform random variables on $\ccint{0,1}$.


The value of $\lambda$ characterizes how close the distribution $\pi^\lambda$ is to the original target $\pi$ and therefore how good the proposal is.
Small values of $\lambda$ provide better approximations to $\pi$ and therefore better proposals (see \cite[Proposition 1]{durmus2018efficient}), while larger values of $\lambda$ provide higher levels of smoothing for non-differentiable distributions (see \cite[Figure 1]{pereyra2016proximal}).
In the case $\lambda = \sigma^2/2$ we obtain the special case of MY-MALA referred to as  P-MALA in \cite{pereyra2016proximal}.

The main contribution of this paper is to analyze the optimal scaling for MY-MALA defined by~\eqref{eq:p_mala_def}.


\subsection{Optimal scaling and related works}
\label{sec:related_work}
We briefly summarize here some examples of MCMC algorithms and their optimal scaling results; a full review is out of the scope of this paper and we only mention algorithms to which we will compare MY-MALA in the development of this work.

Popular examples of Metropolis MCMC are RWM and MALA. RWM uses as a
proposal the transition density
$(\bm{x},\bm{y}) \mapsto \varphibf(\bm{y}\ ; \bm{x}, \sigma^2 \operatorname{I}_{d})$, where $\sigma^2 >0$. The
MALA scheme uses as proposal
$(\bm{x},\bm{y}) \mapsto \varphibf(\bm{y}\ ; \bm{x} + (\sigma^2/2) \nabla \log \pi(\bm{x}), \sigma^2
\operatorname{I}_{d})$.
As we will show in Section~\ref{sec:differentiable}, MY-MALA can be considered as an extension of MALA.

A natural question to address when implementing Metropolis adjusted algorithms is how to set the parameter $\sigma^2$ (variance parameter for RWM, step size parameter for MALA) to maximize the efficiency of the algorithm.
Small values of $\sigma^2$ result in higher acceptance probability 
 and cause the chain to move slowly, while large values of $\sigma^2$ result in a high number of rejections with the chain $(X_k)_{k\geq0}$ moving slowly \cite{roberts2001optimal}.
Optimal scaling studies aim to address this question  by investigating how $\sigma^2$ should behave with respect to the dimension $d$ of the support of $\pi$ in the high dimensional setting $d\to\infty$, to obtain the best compromise.

The standard optimal scaling set-up considers the case of $d$-dimensional targets $\pi_d$ which are product form, i.e.,
\begin{align}
  \label{eq:pi_d_iid}
\pi_d(\bm{x}^d)=\prod_{i=1}^d \pi(x_i^d) ,
\end{align}
where $x_i^d$ stands for the $i$-th component of $\bmx^d$ and $\pi$ is a one-dimensional probability density with respect to the Lebesgue measure. Under appropriate assumptions on the regularity of $\pi$, and assuming that the MCMC algorithm is initialized at stationarity, the optimal value of $\sigma^2$ scales as $\ell^2/d^{2\alpha}$ with $\ell >0$, $2\alpha =1$ for RWM \cite{gelman1997weak} and $2\alpha = 1/3$ for MALA \cite{roberts1998optimal}.

By setting $\alpha$ to these values, it is then possible to show that as $d\to\infty$ each 1-dimensional component of the Markov chain defined by RWM and MALA, appropriately rescaled in time, converges to the Langevin diffusion
\begin{align}
\textrm{d}L_t = h(\ell)^{1/2}\textrm{d}B_t-\frac{h(\ell)}{2} [\log\pi]{^\prime}(x)\textrm{d}t ,
\end{align}
where $(B_t)_{t\geq 0}$ is a standard Brownian motion and  $h(\ell)$, referred to as speed function of the diffusion, is a function of the parameter $\ell>0$ that we may tune. Indeed, it is well-known that $(L_{h(\ell)t})_{t \geq 0}$ is a solution of the Langevin diffusion~\eqref{eq:LangDiff}. As a result, we may identify the values of $\ell$ maximizing $h(\ell)$ for the algorithms at hand to approximate the fastest version of the Langevin diffusion. The optimal values for $\ell$ results in an optimal average acceptance probability of $0.234$ for RWM and $0.574$ for MALA.

The scaling properties allow to get an intuition of the efficiency of the corresponding algorithms: RWM requires $\bigO(d)$ steps to achieve convergence on a $d$-dimensional target, i.e., its efficiency is $\bigO(d^{-1})$, while MALA has efficiency $\bigO(d^{-1/3})$.
While these results are asymptotic in $d$, the insights obtained by considering the limit case $d\to \infty$ prove to be useful in practice \cite{roberts2001optimal}.

In the context of non-smooth and even discontinuous target
distributions, studying the simpler RWM algorithm applied to a class
of distributions on compact intervals,
\citep{neal2008optimal,neal2012optimal} show that the lack of
smoothness affects the optimal scaling of RWM with respect to
dimension $d$. More precisely, they show that for a class of
discontinuous densities which includes the uniform distribution on
$\ccint{0,1}$, the optimal scaling of RWM is of order
$\bigO(d^{-2})$. On the other hand, in the case where the set of non-differentiability $\msd$ of $\pi$ is a null
set with respect to the Lebesgue measure, \cite{durmus2017optimal}
shows that under appropriate conditions, including $\mrl^p$
differentiability, the optimal scaling of RWM is of order
$\bigO(d^{-1})$ still.

The scaling properties of MY-MALA have been partially investigated in \cite{pillai2022optimal}, which shows that  P-MALA, obtained when $\lambda=\sigma^2/2$, has the same scaling properties of MALA for the finite dimensional Gaussian density and for a class of infinite dimensional target measures (Theorem 2.1 and Theorem 5.1 therein, respectively).

\section{Optimal scaling of MY-MALA}

We consider the same set up as \cite{roberts1998optimal} and briefly recalled above. Given a real-valued function $g:\real\to \real$ satisfying \Cref{ass:a0} we consider the i.i.d. $d$-dimensional target specified by~\eqref{eq:pi_d_iid} with 
\begin{align}
\label{eq:target_differentiable}
 \pi(x)  \propto \exp(- g(x)) .
\end{align}
Since for any $\bm{x}^d$, $G(\bm{x}^d)=\sum_{i=1}^d g(x_i^d)$, we have by  \cite[Section 2.1]{parikh2014proximal}
\begin{align}
  \label{eq:proximal_laplace}
\prox_{G}^\lambda (\bm{x}^d) = (\prox_{g}^\lambda(x_1^d),\dots, \prox_{g}^{\lambda}(x_d^d))^\top .
\end{align}
It follows that the distribution of the proposal with target $\pi_d$ in~\eqref{eq:pi_d_iid}-\eqref{eq:target_differentiable} is also product form $q_d(\bm{x}^d, \bm{y}^d)= \prod_{i=1}^d q(x_i^d, y_i^d)$ with
\begin{align}
    \nonumber
\textstyle  
 q(x_i^d, y_i^d) = \frac{1}{(2\uppi \sigma^2)^{1/2}}\exp\left( -\frac{\left(y_i^d-(1-\sigma^2/(2\lambda))x_i^d - \sigma^2 \prox_{g}^{\lambda}(x_i^d)/(2\lambda)\right)^2}{2\sigma^2}\right) ,
\end{align}
and $\lambda>0$. For any dimension $d \in\nsets$, we denote by $(X_k^d)_{k\in\nset}$ the Markov chain defined by the Metropolis recursion~\eqref{eq:p_mala_def} with target distribution $\pi_d$ and proposal density $q_d$ and associated to the sequence of candidate moves
\begin{align}
  \label{eq:proposal_pmala}
  Y_{k+1}^d=\left(1-\frac{\sigma^2}{2\lambda}\right)X_k^d+\frac{\sigma^2}{2\lambda} \prox_{G}^\lambda(X_k^d)+\sigma Z^d_{k+1} . 
\end{align}

As mentioned in the introduction, the focus of this work is on investigating the optimal dependence of the proposal variance $\sigma^2$ on the dimension $d$ of the target $\pi$. In this section, we make the dependence of the proposal variance on the dimension explicit and let $\sigma_d^2 = \ell^2 /d^{2\alpha}$ and $\lambda_d = c^2/2d^{2\beta}$ for some $\alpha, \beta>0$ and some constants $c, \ell$ independent on $d$.
Thus, we can write $\lambda_d$ as a function of $\sigma_d$, $\lambda_d =\sigma_d^{2r}m/2
$,
where we defined the \emph{relative velocity} at which $\sigma_d^2$ and $\lambda_d$ converge to 0 as $d\to \infty$ by $v:= \beta/\alpha$ and the \emph{ratio of constants} $r := c^2/\ell^{2v}>0$.
When $v=1$, $\sigma_d^2$ and $\lambda_d$ decay to 0 at the same rate, for $v>1$ the decay of $\lambda_d$ is faster than that of $\sigma_d^2$ and for $v<1$ the decay of $\lambda_d$ is slower than that of $\sigma_d^2$. The parameter $r$ allows to refine the comparison between $\sigma_d^2$ and $\lambda_d$ as $\beta=\alpha$.

By writing $\lambda_d$ as a function of $\sigma_d$ we can decouple the effect of the constants $c, \ell$ from that of the dependence on $d$ (i.e., $\alpha, \beta$). 
In the case $v=1, r=1$ we get the P-MALA algorithm studied in \cite{pereyra2016proximal, pillai2022optimal}, while for all other values of $r, v$ we have a family of proposals whose behaviour depends on $r$ and $v$.



\subsection{Regular targets}
\label{sec:differentiable}

We start with the case where $\pi$ is continuously differentiable. 
Since MALA can be applied to this class of targets, the results obtained in this section allow direct comparison of MY-MALA algorithms with MALA and thus between gradient-based algorithms (MALA) and algorithms that use proximal operator-based approximations of the gradient (MY-MALA).
If $G = -\log \pi$ is continuously differentiable, using \cite[Corollary 17.6]{bauschke2011convex},  $\prox_{G}^\lambda(\bm{x}) = -\lambda \nabla G(\prox_{G}^\lambda(\bm{x}))+\bm{x}$, and~\eqref{eq:lambda_choice} reduces to
\begin{align}
\label{eq:differentiable_ula}
\xi_{k+1}=\xi_k-\frac{\sigma^2}{2} \nabla G(\prox_{G}^\lambda(\xi_k))+\sigma Z_{k+1} .
\end{align}
Hence, the value of $\lambda$ controls how close to $\xi_k$ is the point at which the gradient is evaluated. For $\lambda\to0$, the MY-MALA proposal becomes arbitrarily close to that of MALA, while, as $\lambda$ increases~\eqref{eq:differentiable_ula} moves away from MALA.

Our main result, Theorem~\ref{theo:differentiable} below, shows that the relative speed of decay (i.e., $v$) influences the optimal scaling of the resulting MY-MALA algorithm, while the constant $r$ influences the speed function  of the limiting diffusion.


We make the following assumptions on the regularity of $g$.
\begin{assumption}
\label{ass:g_smooth}
 $g$ is a $\rmC^8$-function whose derivatives are bounded by some polynomial: there exists $k_0 \in\nset$ such that
 \begin{align}
   \sup_{x \in\rset} \max_{i\in\{0,\ldots,8\}} [g^{(i)}(x)/(1+\vert x \vert ^{k_0})] < \infty .
\end{align}
\end{assumption}
Note that under \Cref{ass:a0} and \Cref{ass:g_smooth} and under the assumption that $\pi$ in \eqref{eq:target_differentiable} is a probability density, \cite[Lemma A.1]{durmus2018efficient}  implies that  $\int_\real x^k\exp(-g(x))\rmd x<\infty$ for any $k\in\mathbb{N}$. We also assume that the sequence of MY-MALA algorithms is initialized at stationarity.

\begin{assumption}
\label{ass:stationarity}
For any   $d \in \nsets$, $X_0^d$ has distribution $\pi_d$.
\end{assumption}

The assumptions above closely resemble those of \cite{roberts1998optimal} used to obtain the optimal scaling results for MALA.
In particular, \Cref{ass:g_smooth} ensures that we can approximate the log-acceptance ratio in~\eqref{eq:p_mala_def} with a Taylor expansion, while~\Cref{ass:stationarity} avoids technical complications due to the transient phase of the algorithm. We discuss how the latter assumption could be relaxed in Section~\ref{sec:discussion}.

For technical reasons, and to allow direct comparisons with the results established in \cite{roberts1998optimal} for MALA, we will also consider the following regularity assumption
\begin{assumption}
\label{ass:g_lipschitz}
The function $g^\prime$ is Lipschitz continuous.
\end{assumption}


We denote by $L_t^d$ the linear interpolation of the first component of the discrete time Markov chain $(X_k^d)_{k\geq0}$ obtained with the generic MY-MALA algorithm described above
\begin{align}
\label{eq:Yt_differentiable}
L_t^d &= (\ceil{d^{2\alpha}t}-d^{2\alpha}t)X^d_{\floor{d^{2\alpha}t},1} + (d^{2\alpha}t- \floor{d^{2\alpha}t})X^d_{\ceil{d^{2\alpha}t},1} , 
\end{align}
where $\floor{\cdot}$ and $\ceil{\cdot}$ denote the lower and upper integer part functions, respectively, and denote by $X_{k,1}^d$ the first component of $X_{k}^d$.
The following result shows that in the limit $d\to\infty$ the properties of MY-MALA depend on the relative speed at which $\sigma_d^2 = \ell^2 /d^{2\alpha}$ and $\lambda_d = c^2/2d^{2\beta}$  converge to $0$. Recall that we set $r=c^2/\ell^{2v}>0$ and  under \Cref{ass:stationarity}, consider for any $d \in\nsets$,
\begin{equation}
\label{eq:ar_finite}
  a_d(\ell,r) = \expe{\frac{\pi_d(Y_1^d)q_d(Y_1^d, X_0^d)}{\pi_d(X_0^d)q_d(X_0^d, Y_1^d)}\wedge 1} .
\end{equation}
\begin{theorem}
\label{theo:differentiable}
Assume \Cref{ass:a0}, \Cref{ass:g_smooth} and~\Cref{ass:stationarity}. For any $d \in\nsets$, let $\sigma_d^2 = \ell^2 /d^{2\alpha}$ and $\lambda_d = c^2/2d^{2\beta}$ with $\alpha,\beta >0$. 
Then, the following statements hold.
\begin{enumerate}[label=(\alph*)]
\item \label{case:a} If $\alpha= 1/4$, $\beta = 1/8$ and $r>0$, we have
$\lim_{d\to \plusinfty}    a_d(\ell, r) = 2\Phi\left(-\ell^2K_1(r)/2\right)$,
where $\Phi$ is the distribution function of a standard normal and
\begin{align}
    K_1^2(r) = \frac{r^2}{4}\mathbb{E}\left[\left\lbrace g^{\dprime}(X_{0,1}^d)g^{\prime}(X_{0,1}^d)\right\rbrace^2\right].
\end{align}
\end{enumerate}
If in addition, \Cref{ass:g_lipschitz} holds.
\begin{enumerate}[label=(\alph*),resume]
\item \label{case:b}  If $\alpha= 1/6$, $\beta = 1/6$ and $r> 0$, we have
$\lim_{d\to \plusinfty}    a_d(\ell, r) = 2\Phi\left(-\ell^3K_2(r)/2\right)$, where $\Phi$ is the distribution function
of a standard normal and
\begin{align}
K_2^2(r) &= 
\left(\frac{r}{8}+\frac{r^2}{4}\right)\mathbb{E}\left[ \{g^{\dprime}(X_{0,1}^d)g^{\prime}(X_{0,1}^d)\}^2\right]+\left(\frac{1}{16}+\frac{r}{8}\right)\mathbb{E}\left[g^{\dprime}(X^d_{0,1})^3\right]\\
&+\frac{5}{48}\mathbb{E}\left[g^{\dprime\prime}(X^d_{0,1})^2\right].
\end{align}
\item \label{case:c} If $\alpha= 1/6$, $\beta > 1/6$ and $r> 0$,
we have
$\lim_{d\to \plusinfty}    a_d(\ell, r) = 2\Phi\left(-\ell^3K_2(0)/2\right)$, 
where $\Phi$ is the distribution function of a standard normal.
\end{enumerate}
In addition, in all these cases,  as $d\to \infty$ the process  $(L_t^d)_{t\geq 0}$ converges weakly to the  Langevin diffusion
\begin{align}
\label{eq:langevin}
\rmd L_t = h(\ell, r)^{1/2}\rmd B_t-\frac{h(\ell, r)}{2}g^\prime(x)\rmd t ,
\end{align}
where $(B_t)_{t\geq 0}$ denotes standard Brownian motion and $h(\ell, r) =
\ell^2 a(\ell,r)$ is the speed of the diffusion, setting $a(\ell,r) = \lim_{d
\to \infty} a_d(\ell,r)$. If $\alpha= 1/4$, $\beta = 1/8$, for any $r>0$, $\ell
\mapsto h(\ell, r)$ is maximized at the unique value of $\ell$ such that $a(\ell, r)= 0.452$; while if $\alpha= 1/6$,
$\beta = v/6$ with $v\geq 1$ and $r>0$, $\ell \mapsto h(\ell, r)$ is maximized
at the unique value of $\ell$ such that $a(\ell, r)= 0.574$.
\end{theorem}
\begin{proof}
The proof follows that of \cite[Theorem 1, Theorem 2]{roberts1998optimal} and is postponed to Appendix A.
\end{proof}

The theorem above shows that the relative speed at which $\lambda_d$ converges to 0 influences the scaling of the resulting proximal algorithm.
In case~\ref{case:c}, $v>1$ and $\lambda_d$ decays with $d$ at a faster rate than $\sigma_d^2$. This causes the proximity map~\eqref{eq:proximity} to collapse onto the identity and therefore the proposal~\eqref{eq:differentiable_ula} is arbitrarily close to that of MALA. The resulting scaling limit also coincides with that of MALA established in \cite[Theorem 1, Theorem 2]{roberts1998optimal}.

If $\lambda_d$ and $\sigma_d^2$ decay at the same rate (case~\ref{case:b}), the amount of gradient information provided by the proximity map is controlled by $r$. Comparing our result for case~\ref{case:b} with \cite[Theorem 1]{roberts1998optimal} we find that
\begin{align}
    K_2^2(0) = \frac{1}{16}\mathbb{E}\left[g^{\dprime}(X^d_{0,1})^3\right]+\frac{5}{48}\mathbb{E}\left[g^{\dprime\prime}(X^d_{0,1})^2\right] = K_{\MALA}^2;
\end{align}
thus, we have
\begin{align}
    K_2^2(r) &= K_2^2(0) +\parenthese{\frac{r}{8} + \frac{r^2}{4}}
    \mathbb{E}\left[\{g^{\dprime}(X^d_{0,1})g^{\prime}(X^d_{0,1})\}^2 \right]
    +\frac{r}{8}\mathbb{E}\left[g^{\dprime}(X^d_{0,1})^3\right] \\
    &= K^2_{\MALA}+\parenthese{\frac{r}{8} + \frac{r^2}{4}}\mathbb{E}\left[\{g^{\dprime}(X^d_{0,1})g^{\prime}(X^d_{0,1})\}^2\right]+\frac{r}{8}\mathbb{E}\left[g^{\dprime}(X^d_{0,1})^3\right] \geq K^2_{\MALA},
\end{align}
since the convexity of $g$ implies that $g^{\dprime}\geq 0$.
In particular, $K_2^2(r)$ is an increasing function of $r$ achieving its minimum when $r\to0$ (i.e., MALA), see Figure~\ref{fig:example_speed}(a).

In case~\ref{case:a}, $v=1/2$ and $\lambda_d$ decays more slowly than $\sigma_d^2$. As a consequence, the gradient information provided by the proximity map is smaller than in cases~\ref{case:b}--\ref{case:c}, and the resulting scaling differs from that of MALA. The value of $K_1^2(r)$ is increasing in $r$ and the speed of the corresponding diffusion also depends on $r$ (see Figure~\ref{fig:example_speed}(a) gray lines and Figure~\ref{fig:example_speed}(b)).

\begin{example}[Gaussian target]
\label{ex:gaussian}
Take $g(x)=x^2/2$, $\prox_\lambda^g(x)=x/(1+\lambda)$. 
In this case, $g^\prime$ is Lipschitz continuous and we have $K_1^2(r)=r^2/4$,
$K_2^2(r) = \left(1+4r+4r^2\right)/16$ and $K_2^2(0)=K_{\MALA}^2 = 1/16$. 
The corresponding speeds are given in Figure~\ref{fig:example_speed}(a).
Optimizing for $v=1, r=0 $ (MALA) and $v=1, r=1$ (P-MALA) we obtain
\begin{align}
    h^{\MALA}( \ell, r) = 1.5639,\qquad 
    h^{\PMALA}( \ell, r) = 0.7519,
\end{align}
achieved with $ \ell^{\MALA} = 1.6503$
and $\ell^{\PMALA} =  1.1443$, respectively. The corresponding acceptance rates are those predicted by Theorem~\ref{theo:differentiable}, $a(\ell, r)= 0.574$.
For Gaussian targets, MALA is geometrically ergodic \cite{durmus:moulines:2021:mala}, and therefore the optimal choice in terms of speed of convergence is MALA which is obtained for $r=0$. The result for $r=1$ and $v=1$ are also given in \cite[Theorem 2.1]{pillai2022optimal}.
\end{example}

\begin{example}[Target with light tails]
Take $g(x)=x^4$, which gives a normalized distribution with normalizing constant $2\Gamma(5/4)$. The proximity map is 
\begin{align}
    \prox^\lambda_g(x) &= \frac{1}{2}\left[\frac{\sqrt[3]{9\lambda^2x+\sqrt{54\lambda^4x^2+3\lambda^3}}}{3^{2/3}\lambda}-\frac{1}{\sqrt[3]{27\lambda^2x+3\sqrt{54\lambda^4x^2+3\lambda^3}}}\right].
\end{align}
In this case $g^\prime$ is not Lipschitz continuous and therefore we only consider~\ref{case:a}, for which we have $K_1^2(r) = 144r^2\Gamma(11/4)/\Gamma(5/4)$.
The corresponding speed is given in Figure~\ref{fig:example_speed}(b).
\end{example}


\begin{figure}
\centering
\begin{tikzpicture}[every node/.append style={font=\normalsize}]
\node (img1) {\includegraphics[width = 0.43\textwidth]{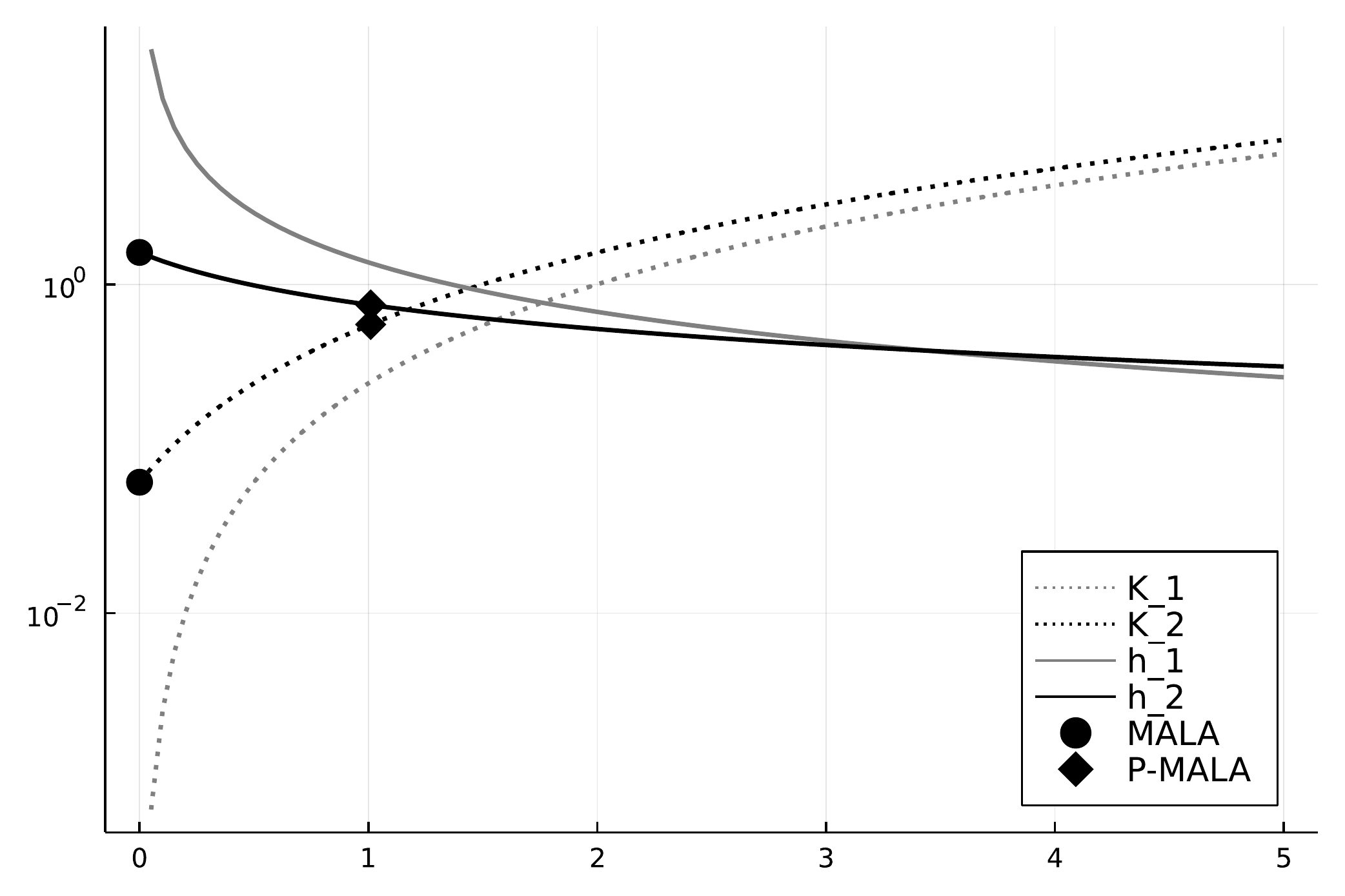}};
\node[below=of img1, node distance = 0, yshift = 1cm] {$r$};
\node[left=of img1, node distance = 0, rotate = 90, anchor = center, yshift = -1cm] {Speed};
\node[below=of img1, node distance = 0, yshift = 0.5cm] {(a) Gaussian target};
\node[right=of img1] (img2) {\includegraphics[width = 0.43\textwidth]{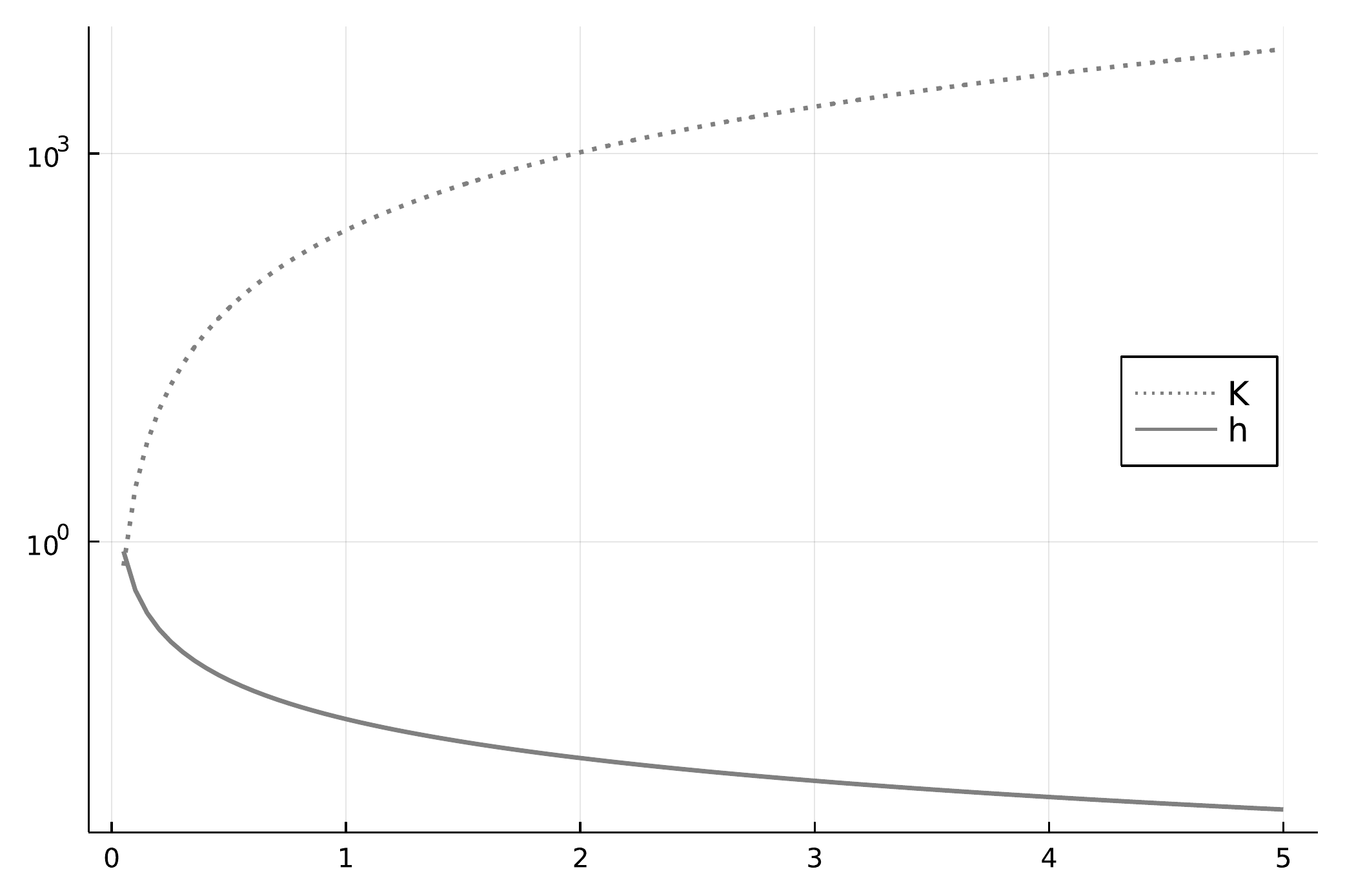}};
\node[below=of img2, node distance = 0, yshift = 1cm] {$r$};
\node[left=of img2, node distance = 0, rotate = 90, anchor = center, yshift = -1cm] {Speed};
\node[below=of img2, node distance = 0, yshift = 0.5cm] {(b) Light tail target};
\end{tikzpicture}
\caption{Value of $K$ for $i=1, 2$ and speed of the corresponding Langevin diffusion as a function of $r$ for a Gaussian target and a light tail target. We denote by $h_1$ the speed obtained in case~\ref{case:a}, by $h_2$ that obtained in~\ref{case:b}. In case~\ref{case:c} both $K_3$ and the speed $h_3$ are constant w.r.t. $r$ and coincide with that of MALA. For the Gaussian target we report the results for case~\ref{case:a}--\ref{case:c} while for the light tail target we only report~\ref{case:a}.}
\label{fig:example_speed}
\end{figure}
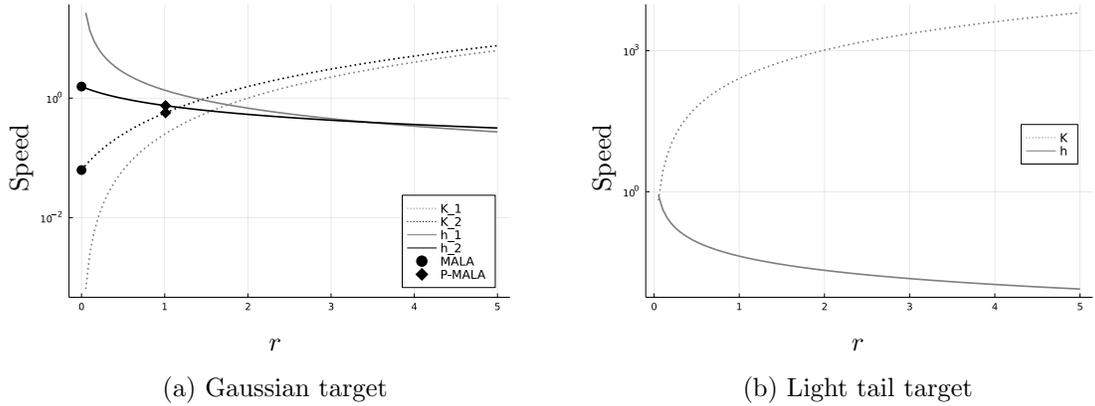


\subsection{Laplace target}
\label{sec:laplace}
As discussed in the introduction, MY-MALA has been widely used to quantify uncertainty in imaging applications, in which target distributions involving the $\ell^1$ norm are particularly common \cite{pereyra2016proximal, durmus2018efficient, agrawal2021optimal, zhou2022proximal}.

Here, we consider $\piLaplace_d$ to be the product of $d$ i.i.d. Laplace distributions as in~\eqref{eq:pi_d_iid},
\begin{equation}
  \label{eq:def_target_laplace}
\text{   $ \piLaplace_d(\bm{x}^d)=\prod_{i=1}^d \piLaplace(x_i^d) $, for $\bm{x}^d \in \rset^d$, where $\piLaplace(x)= 2^{-1} \exp(-\vert x\vert)$ }.
\end{equation}
For this particular choice of one-dimensional target distribution, the corresponding potential $G$ is $x\mapsto \vert x \vert$ and satisfies \Cref{ass:a0}. Then, the proximity map is given by the soft thresholding operator \cite[Section 6.1.3]{parikh2014proximal}
\begin{align}
\label{eq:soft}
    \prox_G^{\lambda}(x) = (x-\sgn(x)\lambda)\indicatorD{\vert x\vert\geq\lambda}{x} ,
\end{align}
where $\sgn : \rset \to \defEns{-1,1}$ is the sign function, given by $\sgn(x) = -1$ if $x< 0$, $\sgn(0)=0$, and $\sgn(x) = 1$ otherwise.
This operator is a continuous but not continuously differentiable map whose non-differentiability points are the extrema of the interval $[-\lambda, \lambda]$ and are controlled by the value of the proximity parameter $\lambda$.

Plugging ~\eqref{eq:soft} in~\eqref{eq:proposal_pmala}, the MY-MALA algorithm applied to $\piLaplace_d$ proposes component-wise for $i=1, \dots, d$
\begin{align}
\label{eq:laplace_proposal}
Y_{k+1,i}^d = X_{k,i}^d -\frac{\sigma^2_d}{2}\sgn(X_{k,i}^d)\indicatorDD{\vert X_{k,i}^d\vert\geq\lambda_d}-\frac{\sigma^2_d}{2\lambda_d}X_{k,i}^d\indicatorDD{\vert X_{k,i}^d\vert< \lambda_d}+\sigma_d Z^d_{k+1,i} .
\end{align}
For $X^d_{k,i}$ close to 0 (i.e., the point of non-differentiability) the MY-MALA proposal is a biased random walk around $X^d_{k,i}$, while outside the region $[-\lambda_d, \lambda_d]$ the proposal coincides with that of MALA. As $\lambda_d\to 0$ the region in which the MY-MALA proposal coincides with that of MALA increases and when $\lambda_d\approx 0$ the region $[-\lambda_d, \lambda_d]$ in which the proposal corresponds to a biased random walk is negligible, as confirmed by the asymptotic acceptance rate in Theorem~\ref{theo:acceptance_laplace}.

We also consider the case $\lambda_d=0$ for any $d$. Then, the proposal~\eqref{eq:laplace_proposal} becomes the proposal for the subgradient version of MALA: $Y_{k+1,i}^d = X_{k,i}^d -({\sigma^2_d}/{2})\sgn(X_{k,i}^d)+\sigma_d Z^d_{k+1,i}$, referred to as sG-MALA.

The proof of the optimal scaling for the Laplace distribution follows the structure of that of \cite{durmus2017optimal} for $\rmL^p$-mean differentiable distributions.
We start by characterizing the asymptotic acceptance ratio of a generic MY-MALA algorithm; contrary to Theorem~\ref{theo:differentiable} for differentiable targets, in the limit $d\to\infty$ the properties of MY-MALA do not depend on the relative speed at which $\sigma_d^2 = \ell^2 /d^{2\alpha}$ and $\lambda_d = c^2/2d^{2\beta}$  converge to $0$, as long as $\lambda_d$ decays at least at the same rate as $\sigma_d^2$. In this regime, the region in which the proposal~\eqref{eq:laplace_proposal} corresponds to a biased random walk proposal is negligible, and therefore we obtain the same scaling obtained with $\lambda_d=0$ and corresponding to sG-MALA.

\begin{theorem}
  \label{theo:acceptance_laplace}
Assume \Cref{ass:stationarity} and consider the sequence of target distributions $\{\piLaplace_d\}_{d \in\nsets}$ given in~\eqref{eq:def_target_laplace}. For any $d \in\nsets$, let $\sigma_d^2 = \ell^2 /d^{2\alpha}$ and $\lambda_d = c^2/2d^{2\beta}$ with $\alpha = 1/3$ and $\beta = v/3$ for $v\geq 1$.
Then, we have
$\lim_{d\to\infty} a_d(\ell,r) = a^{\rmL}(\ell) =
2\Phi(-\ell^{3/2}/(72\uppi)^{1/4})$, where $(a_d(\ell,r))_{d\in\nsets}$ is
defined in~\eqref{eq:ar_finite}, with $r=c^2/\ell^{2v}$, and $\Phi$ is the distribution function of a standard normal.
\end{theorem}
\begin{proof}
The proof is postponed to Appendix C.1.
\end{proof}
Theorem~\ref{theo:acceptance_laplace} shows that the asymptotic average acceptance rate $a^{\rmL}(\ell)$ does not depend on $r$ and as a result on $c$.

Having identified the possible scaling for MY-MALA with Laplace target, we are now ready to show weak convergence to the appropriate Langevin diffusion. To this end, we adapt the proof strategy followed in \cite{jourdain2014optimal} and \cite{durmus2017optimal}.

As for the differentiable case, consider the linear interpolation $(L_{t}^d)_{t\geq0}$ of the first component of the Markov chain $(X^d_k)_{k\geq0 }$ given in~\eqref{eq:Yt_differentiable}. 
For any $d \in \nsets$, denote by $\nu_d$ the law of the process
$(L_{t}^d)_{t\geq0}$ on the space of continuous functions from $\rset_+$ to
$\rset$, $\mathrm{C}(\real^+, \real)$, endowed with the topology of uniform
convergence over compact sets and its corresponding $\sigma$-field.
We first show that the sequence  $(\nu_d)_{d\in \nsets}$, admits a weak limit point as $d\to\infty$.
\begin{prop}
\label{prop:tightness}
Assume \Cref{ass:stationarity} and consider the sequence of target distributions $\{\piLaplace_d\}_{d \in\nsets}$ given in~\eqref{eq:def_target_laplace}. For any $d \in\nsets$, let $\sigma_d^2 = \ell^2 /d^{2\alpha}$ and $\lambda_d = c^2/2d^{2\beta}$ with $\alpha = 1/3$ and $\beta = v/3$.
The sequence $(\nu_d)_{d\in \nsets}$ is tight in $\mathsf{M}^1\left( \mathrm{C}(\real^+, \real)\right)$, the set of probability measures acting on $\mathrm{C}(\real^+, \real)$.
\end{prop}
\begin{proof}
See Appendix C.2.
\end{proof}

By Prokhorov's theorem, the tightness of $(\nu_d)_{d\in \nsets}$ implies existence of a weak limit point $\nu$. 
In our next result, we give a sufficient condition to show that any limit point of $(\nu_d)_{d\in \nsets}$ coincides with the law of a solution of:
\begin{align}
\label{eq:sde_langevin_laplace}
\rmd L_t = [h^\rmL(\ell)]^{1/2}\rmd B_t - \frac{h^\rmL(\ell)}{2}\sgn(L_t)\rmd t .
\end{align}

To this end, we consider the martingale problem (see \cite{stroock:varadhan:1979}) associated with~\eqref{eq:sde_langevin_laplace}, that we now present.
Let us denote by $\rmC_{\rmc}^\infty(\real,\real)$ the subset of functions of $\rmC(\real,\real)$ which are infinitely many times differentiable and with compact support, and define the generator  of~\eqref{eq:sde_langevin_laplace} for $V \in \rmC_{\rmc}^\infty(\real,\real)$  by
\begin{align}
\label{eq:laplace_generator}
\rmL V (x) = \frac{h^\rmL(\ell)}{2}\left[V^{\dprime}(x) -\sgn(x)  V^\prime(x)\right] .
\end{align}
Denote by $(W_t)_{t\geq 0}$ the canonical process on $\rmC(\real_+,\real)$, $W_t : \{w_s\}_{s \geq 0} \mapsto w_t$ and the corresponding filtration by $(\mathfrak{F}_t)_{t\geq 0}$.
A probability measure $\nu$ is said to solve the martingale problem associated with~\eqref{eq:sde_langevin_laplace} with initial distribution $\piLaplace$,  if the pushforward of $\nu$ by $W_0$ is $\piLaplace$ and if for all $V\in \rmC_{\rmc}^\infty(\real,\real)$, the process
\begin{align}
\left(V(W_t)-V(W_0)-\int_0^t \rmL V(W_u)\rmd u\right)_{t\geq 0}
\end{align}
is a martingale with respect to $\nu$ and the filtration $(\mathfrak{F}_t)_{t\geq 0}$. The following proposition gives a sufficient condition to prove that $\nu$ is a solution of the martingale problem:

\begin{prop}\label{prop:reduction}
Suppose that for any $V \in \rmC_{\rmc}^\infty(\real,\real)$, $m\in \mathbb{N}$,  $\rho: \real^m \rightarrow \real$ bounded and continuous, and for any $0 \leq t_1 \leq ... \leq t_m \leq s \leq t$:
\begin{equation}\label{eq:martingale_reduction}
\lim_{d\to +\infty}\mathbb{E}^{\nu_d} \left[ \left( V(W_t)-V(W_s)-\int_s^t \rmL V(W_u) \rmd u \right)\rho(W_{t_1},...,W_{t_m}) \right] = 0.
\end{equation}
Then any limit point of $(\nu_d)_{d\in \nsets}$ on $\mathsf{M}^1\left( \rmC(\real^+, \real)\right)$ is a solution to the martingale problem associated with~\eqref{eq:sde_langevin_laplace}.
\end{prop}
\begin{proof}
See Appendix C.3.
\end{proof}

Finally, we use this sufficient condition to establish that any limit point of $(\nu_d)_{d\in \nsets}$ is a solution of the martingale problem for~\eqref{eq:sde_langevin_laplace}. Uniqueness in law of solutions of~\eqref{eq:sde_langevin_laplace} allows to conclude that $(L_{t}^d)_{t\geq 0}$ converges weakly to the Langevin diffusion~\eqref{eq:sde_langevin_laplace}, which establishes our main result.

\begin{theorem}
\label{theo:laplace_diffusion}
The sequence of processes $\{(L_{t}^d)_{t\geq 0} \, :\, d \in \nsets\}$ converges in distribution towards $(L_t)_{t\geq 0}$, solution of~\eqref{eq:sde_langevin_laplace} as $d\to\infty$, with $h^\rmL(\ell)=\ell^2a^{\rmL}(\ell)$ and $a^{\rmL}$ defined in Theorem~\ref{theo:acceptance_laplace}.
In addition, $h^\rmL$ is maximized at the unique value of $\ell$ such that $a^{\rmL}(\ell)= 0.360$.
\end{theorem}
\begin{proof}
See Appendix C.4.
\end{proof}


\section{Practical implications and numerical simulations}
\label{sec:implications}

\subsection{Practical implications}
The optimal scaling results in Sections~\ref{sec:differentiable} and~\ref{sec:laplace} provide some guidance on the choice of the parameters $\sigma$ and $\lambda$ of MY-MALA algorithms, suggesting that smaller values of $\lambda$ provide better efficiency in terms of number of steps necessary to convergence (Theorem~\ref{theo:differentiable}).

However, a number of other factors must be taken into account.
First, as shown in \cite{mengersen1996rates, roberts1996geometric, roberts1996exponential,jarner2007convergence} the convergence properties of Metropolis adjusted algorithms are influenced by the shape of the target distribution and, in particular, by its tail behavior.
Secondly, when comparing MY-MALA algorithms with gradient-based methods (e.g. MALA) one must take into account the cost of obtaining the gradients, whether this comes from automatic differentiation algorithms or from evaluating a potentially complicated gradient function.
On the other hand, proximity mappings can be quickly found or approximated solving convex optimization problems which have been widely studied in the convex optimization literature (e.g. \cite[Chapter 6]{parikh2014proximal}, \cite{combettes2011proximal} and \cite[Section 3.2.3]{pereyra2016proximal}).

In terms of convergence properties, we are usually interested in the family of distributions for which the discrete time Markov chain produced by our algorithm is geometrically ergodic, together with the optimal scaling results briefly recalled in Section~\ref{sec:related_work}.
Normally, the ergodicity results are given by considering the one-dimensional class of distributions $\mathcal{E}(\beta, \gamma)$ introduced in \cite{roberts1996exponential} and defined for $\gamma>0$ and $0<\beta<\infty$ by
\begin{align}
    \mathcal{E}(\beta, \gamma):\left\lbrace\pi:\real\to [0, +\infty):\pi(x)\propto \exp\left(-\gamma\vert x\vert^{\beta}\right), \vert x\vert > x_0 \textrm{ for some } x_0>0\right\rbrace.
\end{align}
As observed by \cite{livingstone2022barker}, there usually is a trade-off between ergodicity and optimal scaling results, algorithms providing better optimal scaling results tend to be geometrically ergodic for a smaller set of targets (e.g. MALA w.r.t. RWM).

As suggested by Theorem~\ref{theo:differentiable}, the scaling properties of MY-MALA on regular targets are close to those of MALA. This leads to a natural comparison between the two algorithms. First, we observe that \Cref{ass:a0} rules out targets for which $G$ is not convex and therefore restricts the families $\mathcal{E}(\beta, \gamma)$ to $\beta\geq1$. To compare MALA with MY-MALA we therefore focus on distributions with $\beta\geq 1$.

It is shown in \cite{roberts1996exponential} that MALA is geometrically ergodic for targets in $\mathcal{E}(\beta, \gamma)$ with $1\leq\beta\leq 2$ (with some caveat for $\beta=2$). Theorem~\ref{theo:differentiable}--\ref{case:b} and~\ref{case:c} show that in this case MY-MALA has the same scaling properties of MALA but in case~\ref{case:b} the asymptotic speed of convergence decays as the constant $r$ increases (Figure~\ref{fig:example_speed}(a)), with the maximum achieved for $r\to 0$, for which MY-MALA collapses onto MALA.
Since MALA is geometrically ergodic, and achieves better (or equivalent) scaling properties than MY-MALA, it would be natural to prefer MALA to MY-MALA for this set of targets.
However, if the gradient is costly to obtain, one might instead consider to use MY-MALA with a small $\lambda$, to retain scaling properties as close as possible to that of MALA but to reduce the computational cost of evaluating the gradient.

In the case of regular targets with light-tails (i.e., $\beta>2$), MALA is known not to be geometrically ergodic \cite[Section 4.2]{roberts1996exponential} while the ergodicity properties of MY-MALA have only been partially studied in \cite[Section 3.2.2]{pereyra2016proximal} for the case $\lambda=\sigma^2/2$ (P-MALA).
As shown in \cite[Section 2.1]{pereyra2016proximal}, given a distribution $\pi\in\mathcal{E}(\beta, \gamma)$ with $\beta\geq 1$, the distribution $\pi_\lambda$ obtained using the potential~\eqref{eq:glambda} belongs to $\mathcal{E}(\beta^{\prime}, \gamma^{\prime})$, where $\beta^{\prime}=\min(\beta, 2)$ and $\gamma'$ depending on $\lambda$.
This suggests that MY-MALA is likely to be geometrically ergodic for appropriate choices of $\lambda$; a first result in this direction is given in \cite[Corollary 3.2]{pereyra2016proximal} for the P-MALA case $\lambda=\sigma^2/2$.
Theorem~\ref{theo:differentiable}--\ref{case:c} restricts the sets of available $\lambda$s showing that for light-tail distributions (for which \Cref{ass:g_lipschitz} does not hold) $\lambda$ should decay at half the speed of $\sigma^2$.
Studying the ergodicity properties of MY-MALA in function of the parameter $\lambda$ is, of course, an interesting problem that we leave for future work.

For the Laplace distribution, Theorem~\ref{theo:acceptance_laplace} shows that the value of $\lambda$ does not influence the asymptotic acceptance ratio of MY-MALA, as long as $\lambda$ decays with $d$ at least as fast as $\sigma^2$.
The scaling properties and the asymptotic speed $h(\ell)$ in Theorem~\ref{theo:laplace_diffusion} do not depend on $\lambda$ and coincide with that of the sG-MALA (obtained for $\lambda=0$). Hence, in terms of optimal scaling, there does not seem to be a difference between MY-MALA and sG-MALA for the Laplace distribution.

\subsection{Numerical experiments}
\label{sec:numerical}

To illustrate the results established in Section~\ref{sec:differentiable} and~\ref{sec:laplace} we consider here a small collection of simulation studies.
The aim of these studies is to empirically confirm the optimal scalings identified in Theorem~\ref{theo:differentiable} and~\ref{theo:acceptance_laplace}, investigate the dimension $d$ at which the asymptotic acceptance ratio $\lim_{d\to \infty}a_d(\ell, r)$ well approximates the empirical average acceptance ratio and, consequently, for which dimensions $d$ we can expect the optimal asymptotic acceptances in Theorem~\ref{theo:differentiable} and~\ref{theo:acceptance_laplace} to guarantee maximal speed $h(\ell, r)$ (approximated by the expected squared jumping distance, see, e.g. \cite{gelman1996efficient}) for the corresponding diffusion.
We summarize here our findings, a more detailed discussion can be found in Appendix E.

For the regular case, we consider the Gaussian distribution in
Example~\ref{ex:gaussian} and four algorithmic settings which correspond to the three cases identified in Theorem~\ref{theo:differentiable} and MALA.
The different values of $r$ and $v$ influence the
dimension required to observe convergence to the theoretical limit in Theorem~\ref{theo:differentiable}: for $r\to 0$ and $v=1$ (MALA) and $v=1/2, r=1$ (corresponding to Theorem~\ref{theo:differentiable}--\ref{case:a}) the theoretical limit is already achieved for $d$ of order $10^2$, while in the cases
$v=3$, $r=2$ and $v=r=1$ (corresponding to Theorem~\ref{theo:differentiable}--\ref{case:c} and \ref{case:b}, respectively) our simulation result match the theoretical limit only for $d$ of order $ 10^{5}$ or higher. 

The results for the Laplace case are similar, with the case $v>1$ requiring a higher $d$ to observe convergence to the theoretical limit.  \Cref{fig:laplace_mala_main}
and \Cref{fig:laplace_pmala_main}
provide numerical simulations of the behavior,  as $d$ increases, of the mean acceptance ratio $(a_d(\ell,r))_{d \in\nsets}$ as a function of $\ell$ and $(\mathrm{ESJD}_d)_{d \in\nsets}$ as a function of $(a_d(\ell,r))_{d \in\nsets}$,  for sG-MALA ($r=0$) and P-MALA ($r=1$) respectively. These confirm our theoretical findings \Cref{theo:acceptance_laplace} and \Cref{theo:laplace_diffusion}.

\begin{figure}
\centering
\begin{tikzpicture}[every node/.append style={font=\tiny}]
\node (img1) {\includegraphics[width = 0.4\textwidth]{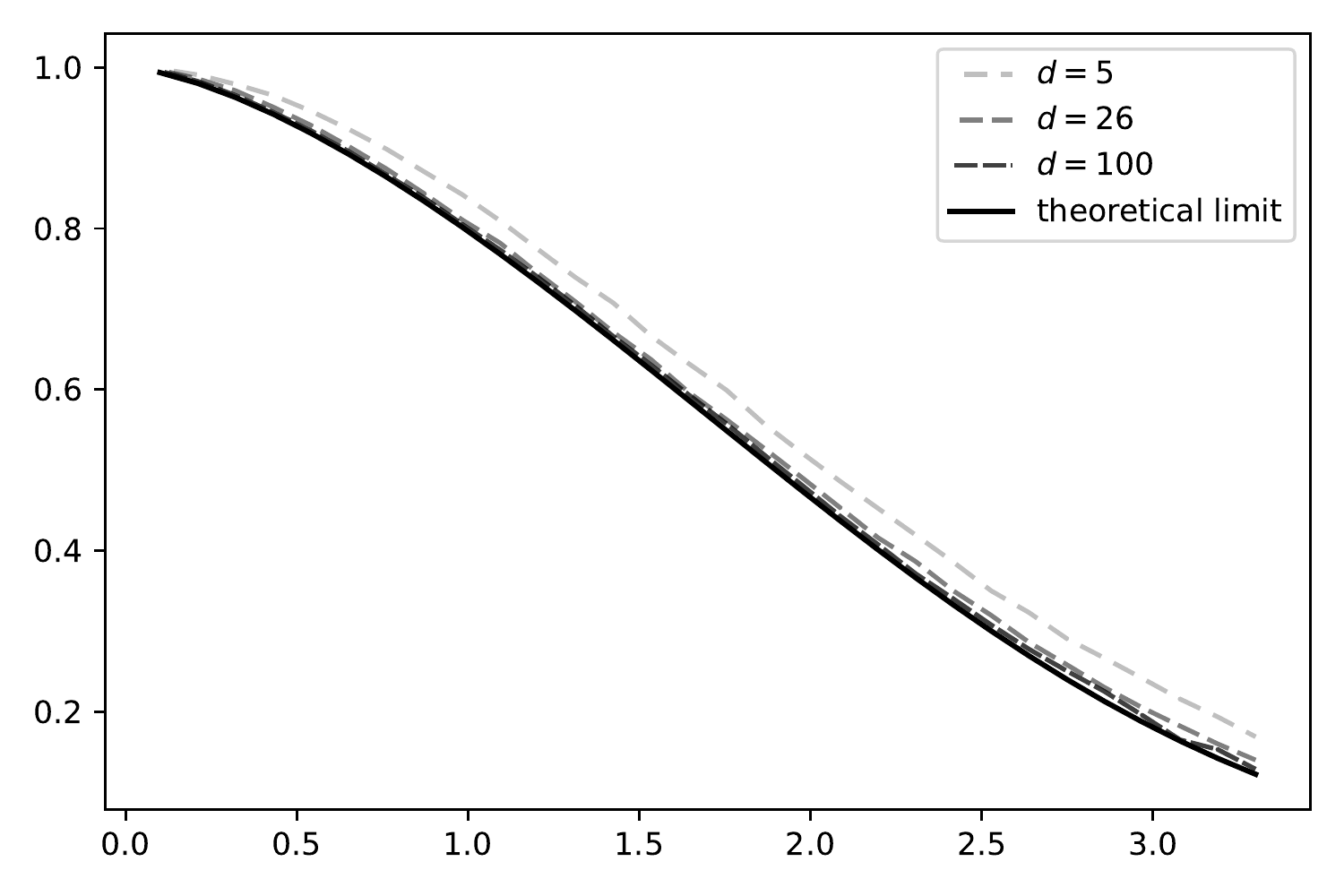}};
\node[below=of img1, node distance = 0, yshift = 1.2cm] {$\ell$};
\node[left=of img1, node distance = 0, rotate = 90, anchor = center, yshift = -1cm] {$a_d(\ell, r)$};
\node[right=of img1, node distance = 0, xshift = -0.7cm] (img3) {\includegraphics[width = 0.4\textwidth]{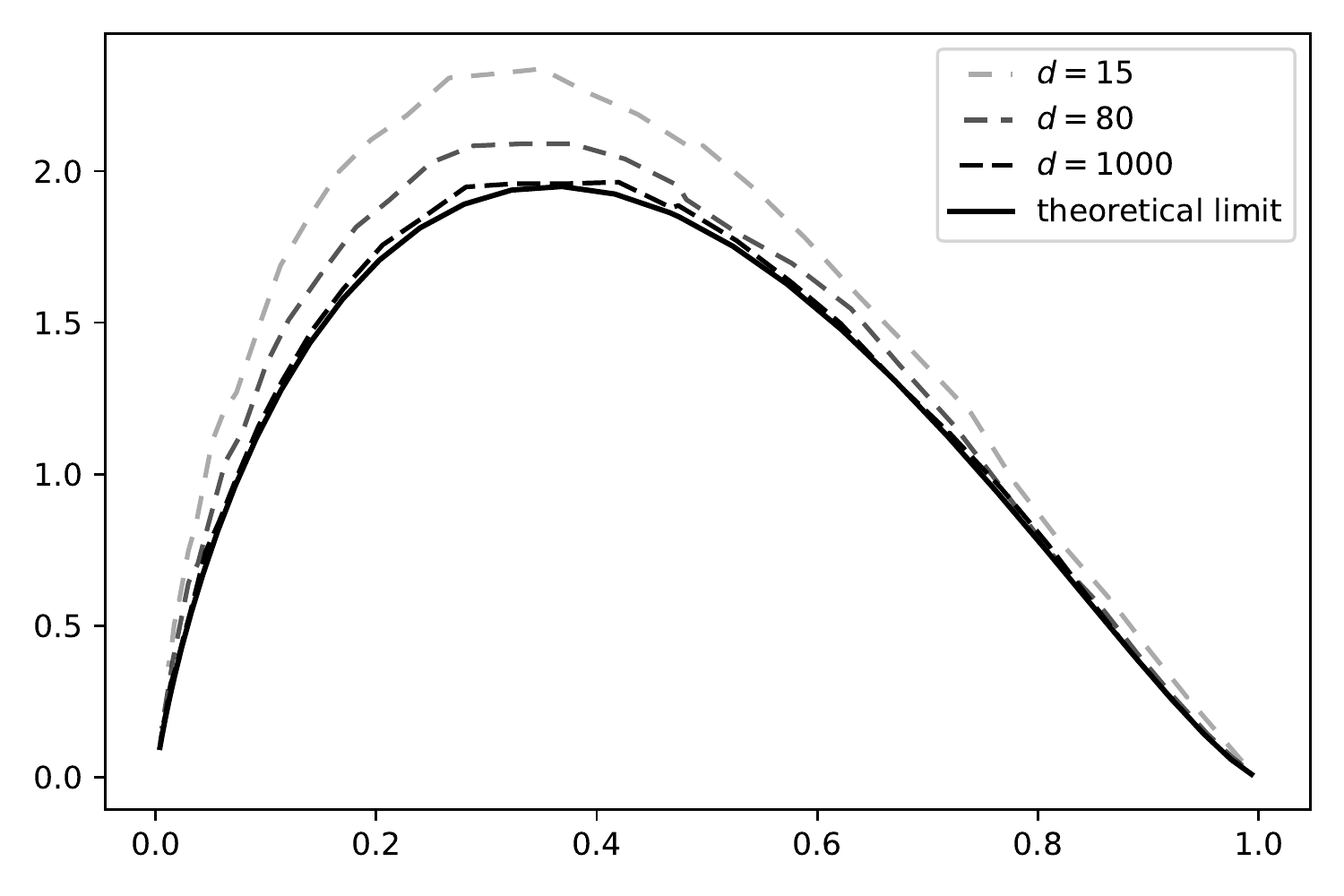}};
\node[below=of img3, node distance = 0, yshift = 1.2cm] {$a_d(\ell, r)$};
\node[left=of img3, node distance = 0, rotate = 90, anchor = center, yshift = -1cm] {$\textrm{ESJD}_d$};
\end{tikzpicture}
\caption{MY-MALA with Laplace target and $v=1, r=0$ (sG-MALA).  Left: acceptance rate as a function of $\ell$ for increasing dimension $d$; Right: $\textrm{ESJD}_d$ as a function of the acceptance rate $a_d(\ell, r)$.}
\label{fig:laplace_mala_main}
\end{figure}
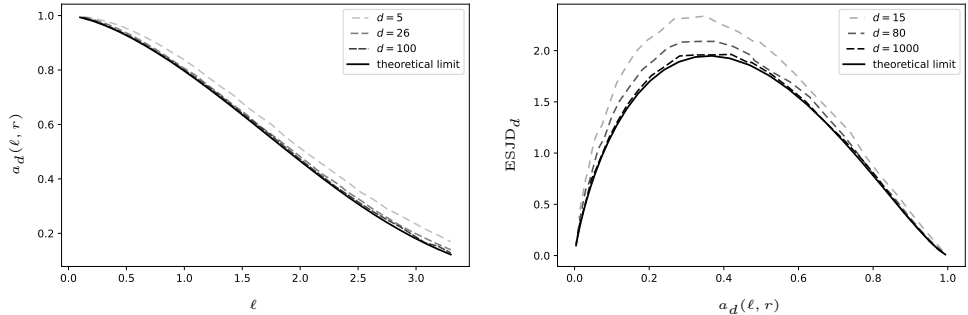

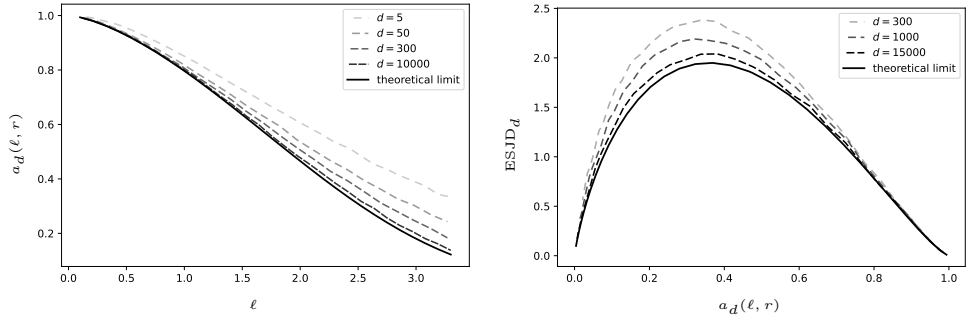
\begin{figure}
\centering
\begin{tikzpicture}[every node/.append style={font=\tiny}]
\node (img2) {\includegraphics[width = 0.4\textwidth]{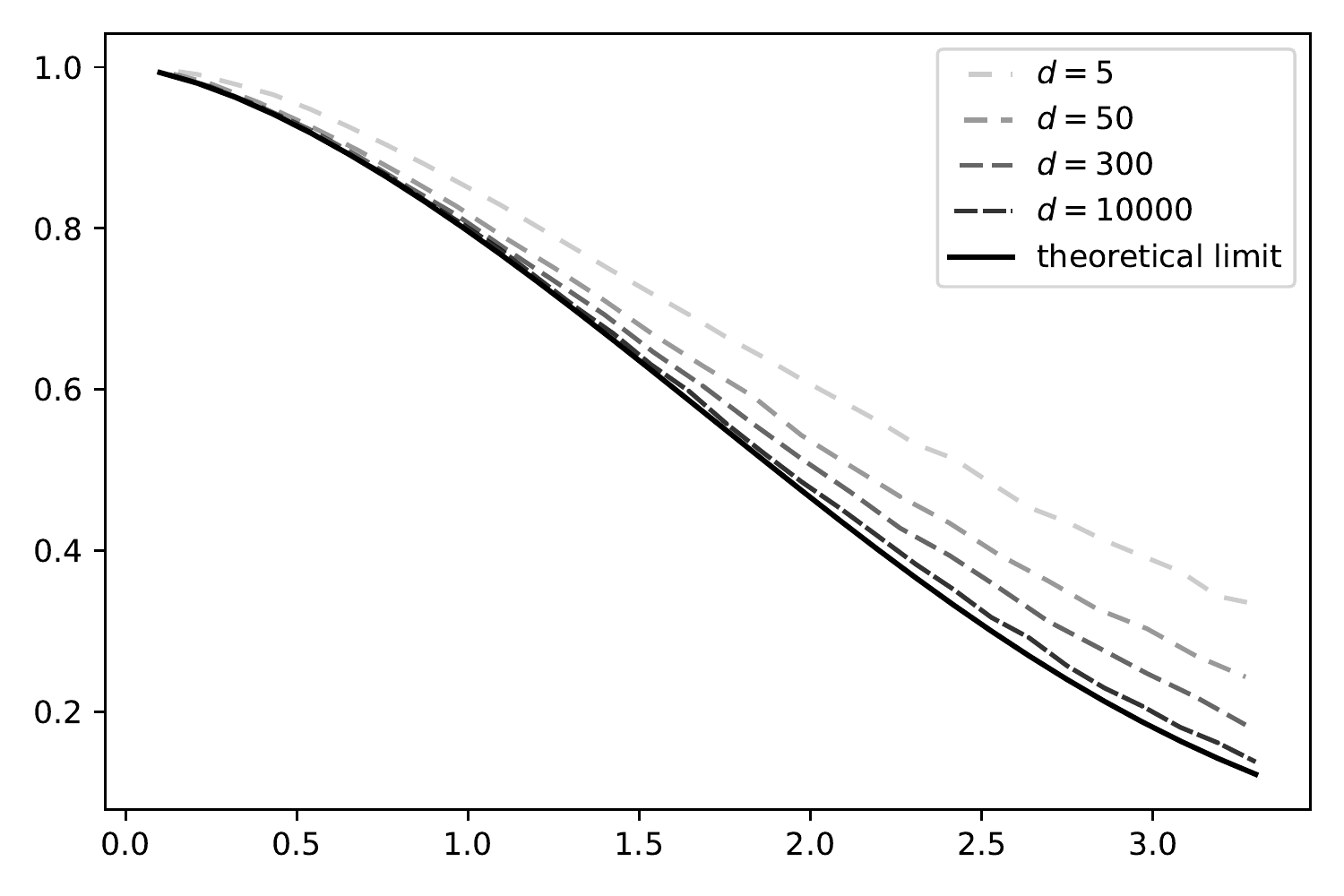}};
\node[below=of img2, node distance = 0, yshift = 1.2cm] {$\ell$};
\node[left=of img2, node distance = 0, rotate = 90, anchor = center, yshift = -1cm] {$a_d(\ell, r)$};
\node[right=of img2, node distance = 0, xshift = -0.7cm] (img3) {\includegraphics[width = 0.4\textwidth]{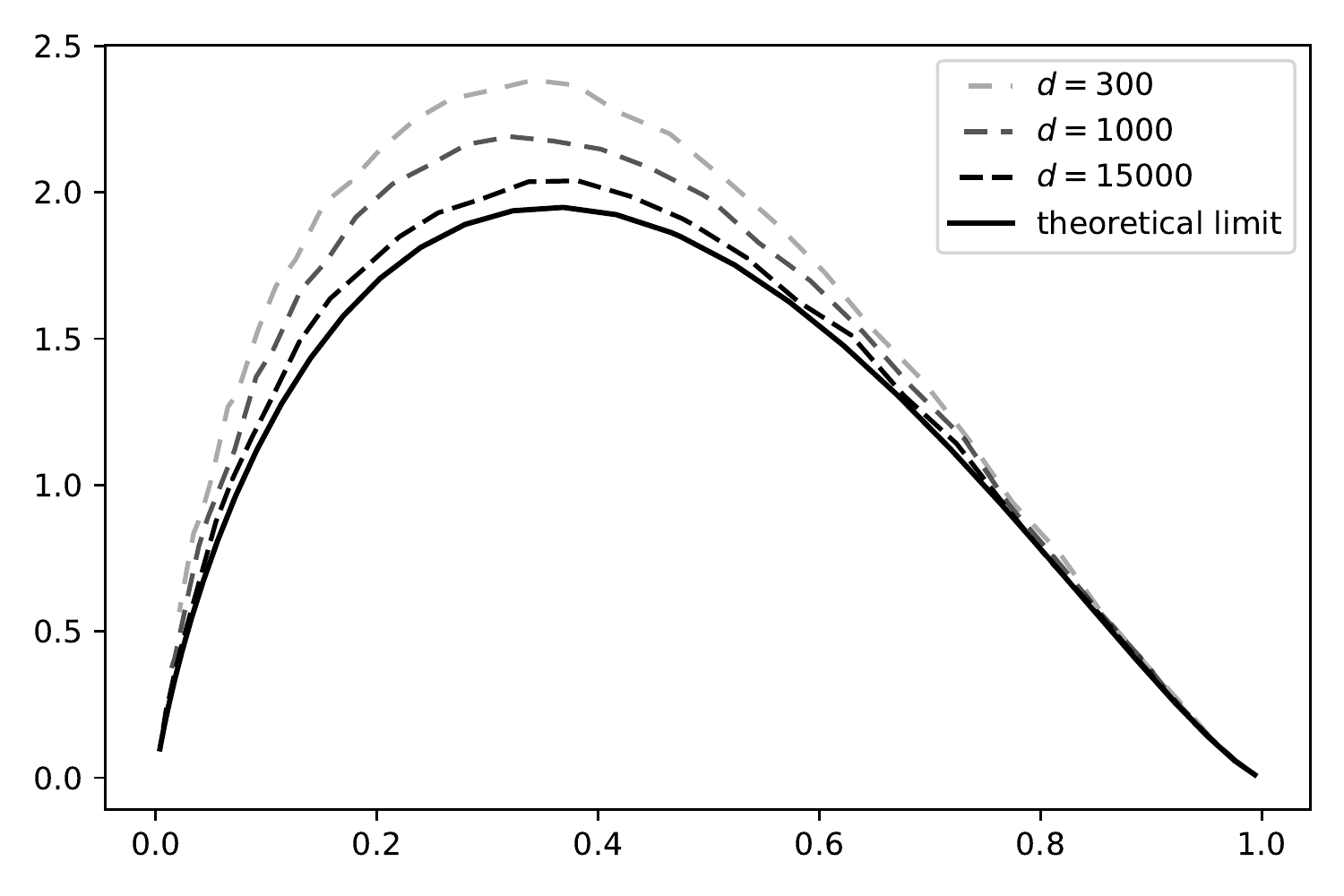}};
\node[below=of img3, node distance = 0, yshift = 1.2cm] {$a_d(\ell, r)$};
\node[left=of img3, node distance = 0, rotate = 90, anchor = center, yshift = -1cm] {$\textrm{ESJD}_d$};
\end{tikzpicture}
\caption{MY-MALA with Laplace target and $v=1, r=1$ (P-MALA). Left: acceptance rate as a function of $\ell$ for increasing dimension $d$; Right: $\textrm{ESJD}_d$ as a function of the acceptance rate $a_d(\ell, r)$.}
\label{fig:laplace_pmala_main}
\end{figure}

In general, we find that the optimal average acceptance ratios in Theorem~\ref{theo:differentiable} and~\ref{theo:laplace_diffusion} guarantee maximal speed $h(\ell, r)$ for $d$ sufficiently large (for small $d$ the optimal acceptance ratio often differs from the optimal asymptotic one, see, e.g.  \cite[Section 2.1]{sherlock2009optimal}).

To further investigate the scaling of MY-MALA to other non-differentiable densities, we empirically study two cases  where the sequence of targets are given by: $x^d \in \rset^d$,
\begin{equation}
\label{eq:def_l2_l1_numerics}
\pi_d^{\mathrm{ GL }}(x^d) = \prod_{i=1}^d \exp(-g(x^d_i)) , \quad g(x)=\abs{x} + x^2/2 , \, \text{ and }
\end{equation}
which, like the Laplace distribution in $0$, is non-differentiable but convex. The study of the potential $g$ is motivated by Bayesian inverse problems considered in \cite{pereyra2016proximal,durmus2018efficient}, for which the posterior distribution arises from Gaussian observations and sparsity-induced priors like the Laplace distribution.
The posterior then has the form (up to a multiplicative constant) $x^d \mapsto \exp(-\Vert y^d - \mathbf{A} x^d \Vert^2 - c_r \sum_{i=1}^d \vert x_i \vert )$. For the choice of target \eqref{eq:def_l2_l1_numerics}, the proposal of MY-MALA is given for any $d\in \nsets$, $k\in \nset$ and $i \in
\iint{1}{d}$,
\begin{equation} Y^d_{k+1,i} = \parenthese{1- \frac{\sigma_d^2}{2\lambda_d} } X^d_{k,i} + \frac{\sigma_d^2}{2\lambda_d}\parenthese{ (X^d_{k,i}-\sgn(X^d_{k,i}))\indicatorDD{|X^d_{k,i}| \geq \lambda_d} - \lambda_d X^d_{k,i}} + \sigma_d Z^d_{k+1,i},
\end{equation}
where $(Z^d_{k+1,i})_{k\in \nset}$ is a sequence of standard normal random variables.

For our second experiment, we go beyond the Laplace case that aim to verify that our scaling  results also holds for other non-smooth distributions. In particular, we consider first a $m$-dimensional distribution, with $m \in \nsets$ defined for any $x^{m}\in \rset^{m}$ by
\begin{equation}
\label{eq:def_smtv_l2_l1_numerics}
\pi_{m}^{\mathrm{ TV }}(x^m) = \exp{\parentheseLigne{ - (x^m_1)^2/2}}\prod_{j=1}^{m-1}\exp\parentheseLigne{-g^{\mathrm{ TV }}(x^m_{i},x^m_{i+1}) + (x^m_{i+1})^2}, \quad g^{\mathrm{TV}}(x,y)=\abs{x-y},
\end{equation}
We chose this type of distribution since it is a sum of a quadratic function and a total variation norm, which have been used in Bayesian image processing \cite{louchet2008total}. Then, the target that we consider is obtained by independently copying this $m$-dimensional distribution $d$ times, with $d\in \nsets$. Let $D=d m$, then for any $x^D \in \rset^D$, 
we define the sequence of target indexed by $d$ by:
\begin{equation}
\label{eq:def_tv_l2_l1_numerics}
\pi_d^{\mathrm{ ITV }}(x^D) = \prod_{i=0}^{d-1} 
\pi_m^{\mathrm{ TV }}\parenthese{x^d_{i m +1}, x^d_{i m +2},\cdots, x^d_{i m +\kappa}},
\end{equation}
where $\pi_m^{\mathrm{ TV }}$ is the $m$-dimensional distribution defined in \eqref{eq:def_smtv_l2_l1_numerics}.
For the choice of target \eqref{eq:def_tv_l2_l1_numerics}, the proximity operator is not explicit, and we use the implementation described in  \cite{conficmlBarbero11,JMLR:v19:13-538}. 
We then repeated the same experiments as for the Laplace distribution increasing the dimension $d$. The results are gathered in Appendix E and \Cref{fig:mixedtv_pmala_main,fig:mixed_general_main}.
These figures seem to indicate that the scaling that we find for the Laplace distribution, i.e., choosing $\sigma_d^2 = \ell/d^{2\alpha}$, $\lambda = \sigma_d^{2v}r/2$ with $\alpha = 1/3$, $r \geq 0$, is also the adequate scaling for \eqref{eq:def_tv_l2_l1_numerics}.
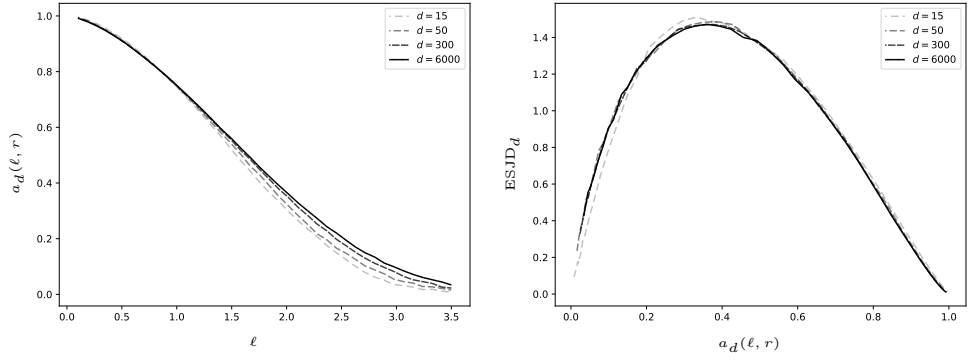
\begin{figure}
\centering
\begin{tikzpicture}[every node/.append style={font=\tiny}]
\node (img1) {\includegraphics[width = 0.4\textwidth]{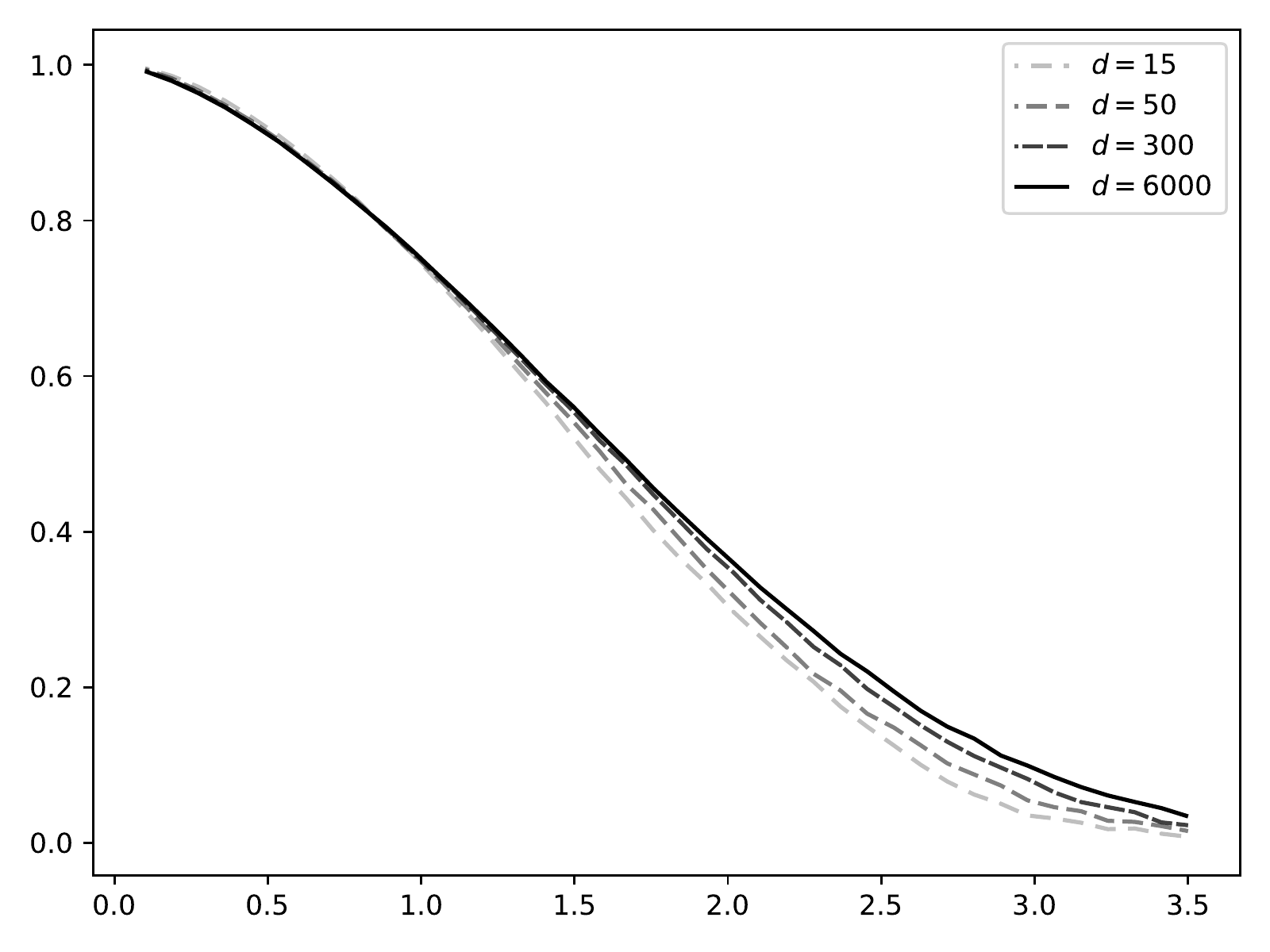}};
\node[below=of img1, node distance = 0, yshift = 1.2cm] {$\ell$};
\node[left=of img1, node distance = 0, rotate = 90, anchor = center, yshift = -1cm] {$a_d(\ell, r)$};
\node[right=of img1, node distance = 0, xshift = -0.7cm] (img3) {\includegraphics[width = 0.4\textwidth]{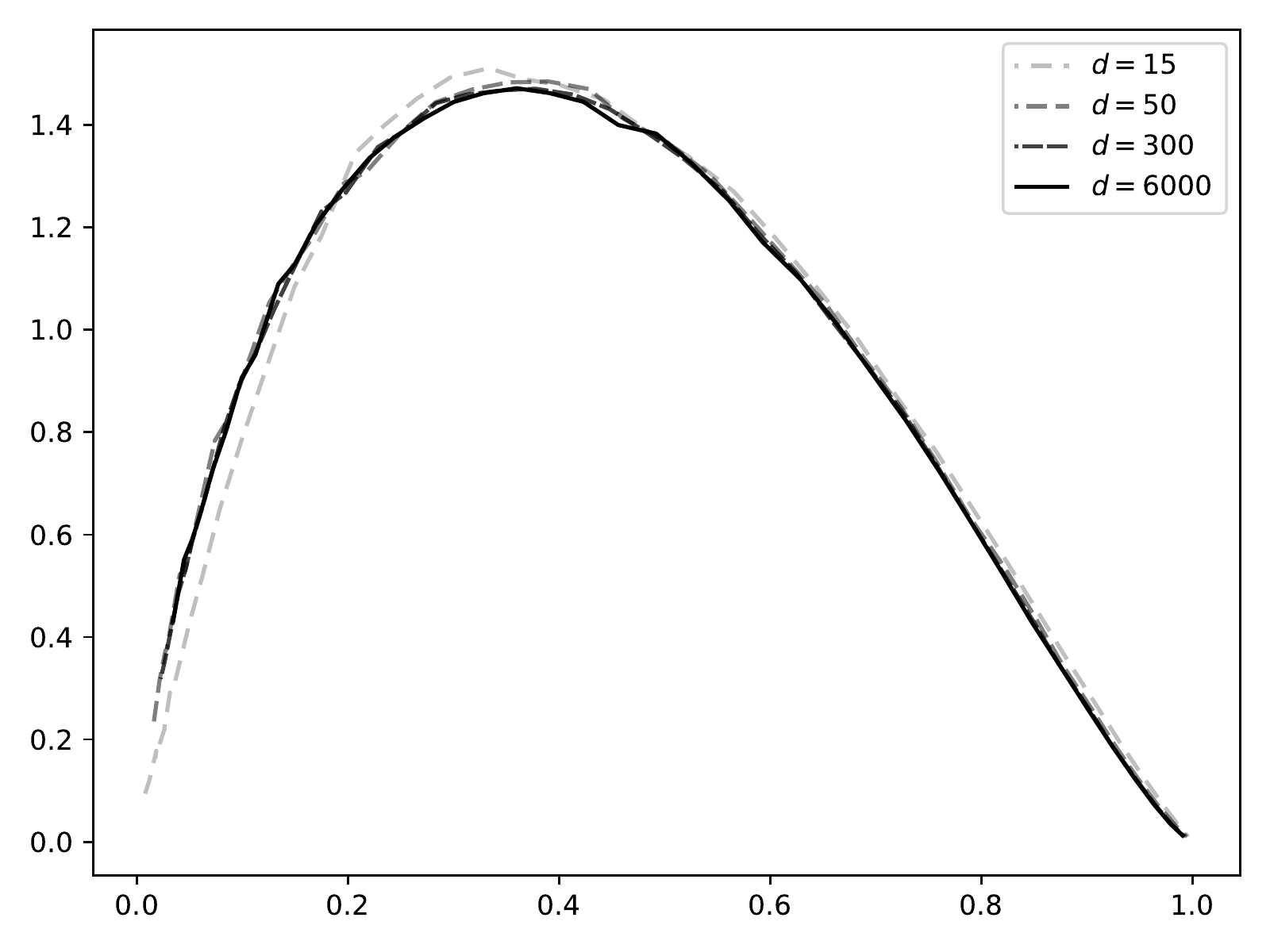}};
\node[below=of img3, node distance = 0, yshift = 1.2cm] {$a_d(\ell, r)$};
\node[left=of img3, node distance = 0, rotate = 90, anchor = center, yshift = -1cm] {$\textrm{ESJD}_d$};
\end{tikzpicture}
\caption{MY-MALA for the target \eqref{eq:def_l2_l1_numerics} and $v=3, r=2$, with $\sigma^2=\ell/d^{2\alpha}$ and $\alpha=1/3$.  Left: acceptance rate as a function of $\ell$ for increasing dimension $d$; Right: $\textrm{ESJD}_d$ as a function of the acceptance rate $a_d(\ell, r)$.}
\label{fig:mixed_general_main}
\end{figure}
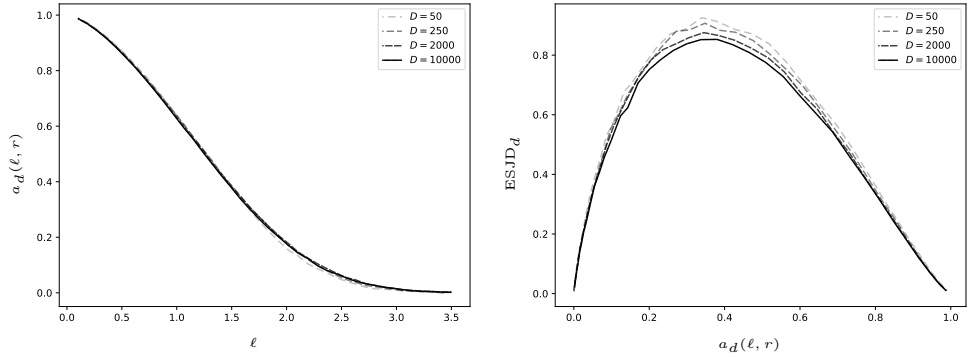
\begin{figure}
\centering
\begin{tikzpicture}[every node/.append style={font=\tiny}]
\node (img1) {\includegraphics[width = 0.4\textwidth]{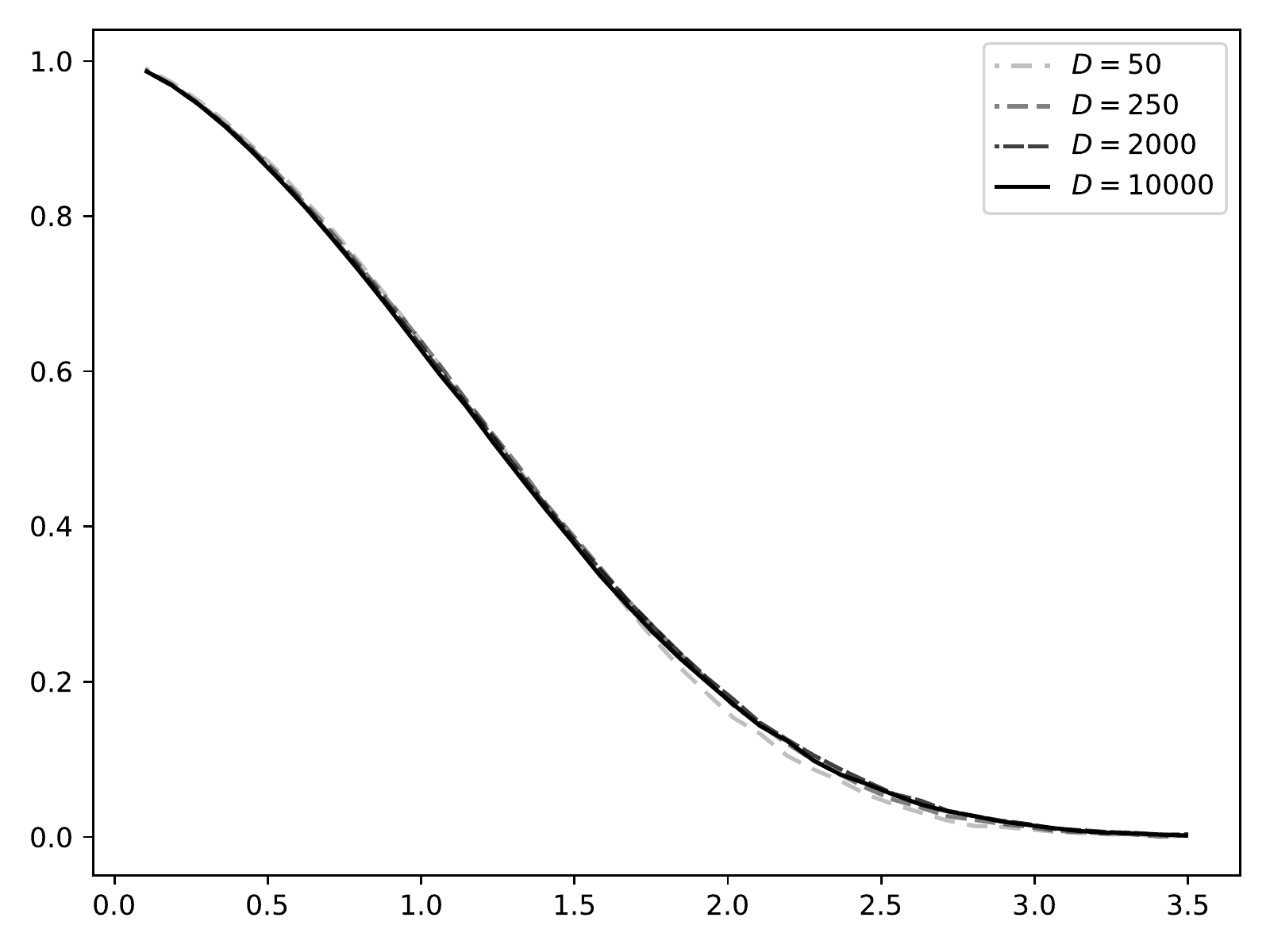}};
\node[below=of img1, node distance = 0, yshift = 1.2cm] {$\ell$};
\node[left=of img1, node distance = 0, rotate = 90, anchor = center, yshift = -1cm] {$a_d(\ell, r)$};
\node[right=of img1, node distance = 0, xshift = -0.7cm] (img3) {\includegraphics[width = 0.4\textwidth]{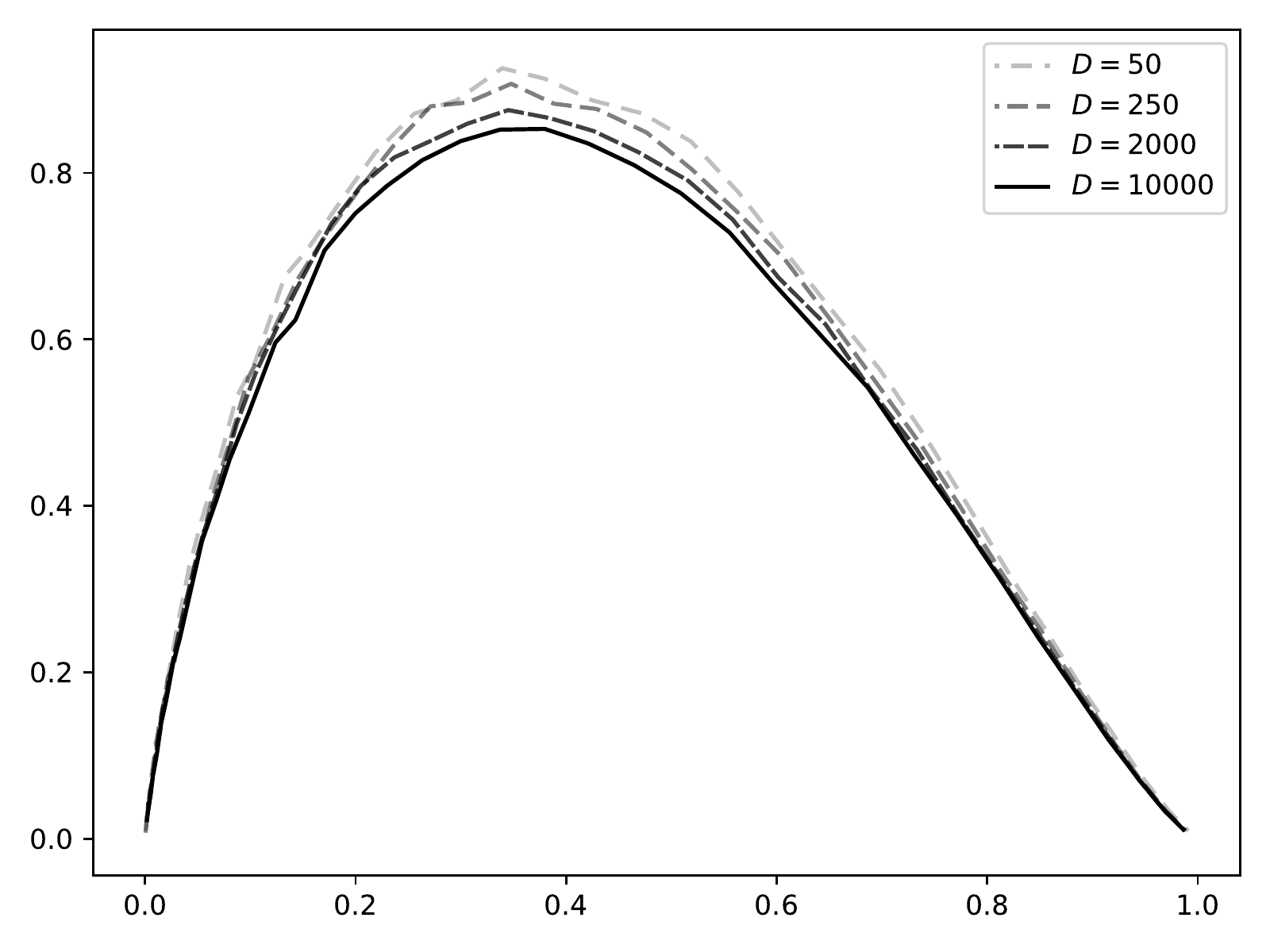}};
\node[below=of img3, node distance = 0, yshift = 1.2cm] {$a_d(\ell, r)$};
\node[left=of img3, node distance = 0, rotate = 90, anchor = center, yshift = -1cm] {$\textrm{ESJD}_d$};
\end{tikzpicture}
\caption{MY-MALA for the target \eqref{eq:def_tv_l2_l1_numerics} with $m=10$ and $v=1$, $r=1$ (P-MALA). Left: acceptance rate as a function of $\ell$ for increasing dimension $d$; Right: $\textrm{ESJD}_d$ as a function of the acceptance rate $a_d(\ell, r)$.}
\label{fig:mixedtv_pmala_main}
\end{figure}

\section{Discussion}
\label{sec:discussion}

In this work we analyze the scaling properties of a wide class of MY-MALA algorithms introduced in \cite{pereyra2016proximal, durmus2018efficient} for smooth targets and for the Laplace distribution.
We show that the scaling properties of MY-MALA are influenced by the relative speed at which the proximal parameter $\lambda_d$ and the proposal variance $\sigma_d$ decay to 0 as $d\to\infty$ and suggest practical ways to choose $\lambda_d$ as a function of $\sigma_d$ to guarantee good results.

In the case of smooth targets, we provide a detailed comparison between MY-MALA and MALA, showing that MY-MALA scales no better than MALA (Theorem~\ref{theo:differentiable}). In particular, Theorem~\ref{theo:differentiable}--\ref{case:a} shows that if $\lambda_d$ is too large w.r.t. $\sigma_d$ then the efficiency of MY-MALA is of order $\bigO(d^{-1/2})$ and therefore worse than the $\bigO(d^{-1/3})$ of MALA, suggesting that $\lambda_d$ should be chosen to decay approximately as $\sigma_d$, if possible.
If $\lambda_d$ decays sufficiently fast, then MALA and MY-MALA have similar scaling properties and, in the case in which the proximity map is cheaper to compute that the gradient, one can build MY-MALA algorithms which are as efficient as MALA in terms of scaling but more computationally efficient. 

In the case of the Laplace distribution, we show that the scaling of MY-MALA is $\bigO(d^{-2/3})$ for any $\lambda_d$ decaying sufficiently fast w.r.t. $\sigma_d$ and, in the case $\lambda_d= 0$, we obtain a novel optimal scaling result for sG-MALA on Laplace targets.

As discussed in Section~\ref{sec:implications}, our analysis provides some guidance on the choice of the parameters that need to be specified to implement MY-MALA, but this analysis should be complemented by an exploration of the ergodicity properties of MY-MALA to obtain a comprehensive description of the algorithms. 
We conjecture that for sufficiently large values of $\lambda$, MY-MALA applied to light tail distributions will be exponentially ergodic; establishing exactly how large should $\lambda$ be to guarantee fast convergence is an interesting question that we leave for future work.
Obtaining these results would open the doors to adaptive tuning strategies for MY-MALA, which are likely to produce better results than those given by the strategies currently used.

We assumed throughout this work that we have access to exact evaluations of the proximity map $\prox_G^\lambda$; while this is true for many functions of interest (e.g. $\ell^p$-norms \cite{parikh2014proximal}), in general the optimization problem~\eqref{eq:proximity} which defines the proximity map needs to be solved approximately. This introduces an
approximation error which could be studied using, e.g., similar tools to those employed in \cite{schmidt2011convergence}.

The set up under which we carried out our analysis closely resembles that of \cite{roberts1998optimal}; we anticipate that \Cref{ass:stationarity} could be relaxed following similar ideas as those in \cite{christensen2005scaling, jourdain2014optimal} and that our analysis could be extended to $d$-dimensional targets $\pi_d$ possessing some dependence structure following the approach of \cite{sherlock2009optimal, bedard2008optimal, yang2020optimal}.
Finally, the analysis carried out for the Laplace distribution could be extended to other piecewise smooth distributions provided that the moments necessary for the proof in Appendix C can be computed.


\section*{Acknowledgments}
F.R.C. and G.O.R. acknowledge support from the EPSRC (grant \#  EP/R034710/1). G.O.R. acknowledges further support from the  EPSRC (grant \# EP/R018561/1) and the Alan Turing Institute. 
A.D. acknowledges support from the Lagrange Mathematics and Computing Research Center.

The authors would like to thank \'{E}ric Moulines for helpful discussions.

For the purpose of open access, the author has applied a Creative Commons Attribution (CC BY) licence to any Author Accepted Manuscript version arising from this submission.
\bibliographystyle{abbrv} 
\bibliography{pmala_biblio}

\appendix

\section{Proof of Theorem~\ref{theo:differentiable}}
\label{app:proof_differentiable}

\noindent The proof of Theorem~\ref{theo:differentiable} follows that of \cite[Theorem 1,
Theorem 2]{roberts1998optimal} and consists of four propositions showing
convergence of the log-acceptance probability to a normal random variable and
(weak) convergence of the process~\eqref{eq:Yt_differentiable} to a Langevin diffusion.

We start by recalling and defining a number of quantities that we will use in the following proofs.
Recall that $\sigma_d = \ell / d^{\alpha}$, that $\lambda_d = \sigma_d^{2v}r/2$ where $v\geq 1/2$ and $r >0$ 
are to be chosen according to the different cases in \Cref{theo:differentiable}.
Recalling the expression of the proposal given in~\eqref{eq:proposal_pmala} and using the
simplification given in~\eqref{eq:differentiable_ula}, we define 
the proposal with starting point $\bm{x}^d\in \rset^d$,
\begin{equation}
    \bm{y}^d(\bm{x}^d,\bm{z}^d) = \bmx^d - \frac{\sigma_d^2}{2} \nabla
    G\parenthese{\prox_{G}^{\sigma_d^{2v}r/2}(\bmx^d)} + \sigma_d
    \bm{z}^d,
\end{equation}
where $\bm{z}^d \in \rset^d$. Since $G(\bm{x}^d)=\sum_{i=1}^d g(x_i^d)$,
the $i$-th component of the proposal only depends on the $i$-th components of
$\bmx^d$ and $\bm{z}^d$. Thus, for any $x,z \in \rset$ we denote
\begin{equation}
    y_d(x,z) = x - \frac{\sigma_d^2}{2} g'\parenthese{\prox_{g}^{\sigma_d^{2v}r/2}(x)}
    + \sigma_d z.
\end{equation}
The proposal for the
chain $(X^d_k)_{k\geq 0}$ is then given by $Y^d_k = \bm{y}^d(X^d_k,Z^d_{k+1}) =
(y_d(X^d_{k,i},Z^d_{k+1,i}))_{i\in\iint{1}{d}}$. 
Let us define the generator of the discrete process $(X_k^d)_{k\geq 0}$ for all
$V\in \rmC_{\rmc}^\infty(\real^d,\real)$, i.e., infinitely differentiable
$\real$-valued multivariate functions with compact support, and any $ \bmx^d\in \rset^d$,
\begin{align}
    \rmL_d V(\bmx^d) &= d^{2\alpha}\expe{\left[V(\bm{y}^d(\bm{x}^d,Z^d_{1})) -
            V(\bm{x}^d)\right]\frac{\pi_d(\bm{y}^d(\bm{x}^d,Z^d_1))q_d(\bm{y}^d(\bm{x}^d,Z^d_1),
    \bm{x}^d)}{\pi_d(\bm{x}^d)q_d(\bm{x}^d, \bm{y}^d(\bm{x}^d,Z^d_1))}\wedge 1}\\
                     &= d^{2\alpha}\expe{\left[V(\bm{y}^d(\bm{x}^d,Z^d_1)) - V(\bm{x}^d)\right]\prod_{i=1}^d \exp\left(\phi_d(x_i^d, Z^d_{1,i})\right)\wedge 1},
\end{align}
where the expectation is w.r.t. $Z^d_{1}=(Z^d_{1,i})_{i \in \iint{1}{d}}$, a
$d$-dimensional standard normal random variable, and where we defined
\begin{align}
\label{eq:log_acceptance}
\phi_d(x, z) &= \log \frac{\pi(y_d(x,z))q(y_d(x,z), x)}{\pi(x)q(x, y_d(x,z))} \\
             &= g(x) - g(y_d(x,z)) +\log q(y_d(x,z), x)  -\log q(x,
             y_d(x,z)).
\end{align}
In the remainder we will work with one-dimensional functions $V\in \rmC_{\rmc}^\infty(\real,\real)$ applied to the first component of $\bm{x}^d$ so that
\begin{align}
    \rmL_d V(\bm{x}^d) &= d^{2\alpha}\expe{\left[V(y_d(x_1^d,Z^d_{1,1})) - V(x_1^d)\right]\frac{\pi_d(\bm{y}^d(\bm{x}^d,Z^d_{1}))q_d(\bm{y}^d(\bm{x}^d,Z^d_1), \bm{x}^d)}{\pi_d(\bm{x}^d)q_d(\bm{x}^d, \bm{y}^d(\bm{x}^d,Z^d_1))}\wedge 1}\\
                       &= d^{2\alpha}\expe{\left[V(y_d(x_1^d,Z^d_{1,1})) - V(x_1^d)\right]\prod_{i=1}^d \exp\left(\phi_d\left(x_i^d, Z^d_{1,i}\right)\right)\wedge 1}.
\label{eq:generator_differentiable}
\end{align}
We also define $ \widetilde{\rmL}_d$ to be a variant of $\rmL_d$ in which the first component of the acceptance ratio is omitted:
\begin{align}
\label{eq:tilde_generator_differentiable}
\widetilde{\rmL}_d V(\bm{x}^d) = d^{2\alpha}\expe{\left[V(y_d(x_1^d,Z^d_{1,1}))
    - V(x_1^d)\right]\prod_{i=2}^d \exp\left(\phi_d\parenthese{x_i^d,
    Z^d_{1,i}}\right)\wedge 1}.
\end{align}
We further define the generator of the Langevin diffusion
\begin{align}
\label{eq:langevin_generator_differentiable}
\rmL V(x) = \frac{h(\ell, r)}{2}\left[V^{\dprime}(x) -g^{\prime}(x) V^{\prime}(x)\right],
\end{align}
where $h(\ell, r) = \ell^2 a(\ell,r)$ is the speed of the diffusion and
$a(\ell, r)=\lim_{d\to \infty} a_d(\ell,r)$ is given in
Theorem~\ref{theo:differentiable}.

We will make use of the derivatives of $g$ in~\eqref{eq:target_differentiable} up to order 8, which we denote by $g^\prime, g^{\dprime}, g^{\dprime\prime}$ and $g^{(k)}$ for all $k=4, \dots, 8$. 
We recall that $( g^\lambda)'$ is Lipschitz continuous with Lipschitz constant
$\lambda^{-1}$ \cite[Proposition 12.19]{rockafellar2009variational} and that $(
g^\lambda)'(x) = \lambda^{-1}(\prox_g^{\lambda}(x)-x)$, hence
$\prox_g^{\lambda}$ is Lipschitz continuous with Lipschitz constant $1$.

\subsection{Identifying the scaling regimes}
The scaling regimes are identified following the approach of \cite{roberts1998optimal}: we approximate $ \phi_d(x,
z)$ with a Taylor expansion about $\sigma_d\to 0$ and find for which values of $d$ this expansion well approximates $ \phi_d(x,
z)$.

We start by decomposing $\phi_d(x, z) = R_1(x, z,
\sigma_d)+R_2(x, z, \sigma_d)$, where
\begin{align}
R_1(x, z, \sigma) &=-g\left[x - \frac{\sigma^2}{2} g^{\prime}\left(\prox_{g}^{\sigma^{2v}r/2}(x)\right)+\sigma z\right]+g(x),\\
R_2(x, z, \sigma) &= \frac{1}{2}z^2
-\frac{1}{2}\left(z-\frac{\sigma}{2}g^\prime\left(\prox_g^{\sigma^{2v}r/2}\left[x+\sigma z
-\frac{\sigma^2}{2}g^\prime\left[\prox_g^{\sigma^{2v}r/2}(x)\right]\right]\right)\right.\\
\label{eq:R}
                  &\qquad\quad \left.-\frac{\sigma}{2}g^\prime\left[\prox_g^{\sigma^{2v}r/2}(x)\right]\right)^2.
\end{align}
We then compute the derivatives of $R_1, R_2$ w.r.t. $\sigma$ at $\sigma = 0$ as shown in Appendix~\ref{app:derivatives_R}. This shows that
\begin{align}
    R_1(x, z, 0)+R_2(x, z, 0)&= 0,\\
 \frac{\partial R_1}{\partial \sigma}(x, z, \sigma)_{\mid \sigma=0}+\frac{\partial R_2}{\partial \sigma}(x, z, \sigma)_{\mid \sigma=0} &=0,\\
\frac{\partial^2 R_1}{\partial \sigma^2}(x, z, \sigma)_{\mid \sigma=0}+\frac{\partial^2 R_2}{\partial \sigma^2}(x, z, \sigma)_{\mid \sigma=0} &=2zg^{\dprime}(x)\frac{\partial}{\partial \sigma}\prox_g^{\sigma^{2v}r/2}(x)_{\mid \sigma=0}.
\end{align}
Looking at the derivatives of the proximity map w.r.t. $\sigma$ in Appendix~\ref{app:derivatives_prox} we find that for $v=1/2$ we have
\begin{align}
\frac{\partial^2 R_1}{\partial \sigma^2}(x, z, \sigma)_{\mid \sigma=0}+\frac{\partial^2 R_2}{\partial \sigma^2}(x, z, \sigma)_{\mid \sigma=0} &=-rzg^{\dprime}(x)g^{\prime}(x).
\end{align}
while for $v>1/2$ we have 
\begin{align}
\frac{\partial^2 R_1}{\partial \sigma^2}(x, z, \sigma)_{\mid \sigma=0}+\frac{\partial^2 R_2}{\partial \sigma^2}(x, z, \sigma)_{\mid \sigma=0} &=0.
\end{align}

Proceeding similarly for higher order derivatives we identify the Taylor expansions in Proposition~\ref{prop:1_case_a}, \ref{prop:1_case_b} and~\ref{prop:1_case_c} below.
\subsection{Auxiliary results for the proof of case~\ref{case:a}}

First, we characterize the limit behaviour of the acceptance ratio~\eqref{eq:ar_finite}.
\begin{prop}
\label{prop:1_case_a}
Under~\Cref{ass:a0}, \Cref{ass:g_smooth} and~\Cref{ass:stationarity}, if $\alpha= 1/4$, $\beta = 1/8$ and $r>0$, then
\begin{enumerate}[label=(\roman*)]
\item the log-acceptance ratio~\eqref{eq:log_acceptance}, when $d\to \infty$, satisfies the following Taylor expansion
\begin{align}
    \phi_d(x, z)
    &= d^{-1/2}C_2(x, z) + d^{-3/4}C_3(x, z)+d^{-1}C_4(x, z)+C_5(x, z, \sigma_d),
\end{align}
where $C_2(x, z)$ is given in~\eqref{eq:C2}, $C_3$ and $C_4$ are
polynomials in $z$ and the derivatives of $g$, such that
$\mathbb{E}[ C_3(X^d_{0,1}, Z^d_{1,1})]=0$ and
$\mathbb{E}[ C_2(X^d_{0,1}, Z^d_{1,1})^2]=-2\mathbb{E}[ C_4(X^d_{0,1}, Z^d_{1,1})]$;
\item there exists sets $F_d \subseteq \real^d$ with $d^{2\alpha}\pi_d(F_d^c)\to 0$ such that 
\begin{align}
\label{eq:expansion_result1}
\lim_{d\to\infty}\sup_{\bm{x}^d\in
    F_d}\expe{\left\lvert\sum_{i=2}^d\phi_d(x_i^d, Z^d_{1,i})
-d^{-1/2}\sum_{i=2}^dC_2(x_i^d,
Z^d_{1,i})+\frac{\ell^4K_1(r)^2}{2}\right\rvert}=0,
\end{align}
where $K_1(r)$ is given in~\Cref{theo:differentiable}--\ref{case:a}.
\end{enumerate}
\end{prop}
\begin{proof}
Take one component of the log-acceptance ratio
\begin{align}
    \phi_d(x, z)= g(x) - g(y_d(x,z))  +\log q(y_d(x,z), x)  -\log q(x,
    y_d(x,z)),
\end{align}
with $y_d(x,z)=  x - \sigma_d^2 g^\prime(\prox_{g}^{\sigma^{2v}_d
r/2}(x))/2+\sigma_d z$. We have that $\phi_d(x, z) = R_1(x, z,
\sigma_d)+R_2(x, z, \sigma_d)$, where
\begin{align}
R_1(x, z, \sigma) &=-g\left[x - \frac{\sigma^2}{2} g^{\prime}\left(\prox_{g}^{\sigma^{2v}r/2}(x)\right)+\sigma z\right]+g(x),\\
R_2(x, z, \sigma) &= \frac{1}{2}z^2
-\frac{1}{2}\left(z-\frac{\sigma}{2}g^\prime\left(\prox_g^{\sigma^{2v}r/2}\left[x+\sigma z
-\frac{\sigma^2}{2}g^\prime\left[\prox_g^{\sigma^{2v}r/2}(x)\right]\right]\right)\right.\\
                  &\qquad\quad \left.-\frac{\sigma}{2}g^\prime\left[\prox_g^{\sigma^{2v}r/2}(x)\right]\right)^2.
\end{align}
Following the approach of \cite{roberts1998optimal} we approximate $ \phi_d(x,
z)$ with a Taylor expansion about $\sigma_d\to 0$.

\begin{enumerate}[label=(\roman*)]
\item Using a Taylor expansion of order 5, we obtain
\begin{align}
\label{eq:expansion1}
\phi_d(x, z) &= d^{-1/2}C_2(x, z) + d^{-3/4}C_3(x, z)+d^{-1}C_4(x, z)+C_5(x, z, \sigma_d),
\end{align}
where
\begin{align}
\label{eq:C2}
C_2(x, z) &= \frac{\ell^2}{2}\left(-rzg^{\dprime}(x)g^{\prime}(x)\right),
\end{align}
$C_3(x, z)$ and $C_4(x, z)$ are given in \Cref{app:taylor_a} and we use the integral form for the remainder
\begin{align}
 C_5(x, z, \sigma_d) = \int_0^{\sigma_d} \left.\frac{\partial ^5}{\partial \sigma^5}R(x, z, \sigma)\right|_{\sigma =u}\frac{(\sigma_d-u)^4}{4!}\rmd u,
\end{align}
with $u$ between 0 and $\sigma_d$ and the derivatives of $R_1$ and $R_2$ given in Appendix~\ref{app:derivatives_R}.
In addition, integrating by parts and using the moments of $Z^d_{1,1}$ we find that
$\mathbb{E}[ C_2(X^d_{0,1}, Z^d_{1,1})]=\mathbb{E}[ C_3(X^d_{0,1},
Z^d_{1,1})]=0$ and
\begin{align}
    2\expe{ C_4(X^d_{0,1}, Z^d_{1,1})}+\expe{ C_2(X^d_{0,1}, Z^d_{1,1})^2}=0.
\end{align}
\item To construct the sets $F_d$, consider, for $j=3, 4$, $F_{d, j} =F_{d, j}^1\cap F_{d, j}^2$ where we define
\begin{align}
    F_{d, j}^1=\left\lbrace \bm{x}^d\in\rset^d: \left\lvert \sum_{i=2}^d
        \expe{C_j(x_i^d,Z^d_{1,i})-C_j(X^d_{0,i},Z^d_{1,i})}\right\rvert \leq
d^{5/8}\right\rbrace,
\end{align}
and
\begin{align}
    F_{d, j}^2=\left\lbrace \bmx^d\in\rset^d : \left\lvert \sum_{i=2}^d V_j(x_i^d)-\expe{V_j(X^d_{0,i})}\right\rvert \leq d^{6/5}\right\rbrace,
\end{align}
where $V_j(x):= \var(C_j(x, Z^d_{1,1}))$.
Using Markov's inequality and the fact that $C_j, V_j$ are bounded by
polynomials since $g$ and its derivatives are bounded by polynomials, it is
easy to show that $d^{1/2}\pi_d((F_{d, j}^1)^c)\to 0$ and $d^{1/2}\pi_d((F_{d,
j}^2)^c)\to 0$, from which follows $d^{1/2}\pi_d(F_{d, j}^c)\to 0$ as $d\to \infty$.
To prove $\mrl^1$ convergence of $C_j$ for $j=3, 4$, observe that
\begin{align}
    &\expe{\left(\sum_{i=2}^d C_j(x_i^d, Z^d_{1,i}) - \expe{C_j(X^d_{0,1},Z^d_{1,1})}\right)^2} \\
    &\quad = \sum_{i=2}^d  V_j(x_i^d) 
    + \left(\sum_{i=2}^d \expe{C_j(x_i^d,Z^d_{1,i})-C_j(X^d_{0,1},Z^d_{1,1})}\right)^2,
\end{align}
and that, for $\bmx^d\in F_{d, j}$, we have
\begin{align}
\expe{\left(\sum_{i=2}^d C_j(x_i^d, Z^d_{1,i}) - \expe{C_j(X^d_{0,1},Z^d_{1,1})}\right)^2}
& \leq \expe{V_j(x_1^d)}(d-1) + d^{6/5}+ d^{5/4}.
\end{align}
Thus, the third and fourth term in the Taylor expansion~\eqref{eq:expansion1} converge in $\mrl^1$ to 0 and $-\ell^4K_1^2(r)/2$ respectively.
Now, consider $C_5(x_i^d, Z^d_{1,i}, \sigma_d)$. We can bound $\frac{\partial^5 R}{\partial \sigma^5}(x, z, \sigma) $ with the derivatives of $g$ evaluated at 
\begin{align}
x+\frac{\sigma^2}{2}\prox_g^{\sigma^{2v}r/2}(x)+\sigma z\qquad\textrm{and}\qquad\prox_g^{\sigma^{2v}r/2}(x).
\end{align}
Under our assumptions, the derivatives of $g$ are bounded by polynomials $M_0$, it follows that there exist polynomials $p$ of the form
$$
A\left(1+\left[\prox_g^{\sigma^{2v}r/2}(x)\right]^N\right)\left(1+z^N\right)\left(1+x^N\right)\left(1+\sigma^N\right),
$$
for sufficiently large $A$ and sufficiently large even integer $N$, such that
\begin{align}
   &\left \lvert g^{(k)}\left[ x+\frac{\sigma^2}{2}\prox_g^{\sigma^{2v}r/2}(x)+\sigma_d z\right]\right\rvert\vee \left \lvert g^{(k)}\left[ \prox_g^{\sigma^{2v}r/2}(x)\right]\right\rvert\\
   &\qquad\qquad\leq  p(\prox_g^{\sigma^{2v}r/2}(x), x, z, \sigma_d).
\end{align}
In addition, $\vert\prox_g^{\sigma^{2v}r/2}(x)\vert \leq C(1+\vert x\vert)$ for some $C\geq 1$, and we can bound
\begin{align}
     p(\prox_g^{\sigma^{2v}r/2}(x), x, z, \sigma) \leq  A\left(1+z^N\right)\left(1+x^{2N}\right)\left(1+\sigma^N\right).
\end{align}
Therefore, we have
\begin{align}
\expe{\left\vert C_5(x_i^d, Z^d_{1,i}, \sigma_d) \right\vert}
&\leq A\expe{1+(Z^d_{1,i})^N}\left(1+(x_i^d)^{2N}\right)\int_0^{\sigma_d} (1+u^N)\frac{(\sigma_d-u)^4}{4!}\rmd u\\
&\leq A\expe{1+(Z^d_{1,i})^N}\left(1+(x_i^d)^{2N}\right)d^{-5/2}\\
&\leq  A\left(1+(x_i^d)^{2N}\right)d^{-5/2},
\end{align}
where the last inequality follows since all the moments of $Z^d_{1}$ are bounded.
Let us denote $p(x)=  A\left(1+x^{2N}\right)$ and
\begin{align}
    F_{d, 5}= \left\lbrace \bm{x}^d\in\real^d: \left\lvert d^{-1}\sum_{i=1}^d p(x_i^d)-\expe{p(X^d_{0,i})}\right\rvert < 1\right\rbrace.
\end{align}
By Chebychev's inequality we have $\pi_d(F_{d, 5}^c)\leq \var(p(X^d_{0,1}))d^{-1}$. Additionally, for all $\bm{x}^d\in F_{d, 5}$,
\begin{align}
    \sum_{i=2}^d \expe{\left\vert C_5(x_i^d, Z^d_{1,i}, \sigma_d) \right\vert}
    &\leq \sum_{i=2}^d d^{-5/2}\parenthese{\expe{p(X^d_{0,1})}+d^{-1}}\\
    &\leq d^{-3/2}\parenthese{\expe{p(X^d_{0,1})}+1}.
\end{align}
Finally, set $F_d= \cap_{j=3}^5 F_{d, j}$. On $F_d$ the last three terms of~\eqref{eq:expansion1} converge uniformly in $\mrl^1$, and~\eqref{eq:expansion_result1} follows using the triangle inequality.
\end{enumerate}
\end{proof}

Next, we compare the generator $\rmL_d $ and $\widetilde{\rmL}_d $ in~\eqref{eq:generator_differentiable} and~\eqref{eq:tilde_generator_differentiable} respectively.
\begin{prop}
\label{prop:2_case_a}
Under~\Cref{ass:a0}, \Cref{ass:g_smooth} and~\Cref{ass:stationarity}, if $\alpha= 1/4$, $\beta = 1/8$ and $r>0$, there exists sets $S_d \subseteq \real^d$ with $d^{2\alpha}\pi_d(S_d^c)\to 0$ such that for any $V\in \rmC_{\rmc}^\infty(\real,\real)$
\begin{align}
    \lim_{d\to \infty}\sup_{\bm{x}^d \in S_d}\left\lvert \rmL_d V(\bm{x}^d)-\widetilde{\rmL}_d V(\bm{x}^d)\right\rvert = 0 ,
\end{align}
and
\begin{align}
    \lim_{d\to \infty}\sup_{\bm{x}^d \in S_d}\expe{\left\lvert
        \parenthese{\exp\parenthese{\sum_{i=1}^d\phi_{d}(x^d_i,Z^d_{1,i})}\wedge 1} -
        \parenthese{\exp\parenthese{\sum_{i=2}^d\phi_{d}(x^d_i,Z^d_{1,i})}\wedge 1}
     \right\rvert} = 0.
\end{align}
\end{prop}
\begin{proof}
The function $x\mapsto \exp(x)\wedge 1$ is Lipschitz continuous with Lipschitz constant 1, hence
\begin{align}
    \left\lvert \rmL_d V(\bm{x}^d)-\widetilde{\rmL}_d V(\bm{x}^d)\right\rvert
    \leq d^{2\alpha}\expe{\left\vert V\left(y_d(x^d_1,Z^d_{1,1})
    \right)-V(x_1^d)\right\vert \vert R(x_1^d, Z^d_{1,1}, \sigma_d)\vert },
\end{align}
where $R(x, z, \sigma) = R_1(x, z, \sigma) + R_2(x, z, \sigma)$ as in~\eqref{eq:R}.
Using a Taylor expansion of order 1 about $\sigma = 0$ with integral remainder:
\begin{align}
R(x, z, \sigma) = R(x, z, 0)+\left.\frac{\partial R}{\partial \sigma}(x, z, \sigma)\right|_{\sigma=0} \sigma + \int_{0}^\sigma \left.\frac{\partial^2 R}{\partial \sigma^2}(x, z, \sigma)\right|_{\sigma =u} (\sigma-u)\rmd u   ,
\end{align}
we obtain 
\begin{align}
    R(x, z, \sigma)  = \int_0^{\sigma}\left.\frac{\partial^2
        R}{\partial \sigma^2}(x, z, \sigma)\right|_{\sigma =u}(\sigma-u)\rmd
        u ,
\end{align}
where $\frac{\partial^2 R}{\partial \sigma^2}(x, z, \sigma) $ is bounded by the derivatives of $g$ evaluated at 
\begin{align}
x+\frac{\sigma^2}{2}\prox_g^{\sigma^{2v}r/2}(x)+\sigma z\qquad\textrm{and}\qquad\prox_g^{\sigma^{2v}r/2}(x).
\end{align}
Under our assumptions, the derivatives of $g$ are bounded by polynomials $M_0$, it follows that there exist polynomials $p$ of the form 
$$A\left(1+\left[\prox_g^{\sigma^{2v}r/2}(x)\right]^N\right)\left(1+z^N\right)\left(1+x^N\right)\left(1+\sigma^N\right),$$ for sufficiently large $A$ and sufficiently large even integer $N$, such that
\begin{align}
   \left \lvert g^{(k)}\left[
   x+\frac{\sigma^2}{2}\prox_g^{\sigma^{2v}r/2}(x)+\sigma
   z\right]\right\rvert\vee \left \lvert g^{(k)}\left[
   \prox_g^{\sigma^{2v}r/2}(x)\right]\right\rvert\leq
   p(\prox_g^{\sigma^{2v}r/2}(x), x, z, \sigma).
\end{align}
Proceeding as in Proposition~\ref{prop:1_case_a}, we can bound
\begin{align}
     p(\prox_g^{\sigma^{2v}r/2}(x), x, z, \sigma) \leq  A\left(1+z^N\right)\left(1+x^{2N}\right)\left(1+\sigma^N\right).
\end{align}
Therefore, we have
\begin{align}
    \left\vert R\left(x_1^d, Z^d_{1,1}, \sigma_d\right) \right\vert 
    &\leq A\left(1+(Z^d_{1,1})^N\right)\left(1+(x_1^d)^{2N}\right)\\
    &\qquad\quad\times\int_0^{\sigma_d} (1+u^N)(\sigma_d-u)\rmd u
    \leq A\left(1+(Z^d_{1,1})^N\right)\left(1+(x_1^d)^{2N}\right)\frac{\sigma_d^2}{2}.
\label{eq:theo3_eq1}
\end{align}
Since $V\in \rmC_\rmc^\infty(\real, \real)$, there exists a constant $C$ such that
\begin{align}
    \left\vert V\left(y_d(x^d_1,Z^d_{1,1})\right)-V(x_1^d)\right\vert 
    &\leq C\left\vert y_d(x^d_1,Z^d_{1,1})-x_1^d \right\vert\\
    &\leq  C\left(\sigma_d\vert Z_{1,1}^d\vert
    +\frac{\sigma_d^2}{2}\left\vert
g^\prime\left(\prox_g^{\sigma_d^{2v}r/2}(x_1^d)\right)\right\vert\right).
\end{align}
Recalling that $g^\prime(\prox_g^{\lambda}(x)) =(g^{\lambda})'(x)$ with
$(g^{\lambda})'$ Lipschitz continuous with Lipschitz constant $\lambda^{-1}$,
we have
\begin{align}
\left\vert g^\prime\left(\prox_g^{\sigma_d^{2v}r/2}(x_1^d)\right)\right\vert \leq \frac{2}{\sigma_d^{2v}r}(1+\vert x_1^d\vert),
\end{align}
and
\begin{align}
\label{eq:theo3_eq2}
  \left\vert V\left(y_d(x^d_1,Z^d_{1,1})\right)-V(x_1^d)\right\vert
  &\leq  C\left(\sigma_d\vert Z^d_{1,1}\vert +\frac{\sigma_d^{2-2m}}{r}\left(1+\vert x_1^d\vert\right)\right)\\
  &\leq  C\sigma_d\left(\vert Z^d_{1,1}\vert +\frac{1}{r}\left(1+\vert x_1^d\vert\right)\right),
\end{align}
since $v= 1/2$.
Combining~\eqref{eq:theo3_eq1} and~\eqref{eq:theo3_eq2} we obtain
\begin{align}
  &d^{2\alpha}\left\vert V\left(y_d(x^d_1,Z^d_{1,1})\right)-V(x_1^d)\right\vert
  \left\vert R(x_1^d, Z^d_{1,1}, \sigma_d)\right\vert \\
  &\qquad\qquad  \leq
  C\sigma_d\left(1+(Z^d_{1,1})^N\right)\left(1+(x_i^d)^{2N}\right)\left(\vert
  Z^d_{1,1}\vert +\frac{1}{r}(1+\vert x_1^d\vert)\right),
\end{align}
for some $C>0$. 

Set $S_d$ to be the set in which $1+(x_1^d)^{2N+1} \leq d^{\alpha/2}$, applying Markov's inequality we obtain
\begin{align}
d^{2\alpha}\pi_d(S_d^c)=
d^{2\alpha}\pi_d\left(\left(1+(x_1^d)^{2N+1}\right)^5\geq d^{5\alpha/2}\right)
\leq d^{-\alpha/2}\expe{\left(1+(x_1^d)^{2N+1}\right)^5}\underset{d\to \infty}{\longrightarrow}0.
\end{align}
Recalling that $\vert Z^d_{1,1}\vert$ and $1+(Z^d_{1,1})^N$ are bounded, we have that
\begin{align}
\sup_{\bm{x}^d \in S_d}\left\lvert \rmL_d V(\bm{x}^d)-\widetilde{\rmL}_d V(\bm{x}^d)\right\rvert \leq Cd^{\alpha/2}\frac{\ell}{d^{\alpha}}
\underset{d\to \infty}{\longrightarrow}0.
\end{align}
The second results follows from~\eqref{eq:theo3_eq1} using the same argument. 
\end{proof}


The following result considers the convergence to the generator of the Langevin diffusion~\eqref{eq:langevin_generator_differentiable}.
\begin{prop}
\label{prop:3_case_a}
Under~\Cref{ass:a0}, \Cref{ass:g_smooth} and~\Cref{ass:stationarity}, if
$\alpha= 1/4$, $\beta = 1/8$ and $r>0$, there exists sets $T_d \subseteq
\real^d$ with $d^{2\alpha}\pi_d(T_d^c)\to 0$ as $d\to \infty$, such that for any $V\in
\rmC_{\rmc}^\infty(\real,\real)$
\begin{align}
    \lim_{d\to \infty}\sup_{\bm{x}^d \in T_d}\left\lvert d^{2\alpha}\expe{
  V\left(y_d(x^d_1,Z^d_{1,1})\right)-V(x_1^d)
} -\frac{\ell^2}{2}(V^{\dprime}(x_1^d)+g^\prime(x_1^d)V^\prime(x_1^d))\right\rvert = 0.
\end{align}
\end{prop}
\begin{proof}
Take
\begin{align}
    y_d(x^d_1,Z^d_{1,1}) = x_1^d+\frac{\sigma_d^2}{2}g^{\prime}\left(\prox_g^{\sigma_d^{2v}r/2}(x_1^d)\right) + \sigma_d Z^d_{1,1},
\end{align}
and use a Taylor expansion of order 2 of
\begin{align}
W(x, z, \sigma) = V\left[x+\frac{\sigma^2}{2}g^{\prime}\left(\prox_g^{\sigma^{2v}r/2}(x)\right) + \sigma z\right],
\end{align}
about $\sigma = 0$ with integral remainder:
\begin{align}
W(x, z, \sigma)& = W(x, z, 0)+\left.\frac{\partial W}{\partial \sigma}(x, z, \sigma)\right|_{\sigma=0} \sigma + \frac{1}{2}\left.\frac{\partial^2 W}{\partial \sigma^2}(x, z, \sigma)\right|_{\sigma=0} \sigma^2 \\
&\quad+ \int_{0}^\sigma \left.\frac{\partial^3 W}{\partial \sigma^3}(x, z, \sigma)\right|_{\sigma =u} \frac{(\sigma-u)^2}{2}\rmd u  .
\end{align}
Using the derivatives
\begin{align}
 W(x, z, 0)&= V(x), \quad
 \left.\frac{\partial W}{\partial \sigma}(x, z, \sigma)\right|_{\sigma=0} =V^\prime(x)z,\\
\frac{\partial^2 W}{\partial \sigma^2}(x, z, \sigma) &= V^{\dprime}(x)z^2+V^\prime(x)g^\prime(x),
\end{align}
and recalling that $\expe{Z^d_{1,1}}=0, \expe{(Z^d_{1,1})^2}=1$, we have
\begin{align}
    \expe{ V\left(y_d(x^d_1,Z^d_{1,1})\right)-V(x_1^d) } 
    &= \frac{\sigma_d^2}{2} \left[V^{\dprime}(x_1^d)+V^\prime(x_1^d)g^\prime(x_1^d)\right] \\
    &\quad+ \expe{\int_{0}^{\sigma_d} \left.\frac{\partial^3 W}{\partial
            \sigma^3}(x_1^d, Z^d_{1,1}, \sigma)\right|_{\sigma =u}
    \frac{(\sigma_d-u)^2}{2}\rmd u}.
\end{align}
Proceeding as in the previous proposition, we can bound
\begin{align}
    \left\lvert \int_{0}^{\sigma_d} \frac{\partial^3 W}{\partial \sigma^3}\left.(x_1^d, Z^d_{1,1}, \sigma)\right|_{\sigma =u} \frac{(\sigma_d-u)^2}{2}\rmd u\right\rvert \leq  A\left(1+(Z^d_{1,1})^N\right)\left(1+(x_i^d)^{2N}\right)d^{-3\alpha}.
\end{align}
Setting $T_d$ to be the set in which $(1+(x_1^d)^{2N}) \leq d^{\alpha/2}$, the
result follows by applying Markov's inequality as in
Proposition~\ref{prop:2_case_a}.
\end{proof}


Before proceeding to stating and proving the last auxiliary result, let us
denote by $\psi_{1}:\real \to [0, +\infty)$ the characteristic function of the
distribution $\mathrm{N}(0, \ell^4K_1^2(r))$, where $K_1^2(r)$ is given
in~\Cref{theo:differentiable}--\ref{case:a},
\begin{align}
\psi_1(t) = \exp(-t^2\ell^{4}K^2_1(r)/2),
\end{align}
and by $\psi_1^d(\bm{x}^d; t) = \int \exp(itw)\mathcal{Q}^d_1(\bm{x}^d;\rmd w)$ the characteristic functions associated with the law
\begin{align}
    \mathcal{Q}_1^d(\bm{x}^d; \cdot) = \mathcal{L}\left\lbrace d^{-1/2}\sum_{i=2}^d C_2(x_i^d, Z^d_{1,i})\right\rbrace.
\end{align}

\begin{prop}
\label{prop:4_case_a}
Under~\Cref{ass:a0}, \Cref{ass:g_smooth} and~\Cref{ass:stationarity}, if $\alpha= 1/4$, $\beta = 1/8$ and $r>0$, there exists a sequence of sets $H_d \subseteq \real^d$ such that 
\begin{enumerate}[label=(\roman*)]
\item $ \lim_{d\to\infty} d^{2\alpha}\pi_d(H_d^c)=0$ ,
\item \label{item:psi} for all $t\in \real$, $\lim_{d\to\infty}\sup_{\bm{x}^d\in H_d}\vert
    \psi_1^d(\bm{x}^d; t)- \psi_1(t)\vert = 0$ ,
\item for all bounded continuous function $\chi : \rset \to \rset$ ,
    \begin{equation}
        \lim_{d\to \infty} \sup_{\bm{x}^d \in H_d} \left\vert
        \int_{\rset} \mathcal{Q}_{1}^d(\bm{x}^d;\rmd u) \chi(u)
            - \parenthese{2\pi \ell^4 K_1^2(r)}^{-1/2}
                \int_{\rset}\chi(u) e^{-u^2/(2\ell^4 K_1^2(r))} \rmd u
                \right\vert= 0 ,
    \end{equation} 
\item in particular,
    \begin{align}
    \lim_{d\to\infty}\sup_{\bm{x}^d\in H_d}\left\vert\expe{1\wedge
\exp\left(d^{-1/2}\sum_{i=2}^d C_2(x_i^d,
Z^d_{1,i})-\frac{\ell^4K_1^2(r)}{2}\right)}-2\Phi\left(-\frac{\ell^2K_1(r)}{2}\right)\right\rvert = 0 ,
\end{align}
where $\Phi$ is the distribution function of the standard normal random variable.
\end{enumerate}
\end{prop}
\begin{proof}
\begin{enumerate}[label=(\roman*)]
\item Define the functions $h_j(x)=\left[-g^{\dprime}(x)g^{\prime}(x)\right]^j$ with $j=1,\dots, 4$ and let $H_d=H_{d, 1}\cap H_{d, 2}$ where 
\begin{align}
H_{d,1} &= \left\lbrace \bm{x}^d\in\real^d:\left\vert \frac{1}{d}\sum_{i=2}^d h_j(x_i^d) - \int_{\rset} h_j(u)\pi(u)\rmd u\right\rvert \leq d^{1/3}\textrm{ for } j=1, \dots, 4\right\rbrace,\\
H_{d,2} &= \left\lbrace \bm{x}^d\in\real^d:\vert h_j(x_i^d)\vert \leq d^{2/3} \textrm{ for } i=1, \dots, d \textrm{ and } j=1, \dots, 4\right\rbrace.
\end{align}
Using Chebychev's inequality, the fact that the derivatives of $g$ are bounded by polynomials and that $\pi$ has finite moments, we have $d^{1/2}\pi_d((H_{d,1})^c)\to 0$ as $d\to \infty$. Similarly, by Markov's inequality we have $d^{1/2}\pi_d((H_{d,2})^c)\to 0$ as $d\to \infty$.
\item 
We follow \cite[Lemma 3(b)]{roberts1998optimal} and decompose 
\begin{align}
\label{eq:characteristic_decomposition}
\vert \psi_1^d(\bm{x}^d; t)- \psi_1(t)\vert &\leq \left\vert \psi_1^d(\bm{x}^d; t)-\prod_{i=2}^d\left(1-\frac{t^2}{2d}v(x_i^d)\right)\right\rvert\\
&+ \left\vert \prod_{i=2}^d\left(1-\frac{t^2}{2d}v(x_i^d)\right)-\prod_{i=2}^d\exp\left(-t^2\frac{v(x_i^d)}{2d}\right)\right\rvert\notag\\
&+ \left\vert \prod_{i=2}^d\exp\left(-t^2\frac{v(x_i^d)}{2d}\right)-\exp\left(-t^2\frac{\ell^4 K_1(r)^2}{2}\right)\right\rvert\notag,
\end{align}
where $v(x_i^d) = \var( C_2(x_i^d, Z^d_{1,i}))=\expeLine{C_2(x^d_i,Z_{1,i}^d)^2}$, where the expectation is taken w.r.t. $Z^d_{1,i}$.
For the first term, decompose the characteristic function $\psi_1^d(\bm{x}^d; t)=\prod_{i=2}^d \theta_i^d(x_i^d; t)$ as the product of the characteristic functions of $d^{-1/2}W_i$ where we define $W_i= C_2(x_i^d, Z^d_{1,i})$, using \cite[equation (3.3.3)]{durrett2019probability} as in the proof of \cite[Theorem 3.4.10]{durrett2019probability} we obtain
\begin{align}
&\left\vert \theta_i^d(x_i^d; t)-\left(1-\frac{t^2}{2d}v(x_i^d)\right)\right\rvert \leq \mathbb{E}\left[\frac{\vert t\vert^3}{d^{3/2}}\frac{\vert W_i\vert^3}{3!}\wedge \frac{2\vert t\vert^2}{d}\frac{\vert W_i\vert^2}{2!}\right]\\
&\qquad\leq \mathbb{E}\left[\frac{\vert t\vert^3}{d^{3/2}}\frac{\vert W_i\vert^3}{3!} ; \vert W_i\vert \leq d^{1/2}\varepsilon\right]+\frac{t^2}{d}\mathbb{E}\left[\vert W_i\vert^2 ; \vert W_i\vert > d^{1/2}\varepsilon\right]\\
&\qquad\leq \frac{\varepsilon \vert t\vert^3}{6d}\mathbb{E}\left[\vert W_i\vert^2\right]+\frac{t^2}{\varepsilon^2 d^2}\mathbb{E}\left[\vert W_i\vert^4\right],
\end{align}
for any $\varepsilon>0$. For sufficiently large $d$, we have that $t^2v(x_i^d)/(2d)\leq 1$ for $\bm{x}\in H_{d, 2}$, and we can use \cite[Lemma 3.4.3]{durrett2019probability}
\begin{align}
&\left\vert \psi_j^d(\bm{x}^d; t)-\prod_{i=2}^d\left(1-\frac{t^2}{2d}v(x_i^d)\right)\right\rvert \leq \sum_{i=2}^d \left(\frac{\varepsilon \vert t\vert^3}{6d}\mathbb{E}\left[\vert W_i\vert^2\right]+\frac{t^2}{\varepsilon^2 d^2}\mathbb{E}\left[\vert W_i\vert^4\right]\right)\\
&\qquad\leq \frac{\varepsilon \ell^4\vert t\vert^3}{6}(K_1^2(r)+D_1 d^{-1/3})+\frac{t^2}{\varepsilon^2d}(\expe{\vert W_i\vert^4}+D_2\ell^8 d^{-1/4}),
\end{align}
where the last inequality follows from the fact that $\bm{x}^d\in H_{d,2}$ and $D_1, D_2$ are positive constants.
For any $\delta>0$ we can chose $\varepsilon$ small enough so that the first term in the above is less than $\delta/2$ and we can chose $d$ sufficiently large to make the second term less than $\delta/2$. Thus, for any $\delta>0$ we can find $\varepsilon>0$ and $d\in\mathbb{N}$ such that
\begin{align}
\left\vert \psi_1^d(\bm{x}^d; t)-\prod_{i=2}^d\left(1-\frac{t^2}{2d}v(x_i^d)\right)\right\rvert <\delta,
\end{align}
the uniform convergence then follows.
The second term in~\eqref{eq:characteristic_decomposition} converges to 0 uniformly for all $\bm{x}^d\in H_{d,1}$; while 
for the third term in~\eqref{eq:characteristic_decomposition} we use again \cite[Lemma 3.4.3]{durrett2019probability}
\begin{align}
\left\vert \prod_{i=2}^d\exp\left(-t^2\frac{v(x_i^d)}{2d}\right)-\exp\left(-t^2\frac{\ell^4 K_1(r)^2}{2}\right)\right\rvert \leq \sum_{i=2}^d \frac{t^4 v(x_i^d)^2}{4d^2},
\end{align}
which goes to zero when $d\to \infty$,
for all $\bm{x}^d\in H_{d, 2}$. The result then follows.
\item Let $\chi : \rset \to \rset$ be a bounded and continuous function. Define
    the sequence $\lbrace \bm{x}^d \,:\, d\in\nsets \rbrace$, where, for any $d
    \in \nsets$, $\bm{x}^d \in H_d$ satisfies, 
    \begin{align}
        &\sup_{y^d \in H_d} \abs{ 
        \int_{\rset} \mathcal{Q}_{1}^d(y_i^d;\rmd u) \chi(u)
             - 
        \parenthese{2 \pi \ell^4 K_1^2(r)}^{-1/2} \int_{\rset}\chi(u) e^{-u^2/(2\ell^4 K_1^2(r))} \rmd u }\\
        & \leq \abs{
        \int_{\rset} \mathcal{Q}_{1}^d(x_i^d;\rmd u) \chi(u)
        - \parenthese{2\pi \ell^4 K_1^2(r)}^{-1/2} \int_{\rset}\chi(u) e^{-u^2/(2\ell^4
        K_1^2(r))} \rmd u } + d^{-1} .
    \end{align}
    Then, using \ref{item:psi} and L\'evy's continuity theorem (e.g. \cite[Theorem 1, page 322]{shiryaev1996probability}), we obtain
    \begin{equation}
        \lim_{d\to \infty} \abs{
        \int_{\rset} \mathcal{Q}_{1}^d(x_i^d;\rmd u) \chi(u)
        - \parenthese{2 \pi \ell^4 K_1^2(r)}^{-1/2} \int_{\rset}\chi(u) e^{-u^2/(2\ell^4
        K_1^2(r))} \rmd u }  = 0 .
    \end{equation}
    Combining this limit with the definition of the sequence $\lbrace\bm{x}^d \,:\, d \in \nsets\rbrace$, we conclude the proof.

\item This statement follows directly from (iii) and \cite[Proposition 2.4]{gelman1997weak}.
\end{enumerate}
\end{proof}

\subsection{Auxiliary results for the proof of case~\ref{case:b}}
First, we characterize the limit behaviour of the acceptance ratio~\eqref{eq:ar_finite}. The following result is an extension of \cite[Lemma 1]{roberts1998optimal}.

\begin{prop}
\label{prop:1_case_b}
Under~\Cref{ass:a0}, \Cref{ass:g_smooth}, \Cref{ass:stationarity} and~\Cref{ass:g_lipschitz}, if $\alpha= 1/6$, $\beta = 1/6$ and $r> 0$, then
\begin{enumerate}[label=(\roman*)]
\item the log-acceptance ratio~\eqref{eq:log_acceptance}, when $d\to \infty$, satisfies the following Taylor expansion
\begin{align}
    \phi_d(x,z) &= d^{-1/2}C_3(x, z)+d^{-2/3}C_4(x, z)\\
    &\qquad\qquad+d^{-5/6}C_5(x, z)+d^{-1}C_6(x, z)+C_7(x, z, \sigma_d),
\end{align}
where $C_3$ is given in~\eqref{eq:C3_case_a}, $C_4, C_5, C_6$
are polynomials in $z$ and the derivatives of $g$, such that
$\mathbb{E}[ C_j(X_{0,1}^d, Z_{1,1}^d)]=0$ for
$j=3, 4, 5$ and $ \mathbb{E}[ C_3(X^d_{0,1},
Z^d_{1,1})^2]=-2\mathbb{E}[C_6(X^d_{0,1}, Z^d_{1,1})] $,
\item there exists sets $F_d \subseteq \real^d$ with $d^{2\alpha}\pi_d(F_d^c)\to 0$ such that 
\begin{align}
    \lim_{d\to\infty}\sup_{\bm{x}^d\in
        F_d}\expe{\left\lvert\sum_{i=2}^d\phi_d(x_i^d, Z^d_{1,i})
-d^{-1/2}\sum_{i=2}^dC_3(x_i^d,
Z^d_{1,i})+\frac{\ell^6K_2(r)^2}{2}\right\rvert}=0,
\end{align}
where $K_2(r)$ is given in~\Cref{theo:differentiable}--\ref{case:b}.
\end{enumerate}
\end{prop}
\begin{proof}
Take one component of the log-acceptance ratio
\begin{align}
    \phi_d(x, z)= g(x) - g(y_d(x,z))  +\log q(y_d(x,z), x)  -\log
    q(x, y_d(x,z)),
\end{align}
with $y_d(x,z) =  x - \sigma_d^2
g^\prime(\prox_{g}^{\sigma^{2v}_dr/2}(x))+\sigma_d z$.
Proceeding as in the proof for case~\ref{case:a}, we have that
$ \phi_d(x, z) = R_1(x, z, \sigma_d)+R_2(x, z, \sigma_d)$ where $R_1, R_2$ are given in~\eqref{eq:R}.
Following the approach of \cite{roberts1998optimal} we approximate $ \phi_d(x, z)$ with a Taylor expansion about $\sigma_d=0$.

\begin{enumerate}[label=(\roman*)]
\item Using a Taylor expansion of order 7, we find that
\begin{align}
\phi_d(x, z) &= d^{-1/2}C_3(x, z)+d^{-2/3}C_4(x, z)+d^{-5/6}C_5(x, z)\\
&+d^{-1}C_6(x, z)+C_7(x, z, \sigma_d),
\end{align}
where
\begin{align}
\label{eq:C3_case_a}
C_3(x, z) &= \frac{\ell^3}{6}\left(\frac{1}{2}g^{\dprime\prime}(x)z^3-\frac{3}{2}g^{\dprime}(x)g^\prime(x)z\left(1+2r\right)\right),
\end{align}
$C_4(x, z), C_5(x, z)$ and $C_6(x, z)$ are given in \Cref{app:taylor_b}
and integral form of the remainder
\begin{align}
 C_7(x, z, \sigma_d) = \int_0^{\sigma_d} \left.\frac{\partial
         ^7}{\partial
 \sigma^7}R(x, z, \sigma)\right|_{\sigma
 =u}\frac{(\sigma_d-u)^6}{6!}\rmd u,
\end{align}
with $u$ between 0 and $\sigma_d$ and the derivatives of $R_1$ and
$R_2$ are given in Appendix~\ref{app:derivatives_R}.
In addition, integrating by parts and using the moments of the standard
normal $Z^d_{1,1}$, we find that $\mathbb{E}[ C_3(X^d_{0,1},
Z^d_{1,1})]=\mathbb{E}[ C_4(X^d_{0,1},
Z^d_{1,1})]=\mathbb{E}[C_5(X^d_{0,1}, Z^d_{1,1})]=0$ and
\begin{align}
    \mathbb{E}\left[C_6(X^d_{0,1}, Z^d_{1,1})\right]
    &=
    -\frac{\ell^6}{16}\left(r+2r^2\right)
    \mathbb{E}\left[\parenthese{g^{\dprime}(X^d_{0,1})g^{\prime}(X^d_{0,1})}^2
    \right]\\
    &\quad -\frac{\ell^6}{16}\left(\frac{1}{2}
    +r\right)\expe{g^{\dprime}(X^d_{0,1})^3}
-\frac{5\ell^6}{96}\expe{g^{\dprime\prime}(X^d_{0,1})^2} ,
\end{align}
which shows that $\mathbb{E}[ C_3(X^d_{0,1}, Z^d_{1,1})^2+
2C_6(X^d_{0,1}, Z^d_{1,1})]=0$.

\item The proof of this result follows using the same steps as case~\ref{case:a} and is analogous to that of \cite[Lemma 1]{roberts1998optimal}.
\end{enumerate}
\end{proof}

Next, we compare the generator $\rmL_d $ and $\widetilde{\rmL}_d $ in~\eqref{eq:generator_differentiable} and~\eqref{eq:tilde_generator_differentiable} respectively, extending \cite[Theorem 3]{roberts1998optimal}.
\begin{prop}
\label{prop:2_case_b}
Under~\Cref{ass:a0}, \Cref{ass:g_smooth}, \Cref{ass:stationarity} and~\Cref{ass:g_lipschitz}, if $\alpha= 1/6$, $\beta = 1/6$ and $r> 0$, there exists sets $S_d \subseteq \real^d$ with $d^{2\alpha}\pi_d(S_d^c)\to 0$ such that for any $V\in \rmC_{\rmc}^\infty(\real,\real)$
\begin{align}
    \lim_{d\to \infty}\sup_{\bm{x}^d \in S_d}\left\lvert \rmL_d V(\bm{x}^d)-\widetilde{\rmL}_d V(\bm{x}^d)\right\rvert = 0,
\end{align}
and
\begin{align}
    \lim_{d\to \infty}\sup_{\bm{x}^d \in S_d}\expe{\left\lvert
        \parenthese{\exp\parenthese{\sum_{i=1}^d\phi_{d}(x^d_i,Z^d_{1,i})}\wedge 1} -
        \parenthese{\exp\parenthese{\sum_{i=2}^d\phi_{d}(x^d_i,Z^d_{1,i})}\wedge 1}
     \right\rvert} = 0.
\end{align}
\end{prop}
\begin{proof}
Proceeding as in Proposition~\ref{prop:2_case_a},
for any $V \in \rmC_\rmc^\infty(\rset,\rset)$,
there exists a constant $C$ such that
\begin{align}
    \left\vert V\left(y_d(x^d_1,Z^d_{1,1})\right)-V(x_1^d)\right\vert 
    &\leq C\left\vert y_d(x^d_1,Z^d_{1,1})-x_1^d \right\vert\\
    &\leq  C\left(\sigma_d\vert Z_{1,1}^d\vert
    +\frac{\sigma_d^2}{2}\left\vert
g^\prime\left(\prox_g^{\sigma_d^{2v}r/2}(x_1^d)\right)\right\vert\right).
\end{align}
Under~\Cref{ass:g_lipschitz}, $g^\prime$ is Lipschitz continuous and we have, for some $C\geq 1$,
\begin{align}
    \abs{ g^\prime\left(\prox_g^{\sigma_d^{2v}r/2}(x_1^d)\right)} \leq
C\parenthese{1+\abs{ \prox_g^{\sigma_d^{2v}r/2}(x_1^d)}}\leq C(1+\vert
    x_1^d\vert),
\end{align}
where we used the fact that $\prox_g^\lambda$ is $1$-Lipschitz continuous for all $\lambda>0$.
The result then follows similarly to \Cref{prop:2_case_a} and
\cite[Theorem 3]{roberts1998optimal}.
\end{proof}

The following result considers the convergence to the generator of the Langevin diffusion~\eqref{eq:langevin_generator_differentiable} and is a generalization of \cite[Lemma 2]{roberts1998optimal}.
\begin{prop}
\label{prop:3_case_b}
Under~\Cref{ass:a0}, \Cref{ass:g_smooth}, \Cref{ass:stationarity}
and~\Cref{ass:g_lipschitz}, if $\alpha= 1/6$, $\beta = 1/6$ and $r> 0$,
there exists sets $T_d \subseteq \real^d$ with
$d^{2\alpha}\pi_d(T_d^c)\to 0$ as $d\to \infty$ such that for any $V\in
\rmC_{\rmc}^\infty(\real,\real)$
\begin{align}
    \lim_{d\to \infty}\sup_{\bm{x}^d \in T_d}\left\lvert
    d^{2\alpha}\expe{V\parenthese{y_d(x^d_1,Z^d_{1,1})}-V(x_1^d)}
-\frac{\ell^2}{2}(V^{\dprime}(x_1^d)+g^\prime(x_1^d)V^\prime(x_1^d))\right\rvert
= 0.
\end{align}
\end{prop}
\begin{proof}
The proof is identical to that of Proposition~\ref{prop:3_case_a}.
\end{proof}

Before proceeding to stating and proving the last auxiliary result, let
us denote by $\psi_{2}:\real \to [0, +\infty)$ the characteristic
function of the distribution $\mathrm{N}(0, \ell^6 K_2^2(r))$, where
$K_2^2(r)$ is given in Theorem~\ref{theo:differentiable}--\ref{case:b},
\begin{align}
\psi_2(t) = \exp(-t^2\ell^{6}K^2_2(r)/2),
\end{align}
and by $\psi_2^d(\bm{x}^d; t) = \int \exp(itw)\mathcal{Q}^d_2(\bm{x}^d;\rmd w)$ the characteristic functions associated with the law
\begin{align}
\mathcal{Q}_2^d(\bm{x}^d; \cdot) = \mathcal{L}\left\lbrace
d^{-1/2}\sum_{i=2}^d C_3(x_i^d, Z^d_{1,i})\right\rbrace.
\end{align}

The following result extends \cite[Lemma 3]{roberts1998optimal}.
\begin{prop}
\label{prop:4_case_b}
Under~\Cref{ass:a0}, \Cref{ass:g_smooth}, \Cref{ass:stationarity} and~\Cref{ass:g_lipschitz}, if $\alpha= 1/6$, $\beta = 1/6$ and $r> 0$, there exists a sequence of sets $H_d \subseteq \real^d$ such that 
\begin{enumerate}[label=(\roman*)]
\item $ \lim_{d\to\infty} d^{2\alpha}\pi_d(H_d^c)=0$ ,
\item \label{item:psi2} for all $t\in \real$, $\lim_{d\to\infty}\sup_{\bm{x}^d\in H_d}\vert
    \psi_2^d(\bm{x}^d; t)- \psi_2(t)\vert = 0$ ,
\item for all bounded continuous function $\chi : \rset \to \rset$ ,
    \begin{equation}
        \lim_{d\to \infty} \sup_{\bm{x}^d \in H_d} \left\vert
        \int_{\rset} \mathcal{Q}^d_2(\bm{x}^d;\rmd u) \chi(u)
            - \parenthese{2\pi \ell^6 K_2^2(r)}^{-1/2}
                \int_{\rset}\chi(u) e^{-u^2/(2\ell^6 K_2^2(r))} \rmd u
                \right\vert= 0 ,
    \end{equation} 
\item in particular,
    \begin{align}
    \lim_{d\to\infty}\sup_{\bm{x}^d\in H_d}\left\vert\expe{1\wedge
\exp\left(d^{-1/2}\sum_{i=2}^d C_3(x_i^d,
Z^d_{1,i})-\frac{\ell^6K_2^2(r)}{2}\right)}-2\Phi\left(-\frac{\ell^3K_2(r)}{2}\right)\right\rvert = 0 ,
\end{align}
where $\Phi$ is the distribution function of the standard normal random variable.
\end{enumerate}
\end{prop}
\begin{proof}
\begin{enumerate}[label=(\roman*)]
\item The proof is analogous to that of Proposition~\ref{prop:4_case_a} and follows the same steps of that of \cite[Lemma 3(a)]{roberts1998optimal}.
\item The proof is analogous to that of Proposition~\ref{prop:4_case_a} and follows the same steps of that of \cite[Lemma 3(b)]{roberts1998optimal}.
\item Following the steps of (iii) in \Cref{prop:4_case_a} and the L\'evy's continuity Theorem (e.g. \cite[Theorem 1, page 322]{shiryaev1996probability}) bring the result.
\item This statement follows directly from (iii) and \cite[Proposition 2.4]{gelman1997weak}.
\end{enumerate}
\end{proof}

\subsection{Auxiliary results for the proof of case~\ref{case:c}}
First, we characterize the limit behaviour of the acceptance ratio~\eqref{eq:ar_finite}.

\begin{prop}
\label{prop:1_case_c}
Under~\Cref{ass:a0}, \Cref{ass:g_smooth}, \Cref{ass:stationarity} and~\Cref{ass:g_lipschitz} and, if $\alpha= 1/6$, $\beta = m/6$ for $v>1$ and $r> 0$, then
\begin{enumerate}[label=(\roman*)]
\item the log-acceptance ratio~\eqref{eq:log_acceptance}, when $d\to \infty$, satisfies the following Taylor expansion
\begin{align}
    \phi_d(x,z) &= d^{-1/2}C_3(x, z)+d^{-2/3}C_4(x, z)\\
    &\qquad\qquad+d^{-5/6}C_5(x, z)+d^{-1}C_6(x, z)+C_7(x, z, \sigma_d),
\end{align}
where $C_3$ is given in~\eqref{eq:C3_case_a}, $C_4, C_5, C_6$
are polynomials in $z$ and the derivatives of $g$, such that
$\mathbb{E}[ C_j(X_{0,1}^d, Z_{1,1}^d)]=0$ for
$j=3, 4, 5$ and $ \mathbb{E}[ C_3(X^d_{0,1},
Z^d_{1,1})^2]=-2\mathbb{E}[C_6(X^d_{0,1}, Z^d_{1,1})] $,
\item there exists sets $F_d \subseteq \real^d$ with $d^{2\alpha}\pi_d(F_d^c)\to 0$ such that 
\begin{align}
    \lim_{d\to\infty}\sup_{\bm{x}^d\in
        F_d}\expe{\left\lvert\sum_{i=2}^d\phi_d(x_i^d, Z^d_{1,i})
-d^{-1/2}\sum_{i=2}^dC_3(x_i^d,
Z^d_{1,i})+\frac{\ell^6K^2_2(0)}{2}\right\rvert}=0,
\end{align}
where $K_2$ is given in~\Cref{theo:differentiable}--\ref{case:b}.
\end{enumerate}
\end{prop}
\begin{proof}
Take one component of the log-acceptance ratio
\begin{align}
    \phi_d(x, z)= g(x) - g(y_d(x,z))  +\log q(y_d(x,z), x)  -\log
    q(x, y_d(x,z)),
\end{align}
with $y_d(x,z) =  x - \sigma_d^2
g^\prime(\prox_{g}^{\sigma^{2v}_dr/2}(x))+\sigma_d z$.
Proceeding as in the proof for case~\ref{case:a}, we have that
$ \phi_d(x, z) = R_1(x, z, \sigma_d)+R_2(x, z, \sigma_d)$ where $R_1, R_2$ are given in~\eqref{eq:R}.
Following the approach of \cite{roberts1998optimal} we approximate $ \phi_d(x, z)$ with a Taylor expansion about $\sigma_d=0$.

\begin{enumerate}[label=(\roman*)]
\item Using a Taylor expansion of order 7, we find that
\begin{align}
 \phi_d(x, z) &= d^{-1/2}C_3(x, z)+d^{-2/3}C_4(x, z)+d^{-5/6}C_5(x, z)\\
&+d^{-1}C_6(x, z)+C_7(x, z, \sigma_d),
\end{align}
where
\begin{align}
\label{eq:C3_case_c}
C_3(x, z) &= 
\frac{\ell^3}{6}\left(\frac{1}{2}g^{\dprime\prime}(x)z^3-\frac{3}{2}g^{\dprime}(x)g^\prime(x)z\right),
\end{align}
$C_4(x, z), C_5(x, z)$ and $C_6(x, z)$ are given in \Cref{app:taylor_c}
and integral form of the remainder
\begin{align}
 C_7(x, z, \sigma_d) = \int_0^{\sigma_d} \left.\frac{\partial
         ^7}{\partial
 \sigma^7}R(x, z, \sigma)\right|_{\sigma
 =u}\frac{(\sigma_d-u)^6}{6!}\rmd u,
\end{align}
with $u$ between 0 and $\sigma_d$ and the derivatives of $R_1$ and $R_2$ are given in Appendix~\ref{app:derivatives_R}.
In addition, integrating by parts and using the moments of $Z^d_{1,1}$ we find that
$\mathbb{E}[ C_3(X^d_{0,1}, Z^d_{1,1})]=\mathbb{E}[ C_4(X^d_{0,1},
Z^d_{1,1})]=\mathbb{E}[ C_5(X^d_{0,1}, Z^d_{1,1})]=0$ and
\begin{align}
\mathbb{E}[ C_6(X^d_{0,1}, Z^d_{1,1})] &= \ell^6\left(
-\frac{1}{32}\expe{g^{\dprime}(X^d_{0,1})^3}-\frac{5}{96}\expe{g^{\dprime\prime}(X^d_{0,1})^2}\right) ,
\end{align}
which shows that $\mathbb{E}[ C_3(X^d_{0,1}, Z^d_{1,1})^2+2C_6(X^d_{0,1}, Z^d_{1,1})]=0$.

\item The proof of this result follows using the same steps as case~\ref{case:a} and is analogous to that of \cite[Lemma 1]{roberts1998optimal}.
\end{enumerate}
\end{proof}

Next, we compare the generator $\rmL_d $ and $\widetilde{\rmL}_d $ in~\eqref{eq:generator_differentiable} and~\eqref{eq:tilde_generator_differentiable} respectively.
\begin{prop}
\label{prop:2_case_c}
Under~\Cref{ass:a0}, \Cref{ass:g_smooth}, \Cref{ass:stationarity} and~\Cref{ass:g_lipschitz}, if $\alpha= 1/6$, $\beta = m/6$ for $v>1$ and $r> 0$, there exists sets $S_d \subseteq \real^d$ with $d^{2\alpha}\pi_d(S_d^c)\to 0$ as $d\to \infty$ such that for any $V\in \rmC_{\rmc}^\infty(\real,\real)$
\begin{align}
    \lim_{d\to \infty}\sup_{\bm{x}^d \in S_d}\left\lvert \rmL_d V(\bm{x}^d)-\widetilde{\rmL}_d V(\bm{x}^d)\right\rvert = 0 ,
\end{align}
and
\begin{align}
    \lim_{d\to \infty}\sup_{\bm{x}^d \in S_d}\expe{\left\lvert
        \parenthese{\exp\parenthese{\sum_{i=1}^d\phi_{d}(x^d_i,Z^d_{1,i})}\wedge 1} -
        \parenthese{\exp\parenthese{\sum_{i=2}^d\phi_{d}(x^d_i,Z^d_{1,i})}\wedge 1}
     \right\rvert} = 0.
\end{align}
\end{prop}
\begin{proof}
The proof is identical to that of Proposition~\ref{prop:2_case_b}.
\end{proof}

The following result considers the convergence to the generator of the Langevin diffusion~\eqref{eq:langevin_generator_differentiable}.
\begin{prop}
\label{prop:3_case_c}
Under~\Cref{ass:a0}, \Cref{ass:g_smooth}, \Cref{ass:stationarity}
and~\Cref{ass:g_lipschitz}, if $\alpha= 1/6$, $\beta = m/6$ for $v>1$ and $r>
0$, there exists sets $T_d \subseteq \real^d$ with $d^{2\alpha}\pi_d(T_d^c)\to
0$ as $d\to \infty$ such that for any $V\in \rmC_{\rmc}^\infty(\real,\real)$
\begin{align}
    \lim_{d\to \infty}\sup_{\bm{x}^d \in T_d}\left\lvert
    d^{2\alpha}\expe{V\parenthese{y_d(x^d_1,Z^d_{1,1})}-V(x_1^d)}
-\frac{\ell^2}{2}(V^{\dprime}(x_1^d)+g^\prime(x_1^d)V^\prime(x_1^d))\right\rvert
= 0.
\end{align}
\end{prop}
\begin{proof}
The proof is identical to that of Proposition~\ref{prop:3_case_a}.
\end{proof}

Before proceeding to stating and proving the last auxiliary result, let us
denote by $\psi_{3}:\real \to [0, +\infty)$ the characteristic function of the
distribution $\mathrm{N}(0, \ell^6 K_2^2(0))$, where $K_2^2$ is given in
Theorem~\ref{theo:differentiable}--\ref{case:b},
\begin{align}
    \psi_3(t) = \exp(-t^2\ell^{6}K_2^2(0)/2),
\end{align}
and by $\psi_3^d(\bm{x}^d; t) = \int \exp(itw)\mathcal{Q}^d_3(\bm{x}^d;\rmd w)$ the characteristic functions associated with the law
\begin{align}
\mathcal{Q}_3^d(\bm{x}^d; \cdot) = \mathcal{L}\left\lbrace d^{-1/2}\sum_{i=2}^d
C_3(x_i^d, Z^d_{1,i})\right\rbrace.
\end{align}

\begin{prop}
\label{prop:4_case_c}
Under~\Cref{ass:a0}, \Cref{ass:g_smooth}, \Cref{ass:stationarity} and~\Cref{ass:g_lipschitz}, if $\alpha= 1/6$, $\beta = m/6$ for $v>1$ and $r> 0$, there exists a sequence of sets $H_d \subseteq \real^d$ such that 
\begin{enumerate}[label=(\roman*)]
\item $ \lim_{d\to\infty} d^{2\alpha}\pi_d(H_d^c)=0$ ,
\item \label{item:psi3} for all $t\in \real$, $\lim_{d\to\infty}\sup_{\bm{x}^d\in H_d}\vert
    \psi_3^d(\bm{x}^d; t)- \psi_3(t)\vert = 0$ ,
\item for all bounded continuous function $\chi : \rset \to \rset$ ,
    \begin{equation}
        \lim_{d\to \infty} \sup_{\bm{x}^d \in H_d} \left\vert
        \int_{\rset} \mathcal{Q}^d_3(\bm{x}^d;\rmd u) \chi(u)
            - \parenthese{2\pi \ell^6 K_2^2(0)}^{-1/2}
                \int_{\rset}\chi(u) e^{-u^2/(2\ell^6 K_2^2(0))} \rmd u
                \right\vert= 0 ,
    \end{equation} 
\item in particular,
    \end{enumerate}
    \begin{equation}
    \lim_{d\to\infty}\sup_{\bm{x}^d\in H_d}\left\vert\expe{1\wedge
\exp\left(d^{-1/2}\sum_{i=2}^d C_3(x_i^d,
Z^d_{1,i})-\frac{\ell^6K_2^2(0)}{2}\right)}-2\Phi\left(-\frac{\ell^3K_2(0)}{2}\right)\right\rvert = 0 ,
\end{equation}
where $\Phi$ is the distribution function of the standard normal random variable.
\end{prop}
\begin{proof}
\begin{enumerate}[label=(\roman*)]
\item The proof is analogous to that of Proposition~\ref{prop:4_case_a} and follows the same steps of that of \cite[Lemma 3(a)]{roberts1998optimal}.
\item The proof is analogous to that of Proposition~\ref{prop:4_case_a} and follows the same steps of that of \cite[Lemma 3(b)]{roberts1998optimal}.
\item Following the steps of (iii) in \Cref{prop:4_case_a} and the L\'evy's continuity Theorem (e.g. \cite[Theorem 1, page 322]{shiryaev1996probability}) bring the result.
\item This statement follows directly from (iii) and \cite[Proposition 2.4]{gelman1997weak}.
\end{enumerate}
\end{proof}

\subsection{Proof of Theorem~\ref{theo:differentiable}}
\begin{proof}[Proof of Theorem~\ref{theo:differentiable}]

\begin{enumerate}[label=(\alph*)]
\item The asymptotic acceptance rate follows by combining Propositions~\ref{prop:1_case_a}--\ref{prop:3_case_a} with part (iv) of Proposition~\ref{prop:4_case_a} as in the proof of \cite[Theorem 1]{roberts1998optimal}.
To prove the weak convergence of the process it suffices to show that there exists events $F_d^\star\in\real^d$ such that for all $t>0$
\begin{align}
    \lim_{d\to \infty}\pr\left(L_t^d\in F_d^\star \textrm{ for all
    }0\leq s\leq t\right)=1,
\end{align}
and
\begin{align}
\lim_{d\to \infty}\sup_{\bm{x}^d\in F_d^\star}\left\vert \rmL_d V(\bm{x}^d)-\rmL V(\bm{x}^d)\right\rvert.
\end{align}
for all $V\in \rmC_{\rmc}^\infty(\real,\real)$ \cite[Chapter 4, Corollary 8.7]{ethier2009markov}.
We take $F_d^\star = F_d\cap S_d\cap T_d \cap H_d$. Then, $d^{2\alpha}\pi_d\left((F_d^\star)^c\right)\to 0$ and
\begin{align}
    \lim_{d\to \infty}\pr\left(L_t^d\in F_d^\star \textrm{ for all }0\leq s\leq t\right)=1,
\end{align}
for all fixed $t$. Combining Propositions~\ref{prop:1_case_a}--\ref{prop:4_case_a} we obtain convergence of the generators.

To obtain the value of $a(\ell, r)$ maximizing the speed, we observe that $K_1^2(r)$ is a function of the ratio $r=c^2/\ell^{2v} = c^2/\ell$ only, we can take $c\propto \ell^{1/2}$ so that $K_1^2(r)$ is constant for given $c$. 
Using the same substitution as in \cite[Theorem 2]{roberts1998optimal} we find that $h(\ell, r)$ is maximized at the unique value of $\ell$ such that $a(\ell, r)= 0.452$.
\item The proof is analogous to that of case~\ref{case:a} replacing Propositions~\ref{prop:1_case_a}, \ref{prop:2_case_a}, \ref{prop:3_case_a} and~\ref{prop:4_case_a} with Propositions~\ref{prop:1_case_b}, \ref{prop:2_case_b}, \ref{prop:3_case_b} and~\ref{prop:4_case_b}. 
To obtain the value of $a(\ell, r)$ maximizing the speed, we observe that $K_2^2(r)$ is a function of the ratio $r=c^2/\ell^{2v}=c^2/\ell^2$ only, we can take $c\propto \ell$ so that $K_2^2(r)$ is constant for given $c$. 
Using the same substitution as in \cite[Theorem 2]{roberts1998optimal} we find that $h(\ell, r)$ is maximized at the unique value of $\ell$ such that $a(\ell, r)= 0.574$.
\item The proof is analogous to that of case~\ref{case:a} replacing
    Propositions~\ref{prop:1_case_a}, \ref{prop:2_case_a},
    \ref{prop:3_case_a} and~\ref{prop:4_case_a} with
    Propositions~\ref{prop:1_case_c}, \ref{prop:2_case_c},
    \ref{prop:3_case_c} and~\ref{prop:4_case_c}. To obtain the value of
    $a(\ell, r)$ maximizing the speed, we observe that $K_2^2(0)$ is
    constant w.r.t. $r$, we can use the same substitution as in
    \cite[Theorem 2]{roberts1998optimal} we find that $h(\ell, r)$ is
    maximized at the unique value of $\ell$ such that $a(\ell, r)=
    0.574$.
\end{enumerate}
\end{proof}


\section{Taylor expansions for the results on regular targets}

\subsection{Coefficients of the Taylor expansion}
We collect here the coefficients of the Taylor expansions in Propositions~\ref{prop:1_case_a}, \ref{prop:1_case_b} and \ref{prop:1_case_c}.
\subsubsection{Case~\ref{case:a}}
\label{app:taylor_a}
If $\alpha= 1/4$, $\beta = 1/8$ and $r>0$, then the log-acceptance ratio~\eqref{eq:log_acceptance} satisfies
\begin{align}
    \phi_d(x, z)
    &= d^{-1/2}C_2(x, z) + d^{-3/4}C_3(x, z)+d^{-1}C_4(x, z)+C_5(x, z, \sigma_d),
\end{align}
where
\begin{align}
C_2(x, z) &= \frac{\ell^2}{2}\left(-rzg^{\dprime}(x)g^{\prime}(x)\right),
\end{align}
and
\begin{align}
C_3(x, z) &= \frac{\ell^3}{6}\left(\frac{z^3}{2}g^{\dprime\prime}(x)-\frac{3}{2}z^2rg^{\prime}(x)g^{\dprime\prime}(x)-\frac{3}{2}rz^2\left[g^{\dprime}(x)\right]^2+\frac{3}{4}zr^2g^{\dprime\prime}(x)\left[g^{\prime}(x)\right]^2\right.\\
&\left.-\frac{3}{2}zg^{\prime}(x)g^{\dprime}(x)+\frac{3}{2}r\left[g^{\prime}(x)\right]^2g^{\dprime}(x) + 3r^2z[g'(x)]^2g^{\dprime}(x)\right),\\
C_4(x, z) &= \frac{\ell^4}{24}\left\lbrace
g^{(4)}(x)\left(z^4-\frac{zr^3}{2}\left[g^{\prime}(x)\right]^3-3z^3rg^{\prime}(x)+\frac{3}{2}z^2r^2\left[g^{\prime}(x)\right]^2\right)\right.\\
&+g^{\prime\dprime}(x)\left(-6z^2g^{\prime}(x)-9rz^3g^{\dprime}(x)+9z^2r^2g^{\prime}(x)g^{\dprime}(x)+9rz\left[g^{\prime}(x)\right]^2\right.\\
&\qquad\left.-\frac{9}{2}zr^3\left[g^{\prime}(x)\right]^2g^{\dprime}(x)-\frac{3}{2}r^2\left[g^{\prime}(x)\right]^3\right)\\
&+12rzg^{\prime}(x)\left[g^{\dprime}(x)\right]^2+3g^{\dprime}(x)\left[g^{\prime}(x)\right]^2-3z^2\left[g^{\dprime}(x)\right]^2+3z^2r^2\left[g^{\dprime}(x)\right]^3\\
&\qquad\left.-6r^2\left[g^{\prime}(x)g^{\dprime}(x)\right]^2-3zr^3g^{\prime}(x)\left[g^{\dprime}(x)\right]^3\right\rbrace,
\end{align}
and we use the integral form for the remainder
\begin{align}
 C_5(x, z, \sigma_d) = \int_0^{\sigma_d} \left.\frac{\partial ^5}{\partial \sigma^5}R(x, z, \sigma)\right|_{\sigma =u}\frac{(\sigma_d-u)^4}{4!}\rmd u,
\end{align}
with $u$ between 0 and $\sigma_d$ and the derivatives of $R_1$ and $R_2$ given in Appendix~\ref{app:derivatives_R}.
\subsubsection{Case~\ref{case:b}}
\label{app:taylor_b}
If $\alpha= 1/6$, $\beta = 1/6$ and $r>0$, the the log-acceptance ratio~\eqref{eq:log_acceptance} satisfies
\begin{align}
    \phi_d(x, z) &= d^{-1/2}C_3(x, z)+d^{-2/3}C_4(x, z)\\
    &\qquad\qquad+d^{-5/6}C_5(x, z)+d^{-1}C_6(x, z)+C_7(x, z, \sigma_d),
\end{align}
where
\begin{align}
C_3(x, z) &= \frac{\ell^3}{6}\left(\frac{1}{2}g^{\dprime\prime}(x)z^3-\frac{3}{2}g^{\dprime}(x)g^\prime(x)z\left(1+2r\right)\right),
\end{align}
and
\begin{align}
C_4(x, z) &= \frac{\ell^4}{24}\left(
z^4 g^{(4)}(x) -6z^2g^{\dprime\prime}(x)g^{\prime}(x)(1+r)\right.\\
&\qquad\qquad\left.-3z^2\left[g^{\dprime}(x)\right]^2(1+2r)+3g^{\dprime}(x)\left[g^{\prime}(x)\right]^2(1+2r)\right),\\
C_5(x, z) &= \frac{\ell^5}{120}\left(
\frac{3}{2}z^5g^{(5)}(x) -15z^3g^{(4)}(x)g^{\prime}(x)(1+r)+15z\left[g^{\prime}(x)\right]^2g^{\dprime\prime}(x)\left(\frac{3}{2}+3r+r^2\right)\right.\\
&\qquad\qquad\left. +15z(1+4r+2r^2)g^{\prime}(x)\left[g^{\dprime}(x)\right]^2 - 15z^3g^{\dprime}(x)g^{\dprime\prime}(x)(1+3r)\right),\\
C_6(x, z) &= \frac{\ell^6}{720}\left(
2z^6g^{(6)}(x)-30\left(1+r\right)z^4g^\prime(x)g^{(5)}(x)+45\left(2+4r+r^2\right)z^2\left[g^\prime(x)\right]^2g^{(4)}(x)\right.\\
&\qquad\qquad+90\left(r+r^2\right)z^2\left[g^{\dprime}(x)\right]^3-15\left(2+6r+3r^2\right)g^{\dprime\prime}(x)\left[g^\prime(x)\right]^3\\
&\qquad\qquad-30\left(1+4r\right)z^4g^{\dprime}(x)g^{(4)}(x) +45\left(3+16r+6r^2\right)z^2g^\prime(x)g^{\dprime}(x)g^{\dprime\prime}(x)\\
&\qquad\qquad\left.-\frac{45}{2}\left(1+4r\right)z^4\left[g^{\dprime\prime}(x)\right]^2-\frac{45}{2}\left(1+8r+8r^2\right)\left[g^\prime(x)g^{\dprime}(x)\right]^2\right),
\end{align}
and integral form of the remainder
\begin{align}
 C_7(x, z, \sigma_d) = \int_0^{\sigma_d} \left.\frac{\partial ^7}{\partial \sigma^7}R(x, z, \sigma)\right|_{\sigma =u}\frac{(\sigma_d-u)^6}{6!}\rmd u,
\end{align}
with $u$ between 0 and $\sigma_d$ and the derivatives of $R_1$ and $R_2$ are given in Appendix~\ref{app:derivatives_R}.

\subsubsection{Case~\ref{case:c}}
\label{app:taylor_c}
If $\alpha= 1/6$, $\beta = v/6$ for $v>1$ and $r> 0$, then the log-acceptance ratio~\eqref{eq:log_acceptance} satisfies
\begin{align}
    \phi_d(x, z) &= d^{-1/2}C_3(x, z)+d^{-2/3}C_4(x, z)\\
    &\qquad\qquad+d^{-5/6}C_5(x, z)+d^{-1}C_6(x, z)+C_7(x, z, \sigma_d),
\end{align}
where
\begin{align}
C_3(x, z) &= 
\frac{\ell^3}{6}\left(\frac{1}{2}g^{\dprime\prime}(x)z^3-\frac{3}{2}g^{\dprime}(x)g^\prime(x)z\right)
\end{align}
and
\begin{align}
C_4(x, z) &= \frac{\ell^4}{24}\begin{cases}
z^4 g^{(4)}(x) -6z^2g^{\dprime\prime}(x)g^{\prime}(x)-3z^2\left[g^{\dprime}(x)\right]^2+3g^{\dprime}(x)\left[g^{\prime}(x)\right]^2 \\
\qquad\qquad- 12zrg^{\prime}(x)g^{\dprime}(x)\qquad\qquad\textrm{if } v=3/2\\
z^4 g^{(4)}(x) -6z^2g^{\dprime\prime}(x)g^{\prime}(x)-3z^2\left[g^{\dprime}(x)\right]^2+3g^{\dprime}(x)\left[g^{\prime}(x)\right]^2 \\
\qquad\qquad\textrm{otherwise }\\
\end{cases},\\
C_5(x, z) &= \frac{\ell^5}{120}\begin{cases}
\frac{3}{2}z^5g^{(5)}(x) -15z^3g^{(4)}(x)g^{\prime}(x) +\frac{45}{2}z\left[g^{\prime}(x)\right]^2g^{\dprime\prime}(x)+15zg^{\prime}(x)\left[g^{\dprime}(x)\right]^2 \\
\qquad - 15z^3g^{\dprime}(x)g^{\dprime\prime}(x)-30z^2r\left[g^{\dprime}(x)\right]^2+30r\left[g^{\prime}(x)\right]^2g^{\dprime}(x)\\
\qquad-30rz^2g^{\prime}(x)g^{\dprime\prime}(x)\\
\qquad\qquad\qquad\textrm{if } v=3/2\\
\frac{3}{2}z^5g^{(5)}(x) -15z^3g^{(4)}(x)g^{\prime}(x) +\frac{45}{2}z\left[g^{\prime}(x)\right]^2g^{\dprime\prime}(x)+15zg^{\prime}(x)\left[g^{\dprime}(x)\right]^2\\
\qquad - 15z^3g^{\dprime}(x)g^{\dprime\prime}(x) -60zrg^{\prime}(x)g^{\dprime}(x)\\
\qquad\qquad\qquad\textrm{if } v=2\\
\frac{3}{2}z^5g^{(5)}(x) -15z^3g^{(4)}(x)g^{\prime}(x) +\frac{45}{2}z\left[g^{\prime}(x)\right]^2g^{\dprime\prime}(x)\\
\qquad+15zg^{\prime}(x)\left[g^{\dprime}(x)\right]^2- 15z^3g^{\dprime}(x)g^{\dprime\prime}(x) \qquad\qquad\qquad\textrm{otherwise }
\end{cases},\\
C_6(x, z) &= \frac{\ell^6}{720}\begin{cases}
 2z^6g^{(6)}(x)-30z^4g^\prime(x)g^{(5)}(x)+90z^2\left[g^\prime(x)\right]^2g^{(4)}(x)-30g^{\dprime\prime}(x)\left[g^\prime(x)\right]^3\\
 \qquad-\frac{45}{2}\left[g^\prime(x)g^{\dprime}(x)\right]^2-30z^4g^{\dprime}(x)g^{(4)}(x) +135z^2g^\prime(x)g^{\dprime}(x)g^{\dprime\prime}(x)\\
\qquad-\frac{45}{2}z^4\left[g^{\dprime\prime}(x)\right]^2-90z^3rg^{\prime}(x)g^{(4)}(x)+540rzg^{\prime}(x)\left[g^{\dprime}(x)\right]^2\\
\qquad+180rzg^{\prime}(x)g^{\dprime}(x)+270rzg^{\dprime\prime}(x)\left[g^{\prime}(x)\right]^2-45zrg^{\dprime}(x)g^{\dprime\prime}(x)\\
\qquad+90rz^3g^{\dprime\prime}(x)\\
\qquad\qquad\qquad\textrm{if } v=3/2\\
 2z^6g^{(6)}(x)-30z^4g^\prime(x)g^{(5)}(x)+90z^2\left[g^\prime(x)\right]^2g^{(4)}(x)-30g^{\dprime\prime}(x)\left[g^\prime(x)\right]^3\\
 \qquad-\frac{45}{2}\left[g^\prime(x)g^{\dprime}(x)\right]^2-30z^4g^{\dprime}(x)g^{(4)}(x) +135z^2g^\prime(x)g^{\dprime}(x)g^{\dprime\prime}(x)\\
\qquad-\frac{45}{2}z^4\left[g^{\dprime\prime}(x)\right]^2-180z^2rg^{\prime}(x)g^{\dprime\prime}(x)-180z^2r\left[g^{\dprime}(x)\right]^2\\
\qquad+180rg^{\dprime}(x)\left[g^{\prime}(x)\right]^2\\
\qquad\qquad\qquad\textrm{if } v=2\\
 2z^6g^{(6)}(x)-30z^4g^\prime(x)g^{(5)}(x)+90z^2\left[g^\prime(x)\right]^2g^{(4)}(x)-30g^{\dprime\prime}(x)\left[g^\prime(x)\right]^3\\
 \qquad-\frac{45}{2}\left[g^\prime(x)g^{\dprime}(x)\right]^2-30z^4g^{\dprime}(x)g^{(4)}(x) +135z^2g^\prime(x)g^{\dprime}(x)g^{\dprime\prime}(x)\\
\qquad-\frac{45}{2}z^4\left[g^{\dprime\prime}(x)\right]^2-360zrg^{\prime}(x)g^{\dprime}(x)\\
\qquad\qquad\qquad\textrm{if } v=5/2\\
 2z^6g^{(6)}(x)-30z^4g^\prime(x)g^{(5)}(x)+90z^2\left[g^\prime(x)\right]^2g^{(4)}(x)-30g^{\dprime\prime}(x)\left[g^\prime(x)\right]^3\\
 \qquad-\frac{45}{2}\left[g^\prime(x)g^{\dprime}(x)\right]^2-30z^4g^{\dprime}(x)g^{(4)}(x) +135z^2g^\prime(x)g^{\dprime}(x)g^{\dprime\prime}(x)\\
\qquad-\frac{45}{2}z^4\left[g^{\dprime\prime}(x)\right]^2\\
\qquad\qquad\qquad\textrm{otherwise}\\
\end{cases},
\end{align}
and integral form of the remainder
\begin{align}
 C_7(x, z, \sigma_d) = \int_0^{\sigma_d} \left.\frac{\partial ^7}{\partial \sigma^7}R(x, z, \sigma)\right|_{\sigma =u}\frac{(\sigma_d-u)^6}{6!}\rmd u,
\end{align}
with $u$ between 0 and $\sigma_d$ and the derivatives of $R_1$ and $R_2$ are given in Appendix~\ref{app:derivatives_R}.
\subsection{Taylor expansions of the log-acceptance ratio}
\label{app:derivatives_R}
\subsubsection{$R_1$}
Recall that $R_1(x, z, \sigma) = -g\left[x - \frac{\sigma^2}{2} g^{\prime}\left(\prox_{g}^{\sigma^{2v}r/2}(x) \right)+ \sigma z\right] + g(x)$. We compute the derivatives of $R_1$ w.r.t. $\sigma$ evaluated at 0:
\begin{align}
 R_1(x, z, 0)&= 0,\\
 \frac{\partial R_1}{\partial \sigma}(x, z, \sigma)_{\mid \sigma=0} &=-g^\prime(x)z,\\
\frac{\partial^2 R_1}{\partial \sigma^2}(x, z, \sigma)_{\mid \sigma=0} &=-z^2g^{\dprime}(x)+\left[g^\prime(x)\right]^2,\\
\frac{\partial^3 R_1}{\partial \sigma^3}(x, z, \sigma)_{\mid \sigma=0}&= -z^3g^{\dprime\prime}(x)+3g^\prime(x)g^{\dprime}(x)\left[z+\frac{\partial}{\partial \sigma}\prox_g^{\sigma^{2v}r/2}(x)_{\mid \sigma =0}\right],\\
\frac{\partial^4 R_1}{\partial \sigma^4}(x, z, \sigma)_{\mid \sigma=0} &= -z^4g^{ (4)}(x)+6z^2g^{\dprime\prime}(x)g^{\prime}(x)-3g^{\dprime}(x)\left[g^{\prime}(x)\right]^2\\
&+12z\left[g^{\dprime}(x)\right]^2\frac{\partial}{\partial \sigma}\prox_g^{\sigma^{2v}r/2}(x)_{\mid \sigma =0}\\
&+6g^\prime(x)\\
&\times\left[g^{\dprime}(x)\frac{\partial^2}{\partial \sigma^2}\prox_g^{\sigma^{2v}r/2}(x)_{\mid \sigma =0}+\left(\frac{\partial}{\partial \sigma}\prox_g^{\sigma^{2v}r/2}(x)_{\mid \sigma =0}\right)^2g^{\dprime\prime}(x)\right].
\end{align}
In addition, for $v>1/2$ we will also use
\begin{align}
\frac{\partial^5 R_1}{\partial \sigma^5}(x, z, \sigma)_{\mid \sigma=0} &=  -z^5g^{ (5)}(x)+10z^3g^{(4)}(x)g^{\prime}(x)-15zg^{\dprime\prime}(x)\left[g^{\prime}(x)\right]^2\\
&+30z\left[g^{\dprime}(x)\right]^2\frac{\partial^2}{\partial^2 \sigma}\prox_g^{\sigma^{2v}r/2}(x)_{\mid \sigma =0} \\
&+ 10g^{\prime}(x)g^{\dprime}(x)\frac{\partial^3}{\partial^3 \sigma}\prox_g^{\sigma^{2v}r/2}(x)_{\mid \sigma =0},\\
\frac{\partial^6 R_1}{\partial \sigma^6}(x, z, \sigma)_{\mid \sigma=0} &=-z^6g^{ (6)}(x)+15z^4g^{(5)}(x)g^{\prime}(x)-45z^2g^{(4)}(x)\left[g^{\prime}(x)\right]^2+15g^{\dprime\prime}(x)\left[g^{\prime}(x)\right]^3\\
&-90g^{\prime}(x)\left[g^{\dprime}(x)\right]^2\frac{\partial^2}{\partial \sigma^2}\prox_g^{\sigma^{2v}r/2}(x)_{\mid \sigma =0}\\
&-60zg^{\dprime}(x)\frac{\partial^3}{\partial \sigma^3}\prox_{g}^{\sigma^{2v}r/2}(x)_{\mid \sigma=0}\\
&+90g^{\dprime}(x)g^{\dprime\prime}(x)\frac{\partial^2}{\partial \sigma^2}\prox_{g}^{\sigma^{2v}r/2}(x)_{\mid \sigma=0}\\
&+15g^{\prime}(x)\\
&\times\left(g^{\dprime}(x)\frac{\partial^4}{\partial \sigma^4}\prox_{g}^{\sigma^{2v}r/2}(x)_{\mid \sigma=0}+3\left[\frac{\partial^2}{\partial \sigma^2}\prox_{g}^{\sigma^{2v}r/2}(x)_{\mid \sigma=0}\right]^2g^{\prime\dprime}(x)\right).
\end{align}

\subsubsection{$R_2$}

Recall that 
\begin{align}
       & R_2(x, z, \sigma) =
        \frac{1}{2}z^2 \notag\\
        &-\frac{1}{2}\left(z-\frac{\sigma}{2}g^\prime\left(\prox_g^{\sigma^{2v}r/2}\left[x+\sigma z -\frac{\sigma^2}{2}g^\prime\left[\prox_g^{\sigma^{2v}r/2}(x)\right]\right]\right)-\frac{\sigma}{2}g^\prime\left[\prox_g^{\sigma^{2v}r/2}(x)\right]\right)^2.\notag
\end{align}
We compute the derivatives of $R_2$ w.r.t. $\sigma$ evaluated at 0:
\begin{align}
R_2(x, z, 0)&= 0,\\
\frac{\partial R_2}{\partial \sigma}(x, z, \sigma)_{\vert \sigma =0} &= z g^\prime(x),\\
\frac{\partial^2 R_2}{\partial \sigma^2}(x, z, \sigma)_{\vert \sigma =0} &= -\left[ g^\prime(x)\right]^2+zg^{\dprime}(x)\left[z+2\frac{\partial}{\partial \sigma}\prox_g^{\sigma^{2v}r/2}(x)_{\mid \sigma=0}\right],\\
\frac{\partial^3 R_2}{\partial \sigma^3}(x, z, \sigma)_{\vert \sigma =0} &= -3g^{\prime}(x)g^{\dprime}(x)\left[z+2\frac{\partial}{\partial \sigma}\prox_g^{\sigma^{2v}r/2}(x)_{\mid \sigma=0}\right]\\
&+\frac{3}{2}z\left(\left[z+\frac{\partial}{\partial \sigma}\prox_g^{\sigma^{2v}r/2}(x)_{\mid \sigma=0}\right]^2g^{\dprime\prime}(x) \right.\\
&\left.+\left[-g^{\prime}(x) +2z\frac{\partial^2}{\partial \sigma\partial x}\prox_g^{\sigma^{2v}r/2}(x)_{\mid \sigma=0} + \frac{\partial^2}{\partial \sigma^2}\prox_g^{\sigma^{2v}r/2}(x)_{\mid \sigma=0}\right]g^{\dprime}(x) \right.\\
&\left. +\left[\frac{\partial}{\partial \sigma}\prox_g^{\sigma^{2v}r/2}(x)_{\mid \sigma=0}\right]^2 g^{\prime\dprime}(x)+\frac{\partial^2}{\partial \sigma^2}\prox_g^{\sigma^{2v}r/2}(x)_{\mid \sigma=0}g^{\dprime}(x)\right),\\
\frac{\partial^4 R_2}{\partial \sigma^4}(x, z, \sigma)_{\vert \sigma =0} &= -3\left[g^{\dprime}(x)\right]^2\left[z+2\frac{\partial}{\partial \sigma}\prox_g^{\sigma^{2v}r/2}(x)_{\mid \sigma=0}\right]^2\\
&-6g^{\prime}(x)\left(\left[z+\frac{\partial}{\partial \sigma}\prox_g^{\sigma^{2v}r/2}(x)_{\mid \sigma=0}\right]^2g^{\dprime\prime}(x) \right.\\
&\left.+\left[-g^{\prime}(x) +2z\frac{\partial^2}{\partial \sigma\partial x}\prox_g^{\sigma^{2v}r/2}(x)_{\mid \sigma=0} + \frac{\partial^2}{\partial \sigma^2}\prox_g^{\sigma^{2v}r/2}(x)_{\mid \sigma=0}\right]g^{\dprime}(x) \right.\\
&\left. +\left[\frac{\partial}{\partial \sigma}\prox_g^{\sigma^{2v}r/2}(x)_{\mid \sigma=0}\right]^2 g^{\prime\dprime}(x)+\frac{\partial^2}{\partial \sigma^2}\prox_g^{\sigma^{2v}r/2}(x)_{\mid \sigma=0}g^{\dprime}(x)\right)\\
&+2zg^{(4)}(x)\left[z+\frac{\partial}{\partial \sigma}\prox_g^{\sigma^{2v}r/2}(x)_{\mid \sigma=0}\right]^3\\
&+6zg^{\dprime\prime}(x)\left[z+\frac{\partial}{\partial \sigma}\prox_g^{\sigma^{2v}r/2}(x)_{\mid \sigma=0}\right]\\
&\qquad\times\left[-g^{\prime}(x)+2z\frac{\partial^2}{\partial \sigma\partial x}\prox_g^{\sigma^{2v}r/2}(x)_{\mid \sigma=0}+\frac{\partial^2}{\partial \sigma^2}\prox_g^{\sigma^{2v}r/2}(x)_{\mid \sigma=0}\right]\\
&+2zg^{\dprime}(x)\left[-3g^{\dprime}(x)\frac{\partial}{\partial \sigma}\prox_g^{\sigma^{2v}r/2}(x)_{\mid \sigma=0}-3g^{\prime}(x)\frac{\partial^2}{\partial \sigma\partial x}\prox_g^{\sigma^{2v}r/2}(x)_{\mid \sigma=0}\right.\\
&\left.\qquad +3z^2\frac{\partial^3}{\partial \sigma\partial x^2}\prox_g^{\sigma^{2v}r/2}(x)_{\mid \sigma=0}+3z\frac{\partial^3}{\partial \sigma^2\partial x}\prox_g^{\sigma^{2v}r/2}(x)_{\mid \sigma=0}\right.\\
&\qquad\left.+\frac{\partial^3}{\partial \sigma^3}\prox_g^{\sigma^{2v}r/2}(x)_{\mid \sigma=0}\right]\\
&+2zg^{(4)}(x)\left[\frac{\partial}{\partial \sigma}\prox_g^{\sigma^{2v}r/2}(x)_{\mid \sigma=0}\right]^3+2zg^{\dprime}(x)\frac{\partial^3}{\partial \sigma^3}\prox_g^{\sigma^{2v}r/2}(x)_{\mid \sigma=0}\\
&+6z\frac{\partial^2}{\partial \sigma^2}\prox_g^{\sigma^{2v}r/2}(x)_{\mid \sigma=0}\frac{\partial}{\partial \sigma}\prox_g^{\sigma^{2v}r/2}(x)_{\mid \sigma=0}g^{\dprime\prime}(x).
\end{align}
We then proceed to get the derivatives needed for $v>1/2$:
\begin{align}
\frac{\partial^5 R_2}{\partial \sigma^5}(x, z, \sigma)_{\vert \sigma =0} &= -15zg^{\dprime}(x)\left(z^2g^{\dprime\prime}(x) +\left[-g^{\prime}(x) +2\frac{\partial^2}{\partial \sigma^2}\prox_g^{\sigma^{2v}r/2}(x)_{\mid \sigma=0}\right]g^{\dprime}(x) \right)\\
&-5g^{\prime}(x)\left(2z^3g^{(4)}(x)+6zg^{\dprime\prime}(x)\left[-g^{\prime}(x)+\frac{\partial^2}{\partial \sigma^2}\prox_g^{\sigma^{2v}r/2}(x)_{\mid \sigma=0}\right]\right.\\
&\left.\qquad+2g^{\dprime}(x)\left[3z\frac{\partial^3}{\partial\sigma^2 \partial x}\prox_g^{\sigma^{2v}r/2}(x)_{\mid \sigma=0}+\frac{\partial^3}{\partial \sigma^3}\prox_g^{\sigma^{2v}r/2}(x)_{\mid \sigma=0}\right]\right.\\
&\left.\qquad+2g^{\dprime}(x)\frac{\partial^3}{\partial \sigma^3}\prox_g^{\sigma^{2v}r/2}(x)_{\mid \sigma=0}\right)\\
&+\frac{5}{2}g^{(5)}(x)z^5+\frac{5}{2}zg^{\dprime}(x)\frac{\partial^4}{\partial \sigma^4}\prox_g^{\sigma^{2v}r/2}(x)_{\mid \sigma=0}\\
&+ 15g^{(4)}(x)z^3\left[-g^{\prime}(x)+\frac{\partial^2}{\partial \sigma^2}\prox_g^{\sigma^{2v}r/2}(x)_{\mid \sigma=0}\right]\\
&+\frac{15}{2}zg^{\dprime\prime}(x)\left[-g^{\prime}(x)+\frac{\partial^2}{\partial \sigma^2}\prox_g^{\sigma^{2v}r/2}(x)_{\mid \sigma=0}\right]^2\\
&+\frac{15}{2}zg^{\dprime\prime}(x)\left[\frac{\partial^2}{\partial \sigma^2}\prox_g^{\sigma^{2v}r/2}(x)_{\mid \sigma=0}\right]^2\\
&+10z^2 g^{\dprime\prime}(x)\left(3z\frac{\partial^3}{\partial \sigma^2\partial x}\prox_g^{\sigma^{2v}r/2}(x)_{\mid \sigma=0}+ \frac{\partial^3}{\partial \sigma^3}\prox_g^{\sigma^{2v}r/2}(x)_{\mid \sigma=0}\right)\\
&+\frac{5}{2}zg^{\dprime}(x)
\left(\frac{\partial^4}{\partial \sigma^4}\prox_g^{\sigma^{2v}r/2}(x)_{\mid \sigma=0}-6g^{\dprime}(x)\frac{\partial^2}{\partial \sigma^2}\prox_g^{\sigma^{2v}r/2}(x)_{\mid \sigma=0}\right.\\
&\qquad\left.-6g^{\prime}(x)\frac{\partial^3}{\partial \sigma^2\partial x}\prox_g^{\sigma^{2v}r/2}(x)_{\mid \sigma=0}+6z^2\frac{\partial^4}{\partial \sigma^2\partial x^2}\prox_g^{\sigma^{2v}r/2}(x)_{\mid \sigma=0}\right.\\
&\left.\qquad+4z\frac{\partial^4}{\partial \sigma^3\partial x}\prox_g^{\sigma^{2v}r/2}(x)_{\mid \sigma=0}\right)
\end{align}
and
\begin{align}
\frac{\partial^6 R_2}{\partial \sigma^6}(x, z, \sigma)_{\vert \sigma =0} &= -\frac{45}{2}\left(z^2g^{\dprime\prime}(x) +\left[-g^{\prime}(x) +2\frac{\partial^2}{\partial \sigma^2}\prox_g^{\sigma^{2v}r/2}(x)_{\mid \sigma=0}\right]g^{\dprime}(x) \right)^2\\
&-15zg^{\dprime}(x)\left( 2z^3g^{(4)}(x)+6zg^{\dprime\prime}(x)
\left[-g^{\prime}(x)+\frac{\partial^2}{\partial \sigma^2}\prox_g^{\sigma^{2v}r/2}(x)_{\mid \sigma=0}\right]\right.\\
&\qquad+2g^{\dprime}(x)\left[3z\frac{\partial^3}{\partial\sigma^2 \partial x}\prox_g^{\sigma^{2v}r/2}(x)_{\mid \sigma=0}+\frac{\partial^3}{\partial \sigma^3}\prox_g^{\sigma^{2v}r/2}(x)_{\mid \sigma=0}\right]\\
&\qquad\left.+2g^{\dprime}(x)\frac{\partial^3}{\partial \sigma^3}\prox_g^{\sigma^{2v}r/2}(x)_{\mid \sigma=0}\right)\\
&+6g^{\prime}(x)A^{(5)} - zA^{(6)},
\end{align}
with
\begin{align}
A^{(5)} &= -\frac{5}{2}g^{(5)}(x)z^4-\frac{5}{2}g^{\dprime}(x)\frac{\partial^4}{\partial \sigma^4}\prox_g^{\sigma^{2v}r/2}(x)_{\mid \sigma=0}\\
&- 15g^{(4)}(x)z^2\left[-g^{\prime}(x)+\frac{\partial^2}{\partial \sigma^2}\prox_g^{\sigma^{2v}r/2}(x)_{\mid \sigma=0}\right]\\
&-\frac{15}{2}g^{\dprime\prime}(x)\left[-g^{\prime}(x)+\frac{\partial^2}{\partial \sigma^2}\prox_g^{\sigma^{2v}r/2}(x)_{\mid \sigma=0}\right]^2-\frac{15}{2}g^{\dprime\prime}(x)\left[\frac{\partial^2}{\partial \sigma^2}\prox_g^{\sigma^{2v}r/2}(x)_{\mid \sigma=0}\right]^2\\
&-10z g^{\dprime\prime}(x)\left(3z\frac{\partial^3}{\partial \sigma^2\partial x}\prox_g^{\sigma^{2v}r/2}(x)_{\mid \sigma=0}+ \frac{\partial^3}{\partial \sigma^3}\prox_g^{\sigma^{2v}r/2}(x)_{\mid \sigma=0}\right)\\
&-\frac{5}{2}g^{\dprime}(x)
\left(\frac{\partial^4}{\partial \sigma^4}\prox_g^{\sigma^{2v}r/2}(x)_{\mid \sigma=0}-6g^{\dprime}(x)\frac{\partial^2}{\partial \sigma^2}\prox_g^{\sigma^{2v}r/2}(x)_{\mid \sigma=0}\right.\\
&\qquad\left.-6g^{\prime}(x)\frac{\partial^3}{\partial \sigma^2\partial x}\prox_g^{\sigma^{2v}r/2}(x)_{\mid \sigma=0}\right.\\
&\qquad\left.+6z^2\frac{\partial^4}{\partial \sigma^2\partial x^2}\prox_g^{\sigma^{2v}r/2}(x)_{\mid \sigma=0}+4z\frac{\partial^4}{\partial \sigma^3\partial x}\prox_g^{\sigma^{2v}r/2}(x)_{\mid \sigma=0}\right),\\
A^{(6)} &= -3\left(10g^{\dprime\prime}(x)\frac{\partial^2}{\partial \sigma^2}\prox_g^{\sigma^{2v}r/2}(x)_{\mid \sigma=0}\frac{\partial^3}{\partial \sigma^3}\prox_g^{\sigma^{2v}r/2}(x)_{\mid \sigma=0}\right.\\
&\left.+g^{\dprime}(x)\frac{\partial^5}{\partial \sigma^5}\prox_g^{\sigma^{2v}r/2}(x)_{\mid \sigma=0}+g^{(6)}(x)z^5\right.\\
& + 10g^{(5)}(x)z^3\left[-g^{\prime}(x)+\frac{\partial^2}{\partial \sigma^2}\prox_g^{\sigma^{2v}r/2}(x)_{\mid \sigma=0}\right]\\
&+ 15zg^{(4)}(x)\left[-g^{\prime}(x)+\frac{\partial^2}{\partial \sigma^2}\prox_g^{\sigma^{2v}r/2}(x)_{\mid \sigma=0}\right]^2\\
&+ 10g^{(4)}(x)z^2\left[\frac{\partial^3}{\partial \sigma^3}\prox_g^{\sigma^{2v}r/2}(x)_{\mid \sigma=0}+3z\frac{\partial^3}{\partial \sigma^2\partial x}\prox_g^{\sigma^{2v}r/2}(x)_{\mid \sigma=0}\right]\\
&+ 10g^{\dprime\prime}(x)\left[-g^{\prime}(x)+\frac{\partial^2}{\partial \sigma^2}\prox_g^{\sigma^{2v}r/2}(x)_{\mid \sigma=0}\right]\\
&\qquad\times\left[\frac{\partial^3}{\partial \sigma^3}\prox_g^{\sigma^{2v}r/2}(x)_{\mid \sigma=0}+3z\frac{\partial^3}{\partial \sigma^2\partial x}\prox_g^{\sigma^{2v}r/2}(x)_{\mid \sigma=0}\right]\\
&+5g^{\dprime\prime}(x)z\left[-6g^{\dprime}(x)\frac{\partial^2}{\partial \sigma^2}\prox_g^{\sigma^{2v}r/2}(x)_{\mid \sigma=0}-6g^{\prime}(x)\frac{\partial^3}{\partial \sigma^2\partial x}\prox_g^{\sigma^{2v}r/2}(x)_{\mid \sigma=0}\right.\\
&\qquad\left. +6z^2\frac{\partial^4}{\partial \sigma^2\partial x^2}\prox_g^{\sigma^{2v}r/2}(x)_{\mid \sigma=0}+3z\frac{\partial^4}{\partial \sigma^3\partial x}\prox_g^{\sigma^{2v}r/2}(x)_{\mid \sigma=0}\right.\\
&\qquad\left.+\frac{\partial^4}{\partial \sigma^4}\prox_g^{\sigma^{2v}r/2}(x)_{\mid \sigma=0}\right]\\
&+g^{\dprime}(x)\left[-10g^{\dprime}(x)\frac{\partial^3}{\partial \sigma^3}\prox_g^{\sigma^{2v}r/2}(x)_{\mid \sigma=0}-10g^{\prime}(x)\frac{\partial^4}{\partial \sigma^3\partial x}\prox_g^{\sigma^{2v}r/2}(x)_{\mid \sigma=0}\right.\\
&\qquad +5z\frac{\partial^5}{\partial \sigma^4\partial x}\prox_g^{\sigma^{2v}r/2}(x)_{\mid \sigma=0}+10z^2\frac{\partial^5}{\partial \sigma^3\partial x^2}\prox_g^{\sigma^{2v}r/2}(x)_{\mid \sigma=0}\\
&\qquad+10z^3\frac{\partial^5}{\partial \sigma^2\partial x^3}\prox_g^{\sigma^{2v}r/2}(x)_{\mid \sigma=0}\\
&\qquad\left.\left. -30g^{\prime}(x)z\frac{\partial^4}{\partial \sigma^2\partial x^2}\prox_g^{\sigma^{2v}r/2}(x)_{\mid \sigma=0}+\frac{\partial^5}{\partial \sigma^5}\prox_g^{\sigma^{2v}r/2}(x)_{\mid \sigma=0}\right]\right).
\end{align}

\subsection{Derivatives of the proximity map for regular targets}
\label{app:derivatives_prox}

Recall that, in the regular case, $\prox_g^{\sigma^{2v}r/2}(x)=-\frac{\sigma^{2v}r}{2} g^\prime(\prox_g^{\sigma^{2v}r/2}(x))+x$
then
\begin{align}
&\frac{\partial}{\partial \sigma} \prox_g^{\sigma^{2v}r/2}(x)_{|\sigma =0} =\begin{cases} 
-\frac{r}{2}g^{\prime}(x)\qquad \textrm{if } v= 1/2\\
0\qquad \textrm{if } v> 1/2\\
\infty \qquad \textrm{otherwise}
\end{cases}\\
&\frac{\partial^2}{\partial \sigma^2} \prox_g^{\sigma^{2v}r/2}(x)_{|\sigma =0} =\begin{cases}
\frac{r^2}{2}g^{\prime}(x)g^{\dprime}(x)\qquad \textrm{if } v = 1/2\\
-rg^{\prime}(x)\qquad \textrm{if } v = 1\\
0 \qquad \textrm{if } m > 1\\
\infty \qquad \textrm{otherwise}
\end{cases}\\
&\frac{\partial^3}{\partial \sigma^3}\prox_g^{\sigma^{2v}r/2}(x)_{|\sigma =0} = 
\begin{cases}
-\frac{3r^3}{8}g^{\dprime\prime}(x)\left[g^{\prime}(x)\right]^2-\frac{3r^3}{4}g^{\prime}(x)\left[g^{\dprime}(x)\right]^2\qquad \textrm{if } v = 1/2\\
-3rg^{\prime}(x)\qquad \textrm{if } v = 3/2\\
0 \qquad \textrm{if }  v=1, m > 3/2\\
\infty \qquad \textrm{otherwise}
\end{cases}\\
&\frac{\partial^4}{\partial \sigma^4}\prox_g^{\sigma^{2v}r/2}(x)_{|\sigma =0} = 
\begin{cases}
<\infty\qquad \textrm{if } v = 1/2\\
6r^2g^{\prime}(x)g^{\dprime}(x)\qquad \textrm{if } v = 1\\
-12rg^{\prime}(x)\qquad \textrm{if } v = 2\\
0 \qquad \textrm{if }v=3/2, v>2\\
\infty \qquad \textrm{otherwise}
\end{cases}\\
&\frac{\partial^5}{\partial \sigma^5}\prox_g^{\sigma^{2v}r/2}(x)_{|\sigma =0} = 
\begin{cases}
<\infty\qquad \textrm{if } v = 1/2\\
-60rg^{\prime}(x)\qquad \textrm{if } v = 5/2\\
0 \qquad \textrm{if } v=1, v=3/2, v=2, v>5/2\\
\infty \qquad \textrm{otherwise}
\end{cases}\\
\end{align}
and
\begin{align}
&\frac{\partial}{\partial x} \prox_g^{\sigma^{2v}r/2}(x)_{|\sigma =0} =1,\qquad\qquad\qquad\frac{\partial^{(k)}}{\partial x^{(k)}} \prox_g^{\sigma^{2v}r/2}(x)_{|\sigma =0} =0,
\end{align}
for all integers $k>1$.
For the mixed derivatives we have
\begin{align}
&\frac{\partial^2}{\partial \sigma \partial x} \prox_g^{\sigma^{2v}r/2}(x)_{|\sigma =0} =\begin{cases}
-\frac{r}{2}g^{\dprime}(x)\qquad\textrm{if } v= 1/2\\
0\qquad\textrm{if }  v>1/2\\
\infty\qquad\textrm{otherwise}
\end{cases}\\
&\frac{\partial^3}{\partial \sigma^2 \partial x} \prox_g^{\sigma^{2v}r/2}(x)_{|\sigma =0} =\begin{cases}
\frac{r^2}{2}\left[g^{\dprime}(x)\right]^2+\frac{r^2}{2}g^{\prime}(x)g^{\dprime\prime}(x)\qquad\textrm{if } v= 1/2\\
-rg^{\dprime}(x)\qquad\textrm{if } v= 1\\
0\qquad\textrm{if }  v>1/2\\
\infty\qquad\textrm{otherwise}
\end{cases}\\
&\frac{\partial^4}{\partial \sigma^3 \partial x} \prox_g^{\sigma^{2v}r/2}(x)_{|\sigma =0} =\begin{cases}
<\infty \qquad\textrm{if } v= 1/2\\
-3rg^{\dprime}(x)\qquad\textrm{if } v= 3/2\\
0\qquad\textrm{if }  v=1, v>3/2\\
\infty\qquad\textrm{otherwise}
\end{cases}\\
&\frac{\partial^5}{\partial \sigma^4 \partial x} \prox_g^{\sigma^{2v}r/2}(x)_{|\sigma =0} =\begin{cases}
<\infty \qquad\textrm{if } v= 1/2\\
6r^2\left(g^{\prime}(x)g^{\dprime\prime}(x)+\left[g^{\dprime}(x)\right]^2\right)\qquad\textrm{if } v= 1\\
-12rg^{\dprime}(x)\qquad\textrm{if } v= 2\\
0\qquad\textrm{if }  v=3/2, v>2\\
\infty\qquad\textrm{otherwise}
\end{cases}\\
\end{align}
and
\begin{align}
&\frac{\partial^3}{\partial \sigma\partial x^2 } \prox_g^{\sigma^{2v}r/2}(x)_{|\sigma =0} =\begin{cases}
-\frac{r}{2}g^{\dprime\prime}(x) \qquad\textrm{if } v= 1/2\\
0\qquad\textrm{if }  v>1/2\\
\infty\qquad\textrm{otherwise}
\end{cases}\\
&\frac{\partial^4}{\partial \sigma\partial x^3 } \prox_g^{\sigma^{2v}r/2}(x)_{|\sigma =0} =\begin{cases}
<\infty \qquad\textrm{if } v= 1/2\\
0\qquad\textrm{if }  v>1/2\\
\infty\qquad\textrm{otherwise}
\end{cases}\\
&\frac{\partial^4}{\partial \sigma^2\partial x^2 } \prox_g^{\sigma^{2v}r/2}(x)_{|\sigma =0} =\begin{cases}
<\infty \qquad\textrm{if } v= 1/2\\
-rg^{\dprime\prime}(x) \qquad\textrm{if } v= 1\\
0\qquad\textrm{if }  v>1/2\\
\infty\qquad\textrm{otherwise}
\end{cases}\\
&\frac{\partial^5}{\partial \sigma^3\partial x^2 } \prox_g^{\sigma^{2v}r/2}(x)_{|\sigma =0} =\begin{cases}
<\infty \qquad\textrm{if } v= 1/2\\
-3rg^{\dprime\prime}(x) \qquad\textrm{if } v= 3/2\\
0\qquad\textrm{if }  v=1, v>3/2\\
\infty\qquad\textrm{otherwise}
\end{cases}\\
&\frac{\partial^5}{\partial \sigma^2\partial x^3 } \prox_g^{\sigma^{2v}r/2}(x)_{|\sigma =0} =\begin{cases}
<\infty \qquad\textrm{if } v= 1/2\\
-rg^{(4)}(x) \qquad\textrm{if } v= 1\\
0\qquad\textrm{if }  v>1\\
\infty\qquad\textrm{otherwise}
\end{cases}
\end{align}


\section{Proof of the result for the Laplace distribution}
\label{sec:laplace_proof}

In this section we prove the results in Section~\ref{sec:laplace} which give the scaling properties of MY-MALA (and sG-MALA) for the Laplace distribution.
We collect technical results (e.g. moment computations, bounds, etc.) in Appendix~\ref{app:laplace_computations}.

We recall that $\sigma_d^2 = \ell^2 /d^{2\alpha}$ and $\lambda_d = c^2/2d^{2\beta}$ for some $\alpha, \beta>0$ and some constants $c, \ell$ independent on $d$.
Thus, we can write $\lambda_d$ as a function of $\sigma_d$, $\lambda_d = \sigma_d^{2v}r/2$,
where we define $r=c^2/\ell^{2v}\geq0$ and $v= \beta/\alpha$. %
 In order to study the scaling limit of MY-MALA with Laplace target, 
 consider the mapping
$b_d:\rset^2 \to \rset$ given by
\begin{align}
\label{eq:martingale_b}
b_d:(x,z)\mapsto z - \frac{\sigma_d}{2}\sgn(x)\1\defEns{\vert x\vert
\geq\sigma_d^{2v}r/2}-\frac{1}{\sigma_d^{2m-1}r}x\1\defEns{\vert x\vert<
\sigma_d^{2v}r/2},
\end{align}
which allows us to write the proposal as  $Y^d_{1,i}=X_{0,i}^d + \sigma_d
 b_d(X_{0,i}^d,Z_{1,i}^d)$ , for any $i \in \iint{1}{d}$.

We consider also the function $\phi_d:\rset^2 \to \rset$, given by 
\begin{align}
\label{eq:laplace_ar}
\phi_d:(x,z)&\mapsto\log \frac{\pi(x+\sigma_d b_d(x,z))q(x+\sigma_d b_d(x,z), x)}{\pi(x)q(x, x+\sigma_d b_d(x,z))} \\
   &= \vert x\vert -\left\lvert x +\sigma_d b_d(x, z) \right\rvert 
    +\frac{z^2}{2}\\
   &\quad-\frac{1}{2\sigma_d^2}\left\lbrace 
    \frac{\sigma_d^2}{2}\sgn\left[x +\sigma_d b_d( x, z)\right]
\1\defEns{\left\vert x +\sigma_d b_d( x,
z)\right\rvert\geq\frac{\sigma_d^{2v}r}{2}}\right. \\
   &\qquad -\sigma_d b_d(x, z)   \\
    & \qquad+\left.\frac{1}{\sigma_d^{2(m-1)}r}\left[x +\sigma_d b_d( x,
z)\right]\1\defEns{\left\vert x +\sigma_d b_d(x,
z)\right\rvert<\frac{\sigma_d^{2v}r}{2}}\right\rbrace^2.
\end{align}

\subsection{Proof of Theorem~\ref{theo:acceptance_laplace}}
\label{sec:acceptance_ratio}

We introduce, for $i\in\iint{1}{d}$, $\phi_{d,i}=\phi_{d}(X^d_{0,i},Z^d_{1,i})$
for the sake of conciseness.
This allows us to rewrite $a_d(\ell,r)$, defined in~\eqref{eq:ar_finite}, in
the following way,
\begin{equation}
    \label{eq:acceptance_phi}
    a_d(\ell,r)= \expe{\exp\parenthese{\sum_{i=1}^d
    \phi_{d,i}} \wedge 1}.
\end{equation}
\begin{remark}
    \label{rmk:iid}
    Under \Cref{ass:stationarity}, the families of random variables
    $(b_d(X^d_{0,i},Z^d_{1,i}))_{i\in\iint{1}{d}}$ and
    $(\phi_{d,i})_{i\in\iint{1}{d}}$ are i.i.d..
\end{remark}
\begin{remark}
\label{rmk:clt}
It is important to see why \Cref{rmk:iid} and the central limit theorem (CLT)
do not allow us to conclude on the limiting law of the sum of the random variables $(\phi_{d,i})_{i\in \iint{1}{d}}$ when $d$ goes to infinity. To use the CLT, we would need a family of i.i.d. random variables $\parentheseLigne{T_i}_{i\in \nsets}$ such that $\expeLine{T_1^2}< +\infty$. In which case, we would obtain the convergence in distribution of the random variable
\begin{equation}
    \left.\sum_{i=1}^d\parenthese{T_i-\expe{T_1}}
    \middle/ \parenthese{\expe{T_1^2}d}^{1/2}
    \right. .
\end{equation}
In our case, $\phi_{d,i}$ would play the role of 
$T_i/ (\expeLine{T_1^2}d)^{1/2}$. However, the comparison is flawed because, as one can see in \eqref{eq:laplace_ar}, there is no easy way to choose a map $f:\nsets \to \rset$ such that $f(d) \phi_{d,i}$ no longer depends on $d$. Indeed, the random variable $\phi_{d,i}$ depends on $d$ in a very intricate way, such that the CLT is not of any help here. This is why we use Lindeberg's CLT for martingale arrays instead (see \cite[Theorem 4, page 543]{shiryaev1996probability}).
\end{remark}

The proof of Theorem~\ref{theo:acceptance_laplace} uses the first three moments
of $\phi_{d,1}$, whose computation is postponed to
Appendix~\ref{app:laplace_moments}, and is an application of Lindeberg's
central limit theorem.

To identify the optimal scaling for the Laplace distribution, we look for those
values of $\alpha$ such that $\sum_{i=1}^d \expeLine{\phi_{d,i}}$ and
$\var\parentheseLigne{\sum_{i=1}^d\phi_{d,i}}$ converge to a finite value. Using
\Cref{rmk:iid}, we have that,
\begin{align}
\label{eq:laplace_expe_var}
\sum_{i=1}^d \expe{\phi_{d,i}} = d\eqsp \expe{\phi_{d,1}} \eqsp
\quad \text{and} \quad
\var\parenthese{\sum_{i=1}^d \phi_{d,i}} = d \var\left(\phi_{d,1}\right) .
\end{align}
Then, using the integrals in Appendix~\ref{app:laplace_moments}, we find that
the only value of $\alpha$ for which~\eqref{eq:laplace_expe_var} converge to a
finite value with the variance strictly positive is
$\alpha=1/3$ as confirmed empirically in Appendix~\ref{sec:laplace_numeric}.

Having identified $\alpha=1/3$, we can then proceed applying Lindeberg's CLT.
\begin{proof}[Proof of Theorem~\ref{theo:acceptance_laplace}]
We start by showing that the acceptance ratio converges to a Gaussian distribution. 
Define $\mu_d= \expeLine{\phi_{d,1}}$ and
$\mathcal{F}_{d,i}=\sigma((X_{0,j}^d, Z^d_{1,j}), 1\leq j\leq i)$, the natural
filtration for $(X_{0,i}^d, Z^d_{1,i})_{d\in\nset, 1\leq i\leq d}$. The
square-integrable martingale sequence
\begin{align}
    \left(\sum_{j=1}^i W_{d, j}, \mathcal{F}_{d,i}\right)_{d \in
    \nsets, 1\leq i\leq d}
\end{align}
where $W_{d,i}= \phi_{d,i} - \mu_d$, 
forms a triangular array, to which we can apply the corresponding CLT (e.g.
\cite[Theorem 4, page 543]{shiryaev1996probability}). In particular, we have that,
\begin{align}
    \lim_{d\to \infty}\sum_{i=1}^d \expe{W_{d, i}^2\mid \mathcal{F}_{d, i-1}
    }&=\lim_{d\to \infty} d \var\left(\phi_{d,1}\right) =
    \frac{2\ell^3}{3\sqrt{2\uppi}},
\end{align}
as shown in Proposition~\ref{prop:laplace_var} in Appendix~\ref{app:laplace_moments}.
It remains to verify Lindeberg's condition: for $\varepsilon>0$, 
\begin{align}
\lim_{d\to \infty} d\expe{W_{d, 1}^2\1\defEns{\vert W_{d, 1}\vert > \varepsilon}}=0.
\end{align}
In order to verify Lindeberg's condition we verify the stronger Lyapunov condition: there exists
$\epsilon>0$ such that 
\begin{align}
\lim_{d\to \infty} d\expe{W_{d, 1}^{2+\epsilon}}=0.
\end{align}
Pick $\epsilon =1$ and expand the cube using $\mu_d=\expeLine{\phi_{d,i}}$,
\begin{align}
\label{eq:linderberg}
\expe{W_{d, 1}^3} &= \expe{\phi_{d,i}^3}-3 \mu_d\expe{\phi_{d,i}^2}
+2\mu_d^3.
\end{align}
By Proposition~\ref{prop:laplace_exp} in Appendix~\ref{app:laplace_moments}, we have
    $\lim_{d\to \infty} d\mu_d^3=0$, $\lim_{d\to \infty} \mu_d=0$,
and, by Proposition~\ref{prop:laplace_var} in Appendix~\ref{app:laplace_moments},
\begin{align}
    \lim_{d\to \infty} d\expe{\phi_{d,i}^2}=\frac{2\ell^3}{3\sqrt{2\uppi}}.
\end{align}
Finally, for the remaining term in~\eqref{eq:linderberg} we use
Proposition~\ref{prop:laplace_third} in Appendix~\ref{app:laplace_moments} to
show that $\lim_{d\to \infty} d\expeLine{\phi_{d,i}^3}=0$. The above and the
fact that, by Proposition~\ref{prop:laplace_exp} in
Appendix~\ref{app:laplace_moments},
\begin{align}
    \lim_{d\to\infty}d \mu_d = -\frac{\ell^3}{3\sqrt{2\uppi}},
\end{align}
show, by Lindeberg's CLT, that the acceptance ratio converges in law to a normal random variable $\widetilde Z$ with mean
$-\ell^3/(3\sqrt{2\uppi})$ and 
variance $2\ell^3/(3\sqrt{2\uppi})$.

To conclude the proof, we apply the continuous mapping theorem to the bounded and continuous function $x \mapsto e^x\wedge 1$ and obtain
\begin{equation}
    \lim_{d\to \infty} \exp\parenthese{\sum_{i=1}^{d} \phi_{d,i}} \wedge 1
    \overset{\rmd}{=} e^{\widetilde{Z}} \wedge 1\eqsp \quad \text{and} \quad
    \lim_{d\to \infty} a_d(\ell,r) = \expe{e^{\widetilde{Z}} \wedge 1},
\end{equation}
where the limit does not depend on $r$. Defining $a^\rmL(\ell) = \lim_{d\to \infty} a_d(\ell,r)$ and
using \cite[Proposition 2.4]{gelman1997weak}, we have the result.

\end{proof}


\subsection{Proof of Proposition~\ref{prop:tightness}}
\label{sec:tightness}
We are interested in the law $\nu_d$ of the linear
interpolant $(L_{t}^d)_{t\geq 0}$, defined in~\eqref{eq:Yt_differentiable},
of the first component of the chain $(X^d_k)_{k\in\nset}$. Let us recall the definition of the chain:
assumption \Cref{ass:stationarity} gives the initial distribution $\pi_d$,
then, for any $k \in \nset$, the proposal $Y^d_{k+1}=(Y^d_{k+1,i})_{1\leq i
\leq d}$ is defined in~\eqref{eq:laplace_proposal} with $\sigma_d^2 =
\ell^2/d^{2\alpha}$, $\lambda_d = \sigma_d^{2v}r/2$ with $\alpha=1/3$ and $v\geq 1$. The proposal~\eqref{eq:laplace_proposal} can be written as
\begin{equation}
    \label{eq:laplace_proposal_bis}
    Y^d_{k+1,i} = X^d_{k,i} + \sigma_d b_d(X^d_{k,i},Z^d_{k+1,i}),
\end{equation}
for any $i \in \iint{1}{d}$, where $b_d$ is defined in~\eqref{eq:martingale_b} and $r=c^2/\ell^{2v}$.
We further define the acceptance event $\acc_{k+1}^d=\defEns{\mathrm{b}_{k+1}^d=1}$ where $\mathrm{b}_{k+1}^d$ is as in~\eqref{eq:p_mala_def}.

We can now expand the expression of the linear interpolant $L_{t}^d$  using~\eqref{eq:p_mala_def}, \eqref{eq:Yt_differentiable} and the definition of
$\acc_{k+1}^d$,
\begin{align}
\label{eq:Yt1}
L_{t}^d&=\begin{cases}
X^d_{\floor{d^{2\alpha}t}, 1} + (d^{2\alpha}t- \floor{d^{2\alpha}t})\left[\sigma_d Z^d_{\ceil{d^{2\alpha}t}, 1}-\frac{\sigma^2_d}{2}\sgn(X^d_{\floor{d^{2\alpha}t}, 1})\right]\mathbbm{1}_{\acc^d_{\ceil{d^{2\alpha}t}}}\\
\qquad\qquad\textrm{if } \vert X^d_{\floor{d^{2\alpha}t}, 1}\vert \geq \frac{\sigma^{2v}_dr}{2}\\
X^d_{\floor{d^{2\alpha}t}, 1} + (d^{2\alpha}t- \floor{d^{2\alpha}t}) \left[\sigma_d Z^d_{\ceil{d^{2\alpha}t}, 1}-\frac{1}{\sigma_d^{2(m-1)}r}X^d_{\floor{d^{2\alpha}t}, 1}\right]\mathbbm{1}_{\acc^d_{\ceil{d^{2\alpha}t}}}\\
\qquad\qquad\textrm{otherwise}
\end{cases},
\end{align}
or, equivalently,
\begin{align}
L_{t}^d&=\begin{cases}
X^d_{\ceil{d^{2\alpha}t}, 1} - (\ceil{d^{2\alpha}t} - d^{2\alpha}t)\left[\sigma_d Z^d_{\ceil{d^{2\alpha}t}, 1}-\frac{\sigma^2_d}{2}\sgn(X^d_{\floor{d^{2\alpha}t}, 1})\right]\mathbbm{1}_{\acc^d_{\ceil{d^{2\alpha}t}}}\\
\qquad\qquad\textrm{if } \vert X^d_{\floor{d^{2\alpha}t}, 1}\vert \geq \frac{\sigma^{2v}_dr}{2}\\
X^d_{\ceil{d^{2\alpha}t}, 1} - (\ceil{d^{2\alpha}t} - d^{2\alpha}t) \left[\sigma_d Z^d_{\ceil{d^{2\alpha}t}, 1}-\frac{1}{\sigma_d^{2(m-1)}r}X^d_{\floor{d^{2\alpha}t}, 1}\right]\mathbbm{1}_{\acc^d_{\ceil{d^{2\alpha}t}}}\\
\qquad\qquad\textrm{otherwise}
\end{cases}.
\end{align}

In order to prove Proposition~\ref{prop:tightness}, we consider Kolmogorov's
criterion for tightness (see \cite[Theorem 23.7]{kallenberg:2021}): the sequence $(\nu_d)_{d\geq 1}$ is tight if the sequence $(L^d_{0})_{d\in \nsets}$ is tight, and
\begin{align}
\expe{(L_{t}^d-L_{s}^d)^4}\leq \gamma(t)(t-s)^2,
\end{align}
for some non-decreasing positive function $\gamma$, all $0\leq s\leq t$ and
all $d\in \nsets$. The condition on $(L^d_{0})_{d\in \nsets}$ is straightforward to check, since by \Cref{ass:stationarity} the distribution of $L^d_0 = X^d_{0,1}$ is $\piLaplace$ for all $d\in\nsets$.
Before proceeding with the proof of the proposition we briefly recall the following H\"older's inequality (which holds for any $p\geq 1$ and any $a_i \geq 0$)
\begin{align}
\label{eq:holder}
\left(\sum_{i=1}^na_i \right)^p \leq n^{p-1}\sum_{i=1}^n(a_i)^p,
\end{align}
of which we will make frequent use.
\begin{proof}[Proof of Proposition~\ref{prop:tightness}]
Consider $\expe{(L_{t}^d-L_{s}^d)^4}$ with $L_t^d, L_s^d$ as in~\eqref{eq:Yt1}.

If $\floor{d^{2\alpha}s}=\floor{d^{2\alpha}t}$ (which implies $\ceil{d^{2\alpha}s}=\ceil{d^{2\alpha}t}$) and $\vert X^d_{\floor{d^{2\alpha}t}, 1}\vert = \vert X^d_{\floor{d^{2\alpha}s}, 1}\vert \geq \sigma_d^{2v}r/2$, we have
\begin{align}
L_{t}^d-L_{s}^d &= (d^{2\alpha}t- d^{2\alpha}s)\left[\sigma_d Z^d_{\ceil{d^{2\alpha}t}, 1}-\frac{\sigma^2_d}{2}\sgn(X^d_{\floor{d^{2\alpha}t}, 1})\right]\mathbbm{1}_{\acc^d_{\ceil{d^{2\alpha}t}}}.
\end{align}
Therefore, 
\begin{align}
\expe{(L_{t}^d-L_{s}^d)^4} &= (d^{2\alpha}t- d^{2\alpha}s)^4\expe{\left[\sigma_d Z^d_{\ceil{d^{2\alpha}t}, 1}-\frac{\sigma^2_d}{2}\sgn(X^d_{\floor{d^{2\alpha}t}, 1})\right]\mathbbm{1}_{\acc^d_{\ceil{d^{2\alpha}t}}}}.
\end{align}
By bounding the indicator function by one and using~\eqref{eq:holder} we obtain
\begin{align}
\expe{(L_{t}^d-L_{s}^d)^4} &\leq  2^3(d^{2\alpha}t- d^{2\alpha}s)^4\expe{\sigma_d^4 (Z^d_{\ceil{d^{2\alpha}t}, 1})^4+\frac{\sigma^8_d}{2^4}\sgn(X^d_{\floor{d^{2\alpha}t}, 1})^4}\\
&\leq C\frac{(d^{2\alpha}t- d^{2\alpha}s)^4}{(d^{2\alpha})^2}\expe{ (Z^d_{\ceil{d^{2\alpha}t}, 1})^4+\frac{\sigma^2_d}{2^4}}\\
&\leq C\frac{(d^{2\alpha}t- d^{2\alpha}s)^4}{(d^{2\alpha})^2}.
\end{align}
where we used the boundedness of the $\sgn$ function and the of the moments of normal distributions and we incorporated all constants in $C$.
Moreover, using $\lfloor d^{2\alpha}s\rfloor = \lfloor d^{2\alpha}t\rfloor$, we have $( d^{2\alpha}t-d^{2\alpha}s)^4\leq ( d^{2\alpha}t-d^{2\alpha}s)^2$. Hence,
\begin{equation}
\expe{(L_{t}^d-L_{s}^d)^4} \leq C(t-s)^2.
\end{equation}
The case $\vert X^d_{\floor{d^{2\alpha}t}, 1}\vert = \vert X^d_{\floor{d^{2\alpha}s}, 1}\vert < \sigma_d^{2v}r/2$ follows using the same strategy and exploiting the boundedness of the moments of normal distributions and the boundedness of $X^d_{\floor{d^{2\alpha}t}, 1}$ itself.

For all $ 0\leq s\leq t$ such that $\ceil{d^{2\alpha}s}\leq\floor{d^{2\alpha}t}$, we can distinguish three cases.
\subsubsection*{Case 1} If $\vert X^d_{\floor{d^{2\alpha}t}, 1}\vert \geq \sigma_d^{2v}r/2$ and $\vert X^d_{\floor{d^{2\alpha}s}, 1}\vert \geq \sigma_d^{2v}r/2$, then
\begin{align}
L_{t}^d-L_{s}^d &= X^d_{\floor{d^{2\alpha}t}, 1} - X^d_{\ceil{d^{2\alpha}s}, 1}\\
&+ (d^{2\alpha}t- \floor{d^{2\alpha}t})\left[\sigma_d Z^d_{\ceil{d^{2\alpha}t}, 1}-\frac{\sigma^2_d}{2}\sgn(X^d_{\floor{d^{2\alpha}t}, 1})\right]\mathbbm{1}_{\acc^d_{\ceil{d^{2\alpha}t}}}\\
&+( \ceil{d^{2\alpha}s}- d^{2\alpha}s)\left[\sigma_d Z^d_{\ceil{d^{2\alpha}s},1}-\frac{\sigma^2_d}{2}\sgn(X^d_{\floor{d^{2\alpha}s}, 1})\right]\mathbbm{1}_{\acc^d_{\ceil{d^{2\alpha}s}}} .
\end{align}
Using H\"older's inequality~\eqref{eq:holder} collecting all constants in $C$, and the fact that $0\leq d^{2\alpha}t- \floor{d^{2\alpha}t} \leq 1$ (and similarly for $s$) we have
\begin{align}
\expe{(L_{t}^d-L_{s}^d)^4}&\leq C\expe{\left(X^d_{\floor{d^{2\alpha}t}, 1}-X^d_{\ceil{d^{2\alpha}s}, 1}\right)^4}\\
&+ C\frac{(d^{2\alpha}t- \floor{d^{2\alpha}t})^2}{d^{4\alpha}}\expe{\left(\ell Z^d_{\ceil{d^{2\alpha}t}, 1}\right)^4+\frac{\ell^8}{2^4d^{4\alpha}}} \\
&+ C\frac{( \ceil{d^{2\alpha}s}- d^{2\alpha}s)^2}{d^{4\alpha}}\expe{\left(\ell Z^d_{\ceil{d^{2\alpha}s},1}\right)^4+\frac{\ell^8}{2^4d^{4\alpha}}}.
\end{align}
Recalling that the moments of $Z^d$ are bounded and that $d^{2\alpha}s\leq\ceil{d^{2\alpha}s}\leq\floor{d^{2\alpha}t}\leq d^{2\alpha}t$, it follows
\begin{align}
\label{eq:tightness1}
\expe{(L_{t}^d-L_{s}^d)^4}\leq C\left((t-s)^2+\expe{\left(X^d_{\floor{d^{2\alpha}t}, 1}-X^d_{\ceil{d^{2\alpha}s}, 1}\right)^4}\right).
\end{align}

\subsubsection*{Case 2} If $\vert X^d_{\floor{d^{2\alpha}t}, 1}\vert \geq \sigma_d^{2v}r/2$ and $\vert X^d_{\floor{d^{2\alpha}s}, 1}\vert < \sigma_d^{2v}r/2$ or $\vert X^d_{\floor{d^{2\alpha}t}, 1}\vert < \sigma_d^{2v}r/2$ and $\vert X^d_{\floor{d^{2\alpha}s}, 1}\vert \geq \sigma_d^{2v}r/2$.
We only describe the argument for the first case, the second case follows from analogous steps.
Take
\begin{align}
L_{t}^d-L_{s}^d &= X^d_{\floor{d^{2\alpha}t}, 1} + (d^{2\alpha}t- \floor{d^{2\alpha}t})\left[\sigma_d Z^d_{\ceil{d^{2\alpha}t}, 1}-\frac{\sigma^2_d}{2}\sgn(X^d_{\floor{d^{2\alpha}t}, 1})\right]\mathbbm{1}_{\acc^d_{\ceil{d^{2\alpha}t}}}\\
                & - X^d_{\ceil{d^{2\alpha}s}, 1} - (\ceil{d^{2\alpha}s} - d^{2\alpha}s)\left( \sigma_d Z^d_{\ceil{d^{2\alpha}s}, 1}-\frac{1}{\sigma_d^{2(m-1)}r}X^d_{\floor{d^{2\alpha}s}, 1}\right)\1_{\acc^d_{\ceil{d^{2\alpha}s}}}.
\end{align}
Proceeding as above, we find that
\begin{align}
\expe{(L_{t}^d-L_{s}^d)^4}&\leq C\left((t-s)^2+\expe{\left(X^d_{\floor{d^{2\alpha}t}, 1}-X^d_{\ceil{d^{2\alpha}s}, 1}\right)^4}\right.\\
&\qquad\qquad\left.+(\ceil{d^{2\alpha}s} - d^{2\alpha}s)^4\expe{\left(\frac{1}{\sigma_d^{2(m-1)}r}X^d_{\floor{d^{2\alpha}s}}\right)^4}\right),
\end{align}
and recalling that $\vert X^d_{\floor{d^{2\alpha}s}, 1}\vert <
\sigma_d^{2v}r/2$ we have that $\vert X^d_{\floor{d^{2\alpha}s},
1}\vert/(r \sigma_d^{2(m-1)}) < \sigma_d^2 /2$. Using this and the same arguments as above, we have
\begin{align}
\label{eq:tightness2}
\expe{(L_{t}^d-L_{s}^d)^4}&\leq C\left((t-s)^2+\expe{\left(X^d_{\floor{d^{2\alpha}t}, 1}-X^d_{\ceil{d^{2\alpha}s}, 1}\right)^4}\right) .
\end{align}

\subsubsection*{Case 3} If $\vert X^d_{\floor{d^{2\alpha}t}, 1}\vert < \sigma_d^{2v}r/2$ and $\vert X^d_{\floor{d^{2\alpha}s}, 1}\vert <  \sigma_d^{2v}r/2$, then
\begin{align}
L_{t}^d-L_{s}^d &= X^d_{\floor{d^{2\alpha}t}, 1} + (d^{2\alpha}t- \floor{d^{2\alpha}t}) \left(\sigma_d Z^d_{\ceil{d^{2\alpha}t}, 1} -\frac{1}{\sigma_d^{2(m-1)}r} X^d_{\floor{d^{2\alpha}t}, 1}\right)\mathbbm{1}_{\acc^d_{\ceil{d^{2\alpha}t}}}\\
&- X^d_{\ceil{d^{2\alpha}s}, 1} +(\ceil{d^{2\alpha}s}-d^{2\alpha}s) \left(\sigma_d Z^d_{\ceil{d^{2\alpha}s}, 1}-\frac{1}{\sigma_d^{2(m-1)}r}X^d_{\floor{d^{2\alpha}s}, 1}\right)\mathbbm{1}_{\acc^d_{\ceil{d^{2\alpha}s}}}.
\end{align}
Using the boundedness of moments of Gaussian distributions and of $X^d_{\floor{d^{2\alpha}t}, 1}, X^d_{\floor{d^{2\alpha}s}, 1}$, we have
\begin{align}
\label{eq:tightness3}
\expe{(L_{t}^d-L_{s}^d)^4}\leq C\left((t-s)^2+\expe{\left(X^d_{\floor{d^{2\alpha}t}, 1}-X^d_{\ceil{d^{2\alpha}s}, 1}\right)^4}\right).
\end{align}
Putting~\eqref{eq:tightness1},~\eqref{eq:tightness2} and~\eqref{eq:tightness3} together and using Lemma~\ref{lemma:4th_moment} below we obtain 
\begin{align}
\expe{ \left( L_{t}^d-L_{s}^d\right)^4} &\leq C \left( \left(t-s\right)^2 + \sum_{p=2}^4 \frac{\left( \floor{ d^{2\alpha}t } - \ceil{ d^{2\alpha}s} \right)^p}{d^{2\alpha p}} \right)\\
&\leq C(t-s)^2 + C \sum_{p=2}^4 \frac{d^{2\alpha p}\left( t - s\right)^p}{d^{2\alpha p}} \leq C\left(2+t+t^2 \right) (t-s)^2,
\end{align}
which concludes the proof.
\end{proof}

We are now ready to state and prove Lemma~\ref{lemma:4th_moment}:
\begin{lemma}
\label{lemma:4th_moment}
There exists $C>0$ such that for any $k_1, k_2\in\mathbb{N}$ with $0\leq k_1 < k_2$,
\begin{equation}
\expe {\left( X^d_{k_2 ,1} - X^d_{k_1 ,1}\right)^4} \leq C  \sum_{p=2}^{4} \frac{(k_2 - k_1)^p}{d^{2\alpha p}},
\end{equation}
where $\alpha=1/3$.
\end{lemma}
\begin{proof}
    Recalling the definition of the proposal in~\eqref{eq:laplace_proposal_bis}
    and the definition of $b_d$ in~\eqref{eq:martingale_b} we can write
\begin{align}
\expe{\left( X_{k_2,1}^d - X^d_{k_1 ,1}\right)^4} = \mathbb{E}\left[\left(
\sum_{k=k_1+1}^{k_2} \sigma_d b_d\parenthese{X_{k-1,1}^d,Z^d_{k,1}}\1_{\acc_k^d}
\right)^4\right].
\end{align} 
Then, we expand all acceptance or rejection terms between $k_1$ and $k_2$ and use H\"older's inequality to obtain
\begin{align}
\label{eq:kolmogorov1}
\expe{\left( X_{k_2,1}^d - X^d_{k_1 ,1}\right)^4} &\leq C \left\lbrace \sigma_d^4\expe{\left( \sum_{k=k_1+1}^{k_2}  b_d\left( X^d_{k-1,1}, Z^d_{k,1}\right)\right)^4}\right.\\
&\left.+\sigma_d^4\expe{\left( \sum_{k=k_1+1}^{k_2} b_d\left(X^d_{k-1,1}, Z^d_{k,1}\right)\1_{(\acc_k^d)^c} \right)^4}\right\rbrace.
\end{align} 
Using again H\"older's inequality, for the first term we have
\begin{align}
&\expe{\left( \sum_{k=k_1+1}^{k_2} b_d\left( X^d_{k-1,1}, Z^d_{k,1}\right)\right)^4}\leq C \left\lbrace \expe{\left( \sum_{k=k_1+1}^{k_2} Z^d_{k,1}\right)^4}\right.\\
&\qquad +\frac{\sigma_d^4}{2^4}\expe{\left(\sum_{k=k_1+1}^{k_2}\sgn\left( X^d_{k-1,1}\right)\1\defEns{\vert X_{k-1, 1}^d\vert \geq \sigma_d^{2v}r/2}\right)^4}\notag\\
&\qquad\left. + \expe{\left(\sum_{k=k_1+1}^{k_2}\frac{1}{\sigma_d^{2v-1} r}X^d_{k-1,1}\1\defEns{\vert X_{k-1, 1}^d\vert < \sigma_d^{2v}r/2}\right)^4}\right\rbrace\notag \\
&\qquad\leq C \left[ 3(k_2-k_1)^2 + \frac{2\sigma_d^4}{2^4} (k_2-k_1)^4 \right],
\label{eq:kolmogorov2}
\end{align}
where the last line follows using the moments of $Z^d_{k, 1}$ and the boundedness of $X_{k-1, 1}^d$ in the set $\{\vert X_{k-1, 1}^d\vert < \sigma_d^{2v}r/2\}$.

Using a Binomial expansion of the rejection term, we obtain
\begin{align} \label{eq:sum_reject}
    \expe{\left( \sum_{k=k_1+1}^{k_2} b_d\left(X^d_{k-1,1}, Z^d_{k,1}\right) \1_{\left( \acc^d_{k}\right)^c} \right)^4 } 
= \sum \expe{\prod_{i=1}^{4}b_d\left(X^d_{m_i-1,1}, Z^d_{m_i,1}\right)\1_{\left( \acc^d_{m_i}\right)^c}},
\end{align}
where the sum is over the quadruplets $(m_i)_{1\leq i \leq 4}$ with $m_i \in \left\lbrace k_1+1,\dots,k_2\right\rbrace$.

We separate the terms in the sum according to their cardinality: we denote by $\vert B\vert $ the cardinality of set $B$ (i.e. the number of distinct elements), and, for $j \in \iint{1}{4}$,
we define
\begin{equation}
\mathcal{I}_j=\left\lbrace (m_1,\dots,m_4)\in \left\lbrace k_1+1,\dots,k_2\right\rbrace^4 : \vert\left\lbrace m_1,\dots,m_4\right\rbrace\vert = j \right\rbrace.
\end{equation}
For any $(m_1,\dots,m_4)\in \left\lbrace
k_1+1,\dots,k_2\right\rbrace^4$, $\widetilde{X}^d_0 = X^d_0$ and for any $i \in
\iint{1}{d}$,
\begin{align}
&\widetilde{X}^d_{k+1,i} = \widetilde{X}^d_{k,i} + \1_{\left\lbrace m_1-1,\dots,m_4-1\right\rbrace^c}(k)\1_{\tacc_{k+1}^d}\sigma_d b_d\left( \widetilde{X}_{k,i}^d, Z^d_{k+1,i}\right),
\end{align}
where
\begin{equation}
\label{eq:4thmoment_acceptance}
\tacc_{k+1}^d = \left\lbrace U_{k+1} \leq \exp \left[ \sum_{i=1}^d \phi_d\left(
\widetilde{X}^d_{k,i},Z_{k+1,i}^d\right) \right]\right\rbrace ,
\end{equation}
and $\phi_d$ in~\eqref{eq:laplace_ar}. 
Denote by $\mathcal{F}$ the $\sigma$-algebra generated by the process
$(\widetilde{X}_k^d)_{k\geq 0}$ and observe that on the event
$\underset{j=1}{\overset{4}{\bigcap}} \left(\acc_{m_j}^d\right)^c$, $X_k^d$ is
equal to $\widetilde{X}_k^d$.

We consider now the terms in the sum~\eqref{eq:sum_reject}.
\begin{enumerate}
\item If $(m_1,\dots,m_4) \in \mathcal{I}_4$, then the $m_i$s are all distinct and
\begin{align}
    \expe{ \prod_{j=1}^{4} b_d\left(X_{m_j-1,1}^d, Z_{m_j,1}^d \right) \1_{\left(\acc_{m_j}^d\right)^c} \middle| \mathcal{F}} = \expe{ \prod_{j=1}^{4} b_d\left(\widetilde{X}_{m_j-1,1}^d, Z^d_{m_j,1} \right) \1_{\left(\tacc_{m_j}^d\right)^c} \middle| \mathcal{F}}.
\end{align}
However, $\lbrace b_d( \widetilde{X}_{m_j-1,1}^d, Z^d_{m_j,1} ) \1_{(\tacc_{m_j}^d)^c} \rbrace_{j=1,\dots,4}$ are independent conditionally on $\mathcal{F}$. Thus,
\begin{align}
&\expe{ \prod_{j=1}^{4} b_d\left(\widetilde{X}_{m_j-1,1}^d, Z^d_{m_j,1} \right) \1_{\left(\tacc_{m_j}^d\right)^c} \middle| \mathcal{F}} = \prod_{j=1}^{4} \expe{ b_d\left(\widetilde{X}_{m_j-1,1}^d, Z^d_{m_j,1} \right) \1_{\left(\tacc_{m_j}^d\right)^c} \middle| \mathcal{F}}\\
&\qquad = \prod_{j=1}^{4} \expe{ b_d\left(\widetilde{X}_{m_j-1,1}^d, Z^d_{m_j,1} \right)\times\left( 1-\exp\left( \sum_{i=1}^d \phi_d\left(
\widetilde{X}^d_{m_j-1,i},Z_{m_j,i}^d\right) \right) \right)_+ \middle| \mathcal{F}},
\end{align}
by integrating the uniform variables $U_{m_j}$
in~\eqref{eq:4thmoment_acceptance}.

Recalling the definition of $b_d$ in~\eqref{eq:martingale_b}, we can bound the expectation above with
\begin{align}
\label{eq:case_a_1}
&\left\lvert\expe{ b_d\left(\widetilde{X}_{m_j-1,1}^d, Z_{m_j,1}^d \right)\left\lbrace 1-\exp\left( \sum_{i=1}^d \phi_d\left(
\widetilde{X}^d_{m_j-1,i},Z_{m_j,i}^d\right) \right) \right\rbrace_+ \middle| \mathcal{F}} \right\rvert\\
&\qquad\leq \left\vert \mathbb{E}\left[ \left(\frac{\sigma_d}{2}\sgn \left( \widetilde{X}^d_{m_j-1,1} \right)\1\defEns{\vert \widetilde{X}_{m_j-1, 1}^d\vert \geq \sigma_d^{2v}r/2}\right.\right.\right. \\
&\qquad\quad- \left.\frac{1}{\sigma_d^{2v-1} r}\widetilde{X}^d_{m_j-1,1}\1\defEns{\vert
\widetilde{X}_{m_j-1, 1}^d\vert < \sigma_d^{2v}r/2} \right) 
\\
&\qquad\qquad\left.\left.\times\left\lbrace 1-\exp\left( \sum_{i=1}^d \phi_d\left(
\widetilde{X}^d_{m_j-1,i},Z_{m_j,i}^d\right) \right) \right\rbrace_+ \middle|
\mathcal{F}\right] \right\vert \\
&\qquad\quad+\left\lvert \expe{ Z_{m_j,1}^d \left\lbrace 1-\exp\left( \sum_{i=1}^d \phi_d\left(
\widetilde{X}^d_{m_j-1,i},Z_{m_j,i}^d\right) \right) \right\rbrace_+  \middle| \mathcal{F}}\right\rvert . 
\end{align}
For the first term, we use the boundedness of the $\sgn$ function and of $\widetilde{X}^d_{m_j-1,1}$ in the set $\{\vert \widetilde{X}_{m_j-1, 1}^d\vert \leq \sigma_d^{2v}r/2\}$ to obtain
\begin{align}
&\left\vert \mathbb{E}\left[ \left(\frac{\sigma_d}{2}\sgn \left(
\widetilde{X}^d_{m_j-1,1} \right)\1\defEns{\vert \widetilde{X}_{m_j-1, 1}^d\vert \geq
\sigma_d^{2v}r/2}-\frac{\widetilde{X}^d_{m_j-1,1}}{\sigma_d^{2v-1} r}\1\defEns{\vert
\widetilde{X}_{m_j-1, 1}^d\vert < \sigma_d^{2v}r/2} \right) \right.\right.\\
&\qquad\qquad\left.\left.\times\left\lbrace 1-\exp\left( \sum_{i=1}^d \phi_d\left(
\widetilde{X}^d_{m_j-1,i},Z_{m_j,i}^d\right) \right) \right\rbrace_+ \middle| \mathcal{F}\right] \right\vert \notag\\
&\leq \frac{\sigma_d}{2}\expe{ \left\vert\left\lbrace 1-\exp\left( \sum_{i=1}^d \phi_d\left(
\widetilde{X}^d_{m_j-1,i},Z_{m_j,i}^d\right) \right) \right\rbrace_+ \right\vert \middle| \mathcal{F}} \leq \frac{\sigma_d}{2}.
\label{eq:case_a_2}
\end{align}
We can write the second term as
\begin{align}
\expe{ Z_{m_j,1}^d \left( 1-\exp\left( \sum_{i=1}^{d}\phi_d\left( \widetilde{X}^d_{m_j-1,i},Z^d_{m_j,i}\right) \right)\right)_+  \middle| \mathcal{F}} \\= \expe{\mathcal{G} \left( \widetilde{X}^d_{m_j-1,1},\sum_{i=2}^{d}\phi_d\left( \widetilde{X}^d_{m_j-1,i},Z^d_{m_j,i}\right)\right) \middle| \mathcal{F}},
\end{align}
where we define $\mathcal{G}(a,b)=\expe{Z\left( 1-\exp\left(\phi_d\left( a,Z\right)+b \right) \right)_+}$ with $Z$ a standard Gaussian. Because the function $x \mapsto (1-\exp(x))_+$ is 1-Lipschitz, we have, using Cauchy-Schwarz and Lemma~\ref{lemma:tightness_moment} in Appendix~\ref{app:tightness_moment},
\begin{align}
\left\vert \expe{Z\left( 1-\exp\left(\phi_d\left( a,Z\right)+b \right) \right)_+} - \expe{Z\left( 1-\exp\left(b \right) \right)_+}\right\vert &\leq \expe{\vert Z \vert \left\vert  \phi_d\left( a,Z\right)\right\vert}\\
& \leq \expe{Z^2}^{1/2} \expe{\phi_d\left( a,Z\right)^2}^{1/2}\\
&\leq \expe{\phi_d\left( a,Z\right)^2}^{1/2}\\
&\leq C d^{-\alpha}. 
\end{align}
However, $\expe{Z\left( 1-\exp\left(b \right) \right)_+} = \expe{Z}\left( 1-\exp\left(b \right) \right)_+ = 0$, and therefore
\begin{equation}\label{eq:bound_g}
\left\vert \expe{ \mathcal{G} \left(
\widetilde{X}^d_{m_j-1,1},\sum_{i=2}^{d}\phi_d\left(
\widetilde{X}^d_{m_j-1,i},Z^d_{m_j,i}\right)\right) \middle| \mathcal{F}}
\right\vert \leq C d^{-\alpha}.
\end{equation}
Combining equations~\eqref{eq:case_a_1},~\eqref{eq:case_a_2} and~\eqref{eq:bound_g} and recalling that $\sigma_d=\ell d^{-\alpha}$, we have
\begin{align}\label{eq:bound_prod_term}
&\left\lvert \expe{b_d \left(\widetilde{X}_{m_j-1,1}^d, Z_{m_j,1}^d \right)\left\lbrace 1-\exp\left( \sum_{i=1}^d \phi_d\left(
\widetilde{X}^d_{m_j-1,i},Z_{m_j,i}^d\right) \right) \right\rbrace_+ \middle| \mathcal{F}}\right\rvert\leq C d^{-\alpha}, 
\end{align}
from which follows that
\begin{align}
\sum_{(m_1,\dots,m_4)\in \mathcal{I}_4} \left\vert \expe{ \prod_{i=1}^{4}b_d\left(X^d_{m_i-1,1} , Z^d_{m_i,1}\right)\1_{\left( \acc^d_{m_i}\right)^c}}\right\vert 
&\leq \sum_{(m_1,\dots,m_4)\in \mathcal{I}_4}\expe{\prod_{j=1}^{4} \frac{C}{d^{\alpha}}} \\
&\leq \binom{k_2-k_1}{4}\frac{C}{d^{4\alpha}}\leq C\frac{(k_2-k_1)^4}{d^{4\alpha}},
\label{eq:case_a}
\end{align}
using that $|\mathcal{I}_4| =\binom{k_2-k_1}{4}$.

\item If $(m_1,..,m_4) \in \mathcal{I}_3$, only three of the $m_i$s take distinct values; without loss of generality, we assume that $m_1$ appears twice, while $m_2,m_3$ appear once. Proceeding as in case 1, we have
\begin{align}
&\left\vert \expe{ \prod_{j=1}^{3} b_d\left( X_{m_j-1,1}^d,  Z^d_{m_j,1} \right)^{1+\delta_{1,j}} \1_{\left(\acc_{m_j}^d\right)^c} \middle| \mathcal{F}} \right\vert \\
&= \prod_{j=1}^{3} \left\vert \expe{ b_d\left(\widetilde{X}_{m_j-1,1}^d,
Z^d_{m_j,1} \right) ^{1+\delta_{1,j}} \left\lbrace 1-\exp\left( \sum_{i=1}^d \phi_d\left(
\widetilde{X}^d_{m_j-1,i},Z_{m_j,i}^d\right) \right) \right\rbrace_+ \middle| \mathcal{F}}
\right\vert ,
\end{align}
where $\delta_{1,j}$ denotes a Dirac's delta.
For the terms $j\neq1$, we use~\eqref{eq:bound_prod_term}, while for the term $j=1$ we bound the indicator function by 1 to obtain
\begin{align}
&\left\vert \expe{ \prod_{j=1}^{3} b_d\left(X_{m_j-1,1}^d,  Z_{m_j,1}^d \right)^{1+\delta_{1,j}} \1_{\left(\acc_{m_j}^d\right)^c} \middle| \mathcal{F}} \right\vert\\ 
&\qquad\leq \left\lvert\expe{ b_d\left(\widetilde{X}_{m_1-1,1}^d,Z_{m_1,1}^d \right)^{2}  \middle| \mathcal{F}} \right\rvert\prod_{j=2}^3 \frac{C}{d^{\alpha}}\\
&\qquad\leq \left( 3+ \frac{2\sigma_d^2}{2^2d^{2\alpha}}\right) \frac{C^2}{d^{2\alpha}}\leq C \frac{1}{d^{2\alpha}},
\end{align}
where the second-to-last inequality follows using the same approach taken
for~\eqref{eq:kolmogorov2} and recalling that $\sigma_d=\ell d^{-\alpha}$.
Hence,
\begin{align}
\label{eq:case_b}
&\sum_{(m_1,\dots,m_4)\in \mathcal{I}_3} \left\vert \expe{ \prod_{i=1}^{4}b_d\left(X^d_{m_i-1,1} ,Z^d_{m_i,1}\right)\1_{\left( \acc^d_{m_i}\right)^c}} \right\vert \\
&\qquad\leq C \binom{k_2-k_1}{3} \frac{1}{d^{2\alpha}}\leq C \frac{(k_2-k_1)^3}{d^{2\alpha}}. 
\end{align}
\item  If $(m_1,..,m_4) \in \mathcal{I}_2$, we have two different cases: the
    $m_i$s take the two values twice or three $m_i$s have the same value.
    For the first one we assume, without loss of generality, that $m_1, m_2$ take distinct values and appear each twice.  Bounding the indicator function with 1,
\begin{align}
&\expe{\expe{ \prod_{j=1}^{2} b_d\left( X_{m_j-1,1}^d, Z_{m_j,1}^d \right)^{2} \1_{\left(\acc_{m_j}^d\right)^c} \middle| \mathcal{F}}} \\
&\qquad\leq \expe{\prod_{j=1}^2 \expe{ b_d\left(\widetilde{X}_{m_j-1,1}^d,Z_{m_j,1}^d \right)^{2} \middle| \mathcal{F}}}.
\end{align}
Since, conditionally on $ \mathcal{F}$, the random variables inside the expectation are normals with bounded mean and variance $1$, we have, using the same approach taken for~\eqref{eq:kolmogorov2},
\begin{align}
&\expe{\prod_{j=1}^2 \expe{ b_d\left(\widetilde{X}_{m_j-1,1}^d,Z_{m_j,1}^d
\right)^{2} \middle| \mathcal{F}}} \leq \left( 1 +
\frac{2\sigma_d^2}{2^2}\right)^2 \leq C.
\end{align}
For the second case we assume, without loss of generality, that $m_1$ appears three times while $m_2$ appears once. Following a similar approach to case 2, we obtain
\begin{align}
&\left\vert\expe{\expe{ \prod_{j=1}^{2} b_d\left( X_{m_j-1,1}^d, Z_{m_j,1}^d \right) ^{1+2\delta_{1,j}} \1_{\left(\acc_{m_j}^d\right)^c} \middle|\mathcal{F}}} \right\vert\\
&\qquad\qquad\leq \expe{\expe{ \prod_{j=1}^{2} \left\vert b_d\left(\widetilde{X}_{m_j-1,1}^d, Z_{m_j,1}^d \right) \right\vert^{1+2\delta_{1,j}}  \middle| \mathcal{F}}} \leq C,
\end{align}
where $\delta_{1,j}$ denotes a Dirac's delta.
Therefore,
\begin{align}
\label{eq:case_c}
&\sum_{(m_1,\dots,m_4)\in \mathcal{I}_2} \left\vert \expe{\prod_{i=1}^{4}\left\lbrace b_d(X^d_{m_i-1,1} , Z^d_{m_i,1}\right) \1_{\left( \acc^d_{m_i}\right)^c}} \right\vert \\
&\qquad\leq C\left(\binom{4}{2}+\binom{4}{3}\right)\binom{k_2-k_1}{2} \leq C (k_2-k_1)^2. 
\end{align}
\item If $(m_1,..,m_4) \in \mathcal{I}_1$ (i.e., all $m_i$s take the same value), we bound the indicator function by 1 and, using the same approach taken for~\eqref{eq:kolmogorov2}, we find
\begin{align}
\expe{b_d\left(X^d_{m_1-1,1} , Z^d_{m_1,1}\right)^4 \1_{\left( \acc^d_{m_1}\right)^c}}
& \leq C \left(3 + \frac{2\sigma_d^4}{2^4} \right)\leq C,
\end{align}
since $\sigma_d = \ell d^{-\alpha}$ and $d\in\mathbb{N}$.
Hence,
\begin{align}
 \label{eq:case_d}
&\sum_{(m_1,\dots,m_4)\in \mathcal{I}_1} \left\vert \expe{\prod_{i=1}^{4}b_d\left(X^d_{m_1-1,1} , Z^d_{m_1,1}\right) \1_{\left( \acc^d_{m_i}\right)^c}} \right\vert \leq C\binom{k_2-k_1}{1}=C (k_2-k_1).
\end{align}
The result follows combining~\eqref{eq:case_a},~\eqref{eq:case_b},~\eqref{eq:case_c} and~\eqref{eq:case_d} in~\eqref{eq:sum_reject}.
\end{enumerate}
\end{proof}

\subsection{Proof of Proposition~\ref{prop:reduction}}
\label{sec:reduction}

We start by proving the following lemma.

\begin{lemma}
\label{lemma:pushforward}
Let $\nu$ be a limit point of the sequence of laws $(\nu_d)_{d\geq 1}$ of $\{(L_{t}^d)_{t\geq 0} \, :\, d \in \nsets\} $. Then for any $t\geq 0$, the pushforward measure of $\nu$ by $W_t$ is $\piLaplace(\rmd x) = \exp(-\vert x\vert)\rmd x / 2$.
\end{lemma}
\begin{proof}
Using~\eqref{eq:Yt1}, we have
\begin{align}
&\expe{\left\vert L_{t}^d-X_{\floor{ d^{2\alpha}t },1}^d \right\vert} \\
&\leq \expe{\left\vert (d^{2\alpha}t- \floor{d^{2\alpha}t})\left[\sigma_d Z^d_{\ceil{d^{2\alpha}t}, 1}-\frac{\sigma^2_d}{2}\sgn(X^d_{\floor{d^{2\alpha}t}, 1})\right]\1\left\lbrace\vert X^d_{\floor{d^{2\alpha}t}, 1}\vert \geq\sigma^{2v}_dr/2\right\rbrace\1_{\acc^d_{\ceil{d^{2\alpha}t}}}\right\vert}\\
& +\expe{\left\lvert(d^{2\alpha}t- \floor{d^{2\alpha}t})\left[\sigma_d Z^d_{\ceil{d^{2\alpha}t}, 1}-\frac{1}{\sigma_d^{2(v-1)}r}X^d_{\floor{d^{2\alpha}t}, 1}\right]\1\left\lbrace\vert X^d_{\floor{d^{2\alpha}t}, 1}\vert < \sigma^{2v}_dr/2\right\rbrace\1_{\acc^d_{\ceil{d^{2\alpha}t}}}  \right\rvert}\\
&\leq (d^{2\alpha}t- \floor{d^{2\alpha}t})\left(\sigma_d\expe{\left\lvert Z^d_{\ceil{d^{2\alpha}t}, 1}\right\rvert}+\frac{\sigma_d^2}{2}\expe{\left\lvert\sgn(X^d_{\floor{d^{2\alpha}t}, 1}) \right\rvert }\right.\\
&\left.\qquad+\frac{1}{\sigma_d^{2(v-1)}r}\expe{\left\lvert X^d_{\floor{d^{2\alpha}t}, 1}\right\rvert\1\left\lbrace\vert X^d_{\floor{d^{2\alpha}t}, 1}\vert < \sigma^{2v}_dr/2\right\rbrace }\right)\\
&\leq (d^{2\alpha}t- \floor{d^{2\alpha}t})
\left(\frac{\ell}{d^{\alpha}}\expe{(Z^d_{\ceil{d^{2\alpha}t}, 1})^2}^{1/2}+\frac{\ell^2}{2d^{2\alpha}}+\frac{1}{\sigma_d^{2(v-1)}r}\expe{\frac{\sigma^{2v}_dr}{2} }\right)\\
&\leq (d^{2\alpha}t- \floor{d^{2\alpha}t})\left(\frac{\ell}{d^{\alpha}}+\frac{\ell^2}{2d^{2\alpha}}+\frac{\ell^2}{2d^{2\alpha}}\right) \leq \frac{C}{d^{\alpha}},
\end{align}
where we used Cauchy-Schwarz inequality and the fact that the moments of $Z^d_{\ceil{d^{2\alpha}t}, 1}$ are bounded.
The above guarantees that, 
\begin{equation}
\lim_{d \rightarrow  \infty} \expe{\left\vert L_{t}^d-X_{\floor{ d^{2\alpha}t } ,1}^d \right\vert} = 0.
\end{equation}
As $(\nu_d)_{d\geq 1}$ converges weakly towards $\nu$, for any Lipschitz bounded function $\psi : \real \to \real$,
\begin{equation}
\lim_{d \rightarrow	\infty} \expe{\psi \left( X_{\floor{ d^{2\alpha}t } ,1}^d\right) } = \lim_{d \rightarrow	\infty} \expe{\psi \left( L_{t}^d\right) } = \mathbb{E}^{\nu} \left[  \psi(W_t)\right]. 
\end{equation}
The result follows since $X_{\floor{ d^{2\alpha}t } ,1}^d$ is distributed according to $\piLaplace(\rmd x) = \exp(-\vert x\vert)\rmd x / 2$ for any $t\geq 0$ and $d \in \mathbb{N}$.
\end{proof}

We are now ready to prove Proposition~\ref{prop:reduction}:
\begin{proof}[Proof of Proposition~\ref{prop:reduction}]
Let $\nu$ be a limit point of $(\nu_d)_{d\geq 1}$. We start by showing that if for any $V \in \rmC^{\infty}_\rmc(\real,\real)$, $m\in \mathbb{N}$, any bounded and continuous mapping $\rho: \real^m \rightarrow \real$ and any $0 \leq t_1 \leq \dots \leq t_m \leq s \leq t$, $\nu$ satisfies
\begin{equation} \label{eq:condition_pb_mart}
\mathbb{E}^{\nu} \left[ \left( V(W_t)-V(W_s)-\int_s^t \rmL V(W_u) \rmd u \right)\rho(W_{t_1},\dots,W_{t_m}) \right] = 0,
\end{equation}
then $\nu$ is a solution to the martingale problem associated with $\rmL$.

Let $\mathfrak{F}_s$ denote the $\sigma$-algebra generated by
\begin{equation}
\left\lbrace \rho(W_{t_1},\dots,W_{t_m})  \, :\, m\in \mathbb{N},\ \rho: \real^m \rightarrow \real\textrm{ bounded and continuous, and }0 \leq t_1 \leq \dots \leq t_m \leq s  \right\rbrace.
\end{equation}
Then,
\begin{equation}
\mathbb{E}^\nu \left[ V(W_t)-V(W_s)-\int_s^t \rmL V(W_u) \rmd u \middle| \mathfrak{F}_s \right] = 0,
\end{equation}
showing that the process 
\begin{equation}
\left(V(W_t)-V(W_0)-\int_0^t \rmL V(W_u) \rmd u \right)_{t\geq 0}
\end{equation}
is a martingale w.r.t. $\nu$ and the filtration $(\mathfrak{F}_t)_{t\geq 0}$.

To prove \eqref{eq:condition_pb_mart}, it is enough to show that for any $V \in \rmC^{\infty}_\rmc(\real,\real)$, $m\in \mathbb{N}$ and any bounded and continuous mapping $\rho: \real^m \rightarrow \real$ and any $0 \leq t_1 \leq \dots \leq t_m \leq s \leq t$, the mapping
\begin{equation}
\Psi_{s,t} : w \longmapsto \left( V(w_t)-V(w_s)-\int_s^t \rmL V(w_u) \rmd u \right) \rho\left( w_{t_1},\dots,w_{t_m}\right),
\end{equation}
is continuous on a $\nu$-almost sure subset of $  \rmC(\rset_+, \real)$.
Let
\begin{equation}
\bmW=\left\lbrace w \in  \rmC(\rset_+, \real)  \, : \,  w_u \neq 0 \textrm{ for almost any } u\in [s,t] \right\rbrace.
\end{equation}
Since $w \in \bmW^c$ if and only if $\int_s^t \1_{\lbrace 0\rbrace}(w_u) \rmd u >0$, using Lemma~\ref{lemma:pushforward} and the Fubini--Tonelli's theorem,
\begin{equation}
\mathbb{E}^\nu \left[ \int_s^t \1_{\lbrace 0\rbrace}(W_u) \rmd u \right] = \int_s^t \mathbb{E}^\nu \left[ \1_{\lbrace 0\rbrace}(W_u) \right] \rmd u =\int_s^t \piLaplace(\lbrace 0\rbrace ) \rmd u = 0,
\end{equation}
and we have that $\nu ( \bmW^c ) = 0$. 

Since $w \mapsto w_u$ is continuous for any $u\geq 0$, so are $w \mapsto V(w_u)$ and $w \mapsto \rho(w_{t_1},\dots,w_{t_m})$. Thus, it is enough to prove that the mapping $w \mapsto \int_s^t \rmL V(w_u) \rmd u$ is continuous. Let $(w^n)_{n\geq 0}$ be a sequence in $\rmC(\rset_+, \real)$ that converges to $w \in \bmW$ in the uniform topology on compact sets. Let $u$ be such that $w_u \neq 0$, therefore, since the $\sgn$ function is continuous in a neighbourhood of $w_u$,
$
\lim_{n \to \infty} \rmL V (w^n_u)= \rmL V (w_u)
$, thus $
\lim_{n \to \infty} \rmL V (w^n_u)= \rmL V(w_u)$  for almost any $u \in [s,t]$.
Finally, using the boundedness of the sequence $( \rmL V (w^n_u) )_{n\geq 0}$ and Lebesgue's dominated convergence theorem,
\begin{equation}
\lim_{n \rightarrow \infty} \int_s^t \rmL V (w^n_u) \rmd u  = \int_s^t \rmL V (w_u) \rmd u,
\end{equation}
which proves that the mappings $\Psi_{s,t}$ are continuous on $\bmW$.
\end{proof}

\subsection{Proof of Theorem~\ref{theo:laplace_diffusion}}
\label{sec:auxiliary_laplace_langevin}

Let us introduce, for any $n\in \mathbb{N}$,  $\mathcal{F}_{n,1}^d = \sigma ( \lbrace X_{k,1}^d, 0\leq k \leq n\rbrace )$, the $\sigma$-algebra generated by the first components of $\{X^d_k  \mid  0\leq k \leq n\}$. We also introduce for any $V \in \rmC^{\infty}_\rmc (\real,\real)$
\begin{align} 
M_n^d (V) & =  \frac{\ell}{d^{\alpha}} \sum_{k=0}^{n-1}V^{\prime}( X_{k,1}^d )\\
\label{eq:def_martingale_approximation}
          &\quad\times\left(b_d\left(X^d_{k,1}, Z^d_{k+1,
          1}\right)\1_{\acc_{k+1}^d}  - \expe{b_d\left(X^d_{k,1}, Z^d_{k+1,
          1}\right)\1_{\acc_{k+1}^d} \middle| \mathcal{F}_{k, 1}^d}\right) \\
&+\frac{\ell^2}{2d^{2\alpha}}\sum_{k=0}^{n-1}V^{\dprime}(X_{k, 1}^d)\notag\\
&\quad\times\left(b_d\left(X^d_{k,1}, Z^d_{k+1,
1}\right)^2\1_{\acc_{k+1}^d} -\expe{b_d\left(
X^d_{k,1}, Z^d_{k+1, 1}\right)^2\1_{\acc_{k+1}^d} \middle|
\mathcal{F}_{k, 1}^d}\right).
\end{align}
where $b_d$ is defined in~\eqref{eq:martingale_b}.

The proof of Theorem~\ref{theo:laplace_diffusion} follows using the sufficient condition in Proposition~\ref{prop:reduction}, the tightness of the sequence $(\nu_d)_{d\geq 1}$ established in Proposition~\ref{prop:tightness} and  Proposition~\ref{prop:martingale_approximation} below.
\begin{proof}
Using Proposition~\ref{prop:tightness},  Proposition~\ref{prop:reduction} and Proposition~\ref{prop:martingale_approximation} below, it is enough to show that for any $V \in \rmC_\rmc^{\infty}(\real,\real), m \geq 1$, any $0\leq t_1 \leq \cdots \leq t_m \leq s \leq t$ and any bounded and continuous mapping $\rho:\real^m \to \real$,
\begin{equation}
\lim_{d \rightarrow \infty}\expe{\left(M^d_{\ceil{ d^{2\alpha}t}}(V) - M^d_{\ceil {d^{2\alpha}s }}(V)\right)\rho(L_{t_1}^d,...,L_{t_m}^d) }=0 ,
\end{equation}
where, for any $n\geq 1$, $M_n^d(V)$  is given by \eqref{eq:def_martingale_approximation}. However, this is straightforwardly obtained by  taking successively the conditional expectations with respect to $\mathcal{F}^d_{k,1}$ for $k=\ceil{ d^{2\alpha}t}, \dots, \ceil{ d^{2\alpha}s}$.
\end{proof}

\begin{prop} \label{prop:martingale_approximation}
For any $0\leq s \leq t$, $V \in \rmC_\rmc(\real,\real)$ we have
\begin{equation}
\label{eq:martingale_approximation}
\lim_{d \rightarrow \infty} \expe{\left\vert V\left( L_{t}^d \right) -  V\left( L_{s}^d \right) - \int_s^t \rmL V \left( L_{u}^d\right) \rmd u - \left( M_{\ceil{ d^{2\alpha}t} }^d \left( V \right) - M_{\ceil{ d^{2\alpha}s } }^d \left( V \right)\right)  \right\vert } = 0,
\end{equation}
where $(L_{t}^d)_{t\geq 0}$ is defined in~\eqref{eq:Yt_differentiable}.
\end{prop}
\begin{proof}
The process $(L^d_t)_{t\geq 0}$ is piecewise linear, thus it has finite variation. For any $\tau\geq 0$, we define
\begin{equation}
\rmd L_{\tau}^d = d^{2\alpha}\sigma_d b_d\left( X^d_{\floor{ d^{2\alpha}\tau},1} , Z^d_{\ceil{ d^{2\alpha}\tau },1})\right) \1_{\acc_{\ceil{ d^{2\alpha}\tau}}^d} \rmd \tau.
\end{equation}
Since $\sigma_d = \ell d^{-\alpha}$ with $\alpha=1/3$ and using the fundamental theorem of calculus for piecewise $C^1$ maps
\begin{align}\label{eq:ito}
V\left( L_{t}^d \right) -  V\left( L_{s}^d \right)  = \ell d^{\alpha} \int_s^t V^{\prime}\left( L_{\tau}^d \right) b_d\left(X^d_{\floor{ d^{2\alpha}\tau},1}, Z^d_{\ceil{d^{2\alpha}\tau}, 1}\right)\1_{\acc_{\ceil{ d^{2\alpha}\tau}}^d} \rmd \tau,
\end{align}
where $b_d$ is defined in~\eqref{eq:martingale_b}.
A Taylor expansion of $V^{\prime}$ with Lagrange remainder about  $X^d_{\floor{ d^{2\alpha}\tau },1}$ gives 
\begin{align}
V^{\prime} \left( L^d_{\tau}\right) &= V^{\prime} \left( X^d_{\floor {d^{2\alpha}\tau},1} \right) \\
&+\frac{\ell}{d^{\alpha}} \left( d^{2\alpha}\tau - \floor{ d^{2\alpha}\tau } \right) V^{\dprime}\left( X^d_{\floor{ d^{2\alpha}\tau },1}\right)  b_d\left(X^d_{\floor{ d^{2\alpha}\tau},1}, Z^d_{\ceil{d^{2\alpha}\tau}, 1}\right)\1_{\acc_{\ceil{ d^{2\alpha}\tau}}^d} \\
&+\frac{\ell^2}{2d^{2\alpha}} \left( d^{2\alpha}\tau - \floor{ d^{2\alpha}\tau } \right)^2 V^{(3)}\left( \chi_\tau \right)  b_d\left(X^d_{\floor{ d^{2\alpha}\tau},1}, Z^d_{\ceil{d^{2\alpha}\tau}, 1}\right)\1_{\acc_{\ceil{ d^{2\alpha}\tau}}^d},
\end{align}
where for any point $\tau \in [s,t]$, there exists $\chi_\tau \in [X^d_{\floor {d^{2\alpha}\tau },1}, L_{\tau}^d ]$.
Substituting the above into~\eqref{eq:ito} we obtain
\begin{align}
&V\left( L_{t}^d \right) -  V\left( L_{s}^d \right) = \ell d^{\alpha} \int_s^t V^{\prime}\left( X^d_{\floor {d^{2\alpha}\tau },1} \right)b_d\left(X^d_{\floor{ d^{2\alpha}\tau},1}, Z^d_{\ceil{d^{2\alpha}\tau}, 1}\right)\1_{\acc_{\ceil{ d^{2\alpha}\tau}}^d} \rmd \tau\\
\label{eq:martingale_taylor}
&\qquad+\ell^2  \int_s^t \left( d^{2\alpha}\tau - \floor{ d^{2\alpha}\tau } \right) V^{\dprime}\left( X^d_{\floor{ d^{2\alpha}\tau },1}\right)  b_d\left(X^d_{\floor{ d^{2\alpha}\tau},1}, Z^d_{\ceil{d^{2\alpha}\tau}, 1}\right)^2\1_{\acc_{\ceil{ d^{2\alpha}\tau}}^d}\rmd \tau \\
& \qquad+\frac{\ell^3}{2d^{\alpha}} \int_s^t\left( d^{2\alpha}\tau - \floor{ d^{2\alpha}\tau } \right)^2 V^{(3)}\left( \chi_\tau \right)  b_d\left(X^d_{\floor{ d^{2\alpha}\tau},1}, Z^d_{\ceil{d^{2\alpha}\tau}, 1}\right)^3\1_{\acc_{\ceil{ d^{2\alpha}\tau}}^d}\rmd \tau.
\end{align}
Since $V^{(3)}$ is bounded, using Fubini-Tonelli's theorem and recalling the definition of $b_d$ in~\eqref{eq:martingale_b}, we have that
\begin{align}
&\frac{\ell^3}{2d^{\alpha}} \expe{\left\lvert\int_s^t\left( d^{2\alpha}\tau - \floor{ d^{2\alpha}\tau } \right)^2 V^{(3)}\left( \chi_\tau \right)  b_d\left(X^d_{\floor{ d^{2\alpha}\tau},1}, Z^d_{\ceil{d^{2\alpha}\tau}, 1}\right)^3\1_{\acc_{\ceil{ d^{2\alpha}\tau}}^d}\rmd \tau\right\rvert} \\
& \qquad\leq C\frac{\ell^3}{2d^{\alpha}} \int_s^t\expe{ \left(\left\lvert Z^d_{\ceil{ d^{2\alpha}\tau},1}\right\rvert + \frac{\ell}{2d^\alpha}\right)^3 }\rmd \tau \underset{d\to \infty}{\longrightarrow} 0,
\end{align}
since the moments of $Z^d_{\ceil{ d^{2\alpha}\tau},1}$ are bounded.

For the second term in~\eqref{eq:martingale_taylor}, we observe that most of the integrand is piecewise constant since the process $X^d_{\floor{ d^{2\alpha}\tau },1}$ evolves in discrete time. Then, for any integer $d^{2\alpha}s \leq k \leq d^{2\alpha}t - 1$,
\begin{align}
&\int_{k/d^{2\alpha}}^{(k+1)/d^{2\alpha}}  \left( d^{2\alpha}\tau - \floor{ d^{2\alpha}\tau } \right) V^{\dprime}\left( X^d_{\floor{ d^{2\alpha}\tau },1}\right) b_d\left(X^d_{\floor{ d^{2\alpha}\tau},1}, Z^d_{\ceil{d^{2\alpha}\tau}, 1}\right)^2 \1_{\acc_{\ceil{ d^{2\alpha}\tau }}^d} \rmd \tau \\
&\qquad =\frac{1}{2d^{2\alpha}}V^{\dprime}\left( X^d_{k,1}\right)  b_d\left(X^d_{k,1}, Z^d_{k+1, 1}\right)^2 \1_{\acc_{k+1}^d} \\
& \qquad =\frac{1}{2}\int_{k/d^{2\alpha}}^{(k+1)/d^{2\alpha}} V^{\dprime}\left( X^d_{\floor{ d^{2\alpha}\tau },1}\right)  b_d\left(X^d_{\floor{ d^{2\alpha}\tau},1}, Z^d_{\ceil{d^{2\alpha}\tau}, 1}\right)^2 \1_{\acc_{\ceil{ d^{2\alpha}\tau }}^d} \rmd \tau .
\end{align}
Thus, we can write
\begin{align}
I &= \int_s^t \left( d^{2\alpha}\tau - \floor{ d^{2\alpha}\tau } \right) V^{\dprime}\left( X^d_{\floor{ d^{2\alpha}\tau },1}\right)  b_d\left(X^d_{\floor{ d^{2\alpha}\tau},1}, Z^d_{\ceil{d^{2\alpha}\tau}, 1}\right)^2 \1_{\acc_{\ceil{ d^{2\alpha}\tau }}^d} \rmd \tau \\
&= I_1 + I_2  ,
\end{align}
where we define
\begin{equation}
I_2 := \frac{1}{2} \int_s^t V^{\dprime}\left( X^d_{\floor{ d^{2\alpha}\tau },1}\right)  b_d\left(X^d_{\floor{ d^{2\alpha}\tau},1}, Z^d_{\ceil{d^{2\alpha}\tau}, 1}\right)^2 \1_{\acc_{\ceil{ d^{2\alpha}\tau }}^d} \rmd \tau ,
\end{equation}
and
\begin{multline}
I_1 := \left[ \int_s^{\ceil{ d^{2\alpha}s} /d^{2\alpha}}+\int_{\floor{ d^{2\alpha}t} /d^{2\alpha}}^t \right] \left( d^{2\alpha}\tau - \floor{ d^{2\alpha}\tau } - \frac{1}{2}\right) V^{\dprime}\left( X^d_{\floor{ d^{2\alpha}\tau },1}\right) \\  \times b_d\left(X^d_{\floor{ d^{2\alpha}\tau},1}, Z^d_{\ceil{d^{2\alpha}\tau}, 1}\right)^2 \1_{\acc_{\ceil{ d^{2\alpha}\tau }}^d} \rmd \tau .
\end{multline}
In addition, we have
\begin{multline}
I_1 = \frac{1}{2 d^{2\alpha}}\left( d^{2\alpha}s - \floor{ d^{2\alpha}s }\right) \left( \ceil{ d^{2\alpha}s } - d^{2\alpha}s \right) V^{\dprime}\left( X^d_{\floor{ d^{2\alpha}s },1}\right) 
b_d\left(X^d_{\floor{ d^{2\alpha}s},1}, Z^d_{\ceil{d^{2\alpha}s}, 1}\right)^2 \1_{\acc_{\ceil{ d^{2\alpha}s }}^d} \\
+ \frac{1}{2d^{2\alpha}}\left( d^{2\alpha}t - \floor{ d^{2\alpha}t }\right) \left( \ceil{ d^{2\alpha}t } - d^{2\alpha}t \right) V^{\dprime}\left( X^d_{\floor{ d^{2\alpha}t },1}\right) 
b_d\left(X^d_{\floor{ d^{2\alpha}t},1}, Z^d_{\ceil{d^{2\alpha}t}, 1}\right)^2 \1_{\acc_{\ceil{ d^{2\alpha}t }}^d},
\end{multline}
and, since $V^{\dprime}$ and the moments of $Z^d_{\ceil{ d^{2\alpha}t},1}$ are bounded, $\lim_{d \to \infty}\expe{\left\vert I_1 \right\vert} = 0$. Thus, 
\begin{equation}
\lim_{d \to  \infty} \expe{\left\vert V\left( L_{t}^d \right) -  V\left( L_{s}^d \right)  - I_{s,t} \right\vert } = 0 ,
\end{equation}
where
\begin{align}
\label{eq:Ist}
    I_{s,t} &= \int_s^t \left\lbrace  \ell d^{\alpha} V^{\prime} \left( X^d_{\floor{ d^{2\alpha}\tau },1} \right) b_d\left(X^d_{\floor{ d^{2\alpha}\tau},1}, Z^d_{\ceil{d^{2\alpha}\tau}, 1}\right)\right. \\
 &\qquad\qquad+ \left.  \frac{\ell^2}{2} V^{\dprime} \left( X^d_{\floor{
d^{2\alpha}\tau },1}\right)b_d\left(X^d_{\floor{ d^{2\alpha}\tau},1},
Z^d_{\ceil{d^{2\alpha}\tau}, 1}\right)^2 \1_{\acc_{\ceil{
d^{2\alpha}\tau }}^d} \right\rbrace \rmd \tau  .
\end{align}
Next, we use~\eqref{eq:laplace_generator} and write
\begin{align}\label{eq:ecriture_generateur}
\int_s^t \rmL V \left( L^d_{\tau}\right) \rmd \tau &= \int_s^t  \frac{h^\rmL(\ell
)}{2}\left[V^{\dprime}\left( X_{\floor{d^{2\alpha}\tau }1}^d \right) - \sgn
\left( X_{\floor{ d^{2\alpha}\tau },1}^d\right)V^{\prime}\left( X_{\floor{
d^{2\alpha}\tau },1}^d \right) \right] \rmd \tau - T_3^d ,
\end{align}
where we define
\begin{equation}
T_3^d = \int_s^t \left( \rmL V \left( X_{\floor{ d^{2\alpha}\tau },1}^d\right) - \rmL V \left( L^d_{\tau}\right) \right) \rmd \tau.
\end{equation}

Finally, we write the difference $M^d_{\ceil{ d^{2\alpha}t }}(V) -  M^d_{\ceil{ d^{2\alpha}s }}(V)$ as the integral of a piecewise constant function
\begin{align} 
\label{eq:app_mart_integrale}
&M^d_{\ceil{ d^{2\alpha}t }}(V) -  M^d_{\ceil{ d^{2\alpha}s }}(V) =I_{s,t}\\
&\qquad- \int_s^t  \left( \ell d^{\alpha} V^{\prime} \left( X_{\floor{
d^{2\alpha}\tau },1}^d \right) \expe{  b_d\left(X^d_{\floor{
d^{2\alpha}\tau},1}, Z^d_{\ceil{d^{2\alpha}\tau}, 1}\right)\1_{\acc_{\ceil{
d^{2\alpha}\tau }}^d} \middle| \mathcal{F}_{\floor{ d^{2\alpha}\tau },1}^d }
\right.\\
&\qquad \left.+ \frac{\ell^2}{2} V^{\dprime} \left( X_{\floor{ d^{2\alpha}\tau
},1}^d \right) \expe{ b_d\left(X^d_{\floor{ d^{2\alpha}\tau},1},
Z^d_{\ceil{d^{2\alpha}\tau}, 1}\right)^2 \1_{\acc_{\ceil{ d^{2\alpha}\tau }}^d}
\middle| \mathcal{F}_{\floor{ d^{2\alpha}\tau },1}^d } \right) \rmd \tau \\
&\qquad\qquad  - T_4^d - T_5^d  , \notag
\end{align}
where $T_4^d$ and $T_5^d$ account for the difference between the sum in~\eqref{eq:def_martingale_approximation} and the integral, and are defined as
\begin{align}
T_4^d &= - \frac{\ell}{d^{\alpha}} \left( \ceil{ d^{2\alpha}t } - d^{2\alpha}t \right) V^{\prime} \left(X_{\floor{ d^{2\alpha}t },1}^d  \right) \left\lbrace b_d\left(X^d_{\floor{ d^{2\alpha}t},1}, Z^d_{\ceil{d^{2\alpha}t}, 1}\right) \1_{\acc_{\ceil{ d^{2\alpha}t }}^d} \right. \\
&\qquad\left. - \expe{  b_d\left(X^d_{\floor{ d^{2\alpha}t},1}, Z^d_{\ceil{d^{2\alpha}t}, 1}\right) \1_{\acc_{\ceil{ d^{2\alpha}t }}^d} \middle| \mathcal{F}_{\floor{ d^{2\alpha}t },1}^d } \right\rbrace \\
& - \frac{\ell^2}{2d^{2\alpha}} \left( \ceil{ d^{2\alpha}t } - d^{2\alpha}t \right) V^{\dprime} \left(X_{\floor{ d^{2\alpha}t },1}^d  \right) \left\lbrace b_d\left(X^d_{\floor{ d^{2\alpha}t},1}, Z^d_{\ceil{d^{2\alpha}t}, 1}\right)^2 \1_{\acc_{\ceil{ d^{2\alpha}t }}^d} \right. \\
&\qquad \left. - \expe{  b_d\left(X^d_{\floor{ d^{2\alpha}t},1}, Z^d_{\ceil{d^{2\alpha}t}, 1}\right)^2 \1_{\acc_{\ceil{ d^{2\alpha}t }}^d} \middle| \mathcal{F}_{\floor{ d^{2\alpha}t },1}^d } \right\rbrace  ,
\end{align}
and $T_5^d = -T_4^d$ with $t$ substituted by $s$.
Putting~\eqref{eq:Ist}, \eqref{eq:ecriture_generateur} and~\eqref{eq:app_mart_integrale} together we obtain 
\begin{equation}
I_{s,t} - \int_s^t \rmL V \left( L^d_{\tau}\right) \rmd \tau - \left( M^d_{\ceil{ d^{2\alpha}t }}(V) -  M^d_{\ceil{ d^{2\alpha}s }}(V) \right) = T_1^d + T_2^d + T_3^d + T_4^d + T_5^d,
\end{equation}
where $T_1^d$ takes into account all the terms involving $V^{\prime}( X_{\floor{ d^{2\alpha}\tau },1}^d )$, and $T_2^d$ the terms involving $V^{\dprime}( X_{\floor{ d^{2\alpha}\tau },1}^d)$:
\begin{align}
T_1^d &= \int_s^t V^{\prime}\left( X_{\floor{ d^{2\alpha}\tau },1}^d \right)\\
&\times\left\lbrace \ell d^{\alpha} \expe{ b_d\left(X^d_{\floor{ d^{2\alpha}\tau},1}, Z^d_{\ceil{d^{2\alpha}\tau}, 1}\right) \1_{\acc_{\ceil{ d^{2\alpha}\tau }}^d} \middle| \mathcal{F}_{\floor{ d^{2\alpha}\tau },1}^d }  +\frac{ h^\rmL(\ell) }{2}\sgn \left( X_{\floor{ d^{2\alpha}\tau },1}^d\right) \right\rbrace \rmd \tau ,\\
T_2^d &= \int_s^t V^{\dprime}\left( X_{\floor{ d^{2\alpha}\tau },1}^d \right)\\
&\times\left\lbrace \frac{\ell^2}{2} \expe{  b_d\left(X^d_{\floor{ d^{2\alpha}\tau},1}, Z^d_{\ceil{d^{2\alpha}\tau}, 1}\right)^2 \1_{\acc_{\ceil{ d^{2\alpha}\tau }}^d} \middle| \mathcal{F}_{\floor{ d^{2\alpha}\tau },1}^d }- \frac{h^\rmL(\ell)}{2} \right\rbrace \rmd \tau .
\end{align}
To obtain~\eqref{eq:martingale_approximation} it is then sufficient to prove that for any $1 \leq i \leq 5$, $\lim_{d \to  \infty}\expe{\left\vert T_i^d\right\vert} = 0$. 

Since $V^{\prime}, V^{\dprime}$ are bounded and $b_d$ is bounded in expectation because the moments of $Z^d_{\ceil{ d^{2\alpha}\tau},1}$ are bounded, it is easy to show that $\lim_{d \to  \infty}\expe{\left\vert T_i^d\right\vert} = 0$ for $i=4, 5$.
For $T_3^d$, we write $T_3^d = h^\rmL(\ell)(T_{3,1}^d - T_{3,2}^d)/2$, where  
\begin{align}
T_{3,1}^d &= \int_s^t \left\lbrace V^{\dprime}\left( X_{\floor {d^{2\alpha}\tau},1}^d \right) - V^{\dprime}\left( L_{\tau}^d \right)\right\rbrace \rmd \tau  , \\
T_{3,2}^d &= \int_s^t \left\lbrace \sgn \left( X_{\floor{ d^{2\alpha}\tau },1}^d\right)V^{\prime}\left( X_{\floor{ d^{2\alpha}\tau},1}^d \right) - \sgn \left( L_{\tau}^d\right)V^{\prime}\left( L_{\tau}^d \right)\right\rbrace \rmd \tau  .
\end{align}
Using Fubini-Tonelli's theorem, the convergence of $X_{\floor{ d^{2\alpha}\tau},1}^d$ to $L_{\tau}^d$ in Lemma~\ref{lemma:pushforward} and Lebesgue's dominated convergence theorem we obtain
\begin{equation}
    \expe{\left\vert T_{3,1}^d \right\vert} \leq \int_s^t \expe{\left\vert V^{\dprime}\left( X_{\floor{ d^{2\alpha}\tau},1}^d \right) - V^{\dprime}\left( L_{\tau}^d \right)  \right\vert  } \rmd \tau  \underset{d \to \infty}{\longrightarrow} 0.
\end{equation}
We can further decompose $T_{3,2}^d$ as
\begin{multline}
T_{3,2}^d = \int_s^t \left\lbrace \sgn \left( X_{\floor{ d^{2\alpha}\tau },1}^d\right) - \sgn \left( L_{\tau}^d\right)\right\rbrace V^{\prime}\left( X_{\floor{ d^{2\alpha}\tau},1}^d \right) \rmd \tau \\ 
+ \int_s^t \sgn \left( L_{\tau}^d\right) \left\lbrace V^{\prime}\left( X_{\floor{ d^{2\alpha}\tau},1}^d \right) - V^{\prime}\left( L_{\tau}^d \right)\right\rbrace \rmd \tau .
\end{multline}
Proceeding as for $T_{3,1}^d$, it is easy to show that the second integral converges to 0 as $d\to \infty$.
We then bound the first integral by
\begin{align}
&\expe{\left\vert \int_s^t \left\lbrace \sgn \left( X_{\floor{ d^{2\alpha}\tau},1}^d\right) - \sgn \left( L_{\tau}^d\right)\right\rbrace V^{\prime}\left( X_{\floor{ d^{2\alpha}\tau},1}^d \right) \rmd \tau \right\vert }\\
&\qquad\leq C\int_s^t\expe{\left\vert \sgn \left( X_{\floor{ d^{2\alpha}\tau},1}^d\right) - \sgn \left( L_{\tau}^d\right) \right\vert} \rmd \tau .
\end{align}
However, since $\lbrace\sgn ( X_{\floor{ d^{2\alpha}\tau},1}^d) \neq \sgn ( L_{\tau}^d) \rbrace \subset \lbrace \sgn ( X_{\floor{ d^{2\alpha}\tau},1}^d) \neq \sgn ( X_{\ceil{ d^{2\alpha}\tau},1}^d) \rbrace$, using Lemma~\ref{lemma:sgn_expe} in Appendix~\ref{app:sgn_expe} we have that
\begin{align}
\expe{\left\vert \sgn \left( X_{\floor{ d^{2\alpha}\tau},1}^d\right) - \sgn
\left( L_{\tau}^d\right) \right\vert} 
&= 2\expe{ \1\left\lbrace \sgn \left( X_{\floor{ d^{2\alpha}\tau},1}^d\right)
\neq \sgn \left( L_{\tau}^d\right) \right\rbrace} \\
&=2\expe{ \1\left\lbrace \sgn \left( X_{\floor{ d^{2\alpha}\tau},1}^d\right)
\neq \sgn \left( X_{\ceil{ d^{2\alpha}\tau},1}^d\right) \right\rbrace}
\underset{d \to \infty}{\longrightarrow} 0 .
\end{align}
The above and the dominated converge theorem show that 
\begin{align}
    \expe{\left\vert \int_s^t \left\lbrace \sgn \left( X_{\floor{ d^{2\alpha}\tau},1}^d\right) - \sgn \left( L_{\tau}^d\right)\right\rbrace V^{\prime}\left( X_{\floor{ d^{2\alpha}\tau},1}^d \right) \rmd \tau \right\vert } \underset{d \to \infty}{\longrightarrow} 0.
\end{align}
Consider then $T_1^d$, recalling that the derivatives of $V$ are bounded, we have
\begin{align}
\expe{\left\vert T_1^d\right\vert } &\leq \int_s^t C \mathbb{E}\left[\left\lvert\ell d^{\alpha} \expe{ b_d\left(X^d_{\floor{ d^{2\alpha}\tau},1}, Z^d_{\ceil{d^{2\alpha}\tau}, 1}\right) \1_{\acc_{\ceil{ d^{2\alpha}\tau }}^d} \middle| \mathcal{F}_{\floor{ d^{2\alpha}\tau },1}^d } \right.\right.\\
&\qquad+ \left.\left.\frac{h^\rmL(\ell)}{2} \sgn \left( X_{\floor{ d^{2\alpha}\tau },1}^d\right)\right\rvert\right]\rmd \tau\\
&\leq \int_s^t C \left\lbrace \expe{\left\vert D_{1,\tau}^{(1)}\right\vert }  + \expe{\left\vert D_{2,\tau}^{(1)}\right\vert } \right\rbrace \rmd \tau,\notag
\end{align}
where we define
\begin{align}
\label{eq:def_dr} 
D_{1,\tau}^{(1)} &=  \ell d^{\alpha} \expe{ Z^d_{\ceil{ d^{2\alpha}\tau} , 1} \1_{\acc^d_{\ceil{ d^{2\alpha}\tau }}}  \middle| \mathcal{F}^d_{\floor{ d^{2\alpha}\tau },1} } , \\
D_{2,\tau}^{(1)} &=   \frac{h^\rmL(\ell)}{2}\sgn \left(  X^d_{\floor{ d^{2\alpha}\tau} ,1}\right) \notag\\
&- \ell d^{\alpha}\left(\frac{\sigma_d}{2}\sgn(X^d_{\floor{ d^{2\alpha}\tau} ,1})\1_{\vert X^d_{\floor{ d^{2\alpha}\tau} ,1}\vert \geq \sigma_d^{2v} r/2}+\frac{ 1}{\sigma_d^{2v-1}r}  X^d_{\floor{ d^{2\alpha}\tau} ,1}\1_{\vert X^d_{\floor{ d^{2\alpha}\tau} ,1}\vert < \sigma_d^{2v} r/2}\right)\notag\\
&\qquad\qquad\times\expe{\1_{\acc^d_{\ceil{ d^{2\alpha}\tau}} } \middle| \mathcal{F}^d_{\floor{ d^{2\alpha}\tau},1} }    .\notag
\end{align}
Let us start with $D_{1, \tau}^{(1)}$:
\begin{equation}
D_{1,\tau}^{(1)} =  \ell d^{\alpha}  \expe{ Z^d_{\ceil{ d^{2\alpha}\tau } , 1} \left( 1 \wedge \exp \left\lbrace \sum_{i=1}^d \phi_d\left( X^d_{\floor{ d^{2\alpha}\tau },i},Z^d_{\ceil{ d^{2\alpha}\tau } , i}\right) \right\rbrace  \right)    \middle| \mathcal{F}^d_{\floor{ d^{2\alpha}\tau },1} } ,
\end{equation}
where $\phi_d$ is given in~\eqref{eq:laplace_ar}.
Then, by independence of the components of $Z^d_{\ceil{ d^{2\alpha}\tau }}$, we have
\begin{multline}
\expe{ Z^d_{\ceil{ d^{2\alpha}\tau } , 1} \left( 1 \wedge \exp \left\lbrace \sum_{i=2}^d \phi_d\left( X^d_{\floor{ d^{2\alpha}\tau },i},Z^d_{\ceil{ d^{2\alpha}\tau } , i}\right) \right\rbrace  \right)    \middle| \mathcal{F}^d_{\floor{ d^{2\alpha}\tau },1} } \\
= \expe{Z^d_{\ceil{ d^{2\alpha}\tau } , 1}} \expe{ 1 \wedge \exp \left\lbrace \sum_{i=2}^d \phi_d\left( X^d_{\floor{ d^{2\alpha}\tau },i},Z^d_{\ceil{ d^{2\alpha}\tau } , i}\right) \right\rbrace   \middle| \mathcal{F}^d_{\floor{ d^{2\alpha}\tau },1} } = 0  .
\end{multline}
This allows us to write
\begin{multline}
\expe{\vert  D_{1,\tau}^{(1)}\vert} \leq \ell d^{\alpha} \expe{ \vert Z^d_{\ceil{ d^{2\alpha}\tau } , 1}\vert  \right. \\ 
\left. \left\lvert 1 \wedge \exp \left\lbrace \sum_{i=1}^d \phi_d\left( X^d_{\floor{ d^{2\alpha}\tau },i},Z^d_{\ceil{ d^{2\alpha}\tau } , i}\right) \right\rbrace - 1 \wedge \exp \left\lbrace \sum_{i=2}^d \phi_d\left( X^d_{\floor{ d^{2\alpha}\tau },i},Z^d_{\ceil{ d^{2\alpha}\tau } , i}\right) \right\rbrace\right\rvert } .
\end{multline}
However, $x \mapsto 1 \wedge \exp(x)$ is a 1-Lipschitz function, thus
\begin{equation}
\expe{\vert D_{1,\tau}^{(1)}\vert} \leq\ell d^{\alpha} \expe{\vert Z^d_{\ceil{ d^{2\alpha}\tau } , 1}\vert \left\lvert \phi_d\left( X^d_{\floor{ d^{2\alpha}\tau },1},Z^d_{\ceil{ d^{2\alpha}\tau } , 1}\right)\right\rvert }, 
\end{equation}
and $D_{1, \tau}^{(1)}\to 0$ as $d \to \infty$ by Lemma~\ref{lemma:integral_D1tau} in Appendix~\ref{app:sgn_expe}.

For $D_{2, \tau}^{(1)}$, we observe that
\begin{align}
\label{eq:x_bound}
&-\frac{\sigma_d}{2}\1_{\vert X^d_{\floor{ d^{2\alpha}\tau} ,1}\vert < \sigma_d^{2v} r/2}\leq \frac{ 1}{\sigma^{2v-1}r}  X^d_{\floor{ d^{2\alpha}\tau} ,1}\1_{\vert X^d_{\floor{ d^{2\alpha}\tau} ,1}\vert < \sigma_d^{2v} r/2}\leq \frac{\sigma_d}{2}\1_{\vert X^d_{\floor{ d^{2\alpha}\tau} ,1}\vert < \sigma_d^{2v} r/2}.
\end{align}
Distinguishing between $X^d_{\floor{ d^{2\alpha}\tau} ,1}< 0$ and $X^d_{\floor{ d^{2\alpha}\tau} ,1}\geq 0$, it follows that
\begin{align}
\vert D_{2,\tau}^{(1)}\vert &\leq \left\lvert \sgn \left(  X^d_{\floor{ d^{2\alpha}\tau} ,1}\right)\right\rvert \\
&\times\left\vert \frac{h^\rmL(\ell)}{2} - \ell d^{\alpha}\left(\frac{\sigma_d}{2}\1_{\vert X^d_{\floor{ d^{2\alpha}\tau} ,1}\vert \geq \sigma_d^{2v} r/2}+\frac{ \sigma_d}{2}  \1_{\vert X^d_{\floor{ d^{2\alpha}\tau} ,1}\vert < \sigma_d^{2v} r/2}\right)\expe{\1_{\acc^d_{\ceil{ d^{2\alpha}\tau}} } \middle| \mathcal{F}^d_{\floor{ d^{2\alpha}\tau},1} } \right\rvert \\
&\leq \frac{1}{2}\left\lvert  h^\rmL(\ell) - \ell^2 \expe{\1_{\acc^d_{\ceil{ d^{2\alpha}\tau}} } \middle| \mathcal{F}^d_{\floor{ d^{2\alpha}r \rfloor,1} }} \right\rvert,
\end{align}
where we recall that $\sigma_d = \ell d^{-\alpha}$ with $\alpha=1/3$.
Using the triangle inequality we obtain
\begin{align}
&2\expe{\vert D_{2,\tau}^{(1)}\vert} \leq \expe{\left\lvert h^\rmL(\ell)- \ell^2 \expe{1\wedge \exp\left(\sum_{i=1}^d \phi_d\left(X_{\floor{ d^{2\alpha}\tau},i}^d,Z_{\ceil{ d^{2\alpha}\tau},i}^d \right)\right)\middle| \mathcal{F}^d_{\floor{ d^{2\alpha}\tau},1}} \right\rvert} \\
& \leq \expe{\left\lvert h^\rmL(\ell)- \ell^2 \expe{1\wedge \exp\left(\sum_{i=2}^d \phi_d\left( X_{\floor{ d^{2\alpha}\tau},i}^d,Z_{\ceil{ d^{2\alpha}\tau},i}^d\right)\right) \middle| \mathcal{F}^d_{\floor{ d^{2\alpha}\tau},1} }\right\rvert} \\
& + \ell^2\expe{\left\lvert 1\wedge \exp\left(\sum_{i=2}^d \phi_d\left( X_{\floor{ d^{2\alpha}\tau},i}^d,Z_{\ceil{ d^{2\alpha}\tau},i}^d\right)\right)-1\wedge \exp\left(\sum_{i=1}^d \phi_d\left(X_{\floor{ d^{2\alpha}\tau},i}^d,Z_{\ceil{ d^{2\alpha}\tau},i}^d\right)\right)\right\rvert} ,
\end{align}
where we used Jensen's inequality to remove the conditional expectation in the last term. 
Recalling that $x\mapsto 1\wedge \exp(x)$ is 1-Lipschitz, we can then bound the second term 
\begin{align}
&\ell^2\expe{\left\lvert 1\wedge \exp\left(\sum_{i=2}^d \phi_d\left(X_{\floor{ d^{2\alpha}\tau},i}^d,Z_{\ceil{ d^{2\alpha}\tau},i}^d\right)\right)-1\wedge \exp\left(\sum_{i=1}^d \phi_d\left(X_{\floor{ d^{2\alpha}\tau},i}^d,Z_{\ceil{ d^{2\alpha}\tau},i}^d\right)\right)\right\rvert} \\
\label{eq:1lipschitz_bound}
&\qquad \leq \ell^2 \expe{\left\vert \phi_d\left(X_{\floor{ d^{2\alpha}\tau},1}^d,Z_{\ceil{ d^{2\alpha}\tau},1}^d\right)\right\vert }  , \\
&\qquad \leq \ell^2 \expe{\phi_d\left(X_{\floor{ d^{2\alpha}\tau},1}^d,Z_{\ceil{ d^{2\alpha}\tau},1}^d\right)^2}^{1/2},
\end{align}
where the final expectation converges to zero as $d\to \infty$ by Proposition~\ref{prop:laplace_var}.
For the remaining term in $D_{2, \tau}^{(1)}$, since $(X^d_{\floor{ d^{2\alpha}\tau},i},Z^d_{\floor{ d^{2\alpha}\tau} ,i})_{2 \leq i \leq n}$ is independent of $\mathcal{F}^d_{\floor{ d^{2\alpha}\tau},1}$, we have
\begin{align}
&\ell^2 \expe{1\wedge \exp\left(\sum_{i=2}^d \phi_d\left( X_{\floor{ d^{2\alpha}\tau},i}^d,Z_{\ceil{ d^{2\alpha}\tau},i}^d\right)\right) \middle| \mathcal{F}^d_{\floor{ d^{2\alpha}\tau},1} } \\
&\qquad\qquad=\ell^2 \expe{1\wedge \exp\left(\sum_{i=2}^d \phi_d\left( X_{\floor{ d^{2\alpha}\tau},i}^d,Z_{\ceil{ d^{2\alpha}\tau},i}^d\right)\right)  },
\end{align}
and, using again the fact that $x \mapsto 1\wedge \exp(x)$ is 1-Lipschitz, we have
\begin{align}
& \left\lvert h^\rmL(\ell)- \ell^2 \expe{1\wedge \exp\left(\sum_{i=2}^d \phi_d\left(X_{\floor{ d^{2\alpha}\tau},i}^d,Z_{\ceil{ d^{2\alpha}\tau},i}^d\right)\right) }\right\rvert\\
 &\qquad\leq \left\lvert h^\rmL(\ell)- \ell^2 \expe{1\wedge \exp\left(\sum_{i=1}^d \phi_d\left(X_{\floor{ d^{2\alpha}\tau},i}^d,Z_{\ceil{ d^{2\alpha}\tau},i}^d}\right)\right)\right\rvert \\
& \qquad\qquad + \ell^2 \expe{\left\vert \phi_d\left(X_{\floor{ d^{2\alpha}\tau},1}^d,Z_{\ceil{ d^{2\alpha}\tau},1}^d\right)\right\vert}  .
\end{align}
The last term goes to 0 as shown in~\eqref{eq:1lipschitz_bound}, and, as $h^\rmL(\ell)=\ell^2 a^\rmL(\ell)$, with
\begin{equation}
    a^\rmL(\ell)=\lim_{d \to \infty} \expe{1\wedge \exp\left(\sum_{i=1}^d \phi_{d,i}\right) } ,
\end{equation}
by Theorem~\ref{theo:acceptance_laplace}, we obtain
\begin{equation}
\lim_{d \to  \infty}\left\lvert h^\rmL(\ell)- \ell^2 \expe{1\wedge \exp\left(\sum_{i=2}^d \phi_d\left(X_{\floor{ d^{2\alpha}\tau},i}^d,Z_{\ceil{ d^{2\alpha}\tau},i}^d\right)\right) \right\rvert} = 0 ,
\end{equation}
showing that $D_{2, \tau}^{(1)}\to 0$ as $d\to \infty$. 
To obtain convergence of $T_1^d$, we observe that for any $\tau\in[s,t]$,
$D_{1,\tau}^{(1)}$ and $D_{2,\tau}^{(1)}$ follow the same distributions as
$D_{1,s}^{(1)}$ and $D_{2,s}^{(1)}$, since for any $k \in \mathbb{N}$, $X_k^d$
has distribution $\piLaplace_d$. Therefore, the convergence towards zero of
$\expeLine{\vert D_{1,\tau}^{(1)}\vert}$ and $\expeLine{\vert D_{2,\tau}^{(1)}\vert}$
is uniform for $\tau\in [s,t]$, which gives us $T_1^d\to 0$ as $d\to \infty$.

Finally, consider $T_2^d$. Using analogous arguments to those used for $T_1^d$, we obtain
\begin{align}
\expe{\vert T_2^d\vert } &\leq C \int_s^t \frac{\ell^2}{2}\expe{\left\lvert \expe{ b_d\left(X^d_{\floor{ d^{2\alpha}\tau},1}, Z^d_{\ceil{d^{2\alpha}\tau}, 1}\right)^2 \1_{\acc_{\ceil{ d^{2\alpha}\tau }}^d} \middle| \mathcal{F}_{\floor{ d^{2\alpha}\tau },1}^d }  - a^\rmL(\ell, r) \right\rvert}\rmd \tau\\
&\leq C  \int_s^t\frac{\ell^2}{2} \left\lbrace \expe{\vert D_{1,\tau}^{(2)}\vert}  + \expe{\vert D_{2,\tau}^{(2)}\vert } \expe{\vert D_{3,\tau}^{(2)}\vert } \right\rbrace \rmd \tau,\notag
\end{align}
where we define
\begin{align}
\label{eq:def_dr_2} 
D_{1,\tau}^{(2)} &=   \expe{ \left(Z^d_{\ceil{ d^{2\alpha}\tau} , 1} \right)^2\1_{\acc^d_{\ceil{ d^{2\alpha}\tau }}}  \middle| \mathcal{F}^d_{\floor{ d^{2\alpha}\tau },1} } -a^\rmL(\ell, r), \\
D_{2,\tau}^{(2)} &=   \left(\frac{\sigma_d}{2}\sgn(X^d_{\floor{ d^{2\alpha}\tau} ,1})\indicatorDD{\vert X^d_{\floor{ d^{2\alpha}\tau} ,1}\vert \geq \sigma_d^{2v} r/2}\right.\\
&\qquad\qquad\quad+\left.\frac{ 1}{\sigma_d^{2v-1}r}  X^d_{\floor{ d^{2\alpha}\tau} ,1}\indicatorDD{\vert X^d_{\floor{ d^{2\alpha}\tau} ,1}\vert < \sigma_d^{2v} r/2}\right)^2\notag
\times\expe{\1_{\acc^d_{\ceil{ d^{2\alpha}\tau}} } \middle| \mathcal{F}^d_{\floor{ d^{2\alpha}\tau},1} } ,\notag\\
D_{3,\tau}^{(2)} &=   2\left(\frac{\sigma_d}{2}\sgn(X^d_{\floor{ d^{2\alpha}\tau} ,1})\indicatorDD{\vert X^d_{\floor{ d^{2\alpha}\tau} ,1}\vert \geq \sigma_d^{2v} r/2}\right.\\
                 &\left.\qquad+\frac{ 1}{\sigma_d^{2v-1}r}  X^d_{\floor{ d^{2\alpha}\tau} ,1}\indicatorDD{\vert X^d_{\floor{ d^{2\alpha}\tau} ,1}\vert < \sigma_d^{2v} r/2}\right)\notag
\expe{Z^d_{\ceil{ d^{2\alpha}\tau} , 1} \1_{\acc^d_{\ceil{ d^{2\alpha}\tau}} } \middle| \mathcal{F}^d_{\floor{ d^{2\alpha}\tau},1} } .\notag
\end{align}
Using~\eqref{eq:x_bound}, Cauchy-Schwarz's inequality and the fact that the moments of $Z^d_{\ceil{ d^{2\alpha}\tau} , 1}$ are bounded we have
\begin{align}
    \expe{\vert D_{2,\tau}^{(2)}\vert} \leq \frac{\sigma_d^2}{4} \underset{d\to \infty}{\longrightarrow} 0,\qquad
\expe{\vert D_{3,\tau}^{(2)}\vert }\leq C\sigma_d \underset{d\to \infty}{\longrightarrow} 0,
\end{align}
since $\sigma_d = \ell d^{-\alpha}$ with $\alpha= 1/3$.
The remaining term is bounded similarly to $D_{2, \tau}^{(1)}$, using the fact that $x\mapsto 1\wedge \exp(x)$ is 1-Lipschitz, we have 
\begin{align}
&\expe{\vert D_{3,\tau}^{(2)}\vert}\\
&\leq \expe{\left \vert \expe{ \left(Z^d_{\ceil{ d^{2\alpha}\tau} , 1} \right)^2\left(1\wedge \exp\left(\sum_{i=2}^d \phi_d\left(X_{\floor{ d^{2\alpha}\tau},i}^d,Z_{\ceil{ d^{2\alpha}\tau},i}^d\right)\right) \right)  \middle| \mathcal{F}^d_{\floor{ d^{2\alpha}\tau },1} } -a^\rmL(\ell, r)\right \rvert}\\
&\qquad+ \expe{ \left(Z^d_{\ceil{ d^{2\alpha}\tau} , 1} \right)^2 \left \vert\phi_d\left(X_{\floor{ d^{2\alpha}\tau},1}^d,Z_{\ceil{ d^{2\alpha}\tau},1}^d\right)\right \rvert   }.
\end{align}
The second expectation is bounded as~\eqref{eq:1lipschitz_bound} using Cauchy-Schwarz's inequality and Proposition~\ref{prop:laplace_var}. For the first expectation, we use the conditional independence of the components of $Z^d_{\ceil{ d^{2\alpha}\tau}}$ and write
\begin{align}
&\expe{ \left(Z^d_{\ceil{ d^{2\alpha}\tau} , 1} \right)^2\left(1\wedge \exp\left(\sum_{i=2}^d \phi_d\left(X_{\floor{ d^{2\alpha}\tau},i}^d,Z_{\ceil{ d^{2\alpha}\tau},i}^d\right)\right) \right)  \middle| \mathcal{F}^d_{\floor{ d^{2\alpha}\tau },1} } \\
&\qquad\qquad = \expe{\left(Z^d_{\ceil{ d^{2\alpha}\tau} , 1} \right)^2}\expe{ \left(1\wedge \exp\left(\sum_{i=2}^d \phi_d\left(X_{\floor{ d^{2\alpha}\tau},i}^d,Z_{\ceil{ d^{2\alpha}\tau},i}^d\right)\right) \right) }\\
&\qquad\qquad = \expe{ \left(1\wedge \exp\left(\sum_{i=2}^d \phi_d\left(X_{\floor{ d^{2\alpha}\tau},i}^d,Z_{\ceil{ d^{2\alpha}\tau},i}^d\right)\right) \right) }.
\end{align}
It follows that $\expeLine{\vert D_{3,\tau}^{(2)}\vert} \to 0$ as $d\to \infty$ since, by Theorem~\ref{theo:acceptance_laplace},
\begin{align}
\left\lvert \expe{ \left(1\wedge \exp\left(\sum_{i=2}^d \phi_d\left(X_{\floor{ d^{2\alpha}\tau},i}^d,Z_{\ceil{ d^{2\alpha}\tau},i}^d\right)\right) \right) } -a^\rmL(\ell, r)\right\rvert \to 0.
\end{align}
Combining the results for $T_i^d$, $i=1,\dots, 5$ we obtain the result.
\end{proof}


\section{Moments and integrals for the Laplace distribution}
\label{app:laplace_computations}
\subsection{Moments of acceptance ratio for the Laplace distribution}
\label{app:laplace_moments}
The indicator functions in the definition of $\phi_d$ identify four different regions:
\begin{align}
    R_1&:=     \left\lbrace (x, z): \vert x\vert< \sigma^{2v}r/2 \wedge \left\lvert \left(1-\frac{1}{\sigma^{2(v-1)}r}\right)x + \sigma z\right\rvert< \sigma^{2v}r/2 \right\rbrace,\\
    R_2&:=    \left\lbrace(x, z): \vert x\vert\geq \sigma^{2v}r/2  \wedge  \left\lvert x-\frac{\sigma^2}{2}\sgn(x)+\sigma z\right\rvert<  \sigma^{2v}r/2 \right\rbrace,\\
    R_3&:=
    \left\lbrace(x, z): \vert x\vert< \sigma^{2v}r/2 \wedge \left\lvert \left(1-\frac{1}{\sigma^{2(v-1)}r}\right)x + \sigma z\right\rvert\geq\sigma^{2v}r/2 \right\rbrace,\\
    R_4&:=    \left\lbrace (x, z):\vert x\vert> \sigma^{2v}r/2  \wedge  \left\lvert x-\frac{\sigma^2}{2}\sgn(x)+\sigma z\right\rvert> \sigma^{2v}r/2 \right\rbrace,
\end{align}
with corresponding acceptance ratios
\begin{align}
    \phi_d^{1}(x, z) &= \vert x\vert -\left\lvert \left(1-\frac{1}{\sigma^{2(v-1)}r}\right)x + \sigma z\right\rvert+\frac{z^2}{2}\\
    &-\frac{1}{2\sigma^2}\left(\left(\frac{2}{\sigma^{2(v-1)}r}-\frac{1}{\sigma^{4(m-1)}r^2}\right)x-\left(1-\frac{1}{\sigma^{2(v-1)}r}\right) \sigma z\right)^2\\
    \phi_d^2(x, z)  &= \vert x\vert -\left\lvert x-\frac{\sigma^2}{2}\sgn(x)+\sigma z\right\rvert+\frac{z^2}{2}\\
    &-\frac{1}{2\sigma^2}\left(\frac{1}{\sigma^{2(v-1)}r}x+\left(1-\frac{1}{\sigma^{2(v-1)}r}\right)\left(\frac{\sigma^2}{2}\sgn(x)-\sigma z\right)\right)^2\\
    \phi_d^3(x, z) &= \vert x\vert -\left\lvert \left(1-\frac{1}{\sigma^{2(v-1)}r}\right)x + \sigma z\right\rvert+\frac{z^2}{2}\\
    &\qquad-\frac{1}{2\sigma^2}\left(\frac{1}{\sigma^{2(v-1)}r}x -\sigma z+\frac{\sigma^2}{2}\sgn\left[\left(1-\frac{1}{\sigma^{2(v-1)}r}\right)x+ \sigma z\right]\right)^2\\
    \phi_d^4(x, z)&= \vert x\vert -\left\lvert x-\frac{\sigma^2}{2}\sgn(x)+\sigma z\right\rvert+\frac{z^2}{2}\\
    &-\frac{1}{2\sigma^2}\left(\frac{\sigma^2}{2}\sgn(x) -\sigma z+\frac{\sigma^2}{2}\sgn\left[x-\frac{\sigma^2}{2}\sgn(x)+\sigma z\right]\right)^2.
\end{align}
Let us denote
\begin{align}
&A_1:= \left\lbrace x: 0\leq x< \frac{\sigma^{2v}r}{2}\right\rbrace,\quad
A_2:= \left\lbrace x: -\frac{\sigma^{2v}r}{2}< x\leq 0\right\rbrace,\\
&A_3:= \left\lbrace x:  x\geq \frac{\sigma^{2v}r}{2}\right\rbrace,\quad
A_4:= \left\lbrace x:  x\leq -\frac{\sigma^{2v}r}{2}\right\rbrace,
\end{align}
and
\begin{align}
&B_1:= \left\lbrace z:0\leq \left(1-\frac{1}{\sigma^{2(v-1)}r}\right)x + \sigma z< \frac{\sigma^{2v}r}{2}\right\rbrace, \\
&B_2:= \left\lbrace z:-\frac{\sigma^{2v}r}{2}< \left(1-\frac{1}{\sigma^{2(v-1)}r}\right)x + \sigma z\leq 0 \right\rbrace, \\
& B_3:= \left\lbrace z: \left(1-\frac{1}{\sigma^{2(v-1)}r}\right)x + \sigma z\geq \frac{\sigma^{2v}r}{2}\right\rbrace,\\
& B_4:= \left\lbrace z:\left(1-\frac{1}{\sigma^{2(v-1)}r}\right)x + \sigma z\leq-\frac{\sigma^{2v}r}{2}\right\rbrace,
\end{align}
and
\begin{align}
&C_1:= \left\lbrace (x, z):0\leq x-\frac{\sigma^2}{2}\sgn(x)+\sigma z< \frac{\sigma^{2v}r}{2} \right\rbrace, \\
&C_2:= \left\lbrace (x, z):-\frac{\sigma^{2v}r}{2} < x-\frac{\sigma^2}{2}\sgn(x)+\sigma z\leq 0\right\rbrace\\
&C_3:= \left\lbrace (x, z): x-\frac{\sigma^2}{2}\sgn(x)+\sigma z\geq \frac{\sigma^{2v}r}{2} \right\rbrace, \\
&C_4:= \left\lbrace (x, z):x-\frac{\sigma^2}{2}\sgn(x)+\sigma z\leq-\frac{\sigma^{2v}r}{2} \right\rbrace,
\end{align}
so that, up to a set of null measure,
\begin{align}
R_1& =(A_1\cup A_2)\cap (B_1\cup B_2), \qquad \qquad R_2=(A_3\cup A_4)\cap (C_1\cup C_2),\\
R_3&=(A_1\cup A_2)\cap (B_3\cup B_4), \qquad\qquad R_4=(A_3\cup A_4)\cap (C_3\cup C_4).
\end{align}

\begin{prop}
\label{prop:laplace_exp}
Take $X$ a Laplace random variable and $Z$ a standard normal random variable independent of $X$, then if $\sigma^2 = \ell^2d^{-2/3}$, we have
\begin{align}
\lim_{d\to +\infty} d \expe{\phi_d(X, Z)} = -\frac{\ell^3}{3\sqrt{2\uppi}}.
\end{align}
\end{prop}
\begin{proof}
Taking expectations of $\phi_d^i\1_{R_i}$ for $i=1, \dots, 4$ and exploiting the symmetry of the laws of $X$ and $Z$, we can write
\begin{align}
    &\expe{\phi_d^1(X, Z)\indicator{R_1}{X, Z}}\\
    &\qquad= 2\expe{\left(\frac{1}{\sigma^{2(v-1)}r}X-\sigma Z\right)\indicator{A_1}{X}\indicator{B_1}{X, Z}}\\
&\qquad+2\expe{\left(2X-\frac{1}{\sigma^{2(v-1)}r}X+\sigma Z\right)\indicator{A_1}{X}\indicator{B_2}{X, Z}}\\
    &\qquad+2 \mathbb{E}\left[\left(\frac{Z^2}{2}-\frac{1}{2\sigma^2}\left(\left(\frac{2}{\sigma^{2(v-1)}r}-\frac{1}{\sigma^{4(m-1)}r^2}\right)X-\left(1-\frac{1}{\sigma^{2(v-1)}r}\right)\sigma Z\right)^2\right)\right.\\
    &\left.\qquad\qquad\times\indicator{A_1}{X}\indicator{B_1\cup B_2}{X, Z}\right],\\
    &\expe{\phi_d^2(X, Z)\indicator{R_2}{X, Z}}\\
&\qquad=2\expe{\left(\frac{\sigma^2}{2}-\sigma Z\right)\indicator{A_3}{X}\indicator{C_1}{X, Z}}\\
&\qquad+2\expe{ \left(2X-\frac{\sigma^2}{2}+\sigma Z\right)\indicator{A_3}{X}\indicator{C_2}{X, Z}}\\
&\qquad+2\mathbb{E}\left[\left(\frac{Z^2}{2}-\frac{1}{2\sigma^2}\left(\frac{1}{\sigma^{2(v-1)}r}X+\left(1-\frac{1}{\sigma^{2(v-1)}r}\right)\left(\frac{\sigma^2}{2}-\sigma Z\right)\right)^2\right)\right.\\
&\qquad\qquad\left.\times\indicator{A_3}{X}\indicator{C_1\cup C_2}{X, Z}\right]  ,\\
    &\expe{\phi_d^3(X, Z)\indicator{R_3}{X, Z}} \\ &\qquad=2\expe{\left(\frac{1}{\sigma^{2(v-1)}r}X-\sigma Z\right)\indicator{A_1}{X}\indicator{B_3}{X, Z}}\\
&\qquad+2\expe{\left(2X-\frac{1}{\sigma^{2(v-1)}r}X+\sigma Z\right)\indicator{A_1}{X}\indicator{B_4}{X, Z}}\\
&\qquad+2\expe{\left(\frac{Z^2}{2}-\frac{1}{2\sigma^2}\left(\frac{1}{\sigma^{2(v-1)}r}X -\sigma Z+\frac{\sigma^2}{2}\right)^2\right)\indicator{A_1}{X}\indicator{ B_3}{X, Z}}  \\
&\qquad+2\expe{\left(\frac{Z^2}{2}-\frac{1}{2\sigma^2}\left(\frac{1}{\sigma^{2(v-1)}r}X -\sigma Z-\frac{\sigma^2}{2}\right)^2\right)\indicator{A_1}{X}\indicator{ B_4}{X, Z}},\\
   &\expe{\phi_d^4(X, Z)\indicator{R_4}{X, Z}}=2\expe{\left(2X-\frac{\sigma^2}{2}+\sigma Z\right)\indicator{A_3}{X}\indicator{C_4}{X, Z}}.
\end{align}
Using the integrals in Appendix~\ref{app:integrals} and Lebesgue's dominated convergence theorem, we find that for $\alpha=\beta =1/3$ and $r\geq0$
\begin{align}
\lim_{d\to+\infty}d\expe{\phi_d^1(X, Z)}&= 0\\
\lim_{d\to+\infty}d\expe{\phi_d^2(X, Z)}&=  -2\frac{\ell^3r}{4\sqrt{2\uppi}}\int_{-\infty}^{0}e^{-z^2/2}z\rmd z=\frac{\ell^3r}{2\sqrt{2\uppi}}\\
\lim_{d\to+\infty}d\expe{\phi_d^3(X, Z)}&= \frac{3\ell^3r}{8\sqrt{2\uppi}}\int_{-\infty}^{0}e^{-z^2/2}z\rmd z-\frac{\ell^3r}{8\sqrt{2\uppi}}\int^{+\infty}_{0}e^{-z^2/2}z\rmd z = -\frac{\ell^3r}{2\sqrt{2\uppi}}\\
\lim_{d\to+\infty}d\expe{\phi_d^4(X, Z)}&= \frac{\ell^3}{6\sqrt{2\uppi}}\int_{-\infty}^{0}e^{-z^2/2}z^3\rmd z = -\frac{\ell^3}{3\sqrt{2\uppi}},
\end{align}
which gives
\begin{align}
\lim_{d\to+\infty}d\expe{\phi_d(X, Z)}&= \lim_{d\to+\infty}d\left(\expe{\phi_d^1(X, Z)}+\expe{\phi_d^2(X, Z)}+\expe{\phi_d^3(X, Z)}+\expe{\phi_d^4(X, Z)}  \right)\\
&=-\frac{\ell^3}{3\sqrt{2\uppi}}.
\end{align}
For $\alpha=1/3, \beta =m/3$ for $m> 1$ and $r\geq0$ we have
\begin{align}
\lim_{d\to+\infty}d\expe{\phi_d^1(X, Z)}&= 0\\
\lim_{d\to+\infty}d\expe{\phi_d^2(X, Z)}&= 0\\
\lim_{d\to+\infty}d\expe{\phi_d^3(X, Z)}&= 0\\
\lim_{d\to+\infty}d\expe{\phi_d^4(X, Z)}&= \frac{\ell^3}{6\sqrt{2\uppi}}\int_{-\infty}^{0}e^{-z^2/2}z^3\rmd z = -\frac{\ell^3}{3\sqrt{2\uppi}},
\end{align}
which gives
\begin{align}
\lim_{d\to+\infty}d\expe{\phi_d(X, Z)}&= \lim_{d\to+\infty}d\left(\expe{\phi_d^1(X, Z)}+\expe{\phi_d^2(X, Z)}+\expe{\phi_d^3(X, Z)}+\expe{\phi_d^4(X, Z)}  \right)\\
&=-\frac{\ell^3}{3\sqrt{2\uppi}}.
\end{align}
\end{proof}

\begin{prop}
\label{prop:laplace_var}
Take $X$ a Laplace random variable and $Z$ a standard normal random variable independent of $X$, then if $\sigma^2 = \ell^2d^{-2/3}$
\begin{align}
\lim_{d\to +\infty}d\var\left(\phi_d(X, Z)\right) = 
\frac{2\ell^3}{3\sqrt{2\uppi}}.
\end{align}
\end{prop}
\begin{proof}
As a consequence of the previous Proposition we have
\begin{align}
\lim_{d\to +\infty}d\expe{\phi_d(X, Z)}^2=0.
\end{align}
Then, because $R_j\cap R_i=\emptyset$ for all $j\neq i$, we have that
\begin{align}
\expe{\phi_d(X, Z)^2} &= \expe{\phi_d^1(X, Z)^2R_1(X, Z)}+\expe{\phi_d^2(X, Z)^2R_2(X, Z)}\\
&+\expe{\phi_d^3(X, Z)^2R_3(X, Z)} +\expe{\phi_d^4(X, Z)^2R_4(X, Z)},
\end{align}
and, exploiting again the symmetry of the laws of $X$ and $Z$, we have
\begin{align}
&\expe{\phi_d^1(X, Z)^2R_1(X, Z)}\\ &=2\mathbb{E}\left[\left(\frac{1}{\sigma^{2(v-1)}r}X-\sigma Z+\frac{Z^2}{2}-\frac{1}{2\sigma^2}\left(\left(\frac{2}{\sigma^{2(v-1)}r}-\frac{1}{\sigma^{4(m-1)}r^2}\right)X-\left(1-\frac{1}{\sigma^{2(v-1)}r}\right)\sigma Z\right)^2\right)^2\right.\\
    &\left.\qquad\times\indicator{A_1}{X}\indicator{B_1}{X, Z}\right],\\
&+2\mathbb{E}\left[\left(2X-\frac{1}{\sigma^{2(v-1)}r}X-\sigma Z+\frac{Z^2}{2}-\frac{1}{2\sigma^2}\left(\left(\frac{2}{\sigma^{2(v-1)}r}-\frac{1}{\sigma^{4(m-1)}r^2}\right)X-\left(1-\frac{1}{\sigma^{2(v-1)}r}\right)\sigma Z\right)^2\right)^2\right.\\
    &\left.\qquad\times\indicator{A_1}{X}\indicator{B_2}{X, Z}\right],\\
   & \expe{\phi_d^2(X, Z)^2\indicator{R_2}{X, Z}}\\
&=2\expe{\left(\frac{\sigma^2}{2}-\sigma Z+\frac{Z^2}{2}-\frac{1}{2\sigma^2}\left(\frac{1}{\sigma^{2(v-1)}r}X+\left(1-\frac{1}{\sigma^{2(v-1)}r}\right)\left(\frac{\sigma^2}{2}-\sigma Z\right)\right)^2\right)^2\indicator{A_3}{X}\indicator{C_1}{X, Z}}\\
&+2\mathbb{E}\left[ \left(2X-\frac{\sigma^2}{2}+\sigma Z+\frac{Z^2}{2}-\frac{1}{2\sigma^2}\left(\frac{1}{\sigma^{2(v-1)}r}X+\left(1-\frac{1}{\sigma^{2(v-1)}r}\right)\left(\frac{\sigma^2}{2}-\sigma Z\right)\right)^2\right)^2\right.\\
    &\left.\qquad\times\indicator{A_1}{X}\indicator{C_2}{X, Z}\right],\\
&    \expe{\phi_d^3(X, Z)^2\indicator{R_3}{X, Z}} \\ &=2\expe{\left(\frac{1}{\sigma^{2(v-1)}r}X-\sigma Z+\frac{Z^2}{2}-\frac{1}{2\sigma^2}\left(\frac{1}{\sigma^{2(v-1)}r}X -\sigma Z+\frac{\sigma^2}{2}\right)^2\right)^2\indicator{A_1}{X}\indicator{ B_3}{X, Z}}  \\
&+2\expe{\left(2X-\frac{1}{\sigma^{2(v-1)}r}X+\sigma Z+\frac{Z^2}{2}-\frac{1}{2\sigma^2}\left(\frac{1}{\sigma^{2(v-1)}r}X -\sigma Z-\frac{\sigma^2}{2}\right)^2\right)^2\indicator{A_1}{X}\indicator{ B_4}{X, Z}},\\
 & \expe{\phi_d^4(X, Z)^2\indicator{R_4}{X, Z}}=2\expe{\left(2X-\frac{\sigma^2}{2}+\sigma Z\right)^2\indicator{A_3}{X}\indicator{C_4}{X, Z}}.
\end{align}
Proceeding as for Proposition~\ref{prop:laplace_exp}, 
using the integrals in Appendix~\ref{app:integrals} and Lebesgue's dominated convergence theorem we can then show that  for $\alpha=1/3, \beta =m/3$ for $m\geq 1$ and $r\geq0$
\begin{align}
\lim_{d\to+\infty}d\var\left(\phi_d(X, Z)\right)&= \frac{2\ell^3}{3\sqrt{2\uppi}}.
\end{align}
\end{proof}

\begin{prop}
\label{prop:laplace_third}
Take $X$ a Laplace random variable and $Z$ a standard normal random variable independent of $X$, then if $\sigma^2 = \ell^2d^{-2/3}$ we have
\begin{align}
\lim_{d\to +\infty}d\expe{\phi_d(X, Z)^3} = 0.
\end{align}
\end{prop}
\begin{proof}
Following the same structure of the previous propositions we have that
\begin{align}
\expe{\phi(X, Z)^3} &= \expe{\phi_d^1(X, Z)^3R_1(X, Z)}+\expe{\phi_d^3(X, Z)^2R_2(X, Z)}\\
&+\expe{\phi_d^3(X, Z)^3R_3(X, Z)} +\expe{\phi_d^4(X, Z)^3R_4(X, Z)},
\end{align}
exploiting again the symmetry of the laws of $X$ and $Z$, using the integrals in Appendix~\ref{app:integrals}, the dominated convergence theorem we can then show that
\begin{align}
\lim_{d\to+\infty}d\expe{\phi_d(X, Z)^3}&=0.
\end{align}
\end{proof}

\subsection{Bound on second moment of acceptance ratio for the Laplace distribution}
\label{app:tightness_moment}

\begin{lemma}
\label{lemma:tightness_moment}
Let $Z$ be a standard normal random variable and $\sigma = \ell/d^{\alpha}$ for
$\alpha=1/3$. Then, there exists a constant $C>0$ such that for all $a\in\real$
and $d\in\mathbb{N}$:
\begin{align}
\expe{\phi_d(a, Z)^2} \leq \frac{C}{d^{2\alpha}}.
\end{align}
\end{lemma}
\begin{proof}
We consider the case $a\geq0$ and $r\geq \sigma^{2(v-1)}$ only, all the other cases follow from identical arguments. As in the derivation of the moments of $\phi_d$ in Appendix~\ref{app:laplace_moments}, we distinguish four regions.
We recall that $\sigma=\ell/d^{\alpha}$ for $\alpha=1/3$ and thus $\sigma^{p+1}\leq \sigma^{p}$ for all $p\in\mathbb{N}$.
Take $r\geq \sigma^{-2(v-1)}$, for $R_1$, we have, using H\"older's inequality multiple times,
\begin{align}
&\expe{\phi_d^1(a, Z)^2\indicator{R_1}{a, Z}} = \expe{\phi_d^1(a, Z)^2\indicator{B_1\cup B_2}{a, Z}} \\
&\qquad\leq C\sigma^2 \int_{-\left(1-1/\sigma^{2(v-1)}r\right)a/\sigma}^{\sigma^{2v-1} r/2-\left(1-1/\sigma^{2(v-1)}r\right)a}\frac{e^{-z^2/2}}{\sqrt{2\uppi}}\left[\left(\frac{1}{\sigma^{2v-1}r}a- z\right)^2\right.\\
&\qquad\qquad\qquad \left.+ \left(\frac{z^2}{2\sigma}-\frac{1}{2\sigma^3}\left(\left(\frac{2}{\sigma^{2(v-1)}r}-\frac{1}{\sigma^{4(m-1)}r^2}\right)a-\left(1-\frac{1}{\sigma^{2(v-1)}r}\right)\sigma z\right)^2\right)^2\right]\rmd z\\
&\qquad+ C\sigma^2\int^{-\left(1-1/\sigma^{2(v-1)}r\right)a/\sigma}_{-\sigma^{2v-1} r/2-\left(1-1/\sigma^{2(v-1)}r\right)a}\frac{e^{-z^2/2}}{\sqrt{2\uppi}}\left[\left(\frac{2a}{\sigma}-\frac{1}{\sigma^{2v-1}r}a+ z\right)^2\right.\\
&\qquad\qquad\qquad\left. + \left(\frac{z^2}{2\sigma}-\frac{1}{2\sigma^3}\left(\left(\frac{2}{\sigma^{2(v-1)}r}-\frac{1}{\sigma^{4(m-1)}r^2}\right)a-\left(1-\frac{1}{\sigma^{2(v-1)}r}\right)\sigma z\right)^2\right)^2\right]\rmd z\\
&\qquad\leq C\sigma^2\int_{-\infty}^{+\infty}\frac{e^{-z^2/2}}{\sqrt{2\uppi}}\left(4\left(\frac{a}{\sigma}\right)^2+2\left(\frac{1}{\sigma^{2v-1}r}a\right)^2+2z^2\right)\rmd z\\
&\qquad +C\sigma^2 \int_{-\sigma^{2v-1} r/2-\left(1-1/\sigma^{2(v-1)}r\right)a}^{\sigma^{2v-1} r/2-\left(1-1/\sigma^{2(v-1)}r\right)a}\frac{e^{-z^2/2}}{\sqrt{2\uppi}}\\
&\qquad\qquad\qquad\times\left(\frac{z^2}{2\sigma}-\frac{1}{2\sigma^3}\left(\left(\frac{2}{\sigma^{2(v-1)}r}-\frac{1}{\sigma^{4(m-1)}r^2}\right)a-\left(1-\frac{1}{\sigma^{2(v-1)}r}\right)\sigma z\right)^2\right)^2\\
&\qquad\leq C\sigma^2+C\sigma^2\int_{-\sigma^{2v-1} r/2-\left(1-1/\sigma^{2(v-1)}r\right)a}^{\sigma^{2v-1} r/2-\left(1-1/\sigma^{2(v-1)}r\right)a}\frac{e^{-z^2/2}}{\sqrt{2\uppi}}\\
&\qquad\qquad\qquad\times\left(\frac{z^4}{4\sigma^2}+\frac{1}{4\sigma^6}\left(\left(\frac{2}{\sigma^{2(v-1)}r}-\frac{1}{\sigma^{4(m-1)}r^2}\right)^4a^4+\left(1-\frac{1}{\sigma^{2(v-1)}r}\right)^4\sigma^4 z^4\right)\right)\rmd z\\
&\qquad\leq C\sigma^2,
\end{align}
where we used the fact that the moments of $Z$ are bounded and $a\leq \sigma^{2v} r/2$ for the first term, and the fact that $z\leq \sigma^{2v-1} r/2$ for the second one.
Proceeding as above, for $R_3$, $a>0$ and $r\geq \sigma^{-2(v-1)}$, we have
\begin{align}
&\expe{\phi_d^3(a, Z)^2\indicator{R_3}{a, Z}} = \expe{\phi_d^3(a, Z)^2\indicator{B_3\cup B_4}{a, Z}}\\
&\qquad\leq C\sigma^2 \int_{\sigma^{2v-1} r/2-\left(1-1/\sigma^{2(v-1)}r\right)a/\sigma}^{+\infty}\frac{e^{-z^2/2}}{\sqrt{2\uppi}}\left[\left(\frac{1}{\sigma^{2v-1}r}a- z\right)^2\right.\\
&\qquad\qquad\qquad\left.+ \left(\frac{z^2}{2\sigma}-\frac{1}{2\sigma^3}\left(\frac{1}{\sigma^{2(v-1)}r}a -\sigma z+\frac{\sigma^2}{2}\right)^2\right)^2\right] \rmd z\\
&\qquad+ C\sigma^2\int^{-\sigma^{2v-1} r/2-\left(1-1/\sigma^{2(v-1)}r\right)a/\sigma}_{-\infty}\frac{e^{-z^2/2}}{\sqrt{2\uppi}} \left[\left(\frac{2a}{\sigma}-\frac{1}{\sigma^{2v-1}r}a+ z\right)^2\right.\\
&\qquad\qquad\qquad\left.+\left(\frac{z^2}{2\sigma}-\frac{1}{2\sigma^3}\left(\frac{1}{\sigma^{2(v-1)}r}a -\sigma z+\frac{\sigma^2}{2}\right)^2\right)^2\right]\rmd z\\
&\qquad\leq C\sigma^2 + C\sigma^2 \int_{-\infty}^{+\infty}\frac{e^{-z^2/2}}{\sqrt{2\uppi}}\left(\frac{1}{2\sigma^3}\left(\frac{a^2}{\sigma^{4(m-1)}r^2} +\frac{\sigma^4}{4}-\sigma^3 z+\frac{a}{\sigma^{2v-4}r} -\frac{2a z}{\sigma^{2v-3}r}\right)\right)^2\rmd z\\
&\qquad \leq C\sigma^2,
\end{align}
where we used again the boundedness of the moments of $Z$, the fact that $a\leq \sigma^{2v} r/2$ and that $\sigma^{p+1}\leq \sigma^{p}$.
For $R_2$ and $a>0$, we have
\begin{align}
&\expe{\phi_d^2(a, Z)^2\indicator{R_2}{a, Z}} = \expe{\phi_d^2(a, Z)^2\indicator{C_1\cup C_2}{a, Z}} \\
&\qquad\leq C\sigma^2 \int_{\sigma/2-a/\sigma}^{\sigma/2+\sigma^{2v-1} r/2-a/\sigma}\frac{e^{-z^2/2}}{\sqrt{2\uppi}}\left[\left(\frac{\sigma^2}{2}-\sigma z\right)^2\right.\\
&\qquad\qquad\qquad\left.+\left(\frac{z^2}{2}-\frac{1}{2\sigma^2}\left(\frac{1}{\sigma^{2(v-1)}r}a+\left(1-\frac{1}{\sigma^{2(v-1)}r}\right)\left(\frac{\sigma^2}{2}-\sigma z\right)\right)^2\right)^2\right]\rmd z\\
&\qquad+ C\sigma^2 \int^{\sigma/2-a/\sigma}_{\sigma/2-\sigma^{2v-1} r/2-a/\sigma}\frac{e^{-z^2/2}}{\sqrt{2\uppi}}\left[\left(2a-\frac{\sigma^2}{2}+\sigma z\right)^2\right.\\
&\qquad\qquad\qquad\left.+\left(\frac{z^2}{2}-\frac{1}{2\sigma^2}\left(\frac{1}{\sigma^{2(v-1)}r}a+\left(1-\frac{1}{\sigma^{2(v-1)}r}\right)\left(\frac{\sigma^2}{2}-\sigma z\right)\right)^2\right)^2\right]\rmd z.
\end{align}
The first integral is bounded using the moments of $Z$, while for the third one let us denote $\chi(a, \sigma, z) := a-\sigma^2/2+\sigma z$, then
\begin{align}
&\int^{\sigma/2+\sigma^{2v-1} r/2-a/\sigma}_{\sigma/2-\sigma^{2v-1} r/2-a/\sigma}\frac{e^{-z^2/2}}{\sqrt{2\uppi}}\left(\frac{z^2}{2\sigma}-\frac{1}{2\sigma^3}\left(\frac{1}{\sigma^{2(v-1)}r}a+\left(1-\frac{1}{\sigma^{2(v-1)}r}\right)\left(\frac{\sigma^2}{2}-\sigma z\right)\right)^2\right)^2\rmd z \\
&\qquad\qquad=\int^{\sigma/2+\sigma^{2v-1} r/2-a/\sigma}_{\sigma/2-\sigma^{2v-1} r/2-a/\sigma}\frac{e^{-z^2/2}}{\sqrt{2\uppi}}\left(\frac{z^2}{2\sigma}-\frac{1}{2\sigma^3}\left(\frac{\chi(a, \sigma, z)}{\sigma^{2(v-1)}r}+\frac{\sigma^2}{2}-\sigma z\right)^2\right)^2\rmd z \\
&\qquad\qquad\leq C\int^{\sigma/2+\sigma^{2v-1} r/2-a/\sigma}_{\sigma/2-\sigma^{2v-1} r/2-a/\sigma}\frac{e^{-z^2/2}}{\sqrt{2\uppi}}\left(\frac{z^2}{2\sigma}-\frac{1}{2\sigma^3}\left(\frac{\sigma^2}{2}-\sigma z\right)^2\right)^2\rmd z\\
&\qquad\qquad +C\int^{\sigma/2+\sigma^{2v-1} r/2-a/\sigma}_{\sigma/2-\sigma^{2v-1} r/2-a/\sigma}\frac{e^{-z^2/2}}{\sqrt{2\uppi}} \left(\frac{\chi(a, \sigma, z)^2}{2\sigma^{4m-1}r^2}\right)^2 + \left(\frac{\chi(a, \sigma, z)}{r\sigma^{2v-1}}\left(\frac{\sigma^2}{2}-\sigma z\right)\right)^2\rmd z;
\end{align}
recalling that in $R_2$ we have $\vert \chi(a, \sigma, z)\vert \leq \sigma^{2v} r/2$, we obtain that this term is also bounded by $C\sigma^2$.
For $R_4$ and $a>0$, we have
\begin{align}
\expe{\phi_d^4(a, Z)^2\indicator{R_4}{a, Z}} 
&= \expe{\phi_d^4(a, Z)^2\indicator{C_4}{a, Z}}\\
&=  \int_{-\infty}^{\sigma/2-\sigma^{2v-1} r/2-a/\sigma}\frac{e^{-z^2/2}}{\sqrt{2\uppi}}\left(2a-\frac{\sigma^2}{2}+\sigma z\right)^2 \rmd z\\
&=  \sigma^2\int_{-\infty}^{\sigma/2-\sigma^{2v-1} r/2-a/\sigma}\frac{e^{-z^2/2}}{\sqrt{2\uppi}}\left(\frac{2a}{\sigma}-\frac{\sigma}{2}+z\right)^2 \rmd z
\end{align}
Collecting all the terms together, we obtain 
\begin{align}
\expe{\phi_d(a, Z)^2} &= \sum_{i=1}^4\expe{\phi_d^i(a, Z)^2\indicator{T_i}{a, Z}}\\
&\leq C\sigma^2 + C\sigma^2\int_{-\infty}^{\sigma/2-a/\sigma}\frac{e^{-z^2/2}}{\sqrt{2\uppi}}\left(\frac{2a}{\sigma}-\frac{\sigma}{2}+z\right)^2 \rmd z.
\end{align}
Recall that $\sigma = \ell d^{-1/3}$. To bound the last integral we use H\"older's inequality
\begin{align}
\int_{-\infty}^{\sigma/2-a/\sigma}\frac{e^{-z^2/2}}{\sqrt{2\uppi}}\left(\frac{2a}{\sigma}-\frac{\sigma}{2}+z\right)^2 \rmd z &\leq C\int_{-\infty}^{\sigma/2-a/\sigma}\frac{e^{-z^2/2}}{\sqrt{2\uppi}}\left[\left(\frac{\sigma}{2}+z\right)^2+4\left(\frac{\sigma}{2}-\frac{a}{\sigma}\right)^2\right] \rmd z \\
&\leq C\int_{-\infty}^{\sigma/2-a/\sigma}\frac{e^{-z^2/2}}{\sqrt{2\uppi}}\left[\left(\frac{\ell^2}{4}+z^2\right)+4\left(\frac{\sigma}{2}-\frac{a}{\sigma}\right)^2\right] \rmd z.
\end{align}
The first term is bounded since the moments of $Z$ are bounded. For the second term we use an estimate of the Gaussian cumulative distribution function. Let $\kappa(\ell,d,a) := \ell d^{-1/3}/2 - a d^{1/3}/\ell$.
When $z<\kappa(\ell,d,a)<0$, we have $1<z/\kappa(\ell,d,a)$ and therefore 
\begin{align}
(2\uppi)^{-1/2}\kappa(\ell,d,a)^2 \int_{-\infty}^{\kappa(\ell,d,a)}e^{-z^2/2}\rmd z &\leq (2\uppi)^{-1/2}\kappa(\ell,d,a) \int_{-\infty}^{\kappa(\ell,d,a)}z e^{-z^2/2}\rmd z, \\
& = (2\uppi)^{-1/2}\kappa(\ell,d,a)\exp(-\kappa(\ell,d,a)^2/2).
\end{align}
However $y \mapsto ye^{-y^2/2}$ is bounded over $\real$, therefore $(a,d) \longmapsto (2\uppi)^{-1/2} \kappa(\ell,d,a)\exp(-\kappa(\ell,d,a)^2/2)$ is bounded over $\real_+^* \times \mathbb{N}$. If $\kappa(\ell,d,a)\geq 0$, then we still have $\kappa(\ell,d,a) < \ell$ and thus have the inequality
\begin{equation}
(2\uppi)^{-1/2}\kappa(\ell,d,a)^2 \int_{-\infty}^{\kappa(\ell,d,a)} e^{-z^2/2}\rmd z  \leq (2\uppi)^{-1/2} \ell^2 \int_{-\infty}^{+\infty}e^{-z^2/2} \rmd z = \ell^2.
\end{equation}
The result then follows since $\sigma = \ell d^{-1/3}$.

\end{proof}
\subsection{Additional integrals for the Laplace distribution}
\label{app:sgn_expe}

We collect here two auxiliary Lemmata which are used in the proof of Proposition~\ref{prop:martingale_approximation}.

\begin{lemma}
\label{lemma:sgn_expe}
Take $X$ a Laplace random variable and $Z$ a standard normal random variable independent of $X$. Let $\tilde{X}:= X-\frac{1}{\sigma^{2(v-1)}r}X\1\left\lbrace\vert X\vert < \sigma^{2v}r/2\right\rbrace-\frac{\sigma^2}{2}\sgn(X)\1\left\lbrace\vert X\vert \geq\sigma^{2v}r/2\right\rbrace+\sigma Z$, then, for $\sigma=\ell d^{-\alpha}$ with $\alpha = 1/3$,
\begin{align}
\expe{ \1\left\lbrace \sgn(X) \neq \sgn(\tilde{X}) \right\rbrace} \to 0
\end{align}
if $d\to \infty$.
\end{lemma}
\begin{proof}
Using the same strategy of Appendix~\ref{app:laplace_moments} and the symmetry of the laws of $X, Z$, we find that 
\begin{align}
\expe{ \1\left\lbrace \sgn(X) \neq \sgn(\tilde{X}) \right\rbrace} &= 2\expe{\indicator{A_1}{X}\indicator{B_2}{X, Z}}+2\expe{\indicator{A_3}{X}\indicator{C_2}{X, Z}}\\
&+2\expe{\indicator{A_1}{X}\indicator{B_4}{X, Z}}+2\expe{\indicator{A_3}{X}\indicator{C_4}{X, Z}}.
\end{align}
Using the same strategy used to obtain the moments of $\phi_d$ in Appendix~\ref{app:laplace_moments}, we find that 
$$\expe{\indicator{A_1}{X}\indicator{B_2}{X, Z}}=o(1),$$
in addition
\begin{align}
\expe{\indicator{A_3}{X}\indicator{C_2}{X, Z}} &= \frac{1}{2\sqrt{2\uppi}}\int_{-\infty}^{\sigma/2-\sigma^{2v-1} r/2} e^{-z^2/2}\int_{\sigma^2/2-\sigma z}^{\sigma^2/2-\sigma z+\sigma^{2v} r/2} e^{-x}\rmd x\ \rmd z + o(1)\\
&= \frac{1}{2\sqrt{2\uppi}}\int_{-\infty}^{\sigma/2-\sigma^{2v-1} r/2} e^{-z^2/2}\left[\frac{\sigma^2 r}{2}\delta_{m1}+...\right] \rmd z + o(1),
\end{align}
where $\delta_{m1}$ is a Dirac's delta, and
\begin{align}
&\expe{\indicator{A_1}{X}\indicator{B_4}{X, Z}} +
\expe{\indicator{A_3}{X}\indicator{C_4}{X, Z}} \\
&\qquad\qquad= \frac{1}{2\sqrt{2\uppi}}\int_{-\infty}^{\sigma/2-\sigma^{2v-1} r} e^{-z^2/2}\int_{0}^{\sigma^2/2-\sigma z-\sigma^{2v} r/2} e^{-x}\rmd x\ \rmd z + o(1)\\
&\qquad\qquad= \frac{1}{2\sqrt{2\uppi}}\int_{-\infty}^{\sigma/2-\sigma^{2v-1} r} e^{-z^2/2}\left[-\sigma z+...\right] \rmd z + o(1).
\end{align}
Since $\sigma = \ell d^{-1/3}$ and the remainder terms of the Taylor expansions are bounded, Lebesque's dominated convergence theorem gives
\begin{align}
\expe{\indicator{A_3}{X}\indicator{C_2}{X, Z}}\to0, \\
\expe{\indicator{A_1}{X}\indicator{B_4}{X, Z}} +
\expe{\indicator{A_3}{X}\indicator{C_4}{X, Z}} \to 0
\end{align}
as $d\to \infty$.
\end{proof}

\begin{lemma}
\label{lemma:integral_D1tau}
Take $X$ a Laplace random variable and $Z$ a standard normal random variable independent of $X$. Then,
\begin{equation}
d^{\alpha}\expe{\vert Z\vert \left\lvert \phi_d\left( X,Z\right)\right\rvert }\to 0
\end{equation}
for $\alpha=1/3$.
\end{lemma}
\begin{proof}
Using Cauchy-Schwarz's inequality we have that
\begin{equation}
\expe{\vert Z\vert \left\lvert \phi_d\left( X,Z\right)\right\rvert }\leq \expe{Z^2}^{1/2}\expe{ \phi_d( X,Z)^2 }^{1/2};
\end{equation}
the first expectation is equal to one, and the second one converges to zero at rate $d^{1/2}$ by Proposition~\ref{prop:laplace_var}. The result follows straightforwardly.
\end{proof}


\subsection{Integrals for moment computations}
\label{app:integrals}
We distinguish the case $v=1/2$ and $v\geq1$ since the integration bounds significantly differ in these two cases. For values between $1/2$ and $1$ the integrals are not finite. 
The expectations below are obtained by integrating w.r.t. $x$ and using a Taylor expansion about $\sigma=0$ to obtain the leading order terms. Using the Lagrange form of the remainder for the Taylor expansions, we find that the remainder terms are all of the form $\sigma^{1/\alpha+1}f(\gamma(\sigma, z))/(1/\alpha+1)!$ where $\gamma(\sigma, z)$ is a point between the limits of integration w.r.t. $x$ and $f:x\mapsto p(x)e^{-x}$, where $p$ is a polynomial.
Therefore, using the boundedness of the remainder and Lebesgue's dominated convergence theorem, the integrals w.r.t. $z$ of the remainder terms all converge to 0.

\subsubsection{First moment}

For simplicity, we only consider the case for $r\geq \sigma^{-2(v-1)}$, the other case follows analogously.

\textbf{Region \texorpdfstring{$R_1$}{2}}
Let us consider $\phi_d^1$ first. We have
\begin{align}
A_1\cap B_1 &=\begin{cases}
0\leq x\leq \frac{\sigma^{2v} r}{2}\qquad\textrm{if } 0\leq z\leq \frac{\sigma}{2}\\
0\leq x\leq \left(\frac{\sigma^{2v} r}{2}-\sigma z\right)\left(1-\frac{1}{\sigma^{2(v-1)}r}\right)^{-1}\qquad\textrm{if } \frac{\sigma}{2}\leq z\leq \frac{ \sigma^{2v-1}r}{2}\\
-\sigma z\left(1-\frac{1}{\sigma^{2(v-1)}r}\right)^{-1}\leq x\leq \frac{\sigma^{2v} r}{2}\qquad\textrm{if } \frac{\sigma}{2}- \frac{\sigma^{2v-1}r}{2}\leq z\leq 0
\end{cases},\\
A_1\cap B_2 &=\begin{cases}
0\leq x\leq\frac{\sigma^{2v} r}{2}\qquad\textrm{if } -\frac{\sigma^{2v-1}r}{2}\leq z\leq \frac{\sigma}{2}-\frac{\sigma^{2v-1}r}{2}\\
0\leq x\leq -\sigma z\left(1-\frac{1}{\sigma^{2(v-1)}r}\right)^{-1}\qquad\textrm{if } \frac{\sigma}{2}-\frac{\sigma^{2v-1}r}{2}\leq z\leq 0\\
-\left(\frac{\sigma^{2v} r}{2}+\sigma z\right)\left(1-\frac{1}{\sigma^{2(v-1)}r}\right)^{-1}\leq x\leq \frac{\sigma^{2v} r}{2}\qquad\textrm{if } \frac{\sigma}{2}-\sigma^{2v-1}r\leq z\leq - \frac{\sigma^{2v-1}r}{2}
\end{cases},
\end{align}
and
\begin{align}
A_1\cap (B_1 \cup B_2) &=\begin{cases}
0\leq x\leq \frac{\sigma^{2v} r}{2}\qquad\textrm{if } -\frac{\sigma^{2v-1}r}{2}\leq z\leq \frac{\sigma}{2}\\
0\leq x\leq \left(\frac{\sigma^{2v} r}{2}-\sigma z\right)\left(1-\frac{1}{\sigma^{2(v-1)}r}\right)^{-1}\qquad\textrm{if } \frac{\sigma}{2}\leq z\leq \frac{\sigma^{2v-1}r}{2}\\
-\left(\frac{\sigma^{2v} r}{2}+\sigma z\right)\left(1-\frac{1}{\sigma^{2(v-1)}r}\right)^{-1}\leq x\leq \frac{\sigma^{2v} r}{2}\\
\qquad\qquad\textrm{if } \frac{\sigma}{2}-\sigma^{2v-1}r\leq z\leq - \frac{\sigma^{2v-1}r}{2}
\end{cases}.
\end{align}
The corresponding expectations are
\begin{align}
&\mathbb{E}\left[\left(\frac{X}{\sigma^{2(v-1)}r}-\sigma Z\right)\indicator{A_1}{X}\indicator{B_1}{X, Z}\right]\\
&\qquad\qquad= \frac{1}{2\sqrt{2\uppi}}\int_{ \sigma/2}^{ \sigma^{2v-1}r/2}e^{-z^2/2}\left[z^{4-2v}\sigma^{4-2v}\xi(r)+\dots\right]\rmd z\\
&\qquad\qquad+ \frac{1}{2\sqrt{2\uppi}}\int_{ \sigma/2-r\sigma^{2v-1}/2}^{0}e^{-z^2/2}\left[z^{4-2v}\sigma^{4-2v}\xi(r)+\dots\right]\rmd z\\
&\qquad\qquad+o(1),\\
&\mathbb{E}\left[\left(2X-\frac{X}{\sigma^{2(v-1)}r}+\sigma Z\right)\indicator{A_1}{X}\indicator{B_2}{X, Z}\right]\\
&\qquad\qquad= \frac{1}{2\sqrt{2\uppi}}\int_{ \sigma/2-\sigma^{2v-1}r/2}^{0}e^{-z^2/2}\left[z^{4-2v}\sigma^{4-2v}\xi(r)+\dots\right]\rmd z\\
&\qquad\qquad+ \frac{1}{2\sqrt{2\uppi}}\int_{ \sigma/2-\sigma^{2v-1}r}^{-\sigma^{2v-1}r/2}e^{-z^2/2}\left[ z^{4-2v}\sigma^{4-2v}\xi(r)+\dots\right]\rmd z\\
&\qquad\qquad+o(1),
\end{align}
where $\xi:[0, +\infty)\to \real$ is a function of $r$ only which might change from one line to the other,
and
\begin{align}
&\mathbb{E}\left[\left(\frac{Z^2}{2}-\frac{1}{2\sigma^2}\left(\left(\frac{2}{\sigma^{2(v-1)}r}-\frac{1}{\sigma^{4(m-1)}r^2}\right)X-\left(1-\frac{1}{\sigma^{2(v-1)}r}\right)\sigma Z\right)^2\right)\indicator{A_1}{X}\indicator{B_1\cup B_2}{X, Z}\right]\\
&\qquad\qquad =\frac{1}{2\sqrt{2\uppi}}\int^{ \sigma/2}_{-\sigma^{2v-1}r/2}e^{-z^2/2}\left[+\dots\right]\rmd z\\
&\qquad\qquad +\frac{1}{2\sqrt{2\uppi}}\int_{ \sigma/2}^{\sigma^{2v-1}r/2}e^{-z^2/2}\left[+\dots\right]\rmd z\\
&\qquad\qquad +\frac{1}{2\sqrt{2\uppi}}\int_{ \sigma/2-\sigma^{2v-1}r}^{-\sigma^{2v-1}r/2}e^{-z^2/2}\left[z^2\sigma^2\xi(r)+\dots\right]\rmd z,
\end{align}
where $\xi:[0, +\infty)\to \real$ is a function of $r$ only which might change from one line to the other.

\textbf{Region \texorpdfstring{$R_2$}{2}}
For $\phi_d^2$, we have
\begin{align}
A_3 \cap C_1&=\begin{cases}
\sigma^2/2-\sigma z \leq x\leq -\sigma z +\sigma^2/2+\sigma^{2v}r/2\qquad\textrm{if } z<\sigma/2 - \sigma^{2v-1} r/2\\
\sigma^{2v}r/2 < x\leq -\sigma z +\sigma^2/2+\sigma^{2v}r/2\qquad\textrm{if } \sigma/2 - \sigma^{2v-1}r/2\leq z\leq \sigma/2
\end{cases},\\
A_3\cap C_2&=\begin{cases}
\sigma^2/2-\sigma z-\sigma^{2v}r/2\leq x\leq \sigma^2/2-\sigma z\qquad\textrm{if } z< \sigma/2-\sigma^{2v-1} r\\
\sigma^{2v}r/2< x\leq \sigma^2/2-\sigma z\qquad\textrm{if } \sigma/2-\sigma^{2v-1} r< z<
	\sigma/2 - \sigma^{2v-1} r/2
\end{cases}
\end{align}
and
\begin{align}
A_3\cap (C_1\cup C_2)=\begin{cases}
\sigma^{2v}r/2\leq x\leq \sigma^{2v}r/2+\sigma^2/2-\sigma z\qquad\textrm{if } \sigma/2- \sigma^{2v-1} r\leq z\leq \sigma/2\\
-\sigma^{2v}r/2+\sigma^2/2-\sigma z< x\leq \sigma^{2v}r/2+\sigma^2/2-\sigma z\\
\qquad\qquad\textrm{if } z< \sigma/2- \sigma^{2v-1} r
\end{cases}.
\end{align}
The corresponding expectations are
\begin{align}
\expe{\left(\frac{\sigma^2}{2}-\sigma Z\right)\indicator{A_3}{X}\indicator{C_1}{X, Z}} 
&=\frac{1}{2\sqrt{2\uppi}}\int_{-\infty}^{\sigma/2 - \sigma^{2v-1} r/2}e^{-z^2/2}\left[-\frac{rz}{2}\sigma^{2v+1}+ \dots\right]\rmd z\\
&+\frac{1}{2\sqrt{2\uppi}}\int_{\sigma/2 - \sigma^{2v-1} r/2}^{\sigma/2}e^{-z^2/2}\left[ z^2\sigma^2+\dots\right]\rmd z,\\
\expe{\left(2X-\frac{\sigma^2}{2}+\sigma Z\right)\indicator{A_3}{X}\indicator{C_2}{X, Z}}
&=\frac{1}{2\sqrt{2\uppi}}\int_{-\infty}^{\sigma/2-\sigma^{2v-1} r}e^{-z^2/2}\left[-\frac{rz}{2}\sigma^{2v+1}+ \dots\right]\rmd z\\
&+\frac{1}{2\sqrt{2\uppi}}\int_{\sigma/2-\sigma^{2v-1} r}^{\sigma/2-\sigma^{2v-1} r/2}e^{-z^2/2}\left[-\frac{rz}{2}\sigma^{2v+1}+ \dots\right]\rmd z,
\end{align}
and
\begin{align}
&\expe{\left(\frac{Z^2}{2}-\frac{1}{2\sigma^2}\left(\frac{1}{\sigma^{2(v-1)}r}X+\left(1-\frac{1}{\sigma^{2(v-1)}r}\right)\left(\frac{\sigma^2}{2}-\sigma Z\right)\right)^2\right)\indicator{A_3}{X}\indicator{C_1\cup C_2}{X, Z}}\\
&\qquad\qquad=\frac{1}{2\sqrt{2\uppi}}\int ^{ \sigma/2- \sigma^{2v-1} r}_{-\infty}e^{-z^2/2}\left[\frac{rz}{2}\sigma^{2v+1}+ \dots\right]\rmd z\\
&\qquad\qquad+\frac{1}{2\sqrt{2\uppi}}\int_{ \sigma/2- \sigma^{2v-1} r}^{\sigma/2}e^{-z^2/2}\left[-\frac{z^3}{2r}\sigma^{3-2v}+ \dots\right]\rmd z.
\end{align}
\textbf{Region \texorpdfstring{$R_3$}{2}}
For $\phi_d^3$, we have, in the case $r\geq \sigma^{-2(v-1)}$,
\begin{align}
A_1\cap B_3&=\begin{cases}
0\leq x\leq \frac{\sigma^{2v}r}{2}\qquad\textrm{if }  z> \frac{\sigma^{2v-1} r}{2}\\
\left(\frac{\sigma^{2v}r}{2}-\sigma z\right)\left(1-\frac{1}{\sigma^{2(v-1)}r}\right)^{-1}< x\leq \frac{\sigma^{2v}r}{2}\qquad\textrm{if } \frac{\sigma}{2}\leq z\leq \frac{\sigma^{2v-1} r}{2}
\end{cases},\\
A_1\cap B_4&=\begin{cases}
0\leq x\leq \frac{\sigma^{2v}r}{2}\qquad\textrm{if } z< \frac{\sigma}{2}-\sigma^{2v-1} r\\
0\leq  x<\left(\frac{\sigma^{2v}r}{2}-\sigma z\right)\left(1-\frac{1}{\sigma^{2(v-1)}r}\right)^{-1}\qquad\textrm{if } \frac{\sigma}{2}-\sigma^{2v-1} r\leq z\leq -\frac{\sigma^{2v-1} r}{2}
\end{cases}.
\end{align}
The corresponding expectations are
\begin{align}
&\expe{\left(\frac{1}{\sigma^{2(v-1)}r}X-\sigma Z+\frac{Z^2}{2}-\frac{1}{2\sigma^2}\left(\frac{1}{\sigma^{2(v-1)}r}X -\sigma Z+\frac{\sigma^2}{2}\right)^2\right)\indicator{A_1}{X}\indicator{ B_3}{X, Z}} \\
&\qquad\qquad=\frac{1}{2\sqrt{2\uppi}}\int_{\sigma^{2v-1} r/2}
^{+\infty}e^{-z^2/2}\left[-\frac{rz}{8}\sigma^{2v+1}+ \dots\right]\rmd z \\
&\qquad\qquad+\frac{1}{2\sqrt{2\uppi}}\int^{\sigma^{2v-1} r/2}_{\sigma/2}e^{-z^2/2}\left[z^{3-2v}\sigma^{3-2v}\xi(r)+ \dots\right]\rmd z ,\\
&\expe{\left(2X-\frac{1}{\sigma^{2(v-1)}r}X+\sigma Z+\frac{Z^2}{2}-\frac{1}{2\sigma^2}\left(\frac{1}{\sigma^{2(v-1)}r}X -\sigma Z-\frac{\sigma^2}{2}\right)^2\right)\indicator{A_1}{X}\indicator{ B_4}{X, Z}} \\
&\qquad\qquad=\frac{1}{2\sqrt{2\uppi}}\int_{-\infty}^{\sigma/2-\sigma^{2v-1} r/2}e^{-z^2/2}\left[\frac{3}{8}\sigma^{2v+1}+\dots\right]\rmd z\\
&\qquad\qquad+\frac{1}{2\sqrt{2\uppi}}\int_{\sigma/2-\sigma^{2v-1} r/2}^{-\sigma^{2v-1}r/2}e^{-z^2/2}\left[z^3\sigma^{3-2v}\xi(r)+\dots\right]\rmd z,
\end{align}
where $\xi:[0, +\infty)\to \real$ is a function of $r$ only which might change from one line to the other.
\textbf{Region \texorpdfstring{$R_4$}{2}}
Finally, for $\phi_d^4$ we have
\begin{align}
A_3 \cap  C_4=\left\lbrace z<\frac{\sigma}{2}-\sigma^{2v-1} r, \frac{\sigma^{2v}r}{2}<x\leq \frac{\sigma^2}{2 }-\sigma z -\frac{\sigma^{2v}r}{2}\right\rbrace,
\end{align} 
and
\begin{align}
    \expe{\left(2X-\frac{\sigma^2}{2}+\sigma Z\right)\indicator{A_3}{X}\indicator{ C_4}{X, Z}}
&=\frac{1}{2\sqrt{2\uppi}}\int_{-\infty}^{\sigma/2-\sigma^{2v-1} r}e^{-z^2/2}\left[\frac{z^3}{6}\sigma^3+ \dots\right]\rmd z.
\end{align}

\subsubsection{Second moment}
For simplicity, we only consider the case for $r\geq \sigma^{-2(v-1)}$, the other case follows analogously.

\textbf{Region \texorpdfstring{$R_1$}{2}}
For $\phi_d^1$ we have
\begin{align}
\mathbb{E}\left[\phi_d^1(X, Z)^2\indicator{A_1}{X}\indicator{B_1}{X, Z}\right]&= \frac{1}{2\sqrt{2\uppi}}\int_{ \sigma/2}^{ \sigma^{2v-1}r/2}e^{-z^2/2}\left[z^3\sigma^3\xi(r)+\dots\right]\rmd z\\
&+ \frac{1}{2\sqrt{2\uppi}}\int_{ \sigma/2-r\sigma^{2v-1}/2}^{0}e^{-z^2/2}\left[z^3\sigma^3\xi(r)+\dots\right]\rmd z\\
&+o(1),\\
\mathbb{E}\left[\phi_d^1(X, Z)^2\indicator{A_1}{X}\indicator{B_2}{X, Z}\right] &= \frac{1}{2\sqrt{2\uppi}}\int_{ \sigma/2-\sigma^{2v-1}r/2}^{0}e^{-z^2/2}\left[z^3\sigma^3\xi(r)+\dots\right]\rmd z\\
&+ \frac{1}{2\sqrt{2\uppi}}\int_{ \sigma/2-\sigma^{2v-1}r}^{-\sigma^{2v-1}r/2}e^{-z^2/2}\left[z^3\sigma^3\xi(r) +\dots\right]\rmd z\\
&+o(1),
\end{align}
where $\xi:[0, +\infty)\to \real$ is a function of $r$ only which might change from one line to the other.
\textbf{Region \texorpdfstring{$R_2$}{2}}
For $\phi_d^2$ we have
\begin{align}
\expe{\phi_d^2(X, Z)^2\indicator{A_3}{X}\indicator{C_1}{X, Z}} &=\frac{1}{2\sqrt{2\uppi}}\int_{-\infty}^{\sigma/2 - \sigma^{2v-1} r/2}e^{-z^2/2}\left[\frac{rz^2}{2}\sigma^{2v+2}+\dots\right]\rmd z\\
&+\frac{1}{2\sqrt{2\uppi}}\int_{\sigma/2 - \sigma^{2v-1} r/2}^{\sigma/2}e^{-z^2/2}\left[-z^3\sigma^3 +\dots\right]\rmd z,\\
\expe{\phi_d^2(X, Z)^2\indicator{A_3}{X}\indicator{C_2}{X, Z}}&=\frac{1}{2\sqrt{2\uppi}}\int_{-\infty}^{\sigma/2-\sigma^{2v-1} r}e^{-z^2/2}\left[z^2\sigma^{2v+2}+ \dots\right]\rmd z\\
&+\frac{1}{2\sqrt{2\uppi}}\int_{\sigma/2-\sigma^{2v-1} r}^{\sigma/2-\sigma^{2v-1} r/2}e^{-z^2/2}\left[-\frac{rz^2}{2}\sigma^{2v+2}+ \dots\right]\rmd z.
\end{align}

\textbf{Region \texorpdfstring{$R_3$}{2}}
For $\phi_d^3$ we have
\begin{align}
&\expe{\left(\frac{1}{\sigma^{2(v-1)}r}X-\sigma Z+\frac{Z^2}{2}-\frac{1}{2\sigma^2}\left(\frac{1}{\sigma^{2(v-1)}r}X -\sigma Z+\frac{\sigma^2}{2}\right)^2\right)^2\indicator{A_1}{X}\indicator{ B_3}{X, Z}} \\
&\qquad\qquad=\frac{1}{2\sqrt{2\uppi}}\int_{\sigma^{2v-1} r/2}
^{+\infty}e^{-z^2/2}\left[\frac{rz^2}{24}\sigma^{2v+2}+ \dots\right]\rmd z \\
&\qquad\qquad+\frac{1}{2\sqrt{2\uppi}}\int^{\sigma^{2v-1} r/2}_{\sigma/2}e^{-z^2/2}\left[\frac{rz^2}{24}\sigma^{2v+2}+ \dots\right]\rmd z ,\\
&\expe{\left(2X-\frac{1}{\sigma^{2(v-1)}r}X+\sigma Z+\frac{Z^2}{2}-\frac{1}{2\sigma^2}\left(\frac{1}{\sigma^{2(v-1)}r}X -\sigma Z-\frac{\sigma^2}{2}\right)^2\right)^2\indicator{A_1}{X}\indicator{ B_4}{X, Z}} \\
&\qquad\qquad=\frac{1}{2\sqrt{2\uppi}}\int_{-\infty}^{\sigma/2-\sigma^{2v-1} r/2}e^{-z^2/2}\left[\frac{7rz^2}{24}\sigma^{2v+2}+\dots\right]\rmd z\\
&\qquad\qquad+\frac{1}{2\sqrt{2\uppi}}\int_{\sigma/2-\sigma^{2v-1} r/2}^{-\sigma^{2v-1}r/2}e^{-z^2/2}\left[\frac{7rz^2}{24}\sigma^{2v+2}+\dots\right]\rmd z,
\end{align}
\textbf{Region \texorpdfstring{$R_4$}{2}}
For $\phi_d^4$ we have
\begin{align}
    \expe{\left(2X-\frac{\sigma^2}{2}+\sigma Z\right)^2\indicator{A_3}{X}\indicator{ C_4}{X, Z}}
&=\frac{1}{2\sqrt{2\uppi}}\int_{-\infty}^{\sigma/2-\sigma^{2v-1} r}e^{-z^2/2}\left[-\frac{z^3}{3}\sigma^3+ \dots\right]\rmd z.
\end{align}

\subsubsection{Third moment}
Having established that the only possible scaling is given by $\alpha=1/3$, $\beta=m/3$ with $m\geq 1$, we now proceed to bound the third moment of $\phi_d$ in this case.
For simplicity, we only consider the case for $r\geq 1$, the other case follows analogously.

Since $m\geq1$, we find that $\expe{\phi_d^1(X, Z)^3\indicator{R_1}{X, Z}} = o(1)$ as $d\to\infty$ since the limits of integration all converge to 0. Then, using H\"older's inequality for $\phi_d^2$, we have
\begin{align}
&\expe{\phi_d^2(X, Z)^3} 
\leq C \expe{\left(\frac{\sigma^2}{2}-\sigma Z\right)^3\indicator{A_3}{X}\indicator{C_1}{X, Z}}\\
&\qquad\qquad +C\expe{\left(2X-\frac{\sigma^2}{2}+\sigma Z\right)^3\indicator{A_3}{X}\indicator{C_2}{X, Z}}\\
&\qquad\qquad + C \mathbb{E}\left[\left(\frac{Z^2}{2}-\frac{1}{2\sigma^2}\left(\frac{1}{\sigma^{2(v-1)}r}X+\left(1-\frac{1}{\sigma^{2(v-1)}r}\right)\left(\frac{\sigma^2}{2}-\sigma Z\right)\right)^2\right)^3\right.\\
&\qquad\qquad\qquad\left.\times\indicator{A_3}{X}\indicator{C_1\cup C_2}{X, Z}\right]\\
&\qquad\qquad = \frac{C}{2\sqrt{2\uppi}}\int_{-\infty}^{\sigma/2 - \sigma^{2v-1} r/2}e^{-z^2/2}\left[-\frac{rz^3}{2}\sigma^5+...\right]\rmd z\\
&\qquad\qquad+\frac{C}{2\sqrt{2\uppi}}\int_{-\infty}^{\sigma/2-\sigma^{2v-1} r}e^{-z^2/2}\left[-\frac{rz^3}{2}\sigma^5+...\right]\rmd z\\
&\qquad\qquad + \frac{C}{2\sqrt{2\uppi}}\int_{-\infty}^{\sigma/2-\sigma^{2v-1} r}e^{-z^2/2}\left[z^3\sigma^5\xi(r)+...\right]\rmd z +o(1),
\end{align}
where $\xi:[0, +\infty)\to \real$ is a function of $r$ only which might change from one line to the other.
For $\phi_d^3$, we have, using again H\"older's inequality,
\begin{align}
&\expe{\left(\frac{1}{\sigma^{2(v-1)}r}X-\sigma Z+\frac{Z^2}{2}-\frac{1}{2\sigma^2}\left(\frac{1}{\sigma^{2(v-1)}r}X -\sigma Z+\frac{\sigma^2}{2}\right)^2\right)^3\indicator{A_1}{X}\indicator{ B_3}{X, Z}} \\
&\qquad\qquad\leq C\expe{\left(\frac{1}{\sigma^{2(v-1)}r}X-\sigma Z\right)^3\indicator{A_1}{X}\indicator{ B_3}{X, Z}} \\
&\qquad\qquad+ C\expe{\left(\frac{Z^2}{2}-\frac{1}{2\sigma^2}\left(\frac{1}{\sigma^{2(v-1)}r}X -\sigma Z+\frac{\sigma^2}{2}\right)^2\right)^3\indicator{A_1}{X}\indicator{ B_3}{X, Z}}\\
&\qquad\qquad=\frac{1}{2\sqrt{2\uppi}}\int_{\sigma^{2v-1} r/2}^{+\infty}e^{-z^2/2}\left[-\frac{z^3r}{2}\sigma^5+...\right] \rmd z \\
&\qquad\qquad+\frac{1}{2\sqrt{2\uppi}}\int_{\sigma^{2v-1} r/2}^{+\infty}e^{-z^2/2}\left[z^3\sigma^5\xi(r)+...\right] \rmd z + o(1)\\
&\expe{\left(2X-\frac{1}{\sigma^{2(v-1)}r}X+\sigma Z+\frac{Z^2}{2}-\frac{1}{2\sigma^2}\left(\frac{1}{\sigma^{2(v-1)}r}X -\sigma Z-\frac{\sigma^2}{2}\right)^2\right)^3\indicator{A_1}{X}\indicator{ B_4}{X, Z}} \\
&\qquad\qquad\leq C\expe{\left(2X-\frac{1}{\sigma^{2(v-1)}r}X+\sigma Z\right)^3\indicator{A_1}{X}\indicator{ B_4}{X, Z}} \\
&\qquad\qquad+ C\expe{\left(\frac{Z^2}{2}-\frac{1}{2\sigma^2}\left(\frac{1}{\sigma^{2(v-1)}r}X -\sigma Z-\frac{\sigma^2}{2}\right)^2\right)^3\indicator{A_1}{X}\indicator{ B_4}{X, Z}}\\
&\qquad\qquad=\frac{1}{2\sqrt{2\uppi}}\int_{\sigma^{2v-1} r/2}^{+\infty}e^{-z^2/2}\left[\frac{z^3r}{2}\sigma^5+...\right] \rmd z \\
&\qquad\qquad+\frac{1}{2\sqrt{2\uppi}}\int_{\sigma^{2v-1} r/2}^{+\infty}e^{-z^2/2}\left[z^3\sigma^5\xi(r)+...\right] \rmd z + o(1),
\end{align}
where $\xi:[0, +\infty)\to \real$ is a function of $r$ only which might change from one line to the other.
Finally, for $\phi_d^4$ we have
\begin{align}
    \expe{\left(2X-\frac{\sigma^2}{2}+\sigma Z\right)^3\indicator{A_3}{X}\indicator{ C_4}{X, Z}}
&=\frac{1}{2\sqrt{2\uppi}}\int_{-\infty}^{\sigma/2-\sigma^{2v-1} r}e^{-z^2/2}\left[\frac{z^3}{10}\sigma^5+...\right]\rmd z.
\end{align}


\section{Numerical experiments}
\label{app:numerical}
\subsection{Differentiable targets}
\label{sec:diff_numeric}

We collect here a number of numerical experiments confirming the results in Section~\ref{sec:differentiable}.
To do so, we consider the Gaussian distribution in
Example~\ref{ex:gaussian} and four algorithmic settings summarized in Table~\ref{tab:differentiable} which correspond to the three cases identified in Theorem~\ref{theo:differentiable} and MALA.
All chains are initialized at stationarity drawing i.i.d samples from the target, i.e., $X_0^d\sim\pi_d$, so that no burnin effects are present.
\begin{table}
\centering
\begin{tabular}{{cllcccc}}
Case & Figure & Algorithm & $\alpha$ & $\beta$ & $v$ &  $r$\\
\hline\noalign{\smallskip}
\ref{case:a} & \ref{fig:gaussian_case_a} & MY-MALA & 1/4 & 1/8 & 1/2 &1 \\
\ref{case:b} & \ref{fig:gaussian_pmala} & P-MALA & 1/6 & 1/6 & 1 &1\\
\ref{case:c} & \ref{fig:gaussian_general} & MY-MALA & 1/6 & 1/2 & 3 &2 \\
--- & \ref{fig:gaussian_mala} & MALA & 1/6 & 1/6 & 1 &$\approx 0$ \\
\end{tabular}
\caption{Algorithm setting for the simulation study on the Gaussian target.}
\label{tab:differentiable}
\end{table}

The first plot in Figure~\ref{fig:gaussian_case_a}--\ref{fig:gaussian_mala} show that for values of $\alpha$ different from those identified in \Cref{theo:differentiable} the acceptance ratio $a_d(\ell,r)$ becomes degenerate as $d$ increases.
For the values of $\alpha$ identified in \Cref{theo:differentiable} we analyze the influence of $\ell$ on the acceptance $a_d(\ell,r)$ (second plot), obtaining for $d
\to +\infty$ the expression given by \Cref{theo:differentiable}--\ref{case:a} for
\Cref{fig:gaussian_case_a}, the expression in \Cref{theo:differentiable}--\ref{case:b} for
\Cref{fig:gaussian_mala,fig:gaussian_pmala} and that in
\Cref{theo:differentiable}--\ref{case:c} for \Cref{fig:gaussian_general}. 

Finally, we consider the relationship between acceptance ratio $a_d(\ell,r)$ and the speed of the diffusion $h(\ell, r)$ approximated by the expected squared jumping distance (see, e.g. \cite{gelman1996efficient})
\begin{align}
    \textrm{ESJD}_d:= d^{2\alpha}\mathbb{E}\left[(X_0^d-X_1^d)^2\right].
\end{align}
Looking at the last plot in Figure~\ref{fig:gaussian_case_a}--\ref{fig:gaussian_mala} we find that, even for relatively small values of $d$, the shape of the plot of $\textrm{ESJD}_d$ as a function of the
acceptance $a_d(\ell,r)$ is similar to that of the theoretical limit.
This suggests that tuning the acceptance ratio to be approximately $0.452$ when $\alpha=1/4, \beta = 1/8$
and approximately
$0.574$ when $\alpha=1/6, \beta=v/6$ with $v\geq 1$ 
should generally guarantee high efficiency.


\begin{figure}
\centering
\begin{tikzpicture}[every node/.append style={font=\tiny}]
\node (img1) {\includegraphics[width = 0.4\textwidth]{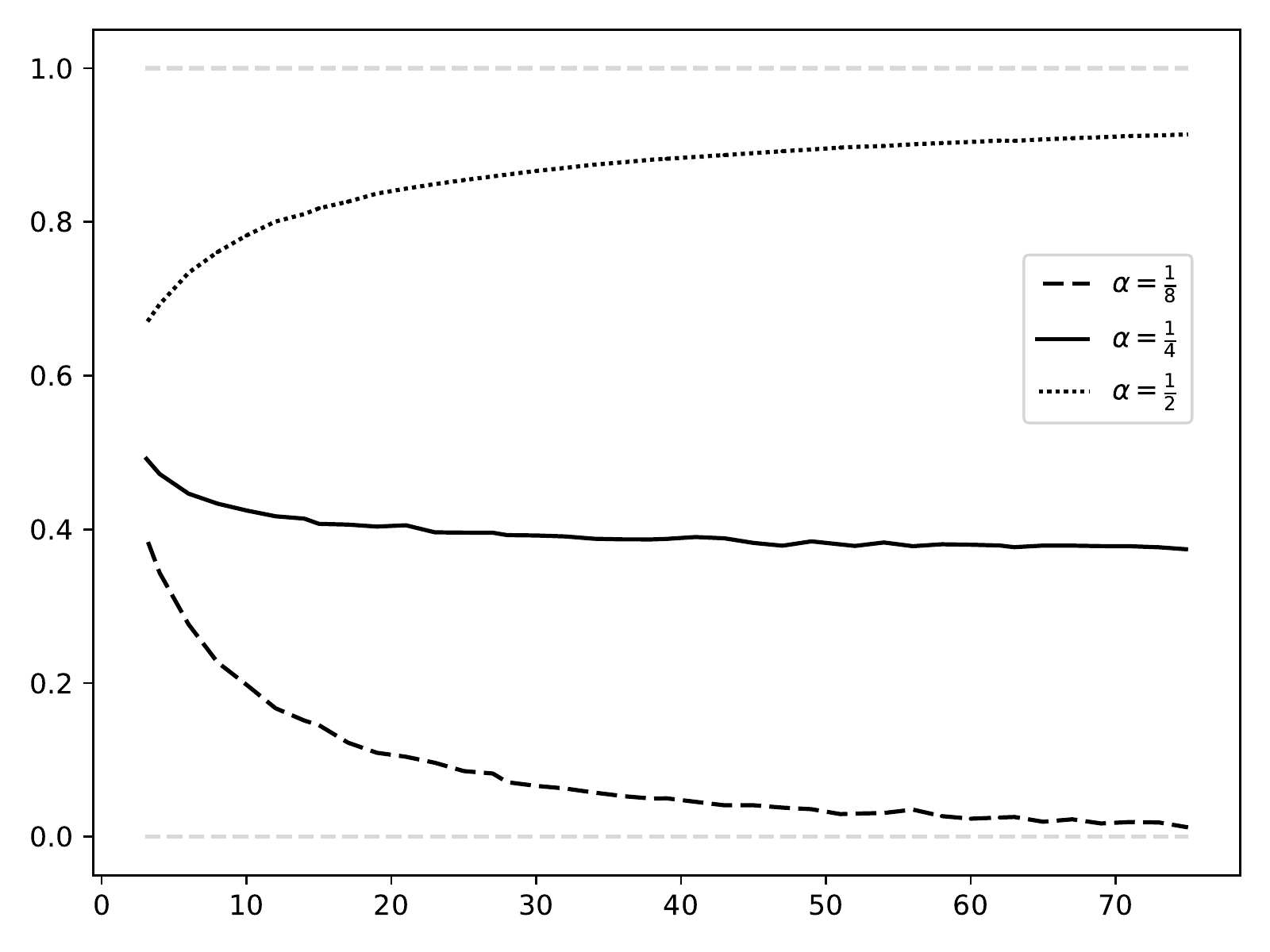}};
\node[below=of img1, node distance = 0, yshift = 1.2cm] {$d$};
\node[left=of img1, node distance = 0, rotate = 90, anchor = center, yshift = -1cm] {$a_d(\ell, r)$};
\node[right=of img1, node distance = 0, xshift = -0.7cm] (img2) {\includegraphics[width = 0.4\textwidth]{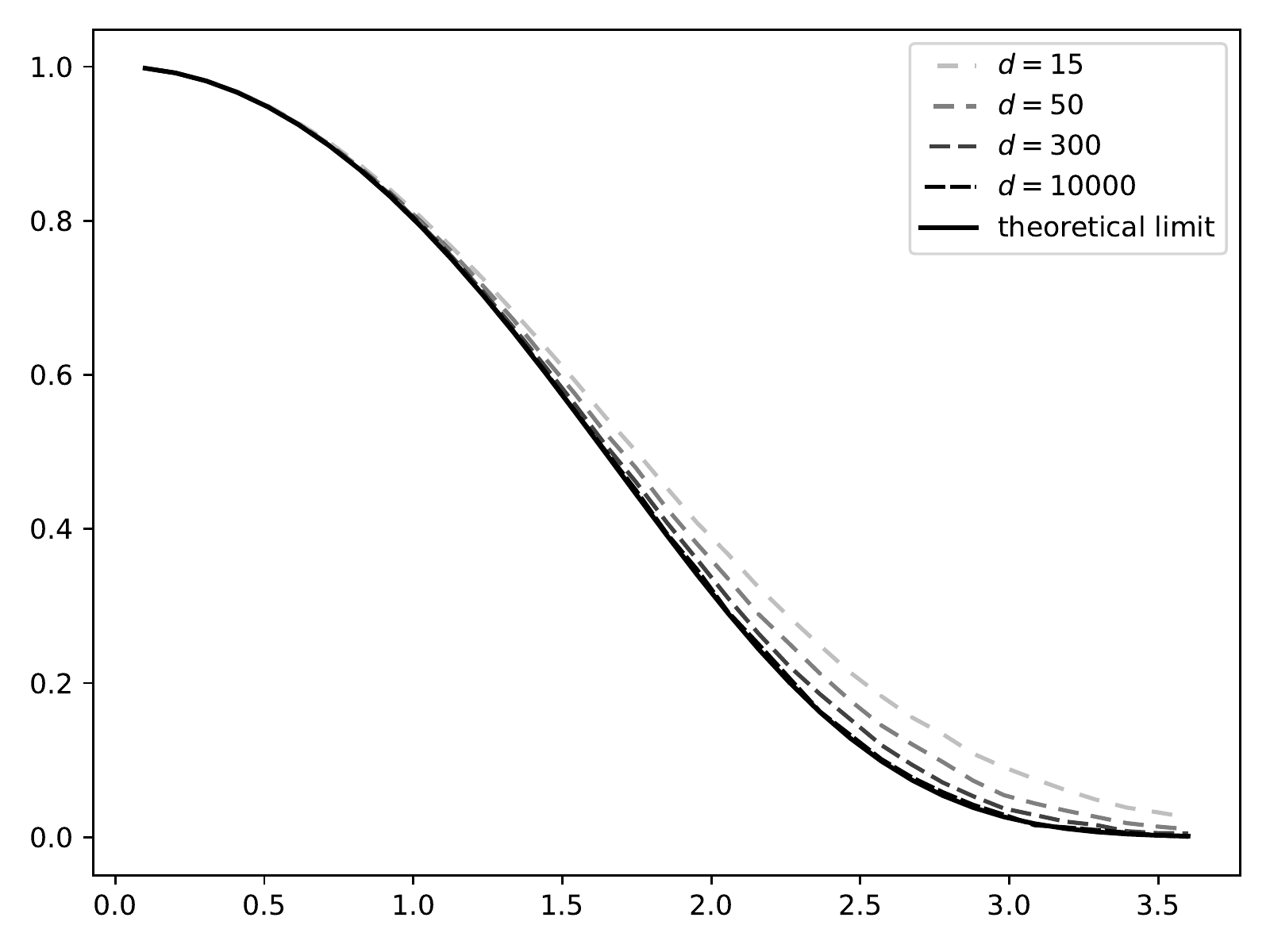}};
\node[below=of img2, node distance = 0, yshift = 1.2cm] {$\ell$};
\node[left=of img2, node distance = 0, rotate = 90, anchor = center, yshift = -1cm] {$a_d(\ell, r)$};
\node[below=of img1, node distance = 0, xshift = 3cm, yshift = 1cm] (img3) {\includegraphics[width = 0.4\textwidth]{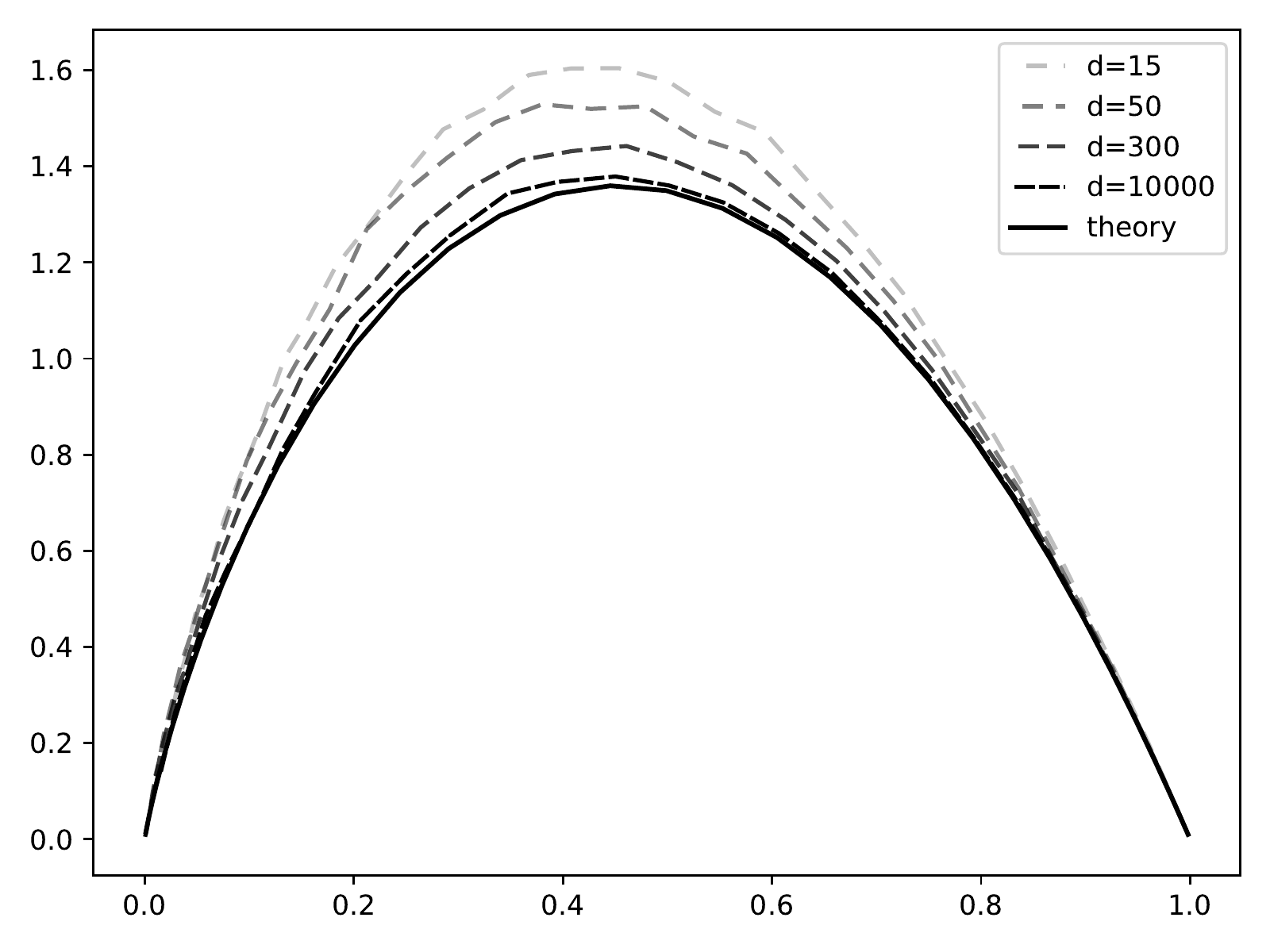}};
\node[below=of img3, node distance = 0, yshift = 1.2cm] {$a_d(\ell, r)$};
\node[left=of img3, node distance = 0, rotate = 90, anchor = center, yshift = -1cm] {$\textrm{ESJD}_d$};
\end{tikzpicture}
\caption{\textbf{Case~\ref{case:a}}: MY-MALA with Gaussian target and $v=1/2, r=1$. Average acceptance rate for different choices of $\alpha$ (first); acceptance rate as a function of $\ell$ for increasing dimension $d$ (second); $\textrm{ESJD}_d$ as a function of the acceptance rate $a_d(\ell, r)$ (third).}
\label{fig:gaussian_case_a}
\end{figure}
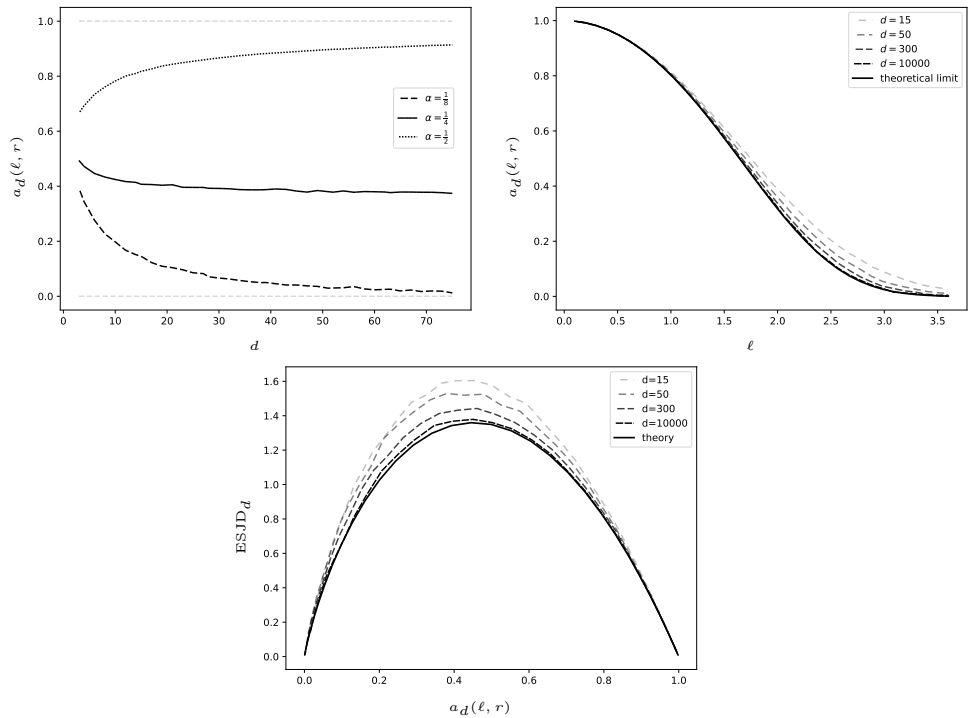

\begin{figure}
\centering
\begin{tikzpicture}[every node/.append style={font=\tiny}]
\node (img1) {\includegraphics[width = 0.4\textwidth]{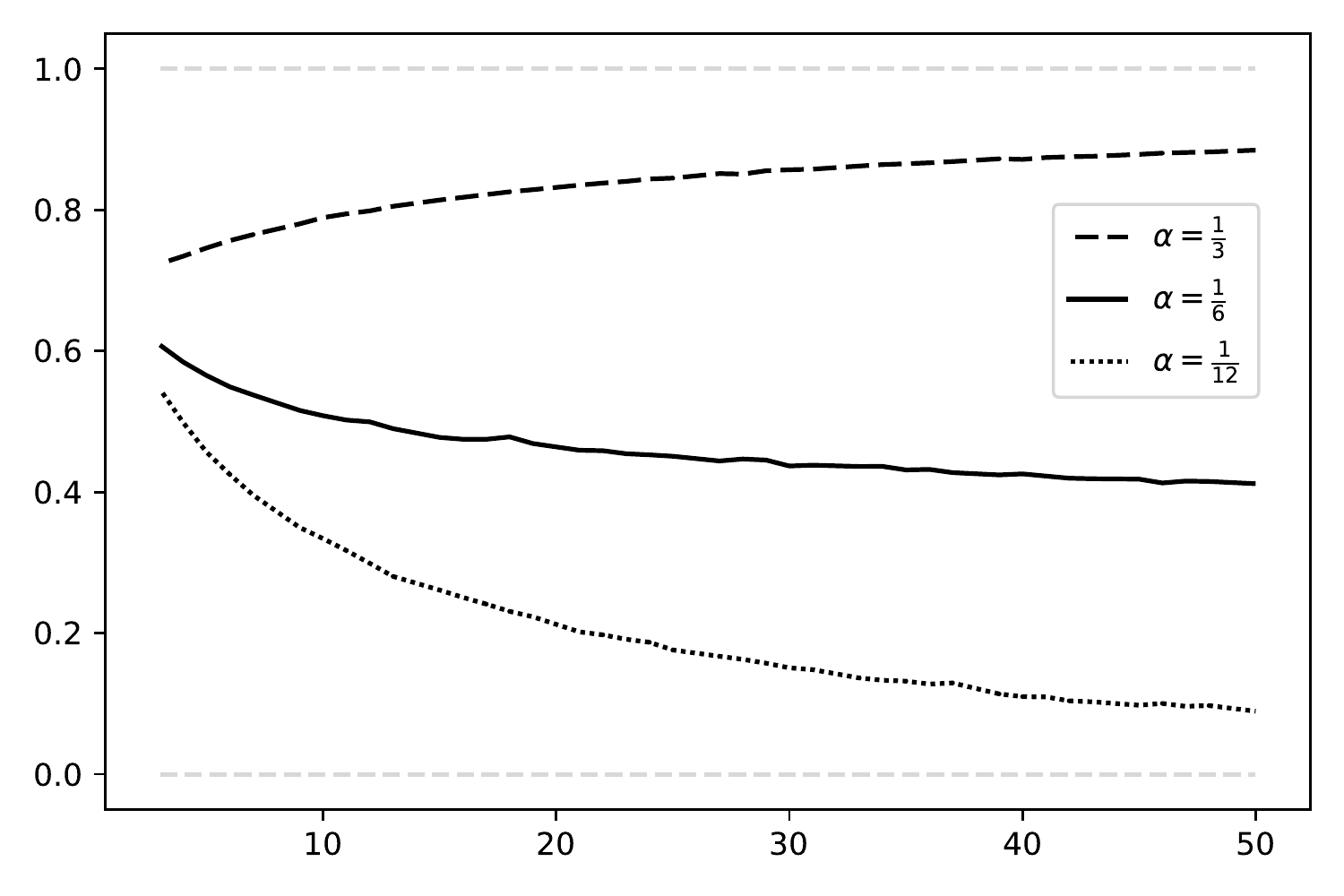}};
\node[below=of img1, node distance = 0, yshift = 1.2cm] {$d$};
\node[left=of img1, node distance = 0, rotate = 90, anchor = center, yshift = -1cm] {$a_d(\ell, r)$};
\node[right=of img1, node distance = 0, xshift = -0.7cm] (img2) {\includegraphics[width = 0.4\textwidth]{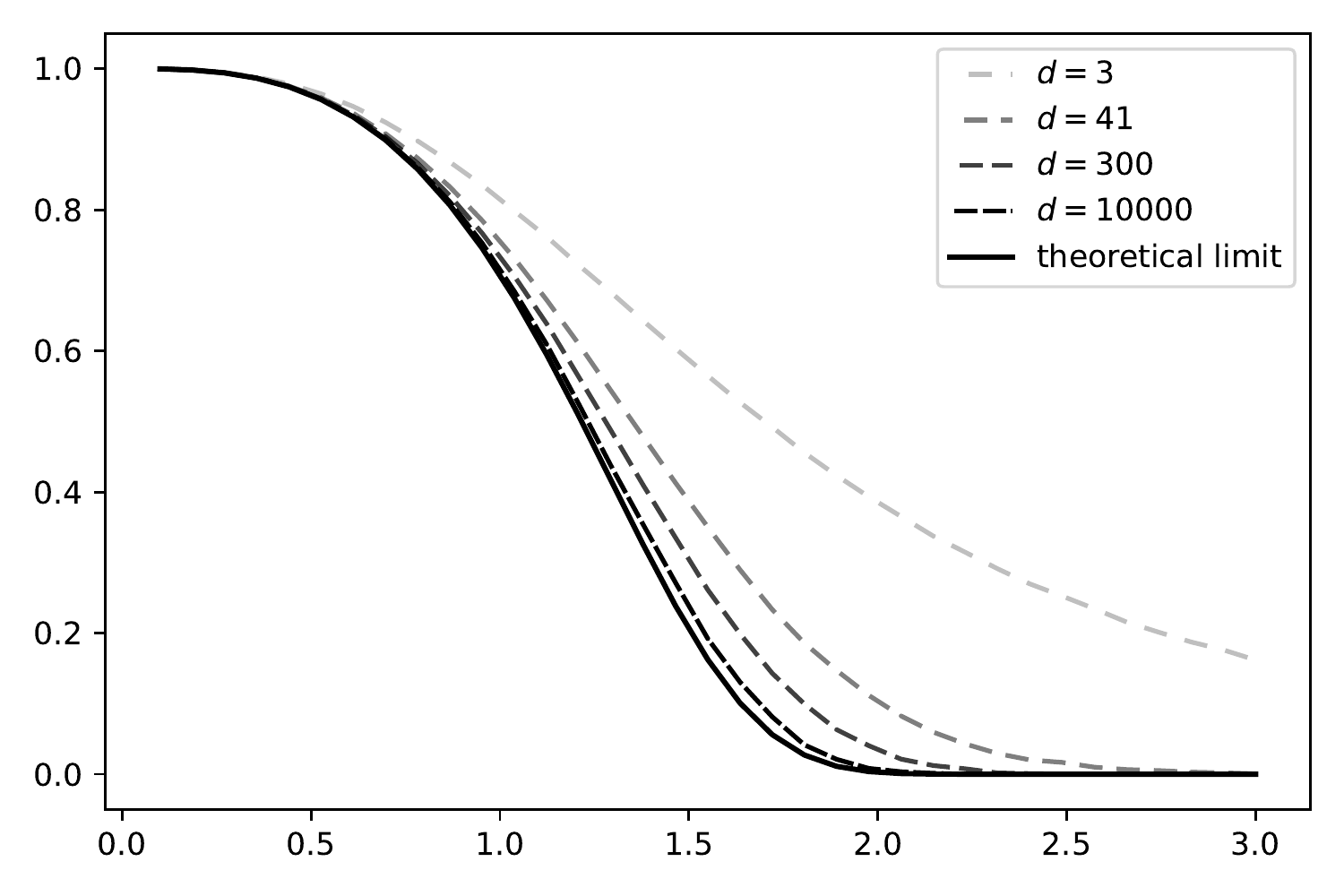}};
\node[below=of img2, node distance = 0, yshift = 1.2cm] {$\ell$};
\node[left=of img2, node distance = 0, rotate = 90, anchor = center, yshift = -1cm] {$a_d(\ell, r)$};
\node[below=of img1, node distance = 0, xshift = 3cm, yshift = 1cm] (img3) {\includegraphics[width = 0.4\textwidth]{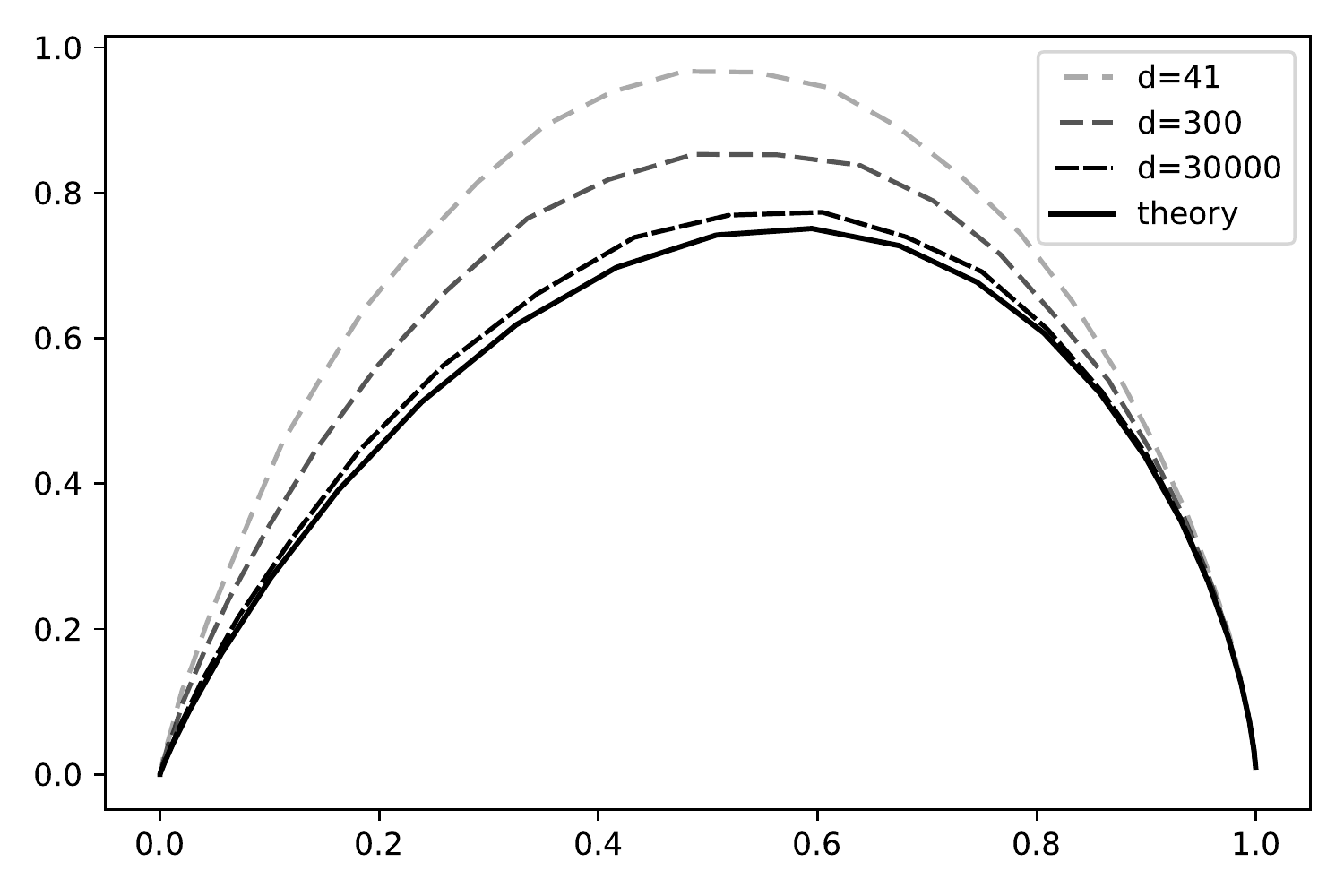}};
\node[below=of img3, node distance = 0, yshift = 1.2cm] {$a_d(\ell, r)$};
\node[left=of img3, node distance = 0, rotate = 90, anchor = center, yshift = -1cm] {$\textrm{ESJD}_d$};
\end{tikzpicture}
\caption{\textbf{Case~\ref{case:b}}: MY-MALA with Gaussian target and $v=1, r=1$ (P-MALA).  Average acceptance rate for different choices of $\alpha$ (first); acceptance rate as a function of $\ell$ for increasing dimension $d$ (second); $\textrm{ESJD}_d$ as a function of the acceptance rate $a_d(\ell, r)$ (third).}
\label{fig:gaussian_pmala}
\end{figure}

\begin{figure}
\centering
\begin{tikzpicture}[every node/.append style={font=\tiny}]
\node (img1) {\includegraphics[width = 0.4\textwidth]{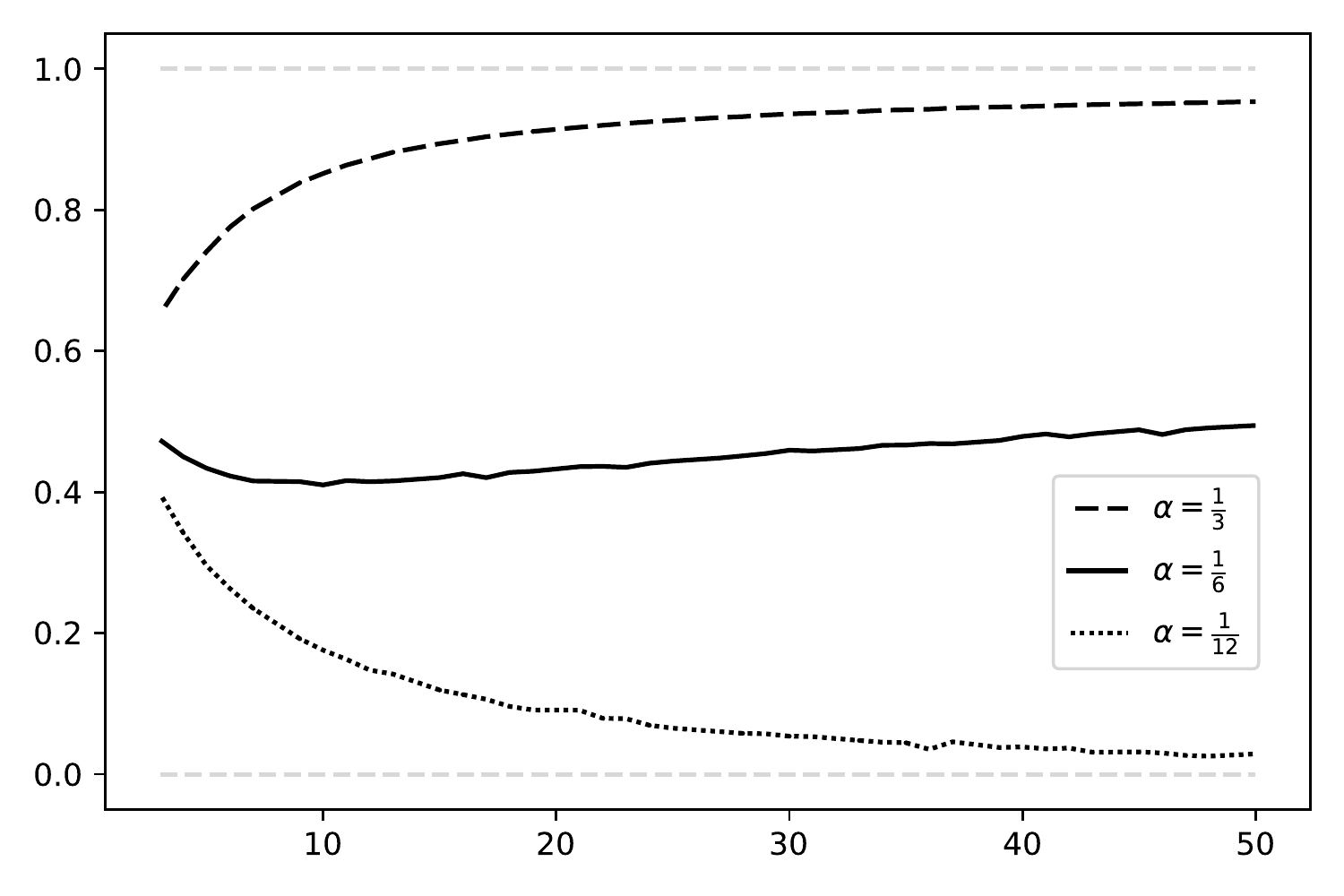}};
\node[below=of img1, node distance = 0, yshift = 1.2cm] {$d$};
\node[left=of img1, node distance = 0, rotate = 90, anchor = center, yshift = -1cm] {$a_d(\ell, r)$};
\node[right=of img1, node distance = 0, xshift = -0.7cm] (img2) {\includegraphics[width = 0.4\textwidth]{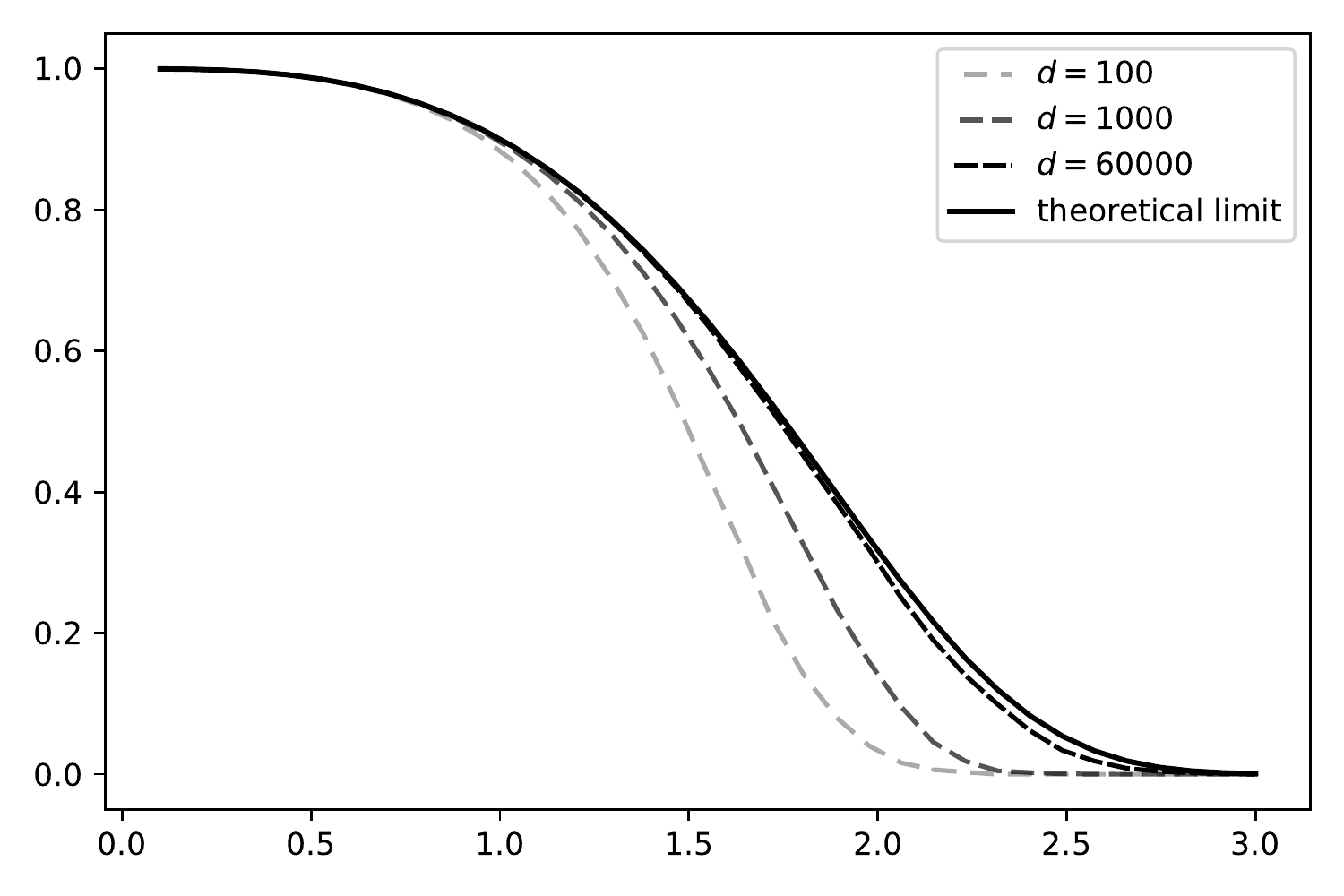}};
\node[below=of img2, node distance = 0, yshift = 1.2cm] {$\ell$};
\node[left=of img2, node distance = 0, rotate = 90, anchor = center, yshift = -1cm] {$a_d(\ell, r)$};
\node[below=of img1, node distance = 0, xshift = 3cm, yshift = 1cm] (img3) {\includegraphics[width = 0.4\textwidth]{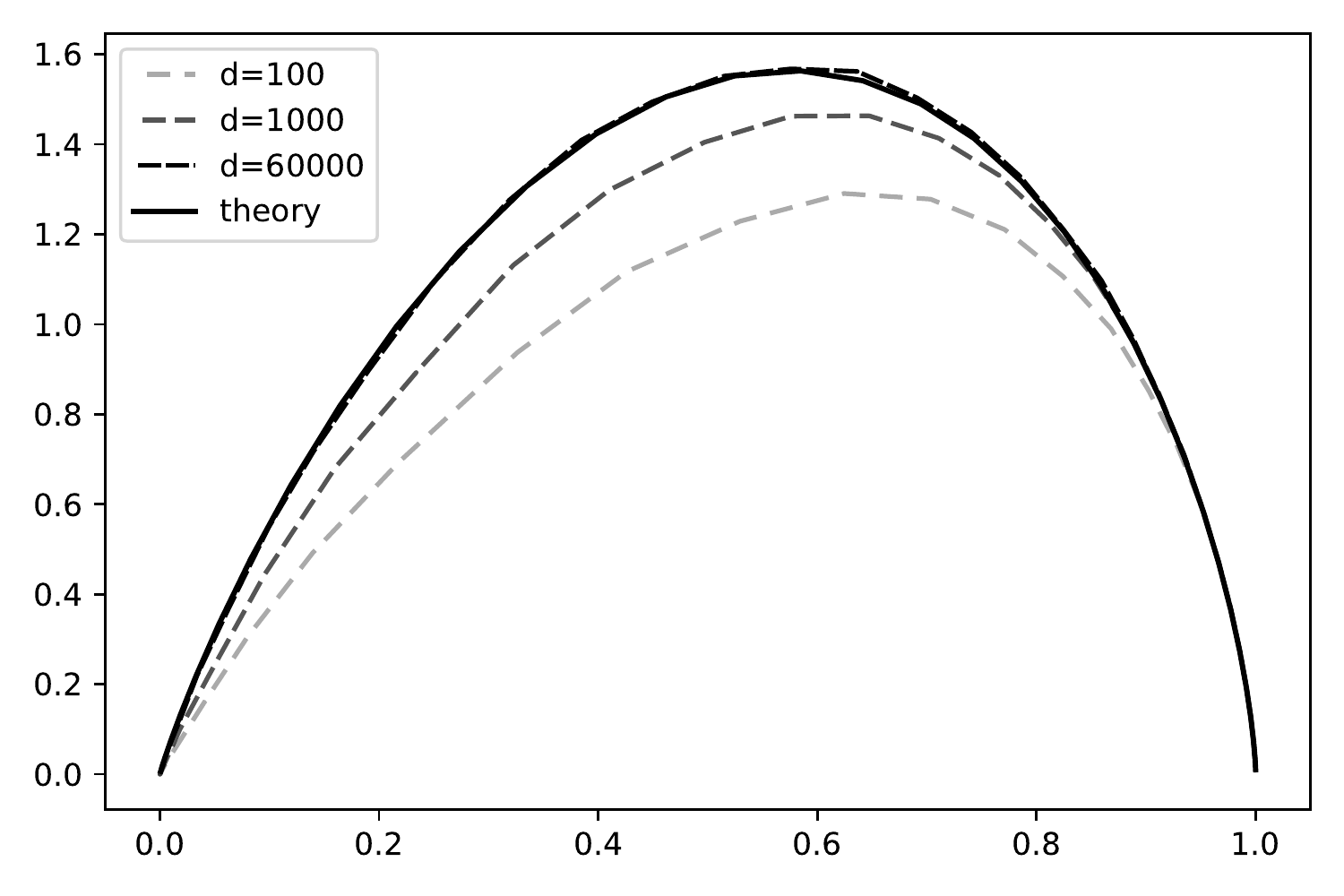}};
\node[below=of img3, node distance = 0, yshift = 1.2cm] {$a_d(\ell, r)$};
\node[left=of img3, node distance = 0, rotate = 90, anchor = center, yshift = -1cm] {$\textrm{ESJD}_d$};
\end{tikzpicture}
\caption{\textbf{Case~\ref{case:c}}: MY-MALA with Gaussian target and $v=3, r=2$.  Average acceptance rate for different choices of $\alpha$ (first); acceptance rate as a function of $\ell$ for increasing dimension $d$ (second); $\textrm{ESJD}_d$ as a function of the acceptance rate $a_d(\ell, r)$ (third).}
\label{fig:gaussian_general}
\end{figure}

\begin{figure}
\centering
\begin{tikzpicture}[every node/.append style={font=\tiny}]
\node (img1) {\includegraphics[width = 0.4\textwidth]{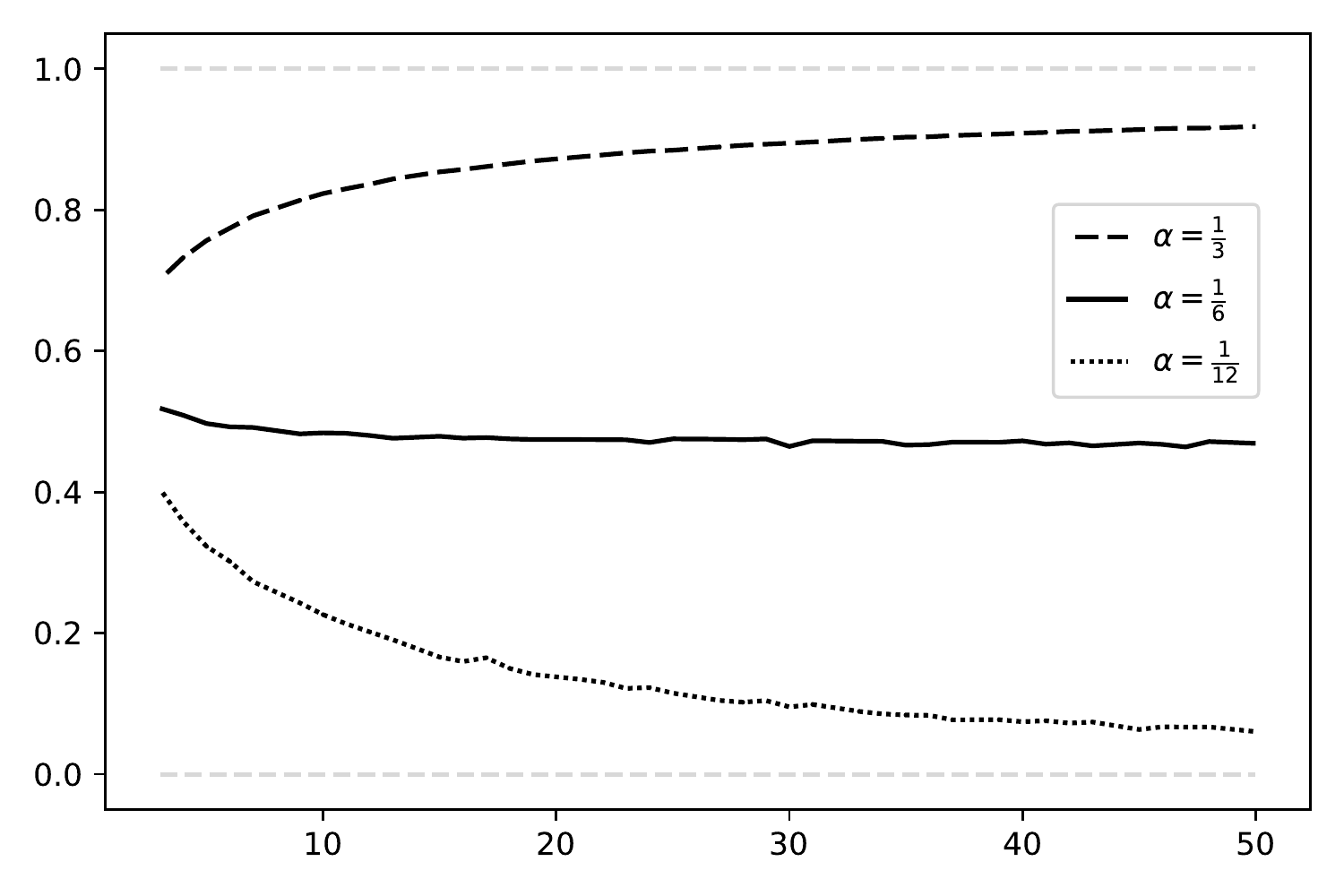}};
\node[below=of img1, node distance = 0, yshift = 1.2cm] {$d$};
\node[left=of img1, node distance = 0, rotate = 90, anchor = center, yshift = -1cm] {$a_d(\ell, r)$};
\node[right=of img1, node distance = 0, xshift = -0.7cm] (img2) {\includegraphics[width = 0.4\textwidth]{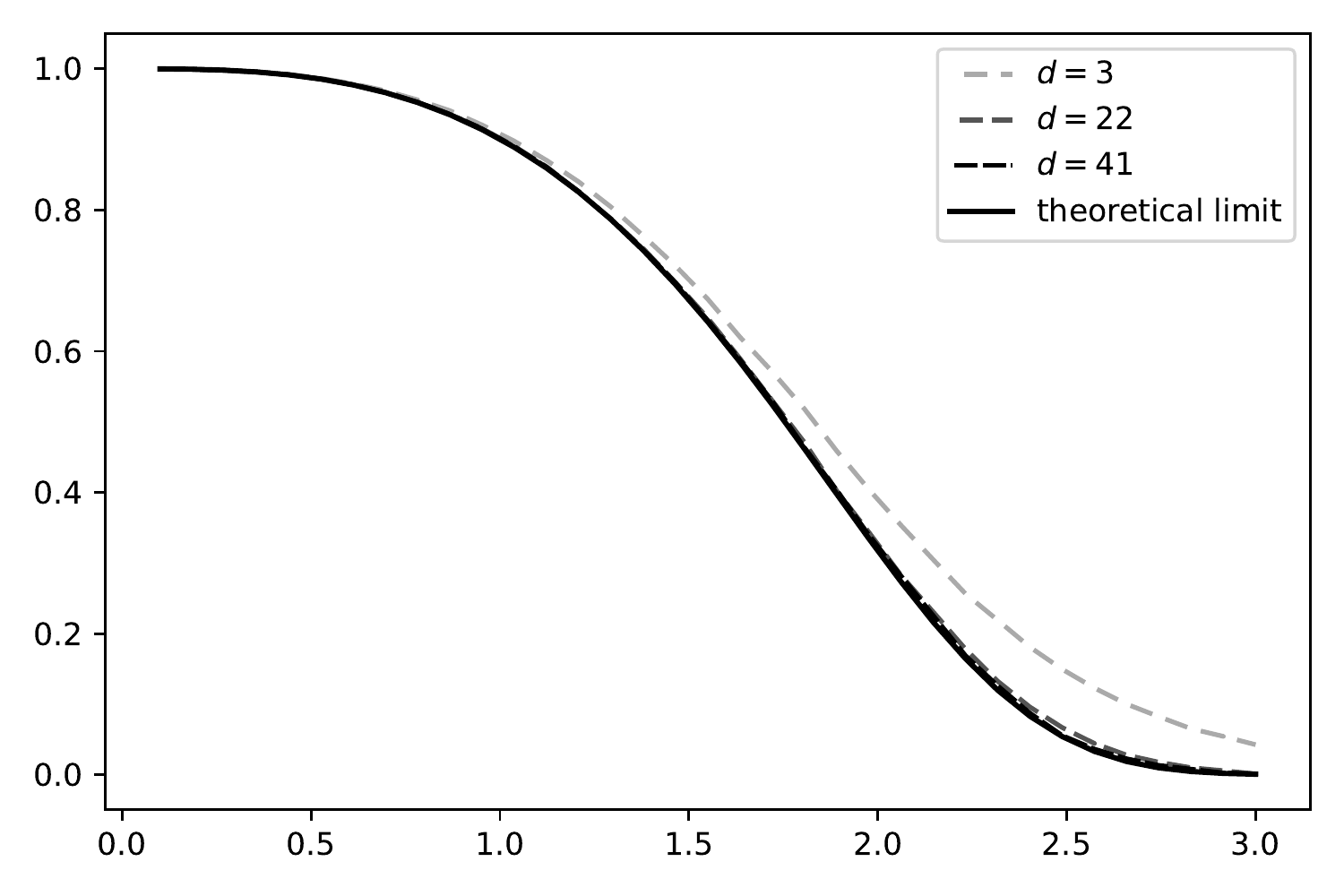}};
\node[below=of img2, node distance = 0, yshift = 1.2cm] {$\ell$};
\node[left=of img2, node distance = 0, rotate = 90, anchor = center, yshift = -1cm] {$a_d(\ell, r)$};
\node[below=of img1, node distance = 0, xshift = 3cm, yshift = 1cm] (img3) {\includegraphics[width = 0.4\textwidth]{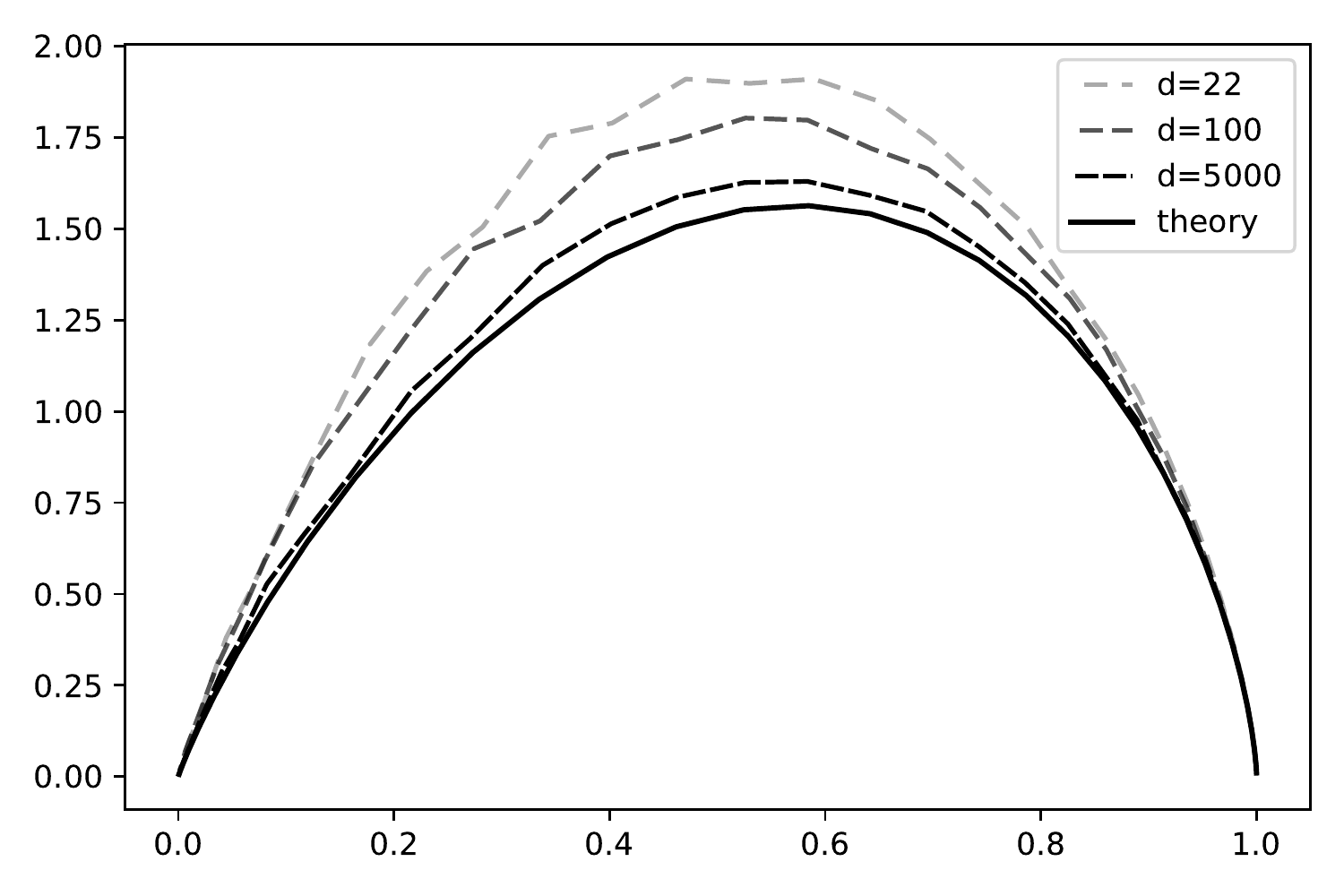}};
\node[below=of img3, node distance = 0, yshift = 1.2cm] {$a_d(\ell, r)$};
\node[left=of img3, node distance = 0, rotate = 90, anchor = center, yshift = -1cm] {$\textrm{ESJD}_d$};
\end{tikzpicture}
\caption{MY-MALA with Gaussian target and $v=1, r\to 0$ (MALA).  Average acceptance rate for different choices of $\alpha$ (first); acceptance rate as a function of $\ell$ for increasing dimension $d$ (second); $\textrm{ESJD}_d$ as a function of the acceptance rate $a_d(\ell, r)$ (third).}
\label{fig:gaussian_mala}
\end{figure}

\subsection{Laplace target}
\label{sec:laplace_numeric}

We collect here a number of numerical experiments confirming the results for the Laplace distribution in Section~\ref{sec:laplace}.
Similarly to \Cref{sec:diff_numeric} we consider three algorithmic settings, summarized in Table~\ref{tab:laplace}. All chains are initialized at stationarity drawing i.i.d samples from the target, i.e., $X_0^d\sim\pi_d$, so that no burnin effects are present.

\begin{table}[!ht]
\centering
\begin{tabular}{{llcccc}}
Figure & Algorithm & $\alpha$ & $\beta$ & $v$ &  $r$ \\
\hline\noalign{\smallskip}
\ref{fig:laplace_mala} & sG-MALA & 1/3 & 1/3 & 1 &$  0$\\
\ref{fig:laplace_pmala} & P-MALA & 1/3 & 1/3 & 1 &1\\
\ref{fig:laplace_general} & MY-MALA & 1/3 & 1 & 3 &2 \\
\end{tabular}
\caption{Algorithm setting for the simulation study on the Laplace target.}
\label{tab:laplace}
\end{table}

The first plot in Figures~\ref{fig:laplace_mala}--\ref{fig:laplace_general} shows that for $\alpha\neq 1/3$ the acceptance ratio $a_d(\ell, r)$ becomes degenerate as $d$ increases; while the second plot shows that $(a_d(\ell, r))_{d \in \nsets}$ and $(\textrm{ESJD}_d)_{d\in\nsets}$ converge to $a^\rmL(\ell)$ and $h^\rmL(\ell)$ given in
\Cref{theo:acceptance_laplace,theo:laplace_diffusion}, respectively.
In the case $v=3, r=2$
in \Cref{fig:laplace_general}, we find that the behaviour for low values of $d$ significantly differs from the limiting one. 
For values of $d$ lower than 130 the $\textrm{ESJD}_d$ achieves its maximum at a value of the acceptance $a_d(\ell, r)$ different from that predicted by \Cref{,theo:laplace_diffusion}. In practice, this might mean that for low dimensional
settings the recommended choice of
$a(\ell,c)=0.360$ is far from optimal. Similar behaviours for small $d$ have also been observed in the case of RWM and MALA (e.g., \cite[Section 2.1]{sherlock2009optimal}).

\begin{figure}
\centering
\begin{tikzpicture}[every node/.append style={font=\tiny}]
\node (img1) {\includegraphics[width = 0.4\textwidth]{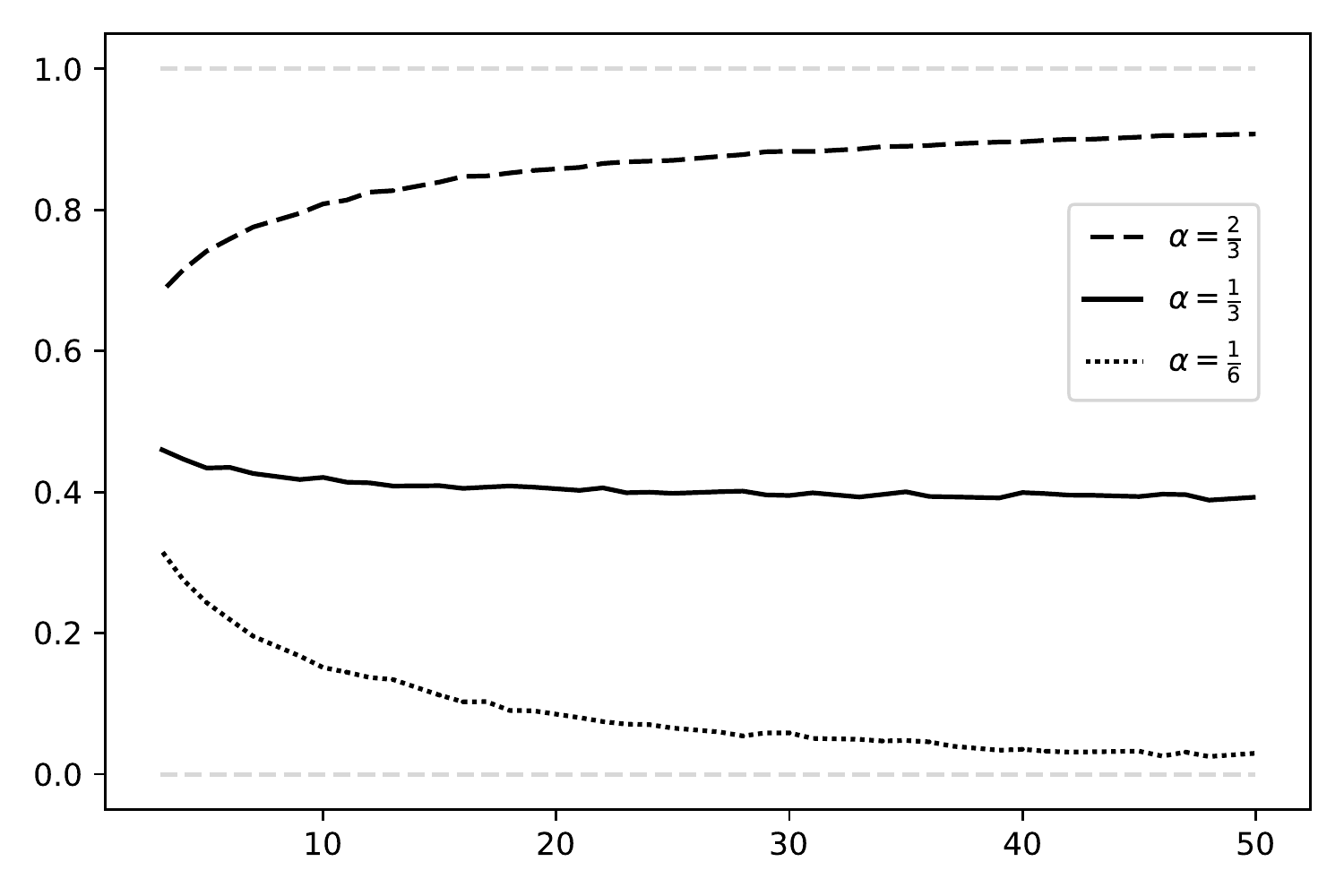}};
\node[below=of img1, node distance = 0, yshift = 1.2cm] {$d$};
\node[left=of img1, node distance = 0, rotate = 90, anchor = center, yshift = -1cm] {$a_d(\ell, r)$};
\node[right=of img1, node distance = 0, xshift = -0.7cm] (img2) {\includegraphics[width = 0.4\textwidth]{laplace_ar_mala.pdf}};
\node[below=of img2, node distance = 0, yshift = 1.2cm] {$\ell$};
\node[left=of img2, node distance = 0, rotate = 90, anchor = center, yshift = -1cm] {$a_d(\ell, r)$};
\node[below=of img1, node distance = 0, xshift = 3cm, yshift = 1cm] (img3) {\includegraphics[width = 0.4\textwidth]{laplace_esjd_mala.pdf}};
\node[below=of img3, node distance = 0, yshift = 1.2cm] {$a_d(\ell, r)$};
\node[left=of img3, node distance = 0, rotate = 90, anchor = center, yshift = -1cm] {$\textrm{ESJD}_d$};
\end{tikzpicture}
\caption{MY-MALA with Laplace target and $v=1, r=0$ (sG-MALA).  Average acceptance rate for different choices of $\alpha$ (first); acceptance rate as a function of $\ell$ for increasing dimension $d$ (second); $\textrm{ESJD}_d$ as a function of the acceptance rate $a_d(\ell, r)$ (third).}
\label{fig:laplace_mala}
\end{figure}

\begin{figure}
\centering
\begin{tikzpicture}[every node/.append style={font=\tiny}]
\node (img1) {\includegraphics[width = 0.4\textwidth]{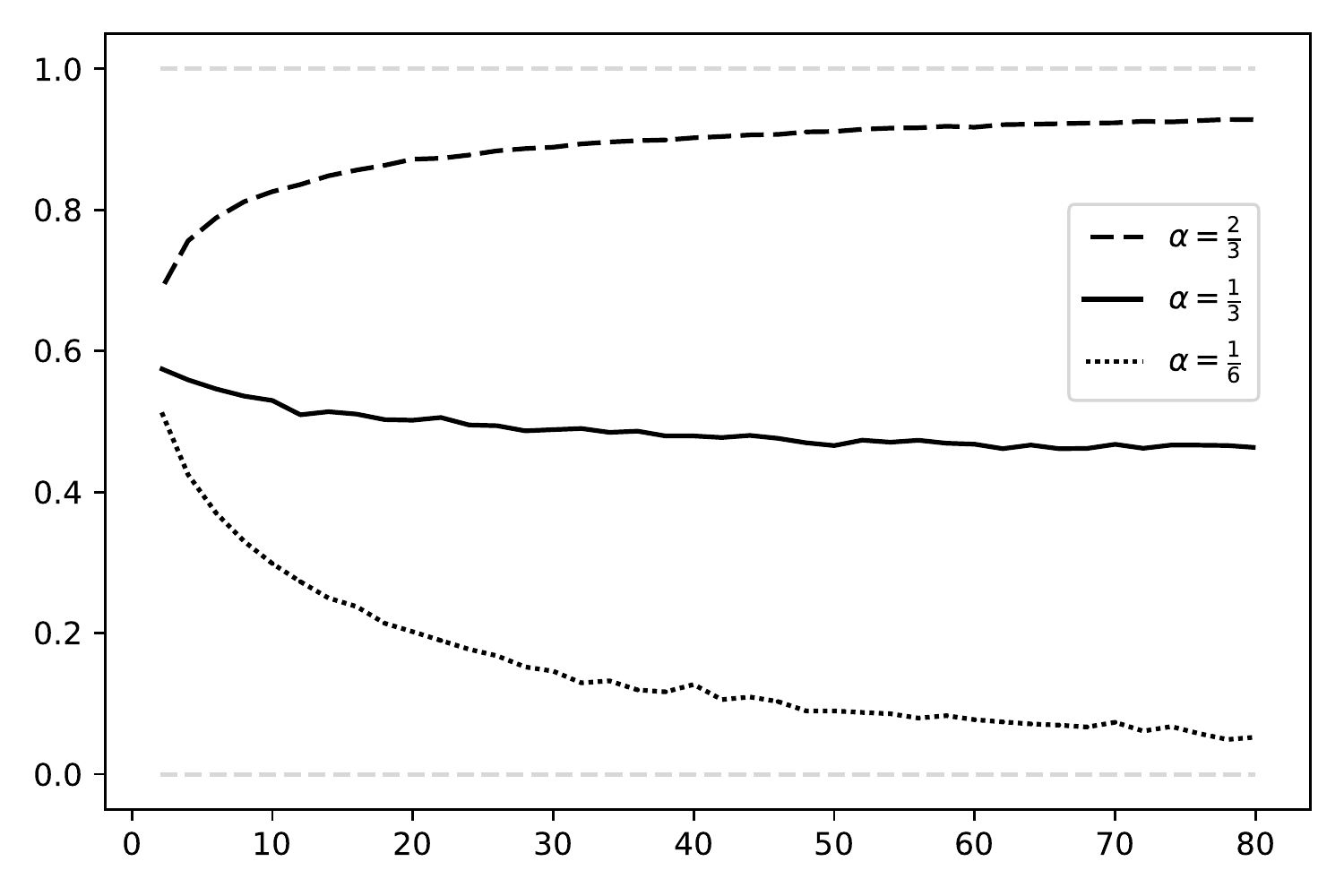}};
\node[below=of img1, node distance = 0, yshift = 1.2cm] {$d$};
\node[left=of img1, node distance = 0, rotate = 90, anchor = center, yshift = -1cm] {$a_d(\ell, r)$};
\node[right=of img1, node distance = 0, xshift = -0.7cm] (img2) {\includegraphics[width = 0.4\textwidth]{laplace_ar_pmala.pdf}};
\node[below=of img2, node distance = 0, yshift = 1.2cm] {$\ell$};
\node[left=of img2, node distance = 0, rotate = 90, anchor = center, yshift = -1cm] {$a_d(\ell, r)$};
\node[below=of img1, node distance = 0, xshift = 3cm, yshift = 1cm] (img3) {\includegraphics[width = 0.4\textwidth]{laplace_esjd_pmala.pdf}};
\node[below=of img3, node distance = 0, yshift = 1.2cm] {$a_d(\ell, r)$};
\node[left=of img3, node distance = 0, rotate = 90, anchor = center, yshift = -1cm] {$\textrm{ESJD}_d$};
\end{tikzpicture}
\caption{MY-MALA with Laplace target and $v=1, r=1$ (P-MALA).  Average acceptance rate for different choices of $\alpha$ (first); acceptance rate as a function of $\ell$ for increasing dimension $d$ (second); $\textrm{ESJD}_d$ as a function of the acceptance rate $a_d(\ell, r)$ (third).}
\label{fig:laplace_pmala}
\end{figure}

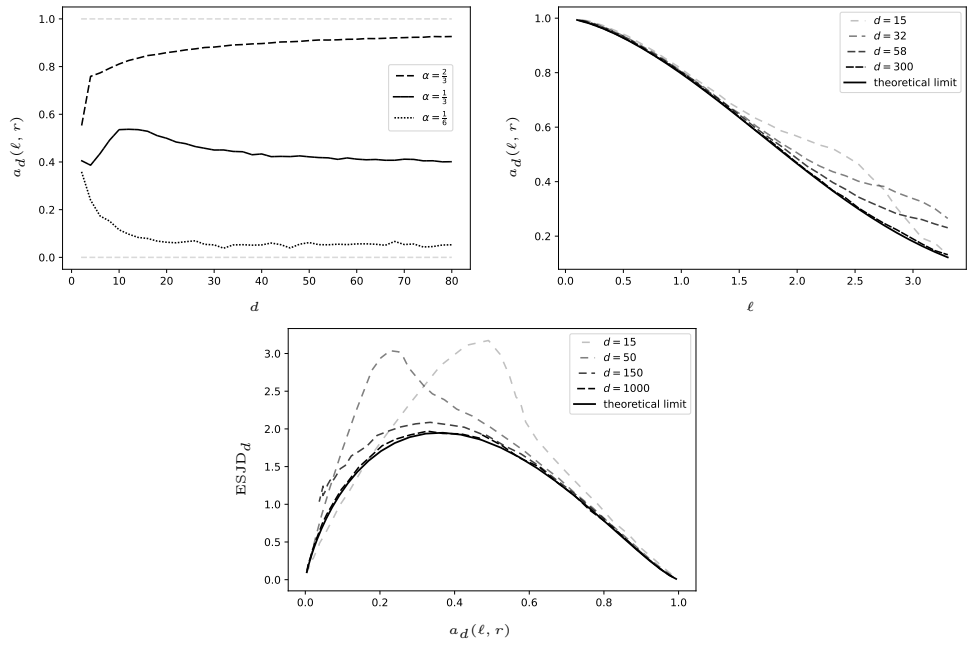
\begin{figure}
\centering
\begin{tikzpicture}[every node/.append style={font=\tiny}]
\node (img1) {\includegraphics[width = 0.4\textwidth]{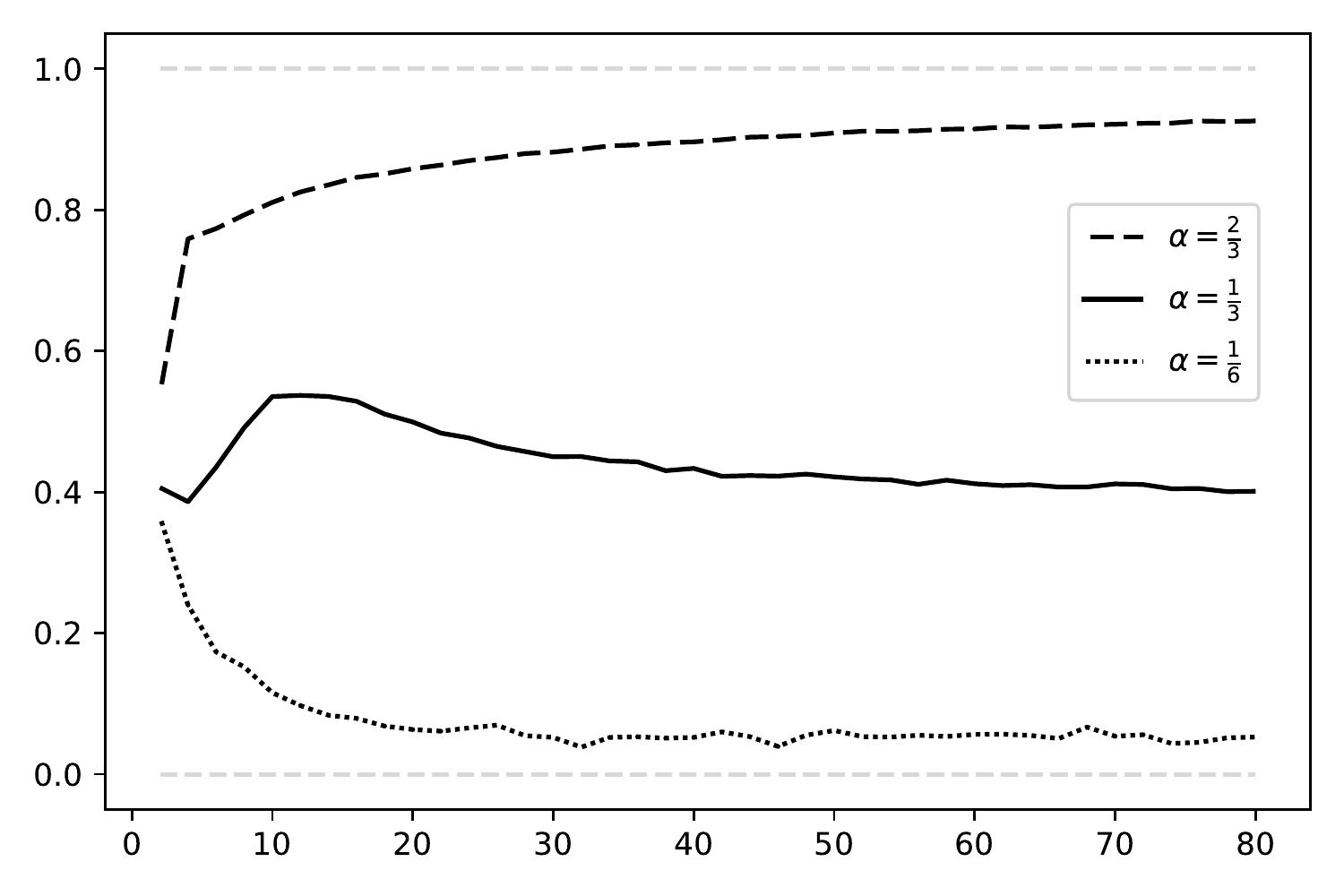}};
\node[below=of img1, node distance = 0, yshift = 1.2cm] {$d$};
\node[left=of img1, node distance = 0, rotate = 90, anchor = center, yshift = -1cm] {$a_d(\ell, r)$};
\node[right=of img1, node distance = 0, xshift = -0.7cm] (img2) {\includegraphics[width = 0.4\textwidth]{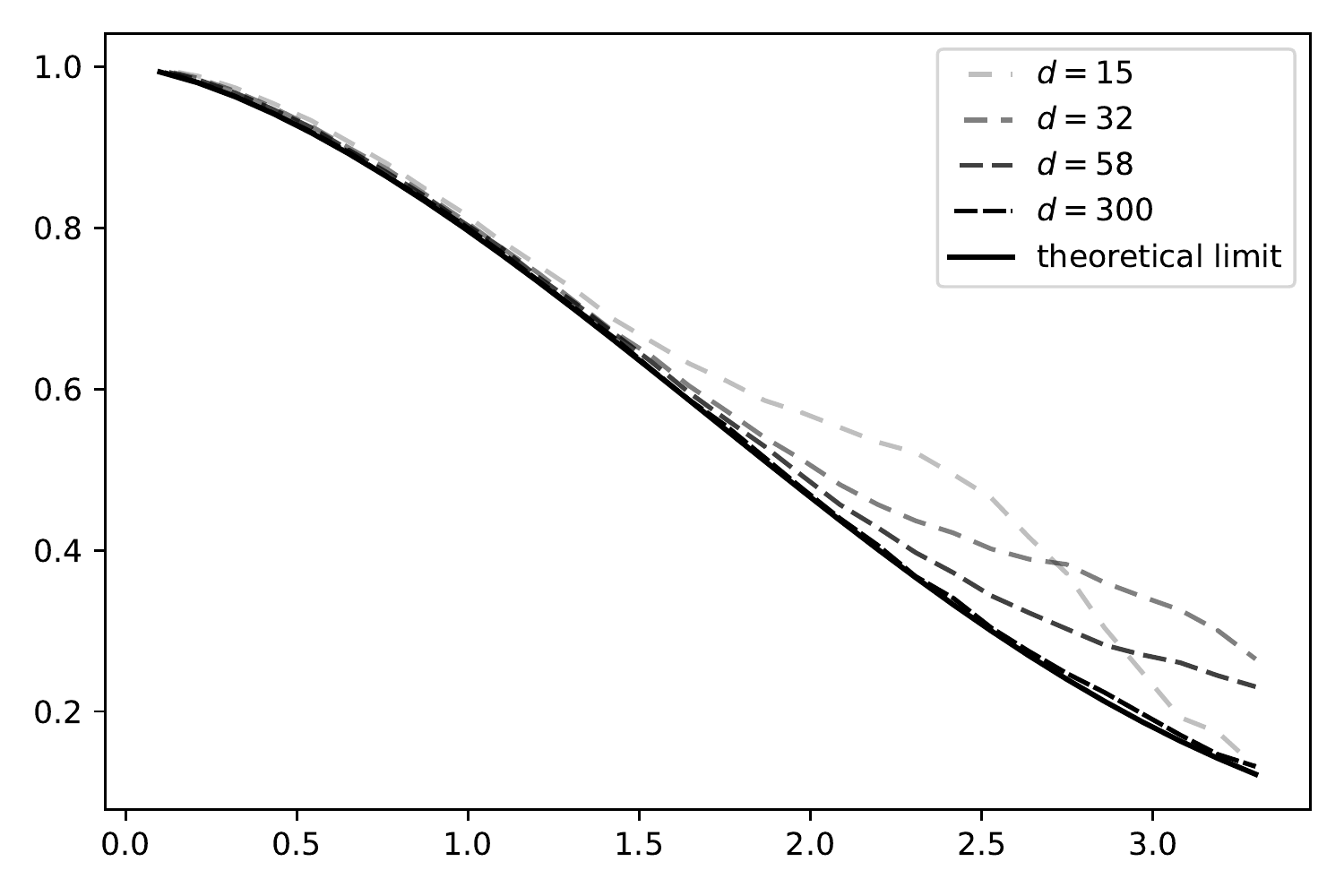}};
\node[below=of img2, node distance = 0, yshift = 1.2cm] {$\ell$};
\node[left=of img2, node distance = 0, rotate = 90, anchor = center, yshift = -1cm] {$a_d(\ell, r)$};
\node[below=of img1, node distance = 0, xshift = 3cm, yshift = 1cm] (img3) {\includegraphics[width = 0.4\textwidth]{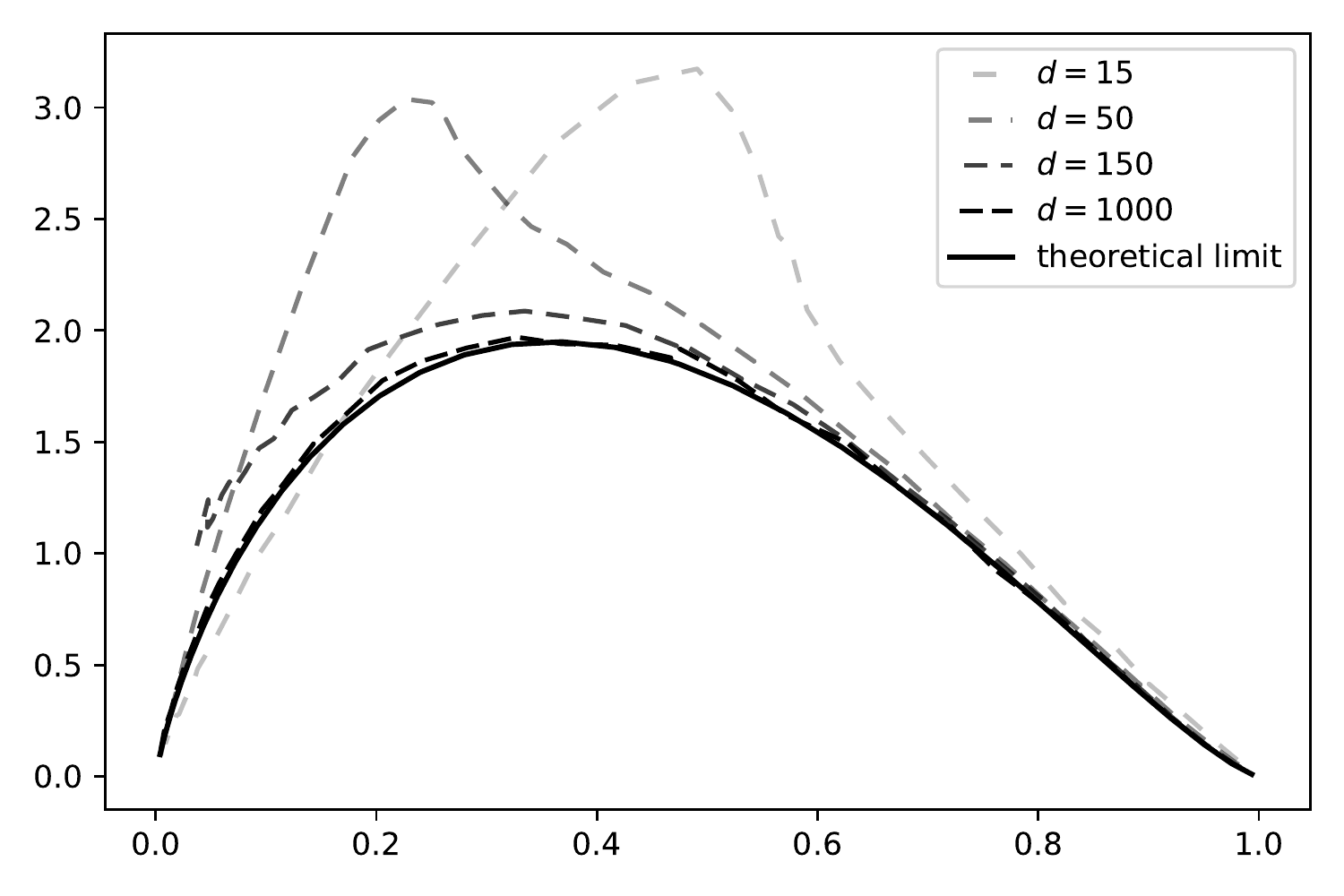}};
\node[below=of img3, node distance = 0, yshift = 1.2cm] {$a_d(\ell, r)$};
\node[left=of img3, node distance = 0, rotate = 90, anchor = center, yshift = -1cm] {$\textrm{ESJD}_d$};
\end{tikzpicture}
\caption{MY-MALA with Laplace target and $v=3, r=2$.  Average acceptance rate for different choices of $\alpha$ (first); acceptance rate as a function of $\ell$ for increasing dimension $d$ (second); $\textrm{ESJD}_d$ as a function of the acceptance rate $a_d(\ell, r)$ (third). }
\label{fig:laplace_general}
\end{figure}

\subsection{Mix of a Laplace and differentiable target}
\label{sec:mixed_numeric}

We collect here the rest of the numerical experiments illustrating the
scaling the results for the density defined in \eqref{eq:def_l2_l1_numerics}. Similarly to
\Cref{sec:laplace_numeric} we consider three algorithmic settings, summarized in
Table~\ref{tab:mixed}. All chains are initialized at stationarity drawing i.i.d samples from the target, i.e., $X_0^d\sim\pi_d$, so that no burnin effects are present.

To sample from the distribution $\pi$ in~\eqref{eq:def_l2_l1_numerics} we employ the following procedure.
Let $X$ be a random variable, distributed according to the density $\pi(x) \propto
e^{-\abs{x} - x^2/2}$. Since this distribution is symmetrical, $X$ has the same
distribution as $\varepsilon Z$, where $Z$ has the same distribution as
$\abs{X}$ and $\pr(\varepsilon = 1)= \pr(\varepsilon = -1)= 1/2$. The density of $Z$ is given by
$\pi_Z(z) \propto e^{-(z +1)^2/2}$ for any positive $z$. However, if $Y$ is a normal
distribution of mean
$-1$ and variance $1$, $Y$ conditionally to the event $(Y>0)$ admits, for all $y>0$, $y \mapsto
(\sqrt{2\pi} \pr(Y>0) )^{-1} e^{-(y +1)^2/2}$ as a
density.
Therefore, we sample $Z$ by sampling the normal variable $Y$ and rejecting whenever $Y<0$. In practice, when $Y<-2$, we don't reject by taking $-Y-2$ instead of $Y$.


\begin{table}[!ht]
\centering
\begin{tabular}{{llcccc}}
Figure & Algorithm & $\alpha$ & $\beta$ & $v$ &  $r$ \\
\hline\noalign{\smallskip}
\ref{fig:mixed_mala} & sG-MALA & 1/3 & 1/3 & 1 &$  0$\\
\ref{fig:mixed_pmala} & P-MALA & 1/3 & 1/3 & 1 &1\\
\ref{fig:mixed_general_main} & MY-MALA & 1/3 & 1 & 3 &2 \\
\end{tabular}
\caption{Algorithm setting for the simulation study on the mixed Laplace-normal target.}
\label{tab:mixed}
\end{table}

The first plot in Figures~\ref{fig:mixed_pmala}--\ref{fig:mixed_mala} shows
that for $\alpha\neq 1/3$ the acceptance ratio $a_d(\ell, r)$ becomes
degenerate as $d$ increases; while the second plot shows that $(a_d(\ell,
r))_{d \in \nsets}$ and $(\textrm{ESJD}_d)_{d\in\nsets}$ converge as $d\to \infty$.

It is interesting to note that the convergence of these plots happens faster
with respect to the dimension $d$, compared to the Laplace distribution. It is also
surprising that out of the three settings tested for the parameters $r$ and $v$, the
one that seems to converge faster, and that looks most stable, is $v=3$ and $r=2$, which 
was never the case for the other densities. This may suggest that for such distributions, 
a general MY-MALA algorithm has better properties of convergence than sG-MALA and P-MALA.

\begin{figure}
\centering
\begin{tikzpicture}[every node/.append style={font=\tiny}]
\node (img1) {\includegraphics[width = 0.4\textwidth]{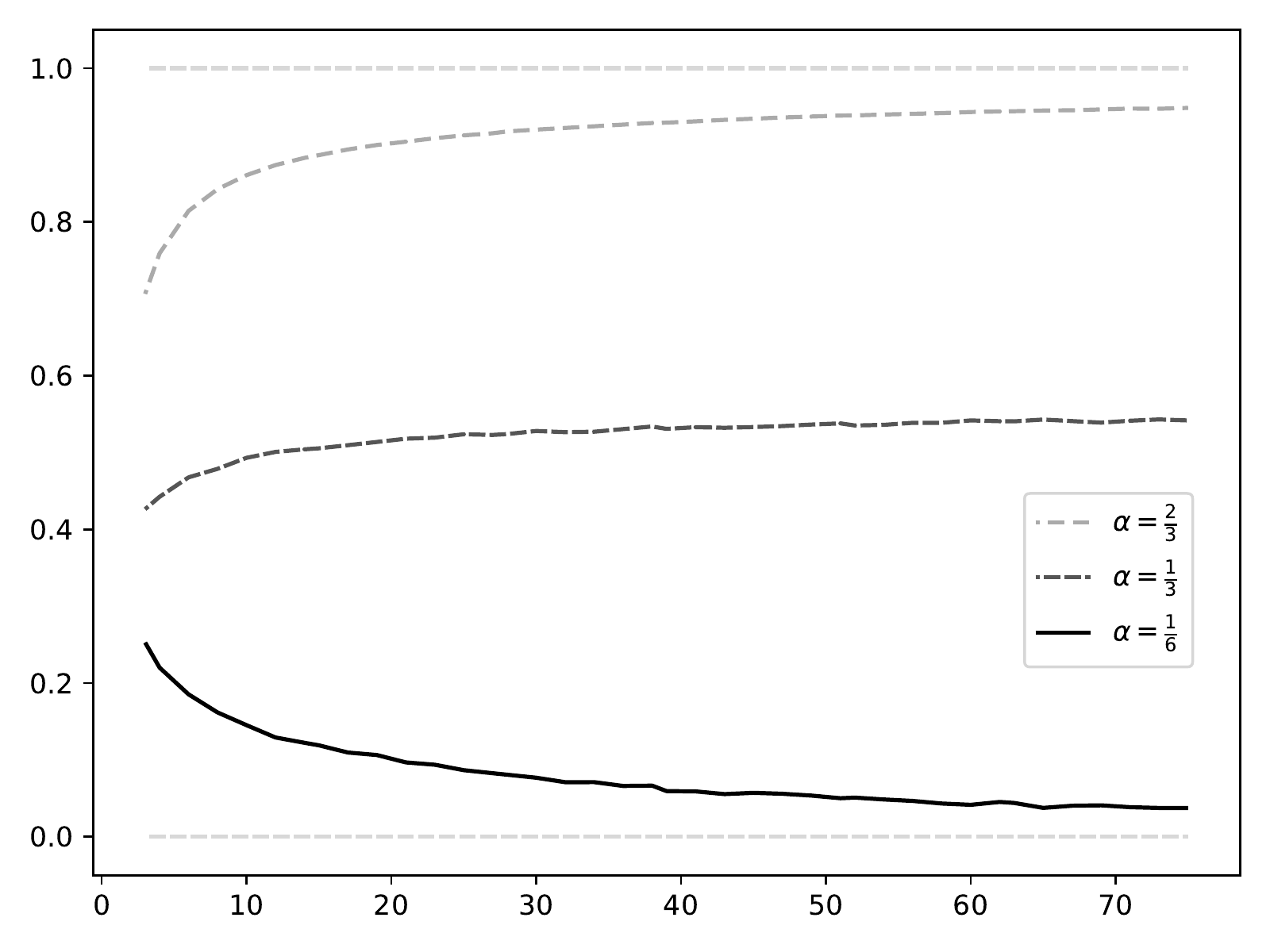}};
\node[below=of img1, node distance = 0, yshift = 1.2cm] {$d$};
\node[left=of img1, node distance = 0, rotate = 90, anchor = center, yshift = -1cm] {$a_d(\ell, r)$};
\node[right=of img1, node distance = 0, xshift = -0.7cm] (img2) {\includegraphics[width = 0.4\textwidth]{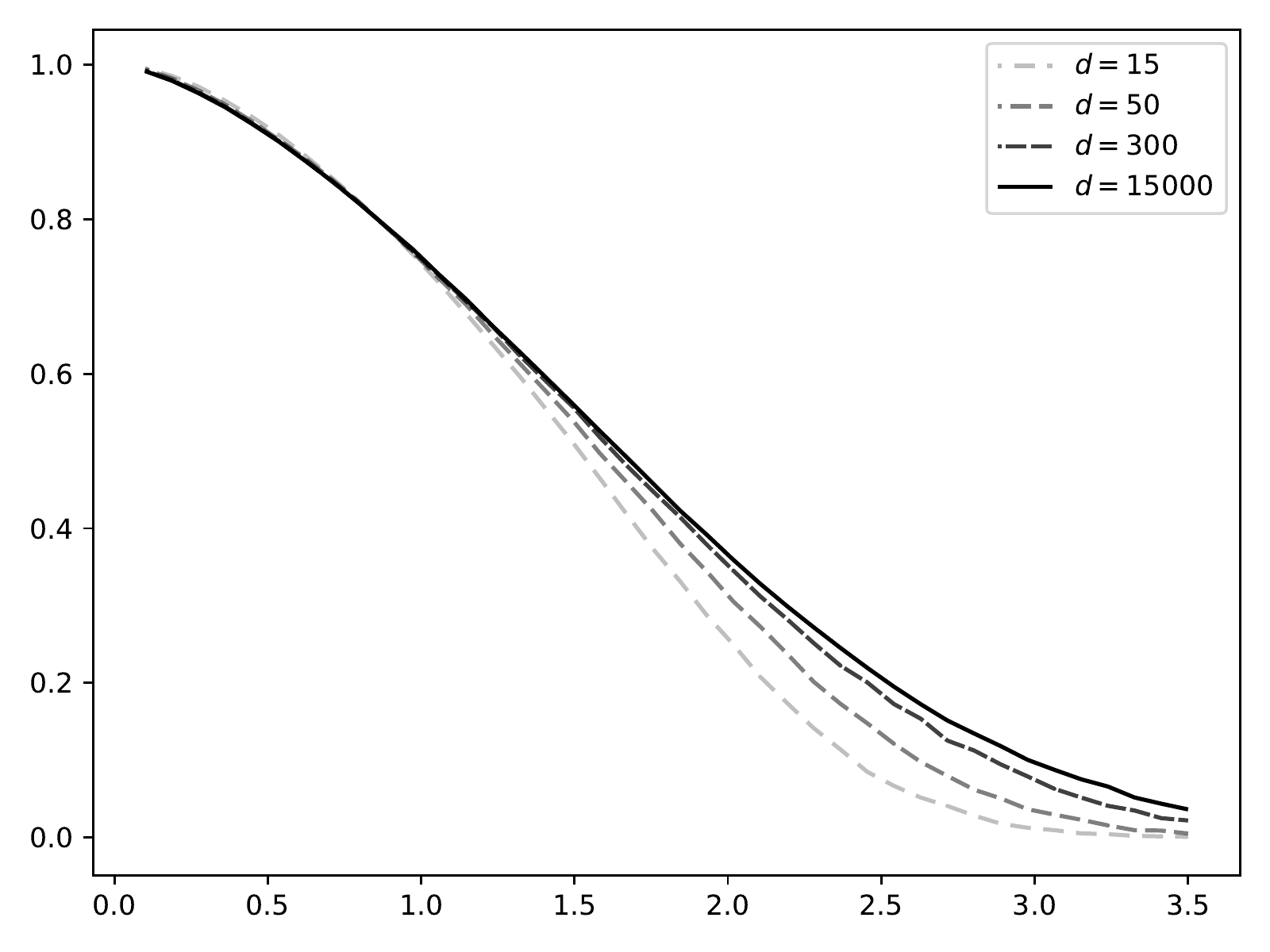}};
\node[below=of img2, node distance = 0, yshift = 1.2cm] {$\ell$};
\node[left=of img2, node distance = 0, rotate = 90, anchor = center, yshift = -1cm] {$a_d(\ell, r)$};
\node[below=of img1, node distance = 0, xshift = 3cm, yshift = 1cm] (img3) {\includegraphics[width = 0.4\textwidth]{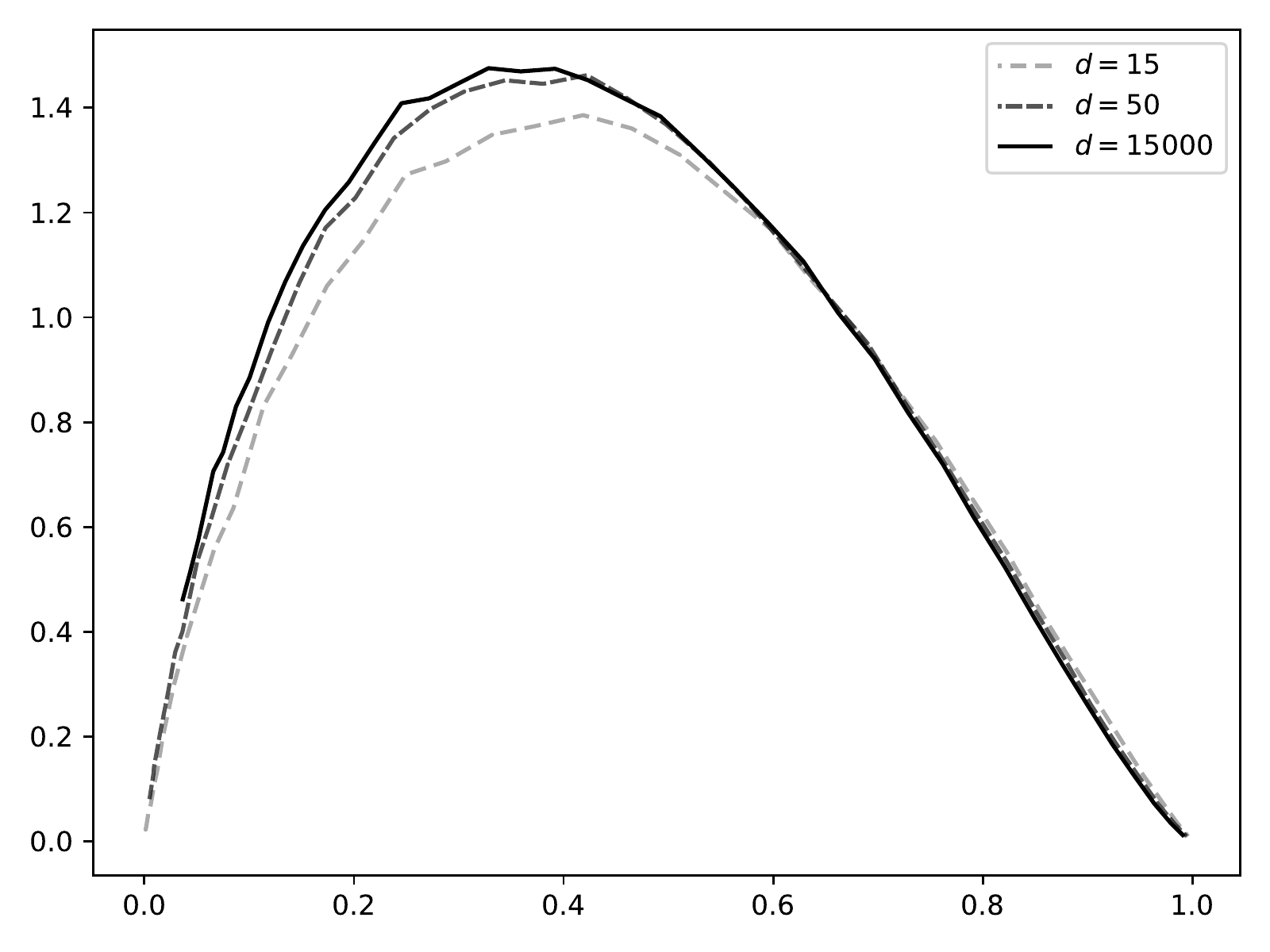}};
\node[below=of img3, node distance = 0, yshift = 1.2cm] {$a_d(\ell, r)$};
\node[left=of img3, node distance = 0, rotate = 90, anchor = center, yshift = -1cm] {$\textrm{ESJD}_d$};
\end{tikzpicture}
\caption{MY-MALA with mixed Laplace-normal target and $v=1, r=0$ (sG-MALA).  Average acceptance rate for different choices of $\alpha$ (first); acceptance rate as a function of $\ell$ for increasing dimension $d$ (second); $\textrm{ESJD}_d$ as a function of the acceptance rate $a_d(\ell, r)$ (third).}
\label{fig:mixed_mala}
\end{figure}
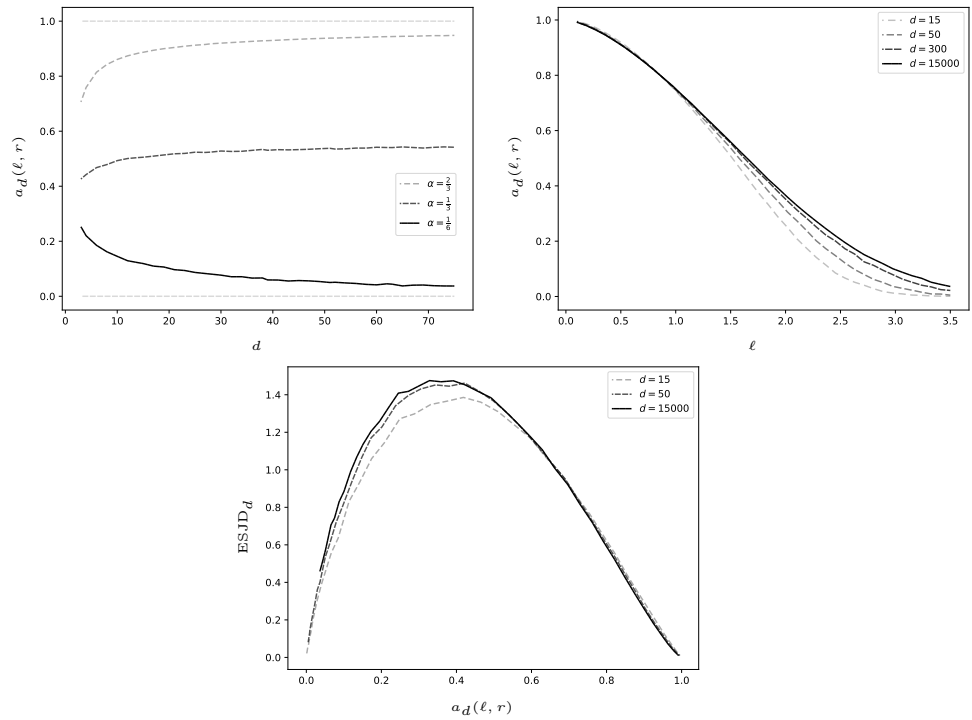

\begin{figure}
\centering
\begin{tikzpicture}[every node/.append style={font=\tiny}]
\node (img1) {\includegraphics[width = 0.4\textwidth]{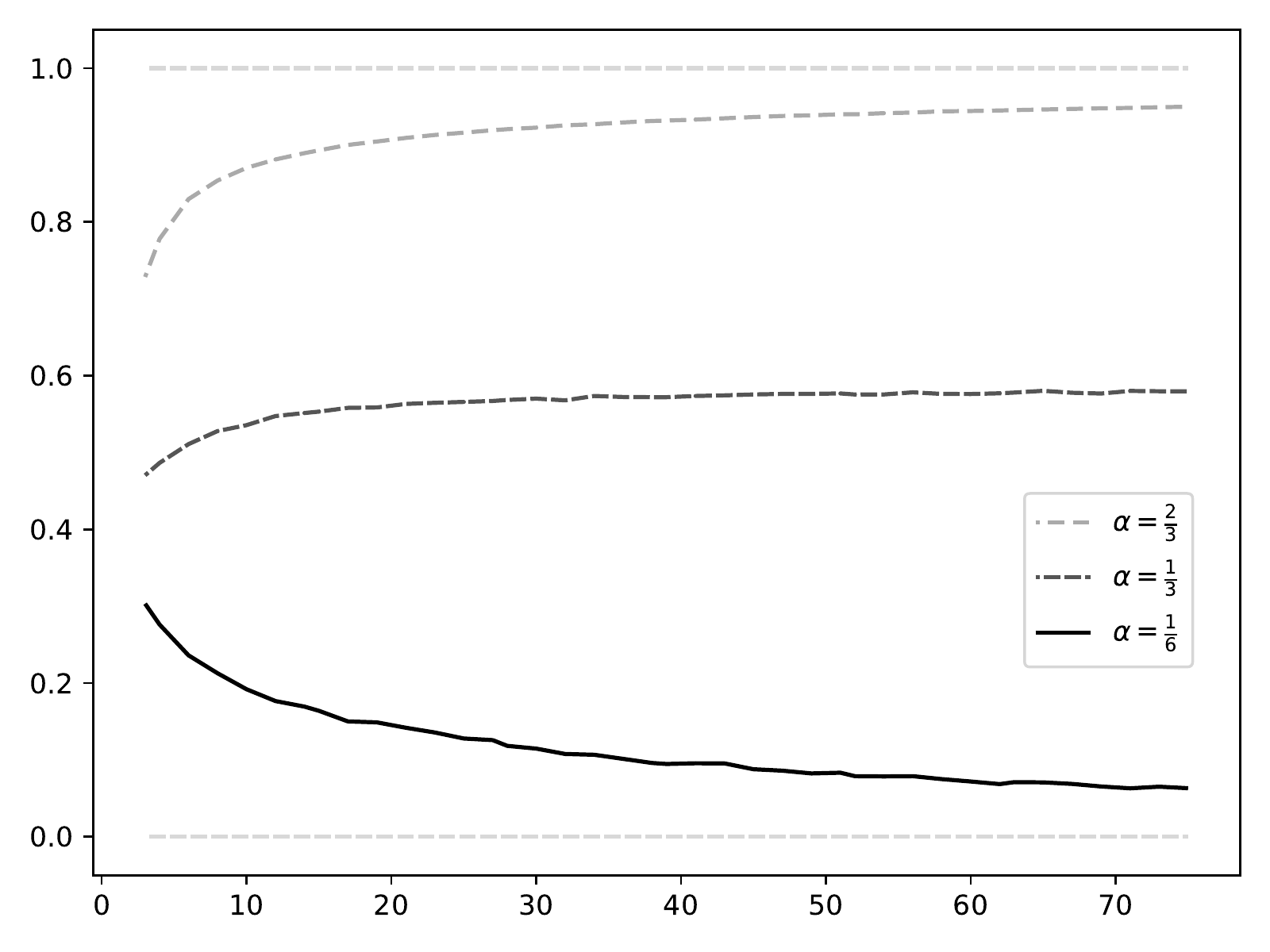}};
\node[below=of img1, node distance = 0, yshift = 1.2cm] {$d$};
\node[left=of img1, node distance = 0, rotate = 90, anchor = center, yshift = -1cm] {$a_d(\ell, r)$};
\node[right=of img1, node distance = 0, xshift = -0.7cm] (img2) {\includegraphics[width = 0.4\textwidth]{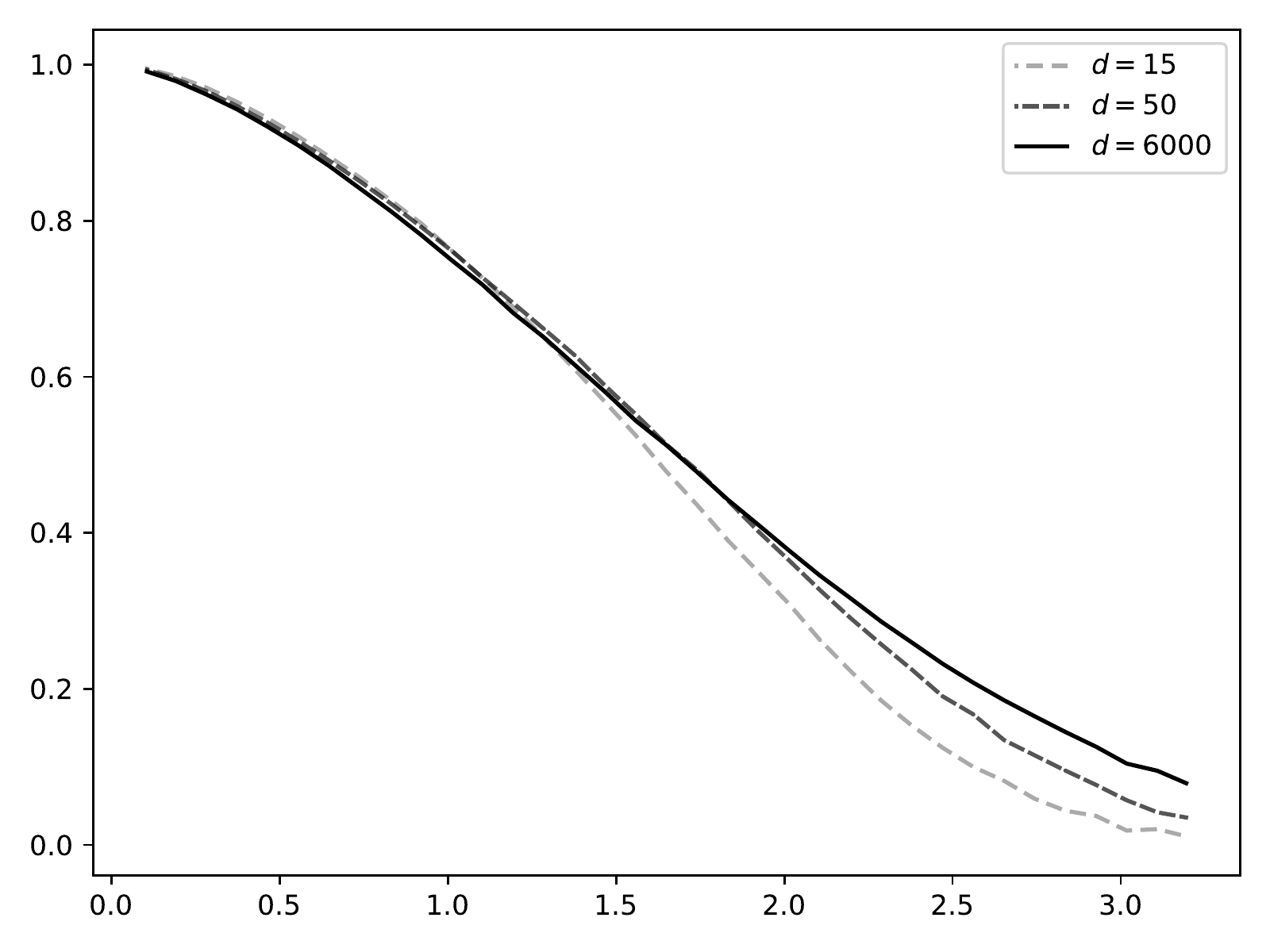}};
\node[below=of img2, node distance = 0, yshift = 1.2cm] {$\ell$};
\node[left=of img2, node distance = 0, rotate = 90, anchor = center, yshift = -1cm] {$a_d(\ell, r)$};
\node[below=of img1, node distance = 0, xshift = 3cm, yshift = 1cm] (img3) {\includegraphics[width = 0.4\textwidth]{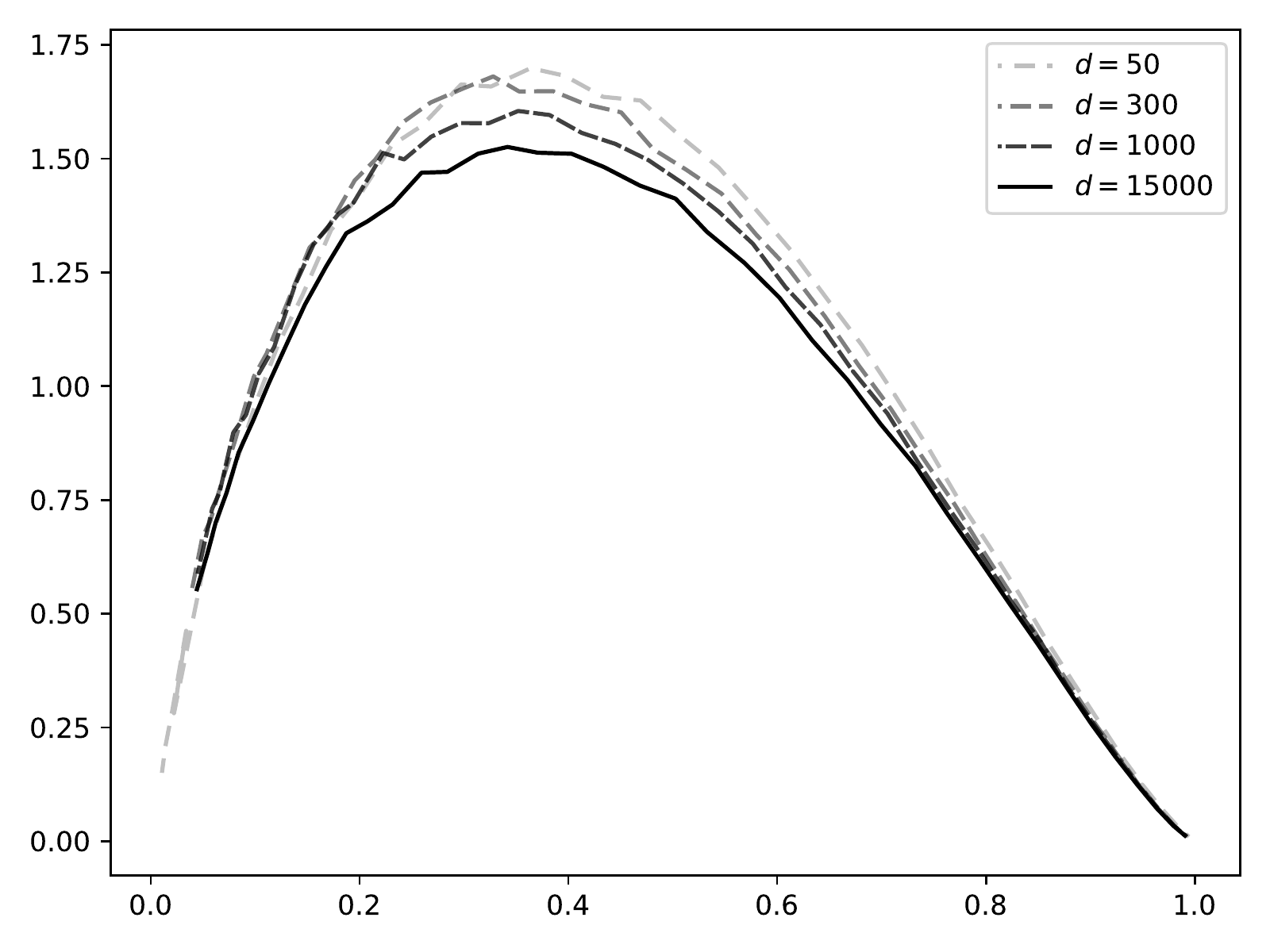}};
\node[below=of img3, node distance = 0, yshift = 1.2cm] {$a_d(\ell, r)$};
\node[left=of img3, node distance = 0, rotate = 90, anchor = center, yshift = -1cm] {$\textrm{ESJD}_d$};
\end{tikzpicture}
\caption{MY-MALA with mixed Laplace-normal target and $v=1, r=1$ (P-MALA).  Average acceptance rate for different choices of $\alpha$ (first); acceptance rate as a function of $\ell$ for increasing dimension $d$ (second); $\textrm{ESJD}_d$ as a function of the acceptance rate $a_d(\ell, r)$ (third).}
\label{fig:mixed_pmala}
\end{figure}





\begin{figure}
\centering
\begin{tikzpicture}[every node/.append style={font=\tiny}]
\node (img1) {\includegraphics[width = 0.4\textwidth]{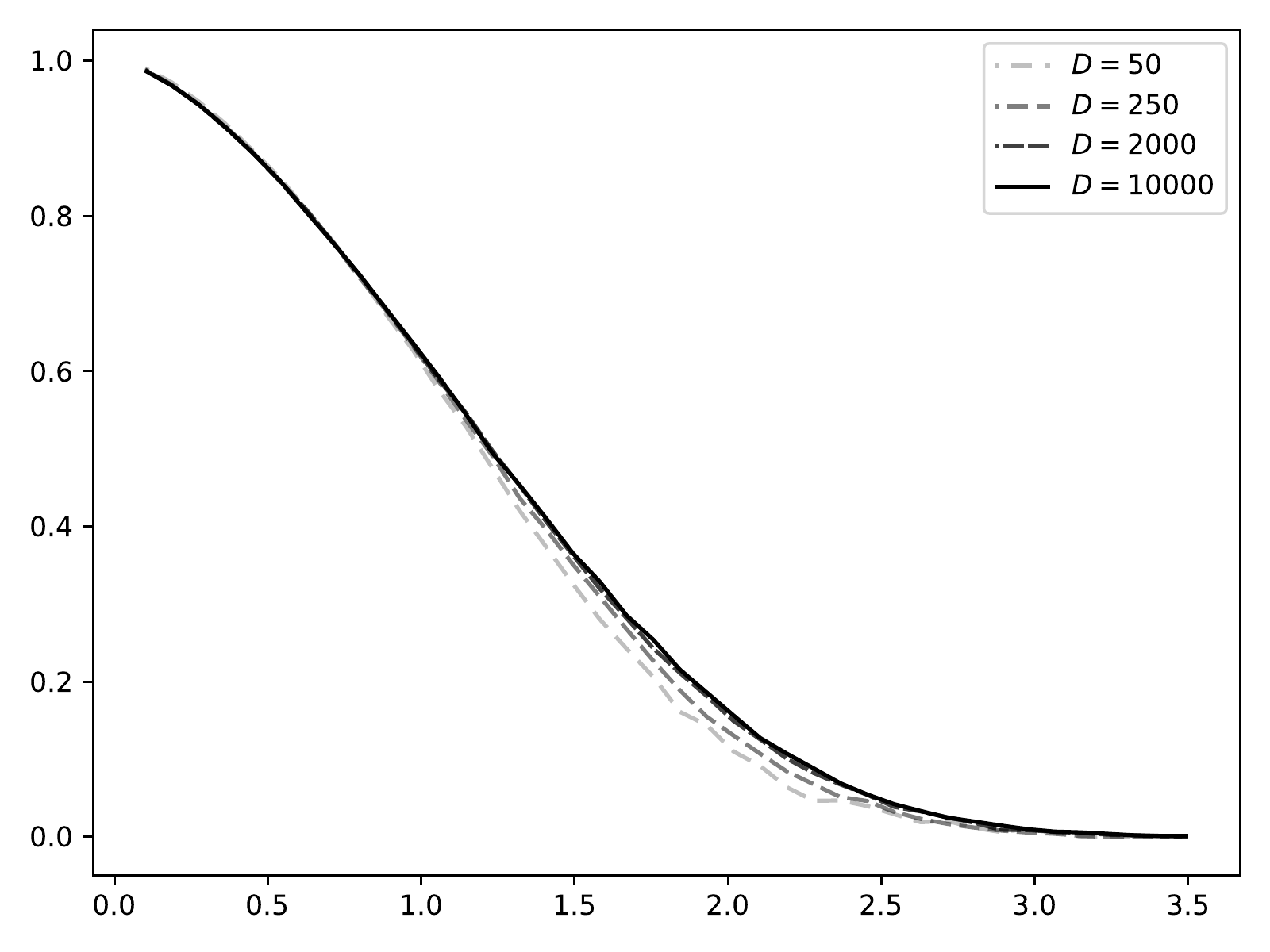}};
\node[below=of img1, node distance = 0, yshift = 1.2cm] {$\ell$};
\node[left=of img1, node distance = 0, rotate = 90, anchor = center, yshift = -1cm] {$a_d(\ell, r)$};
\node[right=of img1, node distance = 0, xshift = -0.7cm] (img3) {\includegraphics[width = 0.4\textwidth]{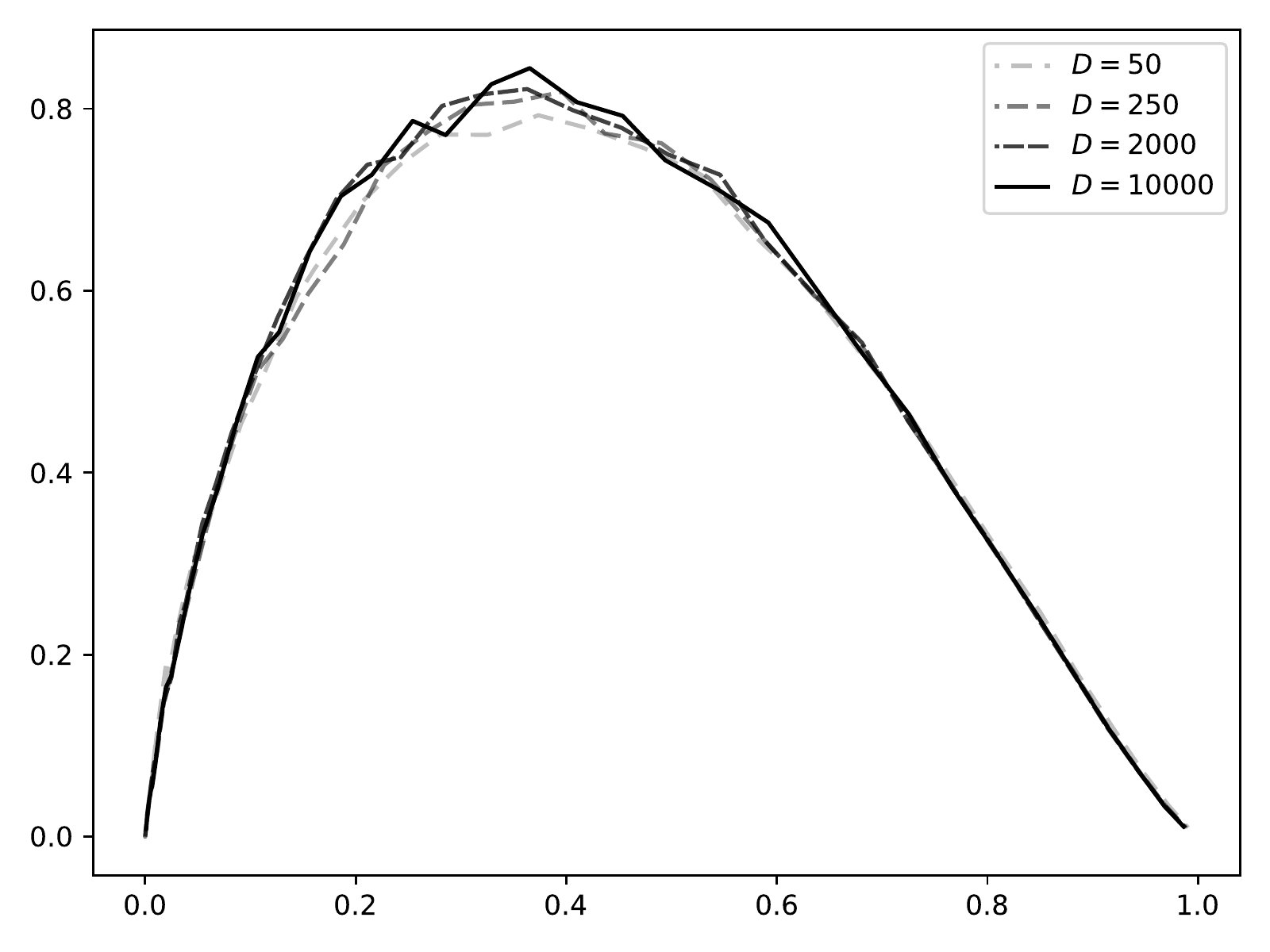}};
\node[below=of img3, node distance = 0, yshift = 1.2cm] {$a_d(\ell, r)$};
\node[left=of img3, node distance = 0, rotate = 90, anchor = center, yshift = -1cm] {$\textrm{ESJD}_d$};
\end{tikzpicture}
\caption{MY-MALA for the target \eqref{eq:def_tv_l2_l1_numerics} with $m=10$ and $v=3$, $r=2$. Left: acceptance rate as a function of $\ell$ for increasing dimension $d$; Right: $\textrm{ESJD}_d$ as a function of the acceptance rate $a_d(\ell, r)$.}
\label{fig:mixedtv_general}
\end{figure}

\end{document}